\documentclass[11pt]{article}
\usepackage[english]{babel}
\usepackage{amsmath,amsfonts,amssymb,amsthm} 
\usepackage{mathtools}
\usepackage{thmtools}
\usepackage{thm-restate}

{\bfseries \upshape}{\itshape}
{\bfseries \upshape}{\itshape}

\newcommand{\mc}{\mathcal}

\usepackage{etex}
\def\tl{\mathop{tl}}
\def\tw{\mathop{tw}}
\newlength{\defbaselineskip}
\usepackage{algorithm}
\usepackage{algpseudocode}
\usepackage{framed}
\usepackage{amssymb}
\usepackage{amsfonts}
\usepackage{amsmath}
\usepackage{graphicx}
\usepackage{caption}
\usepackage[font=footnotesize]{subcaption}
\usepackage{url}
\usepackage{rotating}
\usepackage{multirow}
\usepackage{color}
\usepackage{xcolor}
\usepackage{enumitem}

\definecolor{darkgreen}{rgb}{0,0.6,0}

 \newtheorem{definition}{Definition}
 
 \newtheorem{theorem}{Theorem}
 
 \newtheorem{proposition}{Proposition}

\begin{document}

\title{%
Tree decompositions and social graphs
}

\author{
Aaron B. Adcock
\thanks{
Department of Electrical Engineering,
Stanford University,
Stanford, CA 94305.
Email: aadcock@stanford.edu
}
\and
Blair D. Sullivan
\thanks{
Department of Computer Science,
North Carolina State University,
Raleigh, NC 27695.
Email: blair\_sullivan@ncsu.edu
}
\and
Michael W. Mahoney
\thanks{
International Computer Science Institute 
and Department of Statistics,
University of California at Berkeley,
Berkeley, CA 94720.
Email:  mmahoney@stat.berkeley.edu
}
}

\date{}
\maketitle

\begin{abstract}
\noindent
Recent work has established that large informatics graphs such as  
social and information networks have non-trivial tree-like structure when 
viewed at moderate size scales.
Here, we present results from the first detailed empirical evaluation of the 
use of tree decomposition (TD) heuristics for structure identification and 
extraction in social graphs. 
Although TDs have historically been used in structural graph theory and 
scientific computing, we show that---even with existing TD heuristics 
developed for those very different areas---TD methods can identify 
interesting structure in a wide range of realistic informatics graphs.  
Our main contributions are the following: we show that TD methods can identify structures that 
correlate strongly with the core-periphery structure of realistic networks, 
even when using simple greedy heuristics; we show that the peripheral bags 
of these TDs correlate well with low-conductance communities (when they 
exist) found using local spectral computations; and we show that several 
types of large-scale ``ground-truth'' communities, defined by demographic
metadata on the nodes of the network, are well-localized in the large-scale 
and/or peripheral structures of the TDs.  
Our other main contributions are the following: we provide detailed empirical results for TD heuristics on toy and synthetic networks to establish a baseline to understand better the behavior of the heuristics on more complex real-world networks; and we prove a theorem providing formal justification for the intuition that the only two impediments to low-distortion hyperbolic embedding are high tree-width and long geodesic cycles.
Our results suggest future directions for improved TD heuristics that are more appropriate for realistic social graphs.
\end{abstract}

\section{Introduction}
\label{sec:intro}

Understanding the properties of realistic informatics graphs such as large 
social and information networks and developing algorithmic and statistical 
tools to analyze such graphs is of continuing interest, and recent work has 
focused on identifying and exploiting what may be termed tree-like 
structure in these real-world graphs.
Since an undirected graph is a tree if any two vertices are connected by 
exactly one path, or equivalently if the graph is connected but has no cycles, 
real-world graphs are clearly not trees in any na\"{\i}ve sense of the 
word.
For example, realistic social graphs have non-zero clustering coefficient, 
indicating an abundance of cycles of length three.
There are, however, more sophisticated notions that can be used to 
characterize the manner in which a graph may be viewed as tree-like.
These are of interest since, e.g., graphs that are trees have many nice 
algorithmic and statistical properties, and the hope is that graphs that are 
tree-like inherit some of these nice properties.
In particular, $\delta$-hyperbolicity is a notion from geometric group 
theory that quantifies a way in which a graph is tree-like in terms of its 
distance or metric properties. 
Alternatively, tree decompositions (TDs) are tools from structural graph 
theory that quantify a way in which a graph is tree-like in terms of its 
cut or partitioning properties. 

Although TDs and $\delta$-hyperbolicity capture very different ways in 
which a general graph can be tree-like, recent empirical work (described 
in more detail below) has shown interesting connections between them.
In particular, for realistic social and information networks, these two 
notions of tree-likeness capture very similar structural properties, at 
least when the graphs are viewed at large size scales; and this structure 
is closely related to what may be termed the nested core-periphery or 
$k$-core structure of these networks. 
Recent work has also shown that computing $\delta$-hyperbolicity exactly 
is extremely expensive, that hyperbolicity is quite brittle and difficult 
to work with for realistic social graphs, and that common methods to 
approximate $\delta$ provide only a very rough guide to its extremal value and 
associated graph properties.
Motivated by this, as well as by the large body of work in linear algebra 
and scientific computing on practical methods for computing TDs, in this 
paper we present results from the first detailed empirical evaluation of 
the use of TD heuristics for structure identification and extraction in 
social graphs. 

A TD (defined more precisely below) is a specialized mapping of an arbitrary input graph $G$ into a tree $H$, where the nodes of $H$ (called bags) consist of overlapping subsets of vertices of $G$. 
Quantities such as the treewidth---the size of the bag in $H$ that contains 
the largest number of vertices from $G$---can be used to characterize how 
tree-like is $G$.
A single bag that contains every vertex from $G$ is a legitimate but trivial TD (since the width is as large as possible for a graph of the given size). 
Thus, one usually focuses on finding ``better'' TDs, where  better typically means 
minimizing the width. 
The problem of finding the treewidth of $G$ and of finding
an optimal TD of $G$ are both NP-hard, and thus most effort has focused on
developing heuristics, e.g., by constructing the TD iteratively by 
choosing greedily vertices of $G$ that minimize quantities such as the degree 
or fill.
Since we are interested in applying TDs on realistic graphs, it is these 
heuristics (to be described in more detail below) that we will use in this 
paper.

Our goals are to describe the behavior of TD heuristics on real-world and synthetic social graphs and to use these TD tools to identify and extract useful structure from these graphs.
In particular, (in Section \ref{sec:real_results}) we show the following.
We first show (in Section~\ref{sec:real_results-core_periphery}) that TD 
methods can identify large-scale structures in realistic networks that 
correlate strongly with the recently-described core-periphery structure 
of these networks, even when using simple greedy TD heuristics.
We do this by relating the global ``core-periphery'' structure of these 
networks, as captured using the $k$-core decomposition, with what we call 
the ``central-perimeter'' structure, which is a measure of the centrality 
or eccentricity of each bag in the TD. 
We also describe how small-scale structures such as the internal bag 
structure of TDs of these networks reflects---depending on the density 
and other properties of these networks---their clustering coefficient and 
other related clustering properties of the original networks.
We next show (in Section~\ref{sec:real_results-ncp}) that the peripheral 
bags of these TDs correlate well with low-conductance communities/clusters 
(when they exist) found using local spectral computations, in the sense 
that these low-conductance (i.e., good-conductance) communities/clusters 
occupy a small number of peripheral bags in the TDs. 
In particular, this shows that in graphs for which the so-called Network
Community Profile (NCP) Plot is upward-sloping (as, e.g., described
in~\cite{LLDM09_communities_IM}, and indicating the presence of good small 
and absence of good large clusters), the small-scale ``dips'' in the NCP are 
localized in clusters that are on a peripheral branch in the TD. 
We finally consider (in Section~\ref{sec:real_results-gt}) how several types 
of large-scale ``ground-truth'' communities/clusters, as defined by 
demographic metadata on the nodes of the network (and that are not 
good-conductance clusters), are localized in the TDs.
In particular, we look at two social network graphs consisting of friendship 
edges between students at a university, we use metadata associated with 
graduation year and residence hall information, and we show that clusters 
defined by these metadata are well-localized in the large-scale central 
and/or small-scale peripheral structures of the TDs.  

A significant challenge in applying existing TD heuristics---which 
have been developed for very different applications in scientific computing 
and numerical linear algebra---is that it can be difficult to determine 
whether one is observing a ``real'' property of the networks or an artifact 
of the particular TD heuristic that was used to examine the network.
Thus, to establish a baseline and to determine their behavior in idealized settings, we have first applied several existing TD heuristics to a range of toy and synthetic data.
(See Section~\ref{sec:toy_results} and Section~\ref{sec:synth_results}, 
respectively.) 
The toy data consist of a binary tree, a lattice, a cycle, a clique, and a dense random graph, i.e., graphs for which optimal TDs are known.
The synthetic data consists of Erd\H{o}s-R\'{e}nyi and power law random graph 
models, which help us understand the effect of noise/randomness on the TDs.
(Other random graph models exhibit similar properties, when their parameters are set to correspondingly sparse values.)
For these graphs, we place a particular emphasis on the properties of the 
TDs as the density parameters (i.e., the connection probability for the 
Erd\H{o}s-R\'{e}nyi graphs and the power law parameter for the power law 
graphs) are varied from very sparse to extremely sparse, and we are 
interested in how this relates to the large-scale core-periphery structure.

Our detailed empirical results for TD heuristics on toy and synthetic networks are important for understanding the behavior of these 
heuristics on more complex real-world networks; but our results 
on synthetic and real-world networks also suggest future directions for the 
development of TD heuristics that are more appropriate for social graph data.
Existing TD heuristics focus on producing minimum-width TDs, which 
are of interest in more traditional graph theory and linear algebra 
applications, but they are not well-optimized for finding structures of 
interest in social graphs.
Although it is beyond the scope of this paper, the development of TD 
heuristics that are more appropriate for social graph applications (e.g., 
understanding how the bag structure of those TDs relates to the output of 
recently-developed local spectral methods that find good small clusters in large
networks) is an important question raised by our results.

The remainder of this paper is organized as follows.  
In Section \ref{sec:background}, we present definitions from graph theory, a detailed discussion of tree decompositions and the algorithms for their construction, and a brief discussion of other prior related work.
Section \ref{sec:data} details the datasets we make use of throughout the paper.
The subsequent four sections provide our main empirical results.
In particular, in Section~\ref{sec:toy_results}, we consider several TD
heuristics applied to toy graphs; and in 
Section \ref{sec:synth_results}, we consider TD heuristics applied to 
synthetic random graphs.
Then, in Section \ref{sec:real_results}, 
we describe the results of applying TDs to a carefully-chosen suite of real-world social graphs.
In Section~\ref{sec:theory_new}, we prove a theoretical result 
connecting treewidth and treelength with the (very different) notion of 
$\delta$-hyperbolicity,
under an assumption on the length of the longest 
geodesic cycle in the graph.
Finally, in Section~\ref{sec:conc}, 
we provide a brief discussion of results and conclusion.

\section{Background and Related Work}
\label{sec:background}

In this section, we will review relevant graph theory, TD ideas, and computational methods, as well as relevant related work.

\subsection{Preliminaries on Graph Theory}
\label{sec:background-prelim}

Let $G=(V,E)$ be a \emph{graph} with vertex set $V$ and edge set 
$E \subseteq V \times V$.  
We often refer to graphs as \emph{networks} and vertices as {\em nodes}, and we will model social and information networks by undirected graphs.
We note that TDs are themselves graphs (constructed from other input graphs).  
The \emph{degree} of a vertex $v$, denoted $d(v)$, is defined as the number 
of vertices that are adjacent to $v$ (or the sum of the weights of adjacent
edges, if the graph is weighted).
The average degree is denoted $\bar{d}$.  
A graph is called \emph{connected} if there exists a path between any two
vertices.  
A graph is called a \emph{tree} if it is connected and has no cycles.  
A vertex in a tree is called a \emph{leaf} if it has degree 1.  
A graph $H = (S,F)$ is a subgraph of $G$ if $S \subseteq V, F \subseteq E$.  
An \emph{induced subgraph} of $G$ on a set of vertices $S \subseteq V$ is 
the graph $G[S] := (S, S \times S \cap E)$.  
Unless otherwise specified, our analyses will always consider the 
\emph{giant component}, i.e., the largest connected subgraph of $G$.

The \emph{diameter} of a graph is the maximum distance between any two 
vertices, and the \emph{eccentricity} of a vertex is the maximum distance 
between that vertex and all other vertices in the graph.  Note that the 
maximum eccentricity of a graph is equal to the diameter.  
The \emph{clustering coefficient} of a vertex is the ratio of the number of 
edges present among its neighbors to the maximum possible number of such edges; 
when we refer to the clustering coefficient of a network, we use the average of the 
clustering coefficient of all its vertices. 
A \emph{cut} is a partitioning of a network's vertex set into two pieces.  
The \emph{volume} of a cut is the sum over vertices in the smaller piece
of the number of incident edges, and the \emph{surface area} 
of a cut is the number of edges with one end-point in each piece.  In this 
case, the \emph{conductance} of a cut---one of the most important measures 
for assessing the quality of a cut---is the surface area divided by the 
volume (that is, we will be following the conventions used in previous 
work~\cite{LLDM09_communities_IM,Jeub15}).

Finally, we will refer to the ``core-periphery'' 
structure of a network.  
Following prior work~\cite{LLDM09_communities_IM,Adcock13_icdm,Jeub15}, 
we use the \emph{$k$-core decomposition} to identify these core nodes.  
The $k$-core of a network $G$ is the maximal induced subgraph 
$H \subseteq G$ such that every node in $H$ has degree at least $k$.  
The $k$-core has the advantage of being easily computable in 
$O(V + E)$ time \cite{Batagelj02,Batagelj03,BZ11}.

\subsection{Preliminaries on Tree Decompositions}
\label{sec:tdprelim}

TDs are combinatorial objects that describe 
specialized mappings of cuts in a network to nodes of a tree. 
Although originally introduced in the context of structural graph 
theory (the proof of the Graph Minors 
Theorem~\cite{Robertson86}), TDs have gained attention in the 
broader community due to their use in efficient algorithms for certain 
NP-hard problems. 
In particular, there are polynomial-time algorithms for solving many such 
problems on all graphs that have TDs whose width (defined below) is bounded 
from above by a constant~\cite{Arnborg89,Bern87}. 
These algorithms have been applied to problems in constraint satisfaction, 
computational biology, linear algebra, probabilistic 
networks, and machine learning~\cite{KosterHK02, Lagergren96, Hicks05, ZhaoMC06, ZhaoCC07, Liu06, Lauritzen88, KS01, Chen04}.  

\begin{definition}
\label{def:td}
A \emph{tree decomposition (TD)} of a graph $G=(V,E)$ is a pair 
$$ \left ( X = \{X_i : i \in I \}, T = (I, F) \right ) ,$$ with each 
$X_i \subseteq V$, and $T$ a tree with the following properties: 
\begin{enumerate}[noitemsep,nolistsep]
\item
$\cup_{i \in I} X_i = V$, 
\item
For all $(v,w) \in E, \exists i \in I$ with $v,w \in X_i$, and
\item
For all $v \in V$, $\{ i\in I : v \in X_i\}$ forms a connected subtree of T. 
\end{enumerate}
The $X_i$ are called the \emph{bags} of the tree decomposition. 
\end{definition}

\noindent
The third condition of the definition is a continuity requirement that 
allows the TD to be used in dynamic programming algorithms for many NP-hard 
problems.%
\footnote{Alternatively, the bags and edges of the TD form separators 
(cuts) in the graph.  The set of vertices contained in any bag, or 
intersection of two adjacent bags, form a separator in $G$.  This 
structural property is important as it allows TDs to be thought of as 
a method of organizing cuts in a network.  This is also related to how 
the treewidth of a network is used to measure how tree-like a network 
is.  Intuitively, a tree has a treewidth of 1 because the graph can be 
separated by the removal of a single edge (or vertex) in the network, 
whereas a cycle requires two edges to be cut and thus has a treewidth 
of 2.  TDs with large widths require larger numbers of vertices to separate 
a network into two disconnected pieces.}
It is equivalent to requiring that for all $i,j,k \in I$, if $j$ is on the 
path from $i$ to $k$ in $T$, then $X_i \cap X_k \subseteq X_j$.%
\footnote{A related aspect of the definition of a TD 
is the overlapping nature of the bags of a TD. Vertices in the 
graph will appear in many bags in the TD.  
This is particularly true of high degree or high $k$-core nodes~\cite{Adcock13_icdm}.} 
The quality of a tree decomposition is often measured in 
terms of its largest bag size. 

\begin{definition} 
Let $\mathcal{T} = \left (\{X_i\}, T=(I,F) \right )$  be a tree 
decomposition of a graph $G$. 
The \emph{width} of $\mathcal{T}$ is defined to be 
$\max_{i \in I} |X_i| - 1$, and the \emph{treewidth} of $G$, denoted 
$tw(G)$, is the minimum width over all valid tree decompositions of $G$. 
A tree decomposition whose width is equal to the treewidth is often referred to as \emph{optimal}.
\end{definition}

\noindent
By this definition, trees have the minimum possible treewidth of $1$ (their 
bags contain the edges of the original tree and thus have size $2$); but, 
in contrast to $\delta$-hyperbolicity (see Section~\ref{sec:background-related}), an $n$-vertex clique is the least 
tree-like graph (attaining the maximum treewidth of $n-1$). 
In fact, the only valid TDs of a clique have all vertices in a single bag. 
Since TDs remain valid under taking subgraphs (once you delete any vertices 
no longer present), if $W$ is any 
complete subgraph of $G$, then every TD of $G$ has some bag that contains all the 
vertices of $W$~\cite{Bodlaender93}. 

Two other canonical examples (to which we will return in detail below) 
are the cycle and the grid, which have vastly differing treewidths. 
All cycles (regardless of the number of vertices) have treewidth $2$ (see 
Figure~\ref{fig:td_example} below). 
The $n \times n$ planar grid, on the other hand, has treewidth $n$, and thus 
it is not tree-like by this measure. 
Grids are particularly noteworthy in the discussion of TDs due to a result 
(described in more detail below) showing that they are essentially the only 
obstruction to having bounded treewidth.%
\footnote{In particular, in the so-called Grid Minor Theorem, Robertson and 
Seymour showed that every graph of treewidth at least $k$ contains a 
$f(k) \times f(k)$ grid as a graph minor, for some integer-valued function $f$.
The original estimate of the function $f$ gave an exponential relationship 
between the treewidth and the grid size, and although several results 
greatly improved the relationship, the question of whether or not it held 
for any polynomial function $f$ remained open for over 25 years. 
Recently, Chekuri and Chuzhoy proved that there is a universal constant 
$\delta > 0$ so that all graphs of treewidth at least $k$ have a grid-minor 
of size $\Omega(k^\delta) \times \Omega(k^\delta)$~\cite{Chekuri14}, 
resolving this conjecture.}

Finding a TD for a given graph whose width is minimal (equal to the 
treewidth) is an NP-hard problem~\cite{Seymour94,Hicks05}. 
Most methods (including those discussed here) for 
constructing TDs were designed to minimize width, as most 
prior work focused on using these structures to reduce computational cost 
for an algorithm/application.%
\footnote{In the context of understanding the intermediate-scale 
structure of real networks and improving inference (e.g., link 
prediction, overlapping community detection, etc.), there are likely more 
appropriate objective functions, although their general identification 
and development is left as future work.} 
Also, although treewidth is a graph invariant, TDs of a network are not 
unique, even under the condition of having minimum width. 
See, e.g., Figure \ref{fig:td_example} below, which shows several distinct 
minimum width TDs of a cycle.

Finally, although it is not standard, we will abuse the term \emph{width} to 
apply it directly to a bag of a tree decomposition (in which case, it takes the value 
of the cardinality of the set minus one), so that we can talk about the 
\emph{maximum width} (which is the equivalent to the usual definition of width), and 
\emph{median width} of a decomposition (which is the median of the widths of the bags).
We will use the term \emph{center} to refer to 
the bag (or bags) associated with the node(s) of minimum eccentricity in the tree underlying a TD, 
and the term \emph{perimeter} for bags associated with nodes of relatively high eccentricity.  
We do this to help provide a framework for discussing the connection in many social and information networks 
between the core (resp. periphery) of the network and the central (resp. perimeter) bags of its 
TD computed with certain heuristics. Note that by definition, a tree will have at most two bags at its center.

\subsection{Constructing Tree Decompositions}
\label{sec:constructing_TD}

Here, we give a brief overview of existing algorithms for constructing 
TDs; more comprehensive surveys can be found in~\cite{Hicks05,Bodlaender10}. 
The algorithms for finding low-width TDs are generally divided into two 
classes: ``theoretical'' and ``computational.''
The former category includes, for example, the linear algorithm of 
Bodlaender~\cite{Bodlaender96}, which checks if a TD of width at most $k$ 
exists (for a fixed constant $k$), and the approximation algorithms of 
Amir~\cite{Amir10}.
These are generally considered (practically) intractable
due to very large hidden constants and complexity of implementation---e.g., 
Bodlaender's algorithm was shown by R\"ohrig~\cite{Rohrig98} to have too 
high a computational cost even when $k = 4$.  
The approximation algorithms of Amir have been tested on graphs with up to 
several hundred vertices, but they require hours of running time even at this 
size scale. 
There has also been work on exact algorithms, the most computational of 
which is perhaps the QuickTree algorithm of Shoikhet and Geiger, which was 
tested on graphs with up to about 100 vertices and treewidth 11~\cite{ShoikhetG97}.
Thus, in practice, most computational work requires the use of heuristic 
approaches (i.e., those which offer no worst-case guarantee on their maximum 
deviation from optimality). 
Since we are interested in applying TDs to real network data, we will focus 
on these ``practical'' algorithms in the remainder of this paper.
We used INDDGO~\cite{Groer12,inddgo}, an 
open source software suite for computing TDs and numerous graph and TD 
parameters.

\subsubsection{Chordal Graph Decomposition}

A common method for constructing TDs is based on 
algorithms for decomposing {\em chordal graphs}. 

\begin{definition} 
A graph $G$ is \emph{chordal} if it has no induced cycles of length greater 
than three (equivalently, every cycle in $G$ with length at least four, has 
a chord).
\end{definition}

\noindent
Chordal graphs are characterized by the existence of an 
ordering  $\pi = (v_1, \ldots, v_n)$ of their vertices so that for each 
$v_i$, the set of its neighbors $v_j$ with $j > i$ form a clique. 
This is a \emph{perfect elimination ordering}, and it gives 
a straightforward construction for a TD (also called the clique-tree) of a 
chordal graph, with bags consisting of the sets of higher-indexed neighbors 
of each vertex.

For a general graph $G$, one common approach for finding TDs is to first 
find a chordal graph $H$ containing $G$, then use the associated TD (since, 
as mentioned earlier, TDs remain valid for all subgraphs on the same 
vertex set).
The typical approach is via {\em triangulation}, a process that uses a 
permutation of the vertex set (called the {\em elimination ordering}) to 
guide the addition of edges, which are referred to as \emph{fill} edges. 
Chordal completions are not unique.  
For example, the complete graph formed on the vertices of $G$ is chordal and 
contains $G$ (although, it is a trivial or the ``worst'' 
triangulation, in the sense that it has the most fill edges and largest 
possible clique subgraph among all triangulations). 

We will use the notation $G^+_\pi$ to denote the triangulation of $G$ using 
ordering $\pi$.  An outline of the process is given in 
Algorithm~\ref{alg:triangulate}. 
The process for finding a TD $T_\pi$ using an elimination order $\pi$ and 
Gavril's algorithm (\cite{Gavril74}) for decomposing chordal graphs is given 
in Algorithm~\ref{alg:gavril}. 
We may refer to the {\em width} of an ordering, by which we mean 
the treewidth of the chordal graph $G^+_\pi$. 
The literature includes several slight variants on Gavril's construction 
routine (such as Algorithm 2 in~\cite{Bodlaender10}), but the overall process 
and width of the TD produced is the same for~each. 

\begin{algorithm}[ht!]
\caption{Triangulate a graph $G$ into a chordal graph $G^+_\pi$}
\label{alg:triangulate}
\begin{algorithmic}[1]
\Statex \textsc{Input:} Graph $G=(V,E)$, and $\pi = (v_1, \ldots, v_n)$, a permutation of $V$ 
\Statex \textsc{Output:} Chordal graph $G^+_\pi \supseteq G$, for which $\pi$ is a perfect elimination ordering 
\State Initialize $G^+_\pi= (V',E')$ with $V' = V$ and $E' = E$
\For{$i = 1$ to $n$} 
\State Let $N_i = \{v_j\, |\, j > i \text{ and } (v_i, v_j) \in E\} $
\For{$\{x,y\}  \subseteq N_i$} 
\If{$x \neq y$ and $(x,y) \not\in N_i$} 
\State $E' = E' \cup \{(x,y)\}$
\EndIf
\EndFor
\EndFor 
\State \textbf{return}  $G^+_\pi$
\end{algorithmic}
\end{algorithm}

\begin{algorithm}[ht!]
\caption{Construct a TD $T_\pi$ of a graph $G$ using elimination ordering $\pi$ and Gavril's algorithm}
\label{alg:gavril}
\begin{algorithmic}[1]
\Statex \textsc{Input:} Graph $G=(V,E)$, $\pi$ a permutation of $V$
\Statex \textsc{Output:} TD $T_\pi=(X, (I,F))$ with $(I,F)$ a tree, and bags $X = \{X_i\}$, $X_i \subseteq V$ 
\State Initialize $T = (X, (I,F))$ with $X = I = F = \emptyset$, $n=|V|$
\State Create an empty $n$-long array $t[]$
\State Use Algorithm 1 to create a triangulation $G^+_{\pi}$ using $\pi$.
\State Let $k = 1$, $I = \{1\}$, $X_1 = \{\pi_n\}$, $t[\pi_n] = 1$
\For{$i = n-1$ to $1$}
\State Find $B_i = \{\text{neighbors of }\pi_i  \text{ in } G^+_\pi\} \cap \{\pi_{i+1}, \ldots, \pi_n\}$
\State Find $m = j$ such that $j \leq k$ for all $\pi_k \in B_i$
\If{$B_i = X_{t[m]}$}
\State $X_{t[m]} = X_{t[m]} \cup \{\pi_i\}$; $t[\pi_i] = t[m]$
\Else
\State $k=k+1$
\State $I = I \cup \{k\}$; $X_{k} = B_i \cup \{\pi_i\}$
\State $F = F \cup \{(k, t[m])\}$; $t[\pi_i] = k$
\EndIf
\EndFor
\State \textbf{return}  $T_\pi = (X, (I,F))$
\end{algorithmic}
\end{algorithm}

Perhaps surprisingly, there always exists some elimination ordering which 
produces an optimal TD (one of minimum width), and this may co-occur with 
high fill. 
The following theorem (see~\cite{Bodlaender10})
presents the connections between treewidth, triangulations, and elimination 
orderings. 

\begin{theorem}~\cite{Bodlaender10}
Let $G = (V,E)$ be a graph, and let $k \le n$ be a non-negative integer.  
Then the following are equivalent.
\begin{enumerate}[noitemsep]
\item $G$ has treewidth at most $k$.
\item $G$ has a triangulation $H$ s.t. any complete subgraph of $H$ (clique) has at most $k+1$ vertices.
\item There is an elimination ordering $\pi$, such that $G^+_\pi$ does not contain any clique on $k+2$ vertices as a subgraph.
\item There is an elimination ordering $\pi$, such that no vertex $v \in V$ has more than $k$ neighbors in
$G^+_\pi$ which occur later in $\pi$.
\end{enumerate}
\end{theorem}

\noindent
Thus, if one can produce a ``good'' elimination ordering (i.e.,
one with a small maximum clique), it is easy to construct a low-width TD,
and such an ordering always exists if the treewidth is bounded. 
The intuition behind fill-reducing orderings to minimize width follows 
from the idea that in order to produce a large clique that wasn't already in 
the network, one ``should'' have to add many fill edges.

\subsubsection{Ordering Heuristics}

Here, we describe the landscape of heuristics for creating 
elimination orderings, focusing on those used in our 
empirical evaluations. 
A more detailed analysis of heuristics as well as theoretical connections between chordal graphs and TDs is available~\cite{Bodlaender10}.
The space of all possible elimination orderings is $O(n!)$ 
for a graph on $n$ vertices, making it impractical to search using brute 
force techniques. 
One possibility for exploring the space is to apply a stochastic local 
search approach like simulated annealing, 
but since this is relatively slow, it is 
not common in practice.  

The first class of specialized methods are known as triangulation 
recognition heuristics, which include lexicographic breadth-first-search 
(\textsc{lex-m}) and maximum cardinality search 
(\textsc{mcs})~\cite{Bodlaender10,Rose75,Berry02,Berry04,RoseTarjanLueker1976,TarjanYannakakis1984,Tarjan85}. 
These methods are guaranteed to provide a perfect elimination ordering for 
chordal graphs, so many believed they would produce low-fill and/or 
low-width orderings for more general graphs.  
In~\cite{Bodlaender10}, the authors report good results with respect to 
width when using these methods on graphs which are already chordal or have 
regular structures, but poor results compared to the greedy heuristics when 
even small amounts of randomness is added to the network. 
Further empirical evaluation in \cite{Groer12} supports these claims. 
Additionally, these heuristics are too computationally expensive to run on 
very large graphs.

A large set of additional heuristics uses the idea of splitting the graph 
(using a small separator), recursively decomposing the resulting pieces, 
and then ``gluing'' the solutions into a single TD~\cite{Amir10, BeckerG01, Bodlaender95,Bouchitte04,Lagergren96,Reed92}.
To quote Bodlaender and Koster~\cite{Bodlaender10}, ``they are significantly 
more complex, significantly slower, and often give bounds that are higher 
than those of simple algorithms.'' 
We do, however, use a related approach that finds a set of nested graph 
partitions, but instead of decomposing the resulting pieces, it places the 
separators into an elimination ordering. 
This approach is called \emph{nested dissection} \cite{George73, Gilbert86}, 
and it is quite popular for computing fill-reducing orderings for sparse 
matrices in numerical linear algebra. 
The algorithm recursively finds a small vertex separator (bisector) in a 
graph, and it ensures that in the resulting elimination ordering, the vertices 
in the two components formed by the bisection all appear before the vertices in 
the separator. 
We use the ``node nested dissection'' algorithm implemented in 
\textsc{METIS}~\cite{karypis98_metis} (called through INDDGO), and we refer 
to this heuristic as \textsc{metnnd}. 
In \textsc{METIS}, the recursion is stopped when the components are smaller 
than a certain size, and some version of minimum degree ordering is then 
applied to the remaining pieces. 
The software ``grows'' each bisection using a greedy node-based 
strategy. 
Since the algorithm is searching for bisections, there is a 
tunable ``balance'' condition (determining how close to 50/50 the split 
needs to be), although for all computations reported in this paper, we left 
the parameter at its default value.

Perhaps the most popular class of elimination ordering routines are 
\emph{greedy} heuristics, named because they make greedy 
decisions to pick the subsequent node in the elimination ordering.  
There are innumerable variations, but the most common use two 
basic concepts:  
choosing a vertex to minimize \emph{fill} (how many new edges will be added 
to the graph if a vertex is chosen to be next in the ordering); or choosing 
a vertex of minimum degree (low-degree vertices have small neighborhoods, 
which also limits the potential fill).  
When applied in their most rudimentary forms, these are the 
\textsc{mindeg}~\cite{Markowitz57}  and \textsc{minfill} orderings. 
Both of these indirectly limit the size of cliques produced in 
the final triangulation (although they were originally designed to minimize 
the number of fill edges added during the triangulation, a quantity which is 
not always correlated).  
For additional heuristics combining these strategies and incorporating 
additional local information, see \cite{Bodlaender10}. 

Even though keeping updated vertex degrees for \textsc{mindeg} during 
triangulation (greedy orderings make their decisions based on a partially 
triangulated graph at each step) is significantly less computationally 
intensive than computing current vertex fills for \textsc{minfill}, there have 
been efforts to reduce further the complexity.  
In particular, the approximate minimum degree or \textsc{amd} 
heuristic~\cite{Amestoy04, Amestoy96}.  
This heuristic computes an upper bound on the degree of each node in each 
pass using techniques based on the quotient graph for matrix factorization, 
and it has been shown to be significantly faster and of similar quality (in 
terms of fill and width minimization). 
We use \textsc{amd} interchangeably with the traditional \textsc{mindeg}, 
especially on larger networks.

\subsection{Additional Related Work}
\label{sec:background-related}

For completeness, we provide here a brief overview of the large body of 
additional related work.  
As already mentioned, TDs played an important role in the proving of the 
graph minor theorem \cite{Robertson86}, but they have also become popular in 
theoretical computer science, as many NP-hard optimization problems have a 
polynomial time algorithm for graphs with bounded treewidth \cite{Arnborg89}.
In addition, bounding the treewidth of the underlying graph of probabilistic 
graphical models allows for fast inference computations \cite{Koller09}.
Additional overviews of TDs and their uses in discrete optimization are 
available \cite{Bod93,Bod05,Bod06,BK07}; and one can also learn more about 
the uses of these methods in numerical linear algebra and sparse matrix 
computations \cite{BP93}, as well as connections with triangulation methods: 
triangulation of minimum treewidth \cite{Amir01}, empirical work on 
treewidth computations \cite{koster01}, the minimum degree heuristic and 
connections with triangulation \cite{Berry03}, and a survey of triangulation 
methods \cite{Heg06}.
Finally, the treewidth of random graphs for various parameter settings 
has been studied~\cite{WLCX11,Gao12}.

A different notion of tree-likeness is provided by $\delta$-hyperbolicity.%
\footnote{Our prior work focused on the use of $\delta$-hyperbolicity~\cite{Adcock13_openmp,Adcock13_icdm,Adcock14_diss}.  It can be a useful tool for describing and analyzing real networks, even though it is expensive to compute, but aside from our theoretical result in Section~\ref{sec:theory_new} relating it to treewidth and treelength, it is not our focus in this paper.}
Early more mathematical work did not consider graphs and 
networks \cite{Gro87,Alonso90}, but more recent more applied work 
has \cite{Jonckheere08,Jonckheere11,CFHM13_IM}.
Computing $\delta$ exactly is very expensive \cite{Adcock13_openmp}, and 
sampling-based methods to approximate it provide only a very rough guide to 
its value and properties \cite{Adcock13_icdm}. 
For many references on $\delta$-hyperbolicity in network analysis, 
see~\cite{CFHM13_IM} (and the more recent paper~\cite{VS14}) and references 
therein.
There has been work on trying to relate hyperbolicity and 
TD-based ideas, often going beyond treewidth to consider other metrics such 
as treelength or chordality or the expansion properties of the graph~\cite{BKM01,WZ11,Dourisboure05,lokshtanov,GM09,KLNS12,Dra13,Ata16,abuata_diss}.

Recent work in network analysis and community structure analysis has pointed 
to some sort of ``core-periphery'' structure in many real
networks~\cite{LLDM09_communities_IM, ST08, Rombach14, Adcock13_icdm, Jeub15}; 
and recently this has been related to the $k$-core decomposition---see, 
e.g.,~\cite{Adcock13_icdm} and references therein.
The $k$-core decomposition is of interest more generally, and additional 
references for $k$-core decompositions, including their use in visualization 
and in larger-scale applications, include 
\cite{Sei83,Alvarez-H05,Alvarez-H08,HJMA08,Pajek03,CKCO11,Colomer13}.
Questions of well-connected or expander-like cores are of particular 
interest in applications having to do with diffusion processes, influential 
spreaders, and related questions of social contagion~\cite{Kitsak10, Ugander12}.

There are a few other papers that have used TDs to investigate the structural 
properties of social and information networks: e.g., to look at the 
tree-likeness of internet latency and bandwidth~\cite{RMKBGA09}; to compare 
hyperbolicity and treewidth on internet networks~\cite{MSL11}; and 
to examine the relationship between hyperbolicity, treewidth, and the 
core-periphery structure in a much wider range of social and information 
networks~\cite{Adcock13_icdm}.  
In particular, \cite{MSL11} concludes that the hyperbolicity is 
small in the networks they examined but the treewidth is relatively large, 
presumably due to a highly connected core; and \cite{Adcock13_icdm} 
concludes that many real social and information networks do have a tree-like 
structure, with respect to both metric-based hyperbolicity and (in spite of 
the large treewidth) the cut-based TDs, that corresponds to the 
core-periphery structure of the network.  
Finally, very recently we became aware of~\cite{MAIK14} and~\cite{CM93}.

\section{Network Datasets}
\label{sec:data}

We have examined the empirical performance of existing TD heuristics on a 
broad set of real-world social and information networks as well as a large corpus of synthetic graphs.  
The real-world networks have been chosen to be representative of a broad 
range of networks, as analyzed in prior
work~\cite{LLDM09_communities_IM,Adcock13_icdm,Jeub15}, and the synthetic 
graphs have been chosen to illustrate the behavior of TD methods in 
controlled settings.  
See Table~\ref{tbl:networks-basic-stats} for a summary of the networks we
have considered.  
The real-world graphs are connected, but we are interested in 
parameter values for the synthetic graphs which might cause the instances 
to be \emph{disconnected}.  In these cases, we work with the giant component, and the statistics 
in Table~\ref{tbl:networks-basic-stats} are for this connected subgraph.

\begin{table}[!htb]
\begin{center}
{\footnotesize
\begin{tabular}{l|r|r|r|r|r|r|}
Network & $n_c$ & $k_{l}$ & $k_{m}$ & $\bar{d}$ & $\bar{C}$ & $D$  \\
\hline \hline
\multicolumn{7}{l}{ER Random Graphs} \\ 
\hline \hline
\textsc{ER(1.6)} & 3210 & 1 & 2 & 2.16 & 0.00 & 38  \\
\textsc{ER(1.8)} & 3617 & 1 & 2 & 2.28 & 9.30 $\times 10^{-4}$ & 34 \\
\textsc{ER(2)}   & 4001 & 1 & 2 & 2.39 & 9.11 $\times 10^{-4}$ & 30 \\
\textsc{ER(4)}   & 4879 & 1 & 3 & 4.05 & 8.96 $\times 10^{-3}$ & 15 \\
\textsc{ER(8)}   & 4998 & 1 & 5 & 8.04 & 1.59 $\times 10^{-3}$ & 7 \\
\textsc{ER(16)}  & 5000 & 4 & 11 & 16.1 & 3.13 $\times 10^{-3}$ & 5 \\
\textsc{ER(32)}  & 5000 & 7 & 23 & 32.1 & 6.39 $\times 10^{-3}$ & 4 \\
\hline \hline
\multicolumn{7}{l}{PL Random Graphs} \\ 
\hline \hline
\textsc{PL(2.50)}  & 4895 & 1 & 4 & 2.78 & 2.46 $\times 10^{-3}$ & 18 \\
\textsc{PL(2.75)}  & 4650 & 1 & 2 & 2.43 & 6.99 $\times 10^{-4}$ & 22 \\
\textsc{PL(3.00)}  & 4071 & 1 & 2 & 2.24 & 1.18 $\times 10^{-3}$ & 29 \\
\hline \hline
\multicolumn{7}{l}{SNAP Social Graphs} \\ 
\hline \hline
\textsc{CA-GrQc}     &  4158 & 1 & 43 &  6.46 & .665  & 17 \\
\textsc{CA-AstroPh}  & 17903 & 1 & 56 & 22.0  & .669  & 14 \\
\textsc{as20000102}  &  6474 & 1 & 12 &  3.88 & .399  &  9 \\
\textsc{Gnutella09}  &  8104 & 1 & 10 &  6.42 & .0137 & 10 \\
\textsc{Email-Enron} & 33696 & 1 & 43& 10.7  & .708  & 13 \\
\hline \hline
\multicolumn{7}{l}{FB Social Graphs} \\ 
\hline \hline
\textsc{FB-Caltech}  &  762 & 1 & 35 & 43.7 & .426 & 6 \\
\textsc{FB-Haverford} & 1446 & 1 & 63 & 82.4 & .327 & 6\\
\textsc{FB-Lehigh}   &  5073 & 1 & 62 & 78.2 & .270 & 6 \\
\textsc{FB-Rice}     & 4083 & 1 & 72 & 90.5 & .300 & 6 \\
\textsc{FB-Stanford} & 11586 & 1 & 91 & 98.1 & .252 & 9 \\
\hline \hline
\multicolumn{7}{l}{Miscellaneous Graphs} \\ 
\hline \hline
\textsc{PowerGrid}   &  4941 & 1 & 5 & 2.67 & 0.107 & 46 \\
\textsc{Polblogs}    &  1222 & 1 & 36 & 27.4 & 0.360 &  8 \\
\textsc{PlanarGrid}  &  2500 & 2 & 2 & 3.92 & 0.00  & 98 \\
\textsc{road-TX} & 1379917 & 1 & 3 & 1.39 & 0.0209 & 1054 \\
\textsc{web-Stanford} & 281903 & 1 & 71 & 8.20 & $2.89 \times 10^{-3}$ & 674 \\
\end{tabular}
}
\caption{
Statistics of analyzed networks: nodes in giant component $n_c$; 
$k_l$ the lowest $k$-core;
$k_m$ the maximum $k$-core; 
average degree $\bar{d}=2E/N$;  
average clustering coefficient $\bar{C}$; 
and diameter~$D$.
}
\label{tbl:networks-basic-stats}
\end{center}
\end{table}

\textbf{Erd\H{o}s-R\'{e}nyi (\textsc{ER}) graphs.}  
Although \textsc{ER} graphs are often criticized for their inability to model
pertinent properties of realistic networks, \emph{extremely} sparse
\textsc{ER} graphs have several structural inhomogeneities that are 
important for understanding tree-like structure in realistic
networks~\cite{Adcock13_icdm}.%
\footnote{The same is true for many other less unrealistic random graph 
models, assuming their parameters are set to analogously sparse values (which they are often not).}  
In particular, in the
extremely sparse regime of $1/n<p<\log(n)/n$, \textsc{ER} graphs are
(w.h.p.) not even fully connected; \textsc{ER} graphs in this regime
have an upward-trending NCP (network community
profile)~\cite{LLDM09_communities_IM}; with respect to their $k$-core
structure, a shallow (but non-trivial) core-periphery structure
emerges~\cite{Adcock13_icdm, Percus08}; and with respect to their
metric properties (as measured with $\delta$-hyperbolicity), graphs in
this regime have non-trivial tree-like properties~\cite{Adcock13_icdm}.
Following previous work~\cite{Adcock13_icdm}, we set the target number
of vertices to $n=5000$, and we choose $p=\frac{d}{n}$ for various values
of $d$ from $d=1.6$ to $d=32$.  We denote these networks using
\textsc{ER($d$)}.  Table~\ref{tbl:networks-basic-stats} clearly shows
that, as a function of \emph{increasing} $d$, i.e., increasing $p$,
the size of the giant component increases to $5000$, the number of
edges increases dramatically, the clustering coefficient remains close
to zero, the average degree $\bar{d}$ increases, and the diameter
decreases dramatically.

\textbf{Power Law (PL) graphs.}  
We also considered the Chung-Lu model~\cite{ChungLu:2006}, an ER-like
random graph model parameterized to have a power law degree
distribution (in expectation) with power law (or heterogeneity)
parameter $\gamma$, which we vary between $2$ and $3$.  We denote
these networks using \textsc{PL($\gamma$)}.  We consider values of the
degree heterogeneity parameter $ \gamma\in \{ 2.50, 2.75, 3.00 \}$.
Table~\ref{tbl:networks-basic-stats} shows that, as a function of
\emph{decreasing} $\gamma$, the size of the giant component increases,
the average degree $\bar{d}$ increases, and the diameter decreases.
Although not shown in Table~\ref{tbl:networks-basic-stats}, as
$\gamma$ decreases, PL graphs also form a rather prominent, and
moderately-deep, $k$-core structure~\cite{Adcock13_icdm}.  These are
all trends that parallel the behavior of \textsc{ER} as $d$
increases.%
\footnote{For sparse ER graphs, this happens since there is not enough
  edges for concentration of measure to occur, i.e., for empirical
  quantities such as the empirical degrees to be very close to their
  expected value.  For PL graphs, an analogous lack of measure
  concentration occurs due to the exogenously-specified heterogeneity
  parameter~$\gamma$.}

\textbf{SNAP graphs.}  
We selected various social/information networks that were used in the 
large-scale empirical analysis that first established the upward-sloping NCP 
and associated nested core-periphery for a broad range of realistic social 
and information graphs~\cite{LLDM09_communities_IM}.  These are available at 
the SNAP website~\cite{snap-index}.  
In particular, the networks we considered are \textsc{CA-GrQc} and 
\textsc{CA-AstroPh} (two collaboration networks); \textsc{as20000102} (an 
autonomous system snapshot); \textsc{Gnutella09} (a peer-to-peer network 
from Gnutella); \textsc{Email-Enron} (an email network from the Enron 
database); as well as the Stanford Web network \textsc{web-Stanford} and the 
Texas road network \textsc{road-TX}. 
These networks are very sparse, e.g., fewer than ca. $10$ edges per node; 
and they exhibit substantial degree heterogeneity, moderately high 
clustering coefficients (except for \textsc{Gnutella09}, 
\textsc{web-Stanford}, and \textsc{road-TX}), and moderately small diameters.  
In addition, although not presented in Table~\ref{tbl:networks-basic-stats}
(and with the exception of \textsc{road-TX}), 
these graphs have a much stronger core-periphery structure, as measured by 
the $k$-core decomposition, than typical synthetic 
networks~\cite{Adcock13_icdm, Jeub15, LLDM09_communities_IM}.

\textbf{Facebook Networks.}  
We selected several representative Facebook graphs out of ca. 100
Facebook graphs from various American universities collected in ca. 2005
\cite{Traud12}.  These data sets range in size from around $700$ vertices
(\textsc{FB-Caltech}) to approximately $30,000$ vertices
(\textsc{FB-Texas84}).  In particular, we examine \textsc{FB-Caltech},
\textsc{FB-Rice}, \textsc{FB-Haverford}, \textsc{FB-Lehigh}, and
\textsc{FB-Stanford} in this paper.  These networks all arise via
similar generative procedures, and thus there are
strong similarities between them.  There are a few distinctive
networks, however, that are worth mentioning.  In particular, several
universities (\textsc{FB-Caltech, FB-Rice, FB-UCSC}) have a
particularly strong resident housing system, and it is known that this
manifests itself in structural properties of the graphs
\cite{Traud12}.  Below, we will use the meta-information associated
with this housing system to provide ``ground-truth'' clusters/communities for
comparison and evaluation.%
\footnote{See \cite{Jeub15} for how this affects the NCP of these networks.}
One important characteristic to observe
from Table~\ref{tbl:networks-basic-stats} is that these Facebook
networks, while sparse, are much denser than any of the SNAP graphs we
consider or that were considered
previously~\cite{LLDM09_communities_IM}.%
\footnote{Among the differences caused by the much higher density of
  Facebook networks is that these networks have a much deeper $k$-core 
  structure than the other real networks, and they tend to lack deep
  cuts, e.g., they lack even good very-imbalanced partitions such as those
  responsible for the upward-sloping
  NCP~\cite{LLDM09_communities_IM,Jeub15}.}

\textbf{Miscellaneous Networks.}  
We also selected a variety of real-world networks that, based on prior
work~\cite{LLDM09_communities_IM, Adcock13_icdm, Jeub15}, are known to have 
very different properties than the SNAP social graphs or the Facebook social 
graphs.  
In particular, we consider \textsc{Polblogs}, a political bloggers 
network~\cite{adamic05election} (a graph constructed from political blogs which are 
linked); the Western US power grid \textsc{PowerGrid} \cite{watts98collective}; and a 
two-dimensional $50 \times 50$ planar grid \textsc{PlanarGrid}.

\section{Tree decompositions of toy networks}
\label{sec:toy_results}

In this section,
we will describe the results of using a variety of TD heuristics on a set of 
very simple ``toy'' networks, on which the optimal-width TDs are known.  
The five toy networks we consider are a binary tree (\textsc{SmallBinary}), 
a small section of the two-dimensional planar grid (\textsc{SmallPlanar}), a 
cycle (\textsc{SmallCycle}), a clique (\textsc{SmallClique}), and an 
Erd\H{o}s-R\'{e}nyi graph with an edge probability of $p=0.5$ 
(\textsc{SmallER}).  
Each of these networks has $100$ nodes (except for \textsc{SmallCycle}, which 
has only $10$ nodes---the reason for this is that the principal change by
having a larger cycle is that the eccentricity of the decompositions becomes 
much larger, which simply makes it more difficult to visualize---and
\textsc{SmallBinary}, which has $128$ nodes to maintain symmetry).
In Figure~\ref{fig:toy_networks}, we provide visualizations%
\footnote{These and other visualizations were created with the GraphViz command \emph{neato}~\cite{Gansner00}, with the help of \cite{Davis11,Malisiewicz10}.}
of each of the 
five networks.

These very simple network topologies illustrate in a controlled way the 
behavior of different TD heuristics in a range of 
settings.  
For example, while \textsc{SmallBinary} is a tree, the other graphs are not; 
the two-dimensional grid is quite different from a tree, as is 
\textsc{SmallCycle} (although, from the treewidth perspective it is fairly 
close to a tree), and both have high-quality well-balanced partitions; and 
both \textsc{SmallClique} and \textsc{SmallER} are expanders (not constant 
degree expanders, but expanders in the sense that they don't have any good 
partitions) and thus very non-tree-like (from the TD perspective), but each 
has important differences with respect to their respective TDs.  
We will focus on which types of structures different 
heuristics tend to capture, as well as how different heuristics deal with 
nodes (not bags) which are associated with the core or periphery of the 
original network.  
Importantly, these toy networks have basic constructions, and they (mostly) 
have known optimal width TDs---e.g., \textsc{SmallPlanar} and 
\textsc{SmallCycle} have several known equivalent minimum width TDs---and the 
\textsc{SmallER} network serves to illustrate some of the effects of 
randomness on a TD.  
The insights we obtain 
here can be used to interpret the output of TD heuristics in much more
complex synthetic and real~networks.

\subsection{TD properties of toy networks}

In Figures~\ref{fig:toy_td_mind}, \ref{fig:toy_td_metnnd}, 
and~\ref{fig:toy_td_lexm}, we show visualizations of TDs produced by various 
heuristics (the greedy \textsc{mindeg} in Figure~\ref{fig:toy_td_mind}; the 
\textsc{metnnd}, nested node dissection via METIS, in 
Figure~\ref{fig:toy_td_metnnd}; and \textsc{lexm} in 
Figure~\ref{fig:toy_td_lexm}) for each of these five toy networks.  
In these visualizations, the size of the bag corresponds with the bag's 
width, and the coloring is based on the fraction of edges present in the 
induced subgraph of the bag.  
In particular, if the nodes in the bag form a clique in the original 
network, then the fraction of edges present is $1.0$ and the bag is
dark red; while if the the nodes are completely disconnected in the 
original network, then the bag is dark blue.

From these figures, we see substantial differences between 
the TDs that different heuristics generate for these five toy networks.  
All heuristics give the same uninteresting results for 
\textsc{SmallClique}; but for all of the other networks, including 
\textsc{SmallBinary}, there are differences in the
decompositions produced by the different heuristics.  
Consider
\textsc{SmallPlanar},
\textsc{SmallCycle}, and \textsc{SmallER}.  
For both \textsc{SmallPlanar} and \textsc{SmallER}, \textsc{mindeg} and
\textsc{metnnd} return TDs with several prominent branches, while 
\textsc{lexm} returns a path for the TD.  
For \textsc{mindeg}, this is due to the tendency of the algorithm to pick
low-degree nodes on the ``outside'' of the network and then work its way 
around the outside of the network.  
For \textsc{metnnd}, this is due to the tendency of the algorithm to cut 
the networks repeatedly into smaller pieces and then recursively ``eat 
away'' at these smaller pieces to form the TD.  
On the other hand, the \textsc{lexm} heuristic works to produce a minimal 
triangulation using lexicographic labelings along paths.
This often results in a path-like TD, as the algorithm 
uses a breadth-first search through the network.  
For \textsc{SmallCycle}, \textsc{metnnd} returns a ``branchy''
TD, while both \textsc{mindeg} and \textsc{lexm} return path-like TDs.


\begin{figure*}[!htb]
\begin{center}
\begin{subfigure}[h]{0.19\textwidth}
\includegraphics[width=\textwidth]{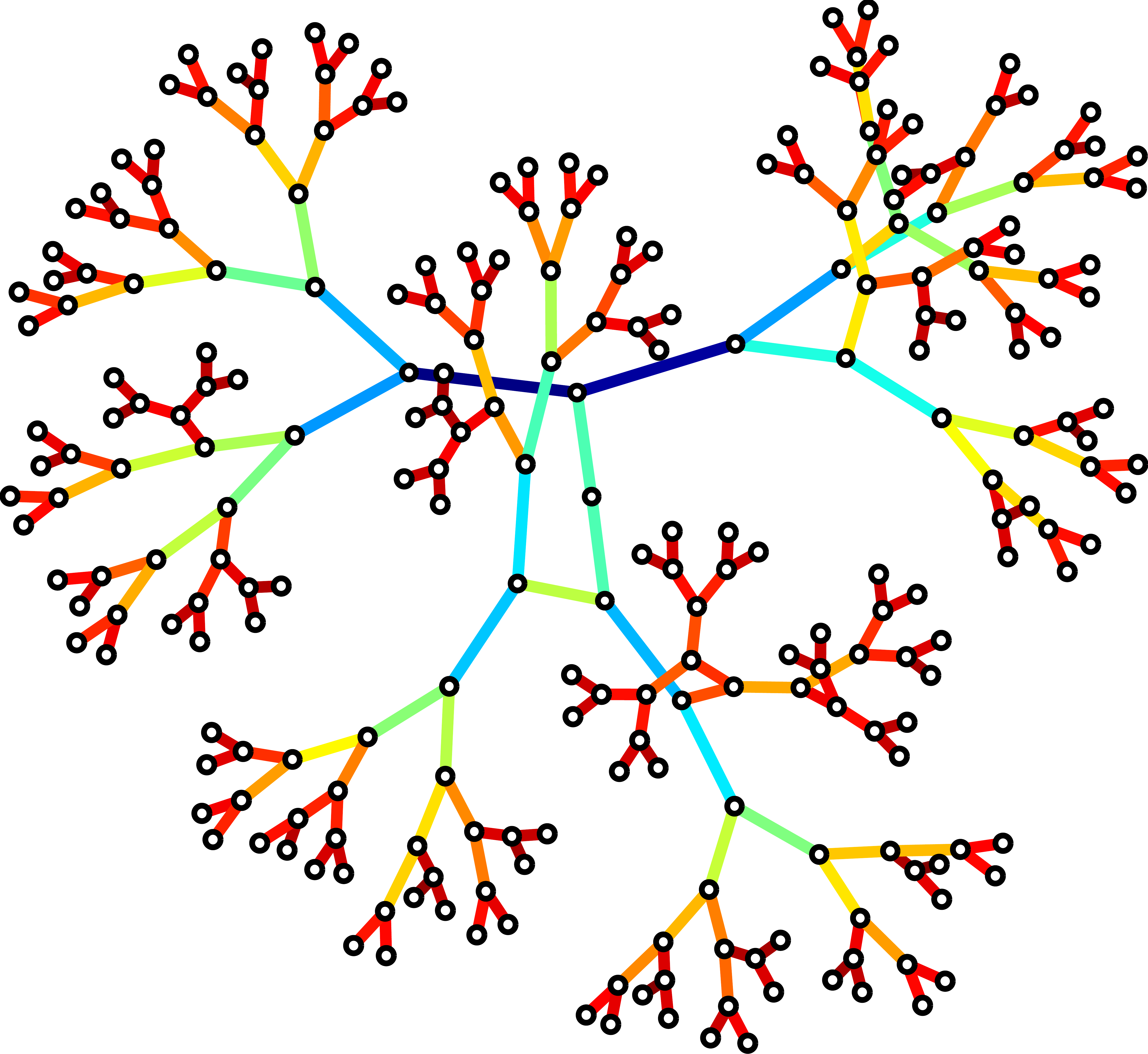}
\caption{\centering \textsc{SmallBinary}}
\end{subfigure}
\begin{subfigure}[h]{0.19\textwidth}
\includegraphics[width=\textwidth]{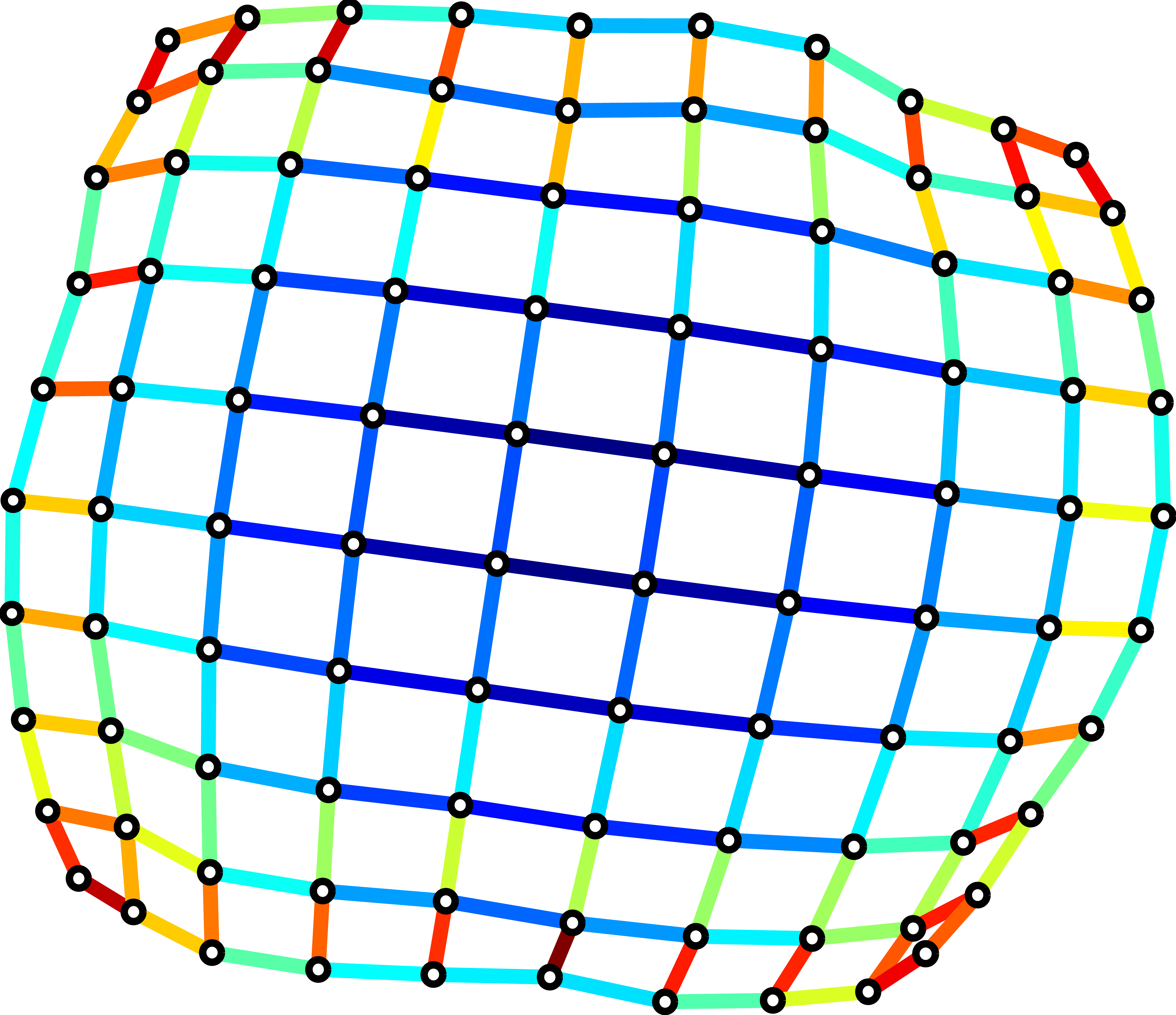}
\caption{\centering $10 \times 10$ \textsc{SmallPlanar}}
\end{subfigure}
\begin{subfigure}[h]{0.12\textwidth}
\includegraphics[width=\textwidth]{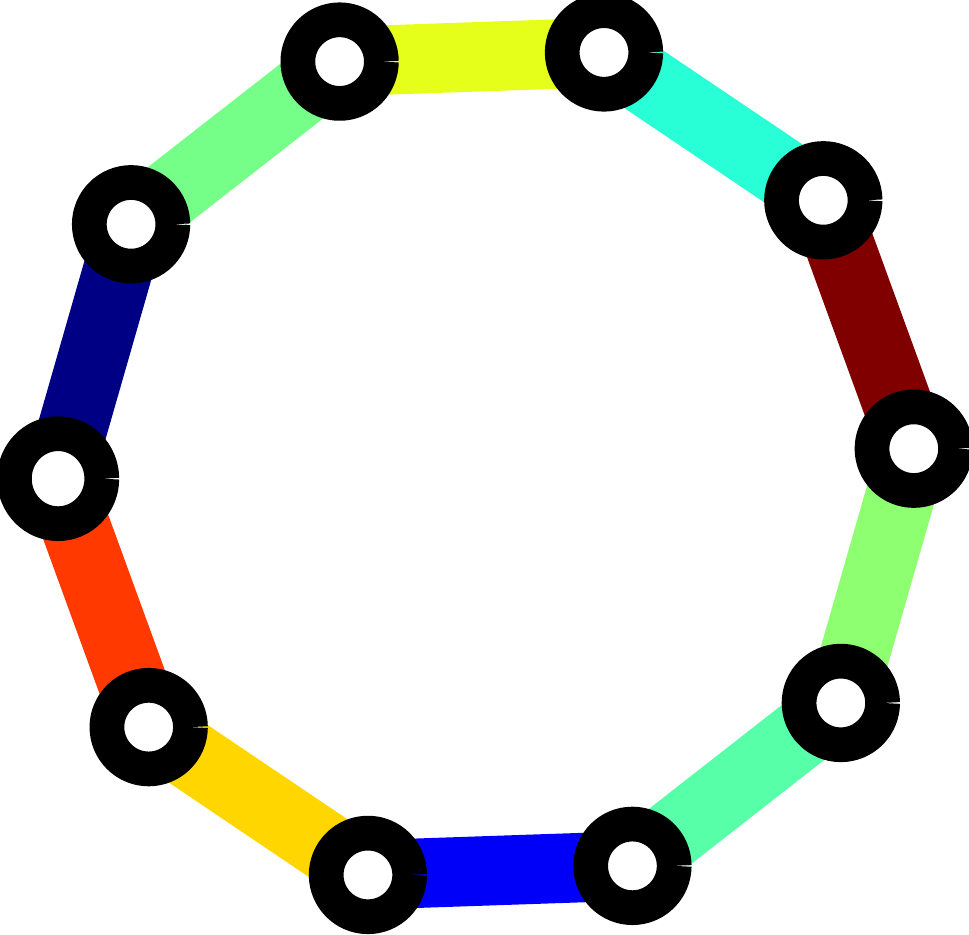}
\caption{\centering \textsc{SmallCycle}}
\end{subfigure}
\begin{subfigure}[h]{0.12\textwidth}
\includegraphics[width=\textwidth]{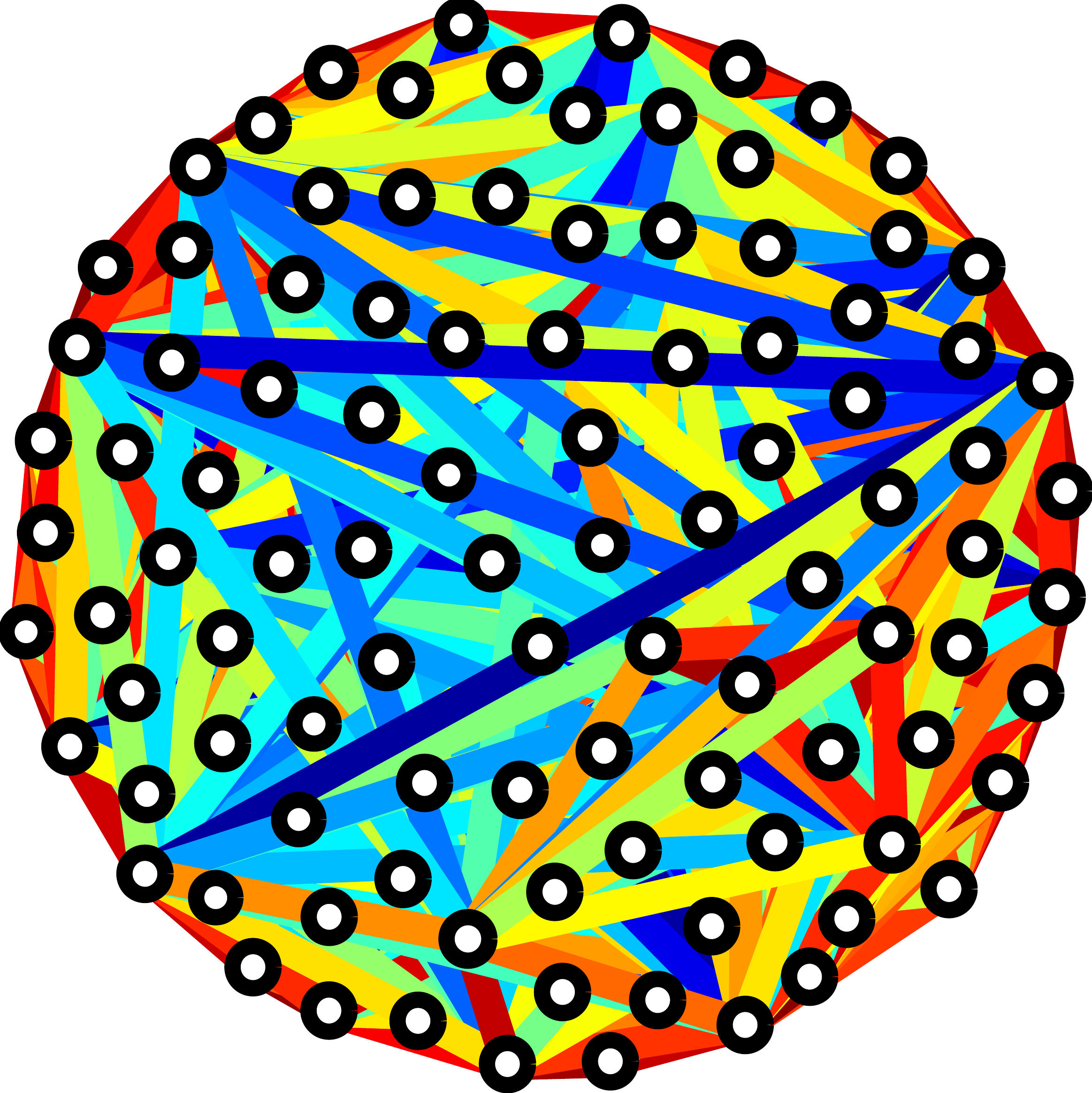}
\caption{\centering $100$ node \textsc{SmallClique}} 
\end{subfigure}
\begin{subfigure}[h]{0.12\textwidth}
\includegraphics[width=\textwidth]{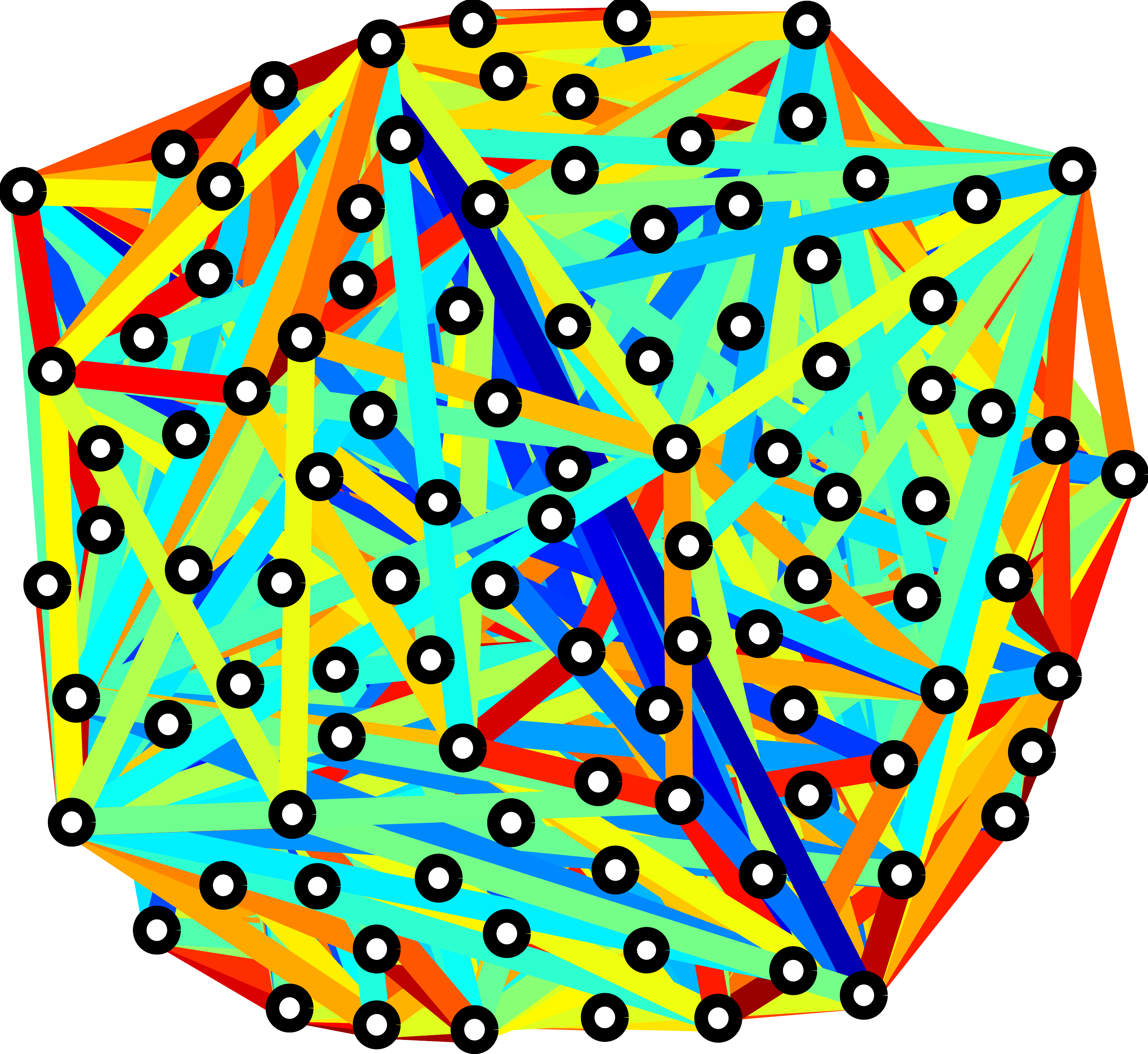}
\caption{\centering \textsc{SmallER}} 
\end{subfigure}
\caption{A set of small networks.  Edges are colored by their length in the planar embedding.}
\label{fig:toy_networks}
\end{center}
\begin{center}
\begin{subfigure}[h]{0.19\textwidth}
\includegraphics[width=\textwidth]{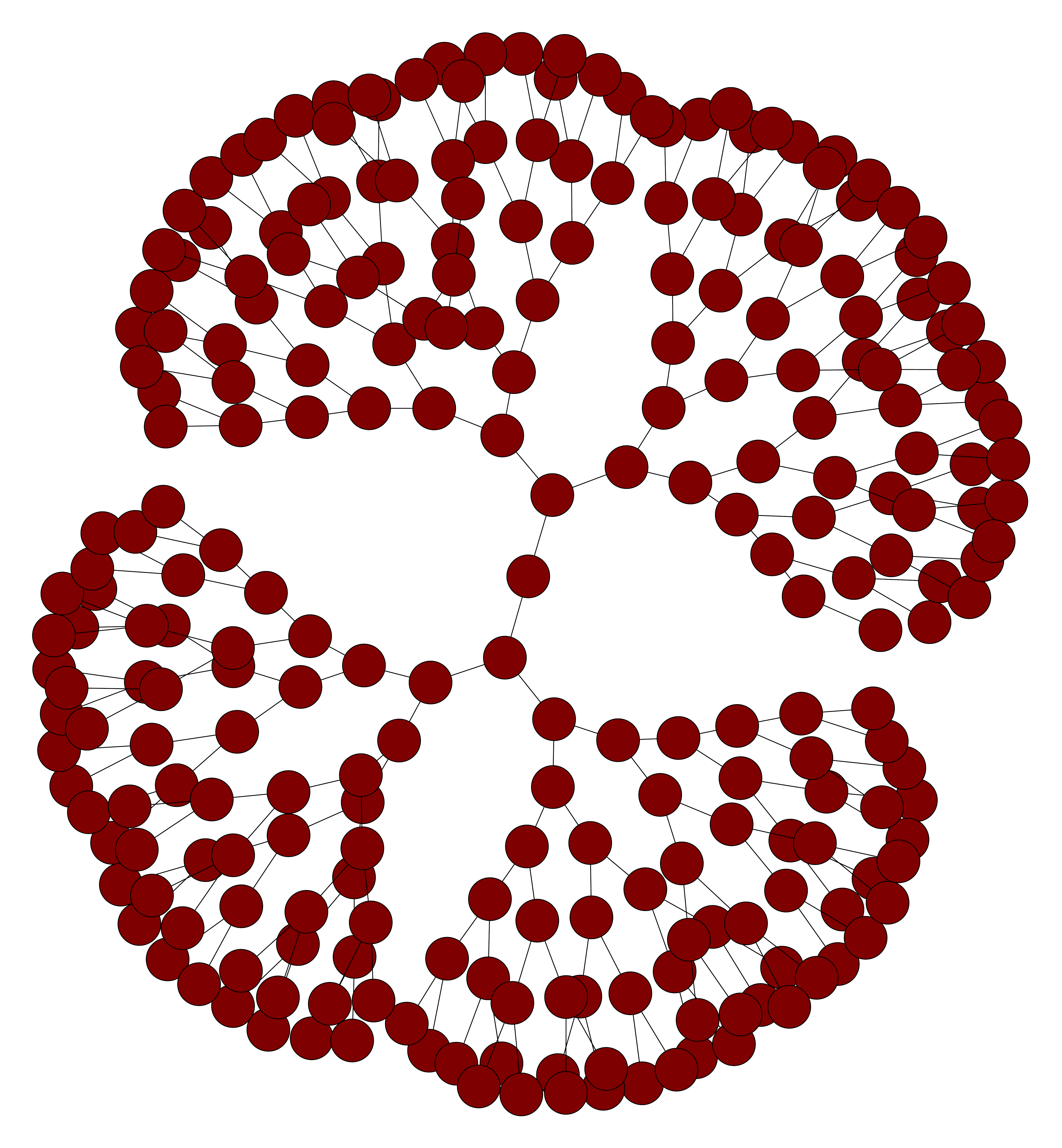}
\caption{\centering \textsc{SmallBinary}}
\end{subfigure}
\begin{subfigure}[h]{0.19\textwidth}
\includegraphics[width=\textwidth]{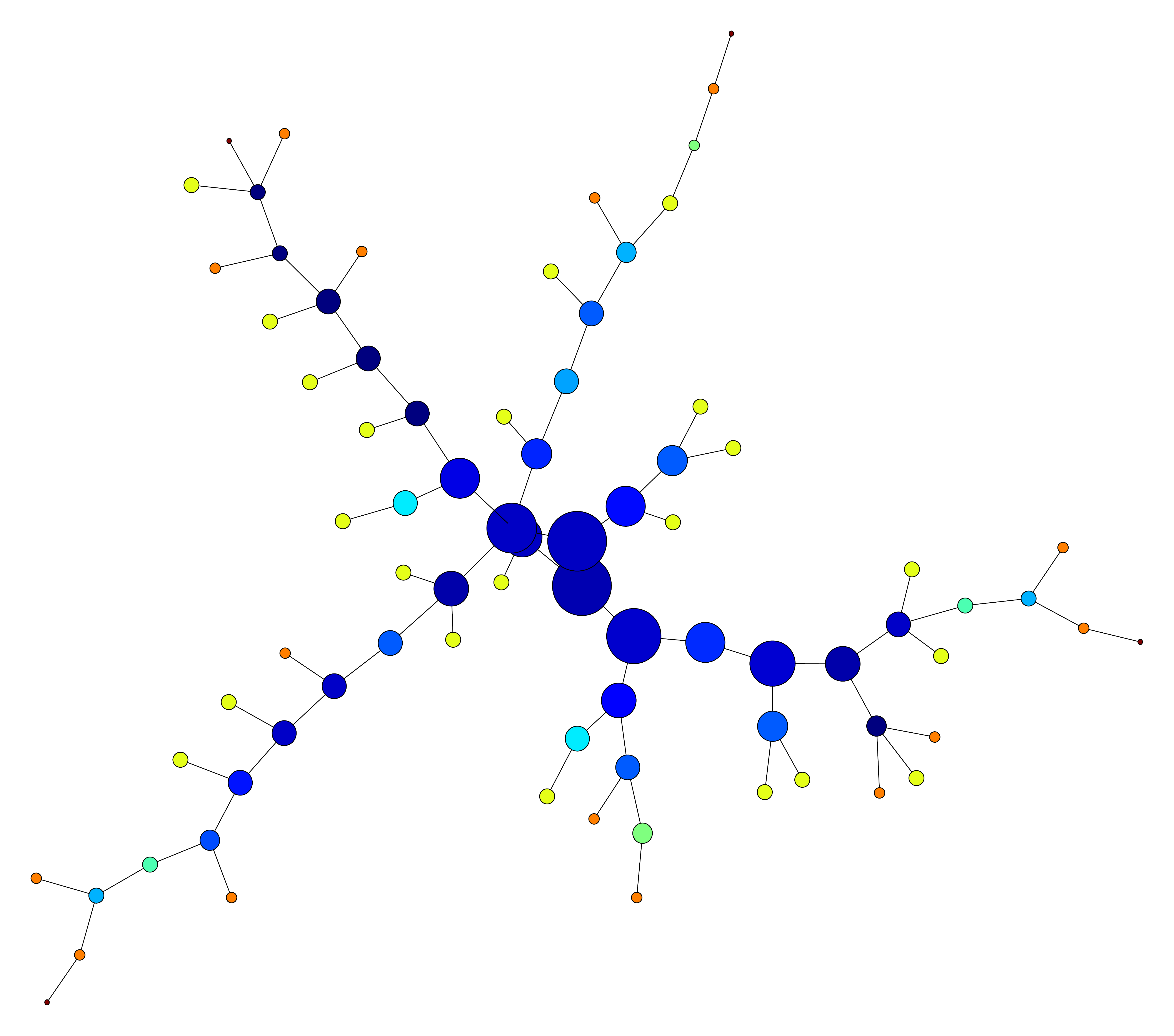}
\caption{\centering \textsc{SmallPlanar}}
\label{fig:mind_10_10_plane}
\end{subfigure}
\begin{subfigure}[h]{0.12\textwidth}
\includegraphics[width=\textwidth]{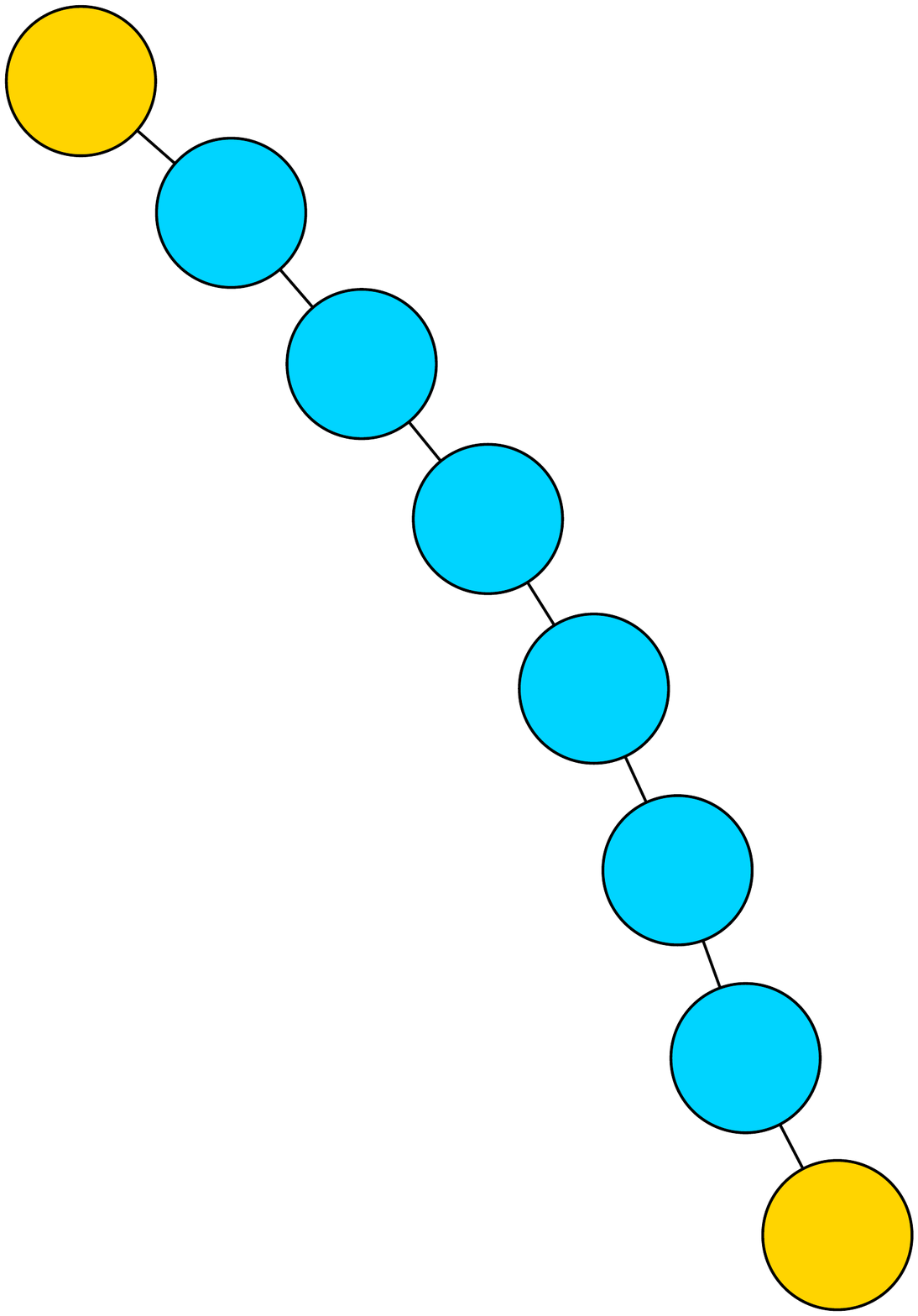}
\caption{\centering \textsc{SmallCycle}}
\end{subfigure}
\begin{subfigure}[h]{0.12\textwidth}
  \includegraphics[width=\textwidth]{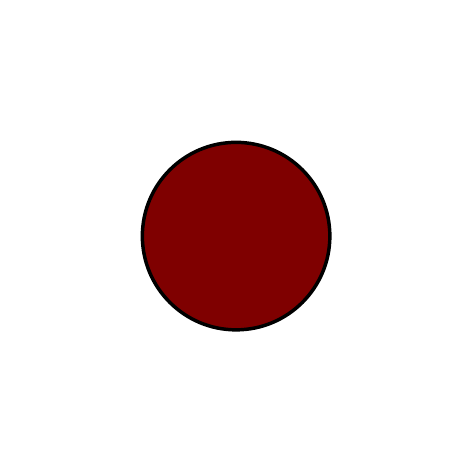}
\caption{\centering \textsc{SmallClique}} 
\end{subfigure}
\begin{subfigure}[h]{0.12\textwidth}
  \includegraphics[width=\textwidth]{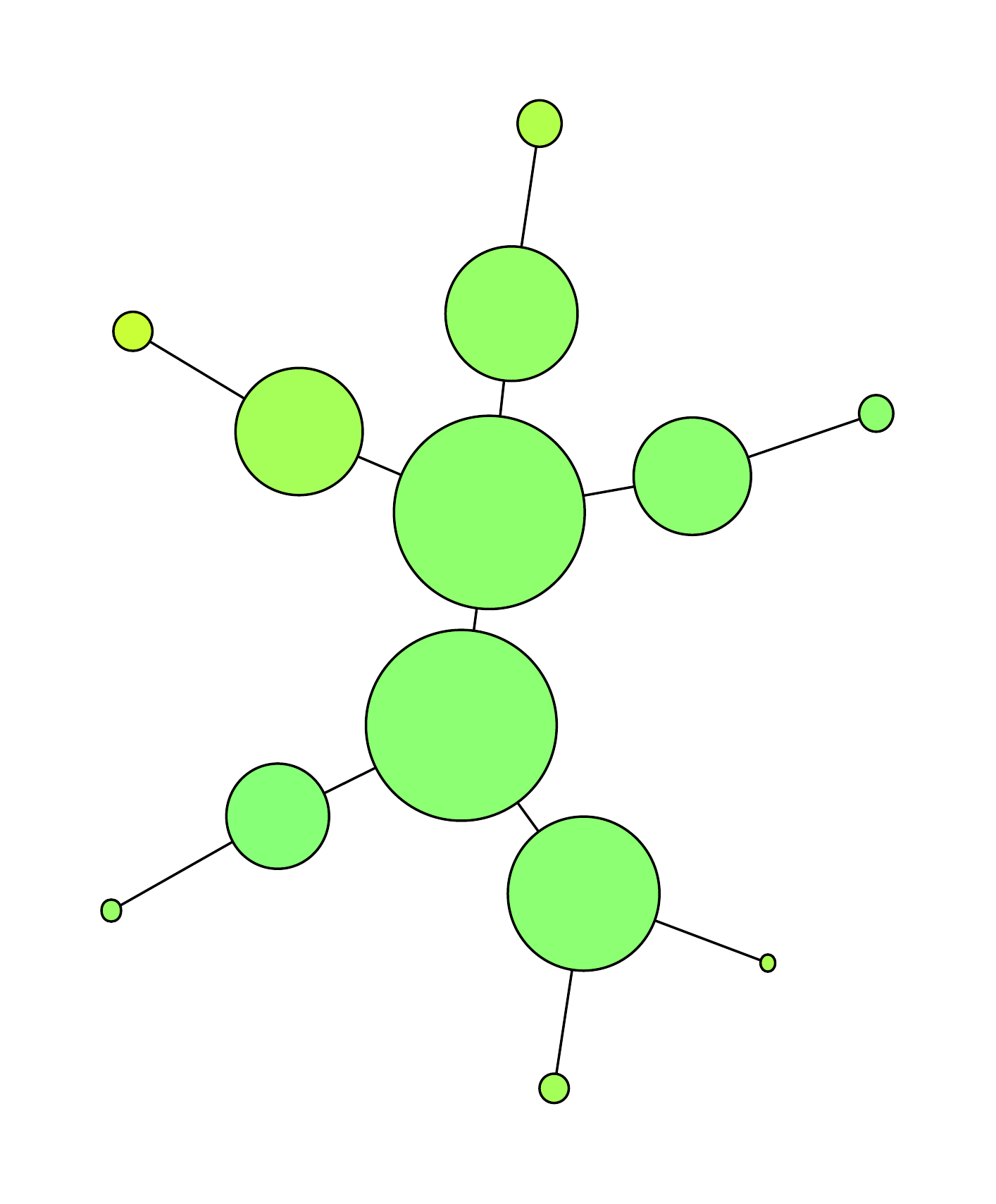}
\caption{\centering \textsc{SmallER}} 
\end{subfigure}
\caption{Greedy \textsc{mindeg} TDs of toy networks.  Bags are colored by the fraction of possible edges present in the bag, with red being denser and blue being less dense.}
\label{fig:toy_td_mind}
\end{center}
\begin{center}
\begin{subfigure}[h]{0.19\textwidth}
\includegraphics[width=\textwidth]{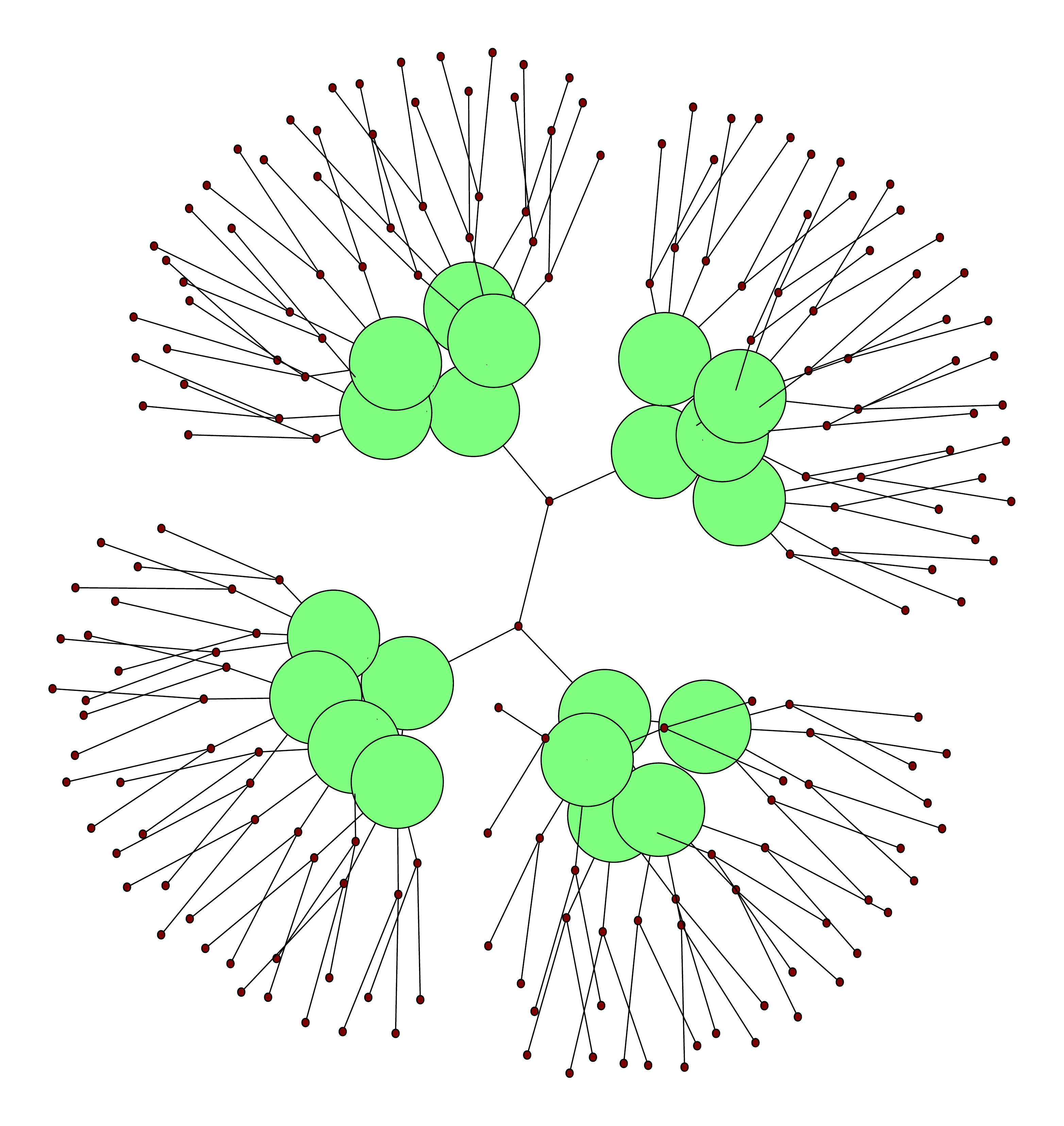}
\caption{\centering \textsc{SmallBinary}}
\end{subfigure}
\begin{subfigure}[h]{0.19\textwidth}
\includegraphics[width=\textwidth]{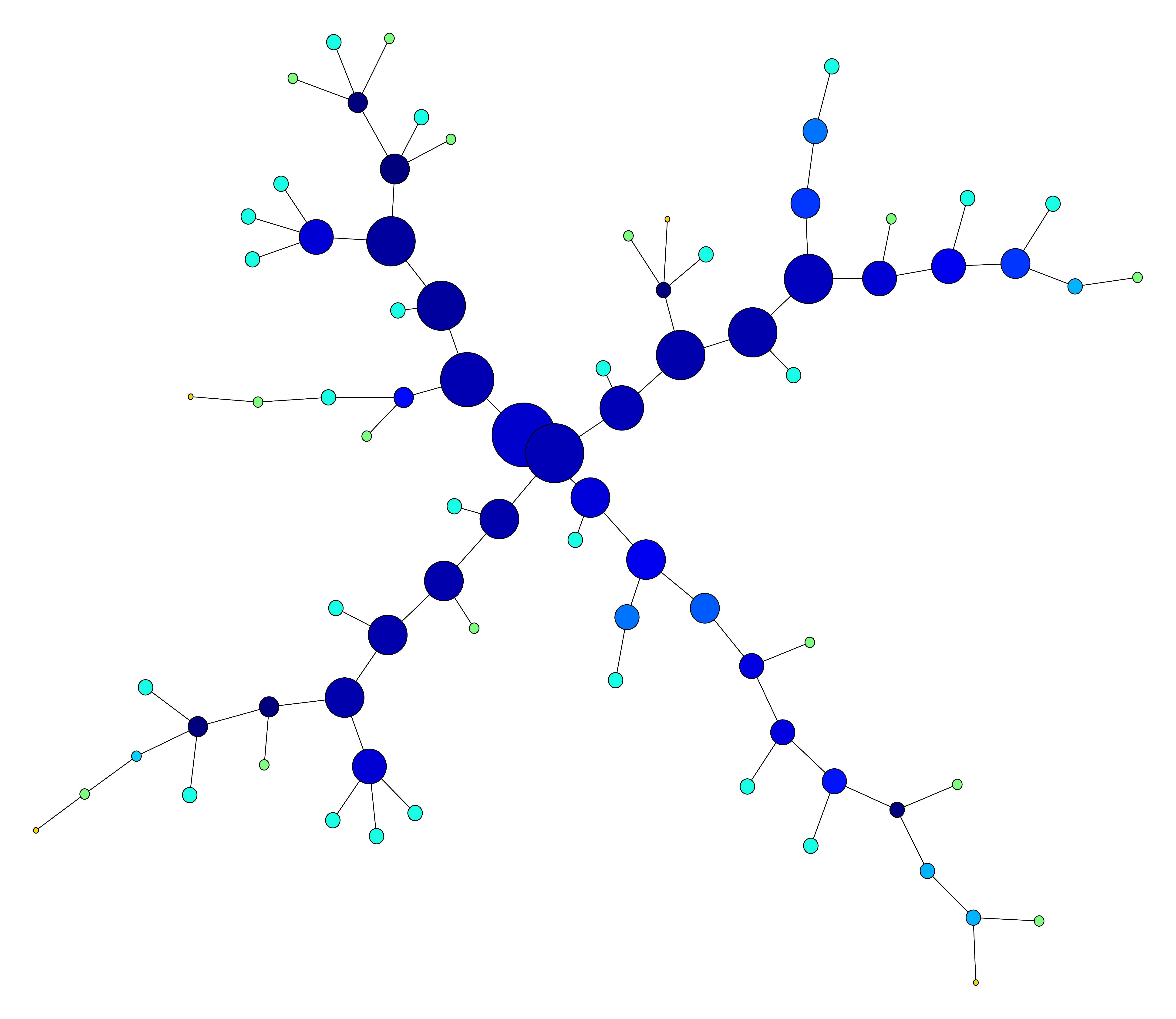}
\caption{\centering \textsc{SmallPlanar}}
\label{fig:metnnd_10_10_plane}
\end{subfigure}
\begin{subfigure}[h]{0.12\textwidth}
\includegraphics[width=\textwidth]{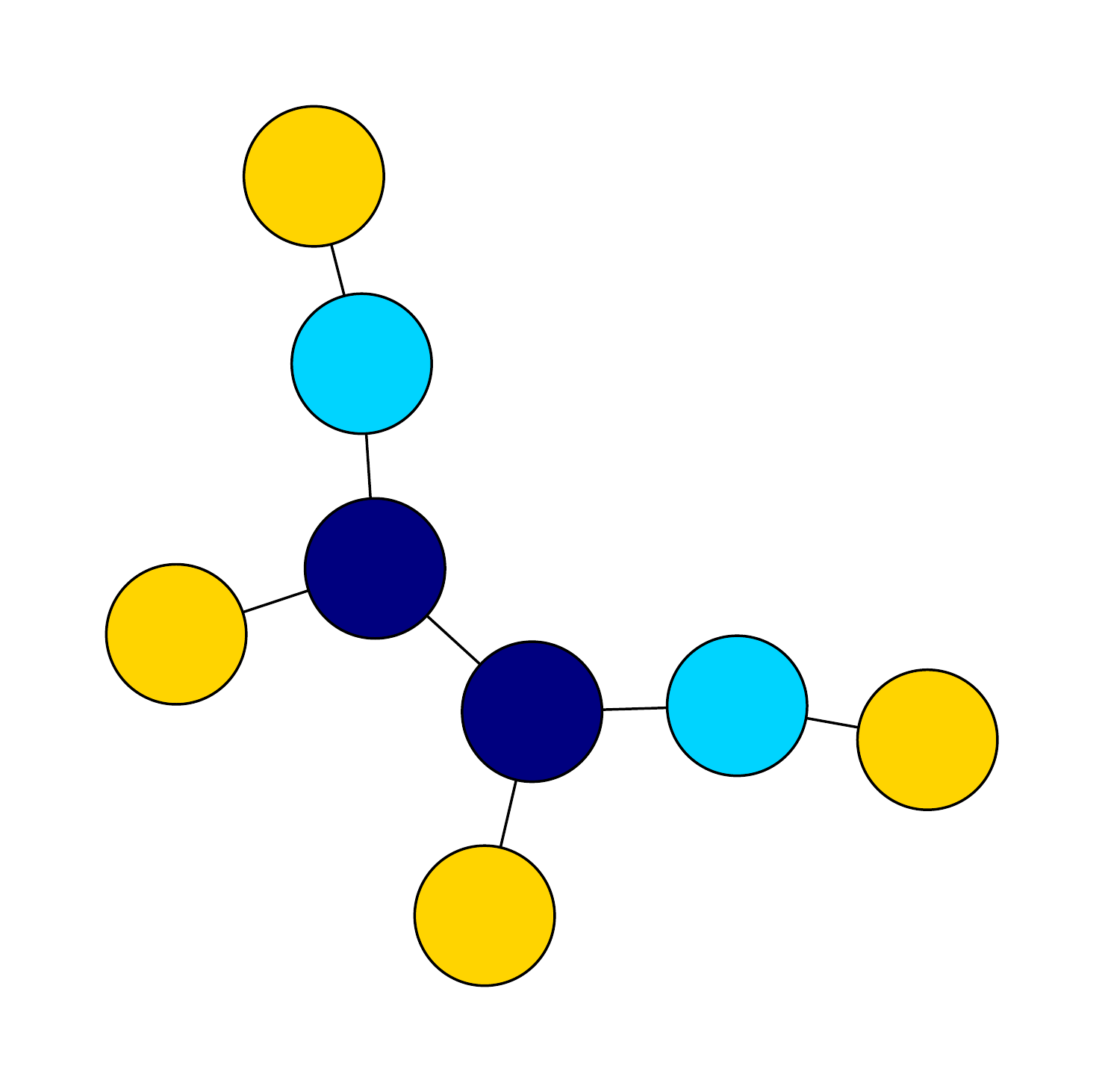}
\caption{\centering \textsc{SmallCycle}}
\end{subfigure}
\begin{subfigure}[h]{0.12\textwidth}
  \includegraphics[width=\textwidth]{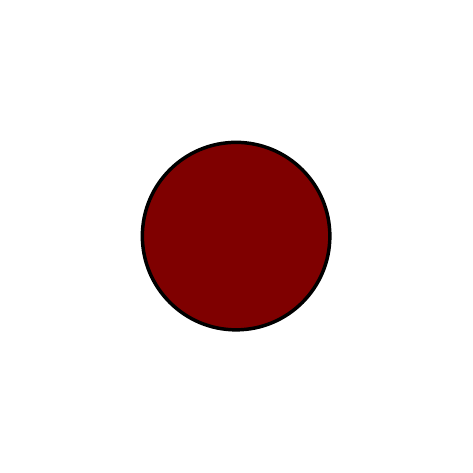}
\caption{\centering \textsc{SmallClique}} 
\end{subfigure}
\begin{subfigure}[h]{0.12\textwidth}
  \includegraphics[width=\textwidth]{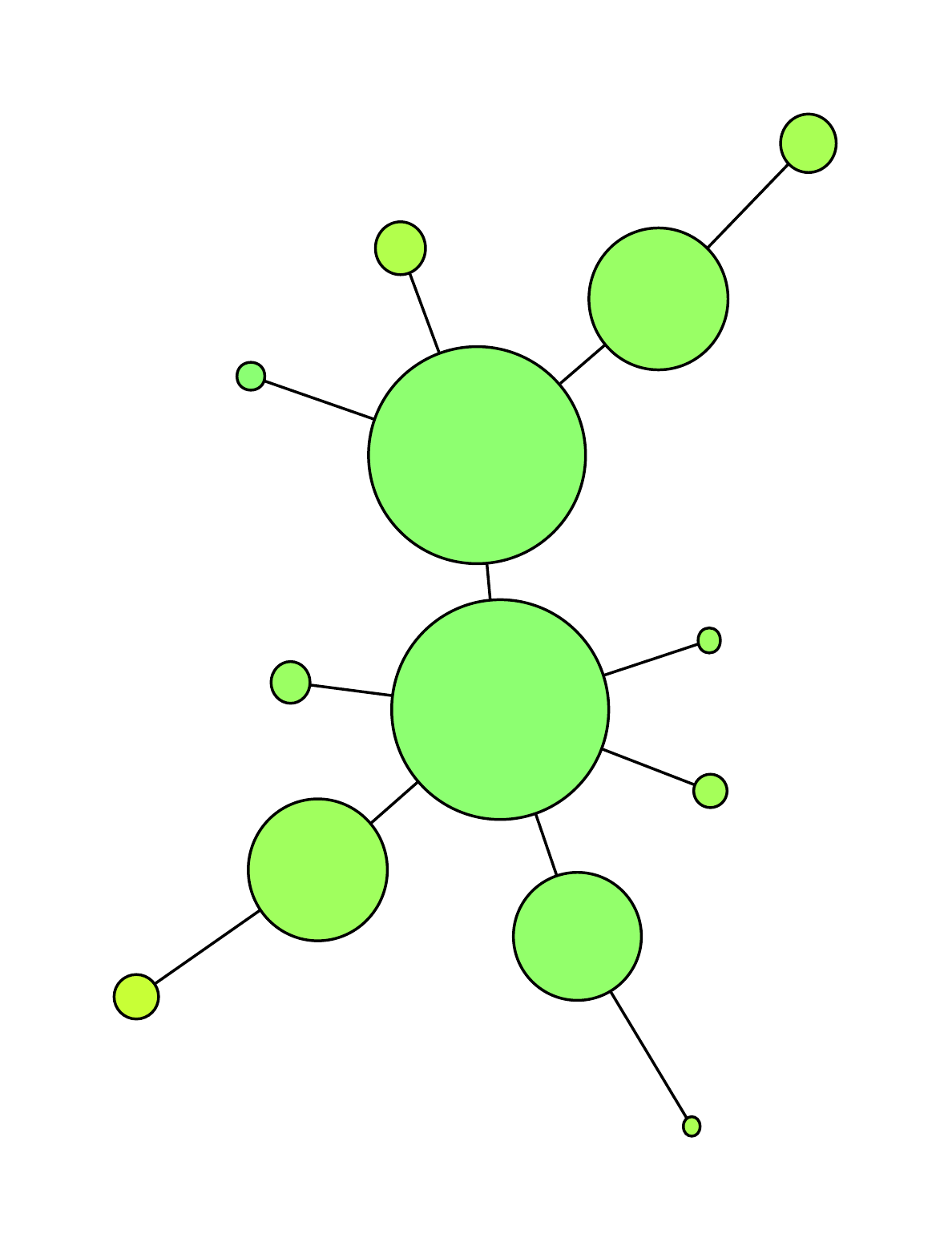}
\caption{\centering \textsc{SmallER}}
\end{subfigure}
\caption{\textsc{metnnd} (nested node dissection via METIS) TDs of toy networks.  Bags are colored by the fraction of possible edges present in the bag, with red being denser and blue being less~dense.}
\label{fig:toy_td_metnnd}
\end{center}
\begin{center}
\begin{subfigure}[h]{0.19\textwidth}
\includegraphics[width=\textwidth]{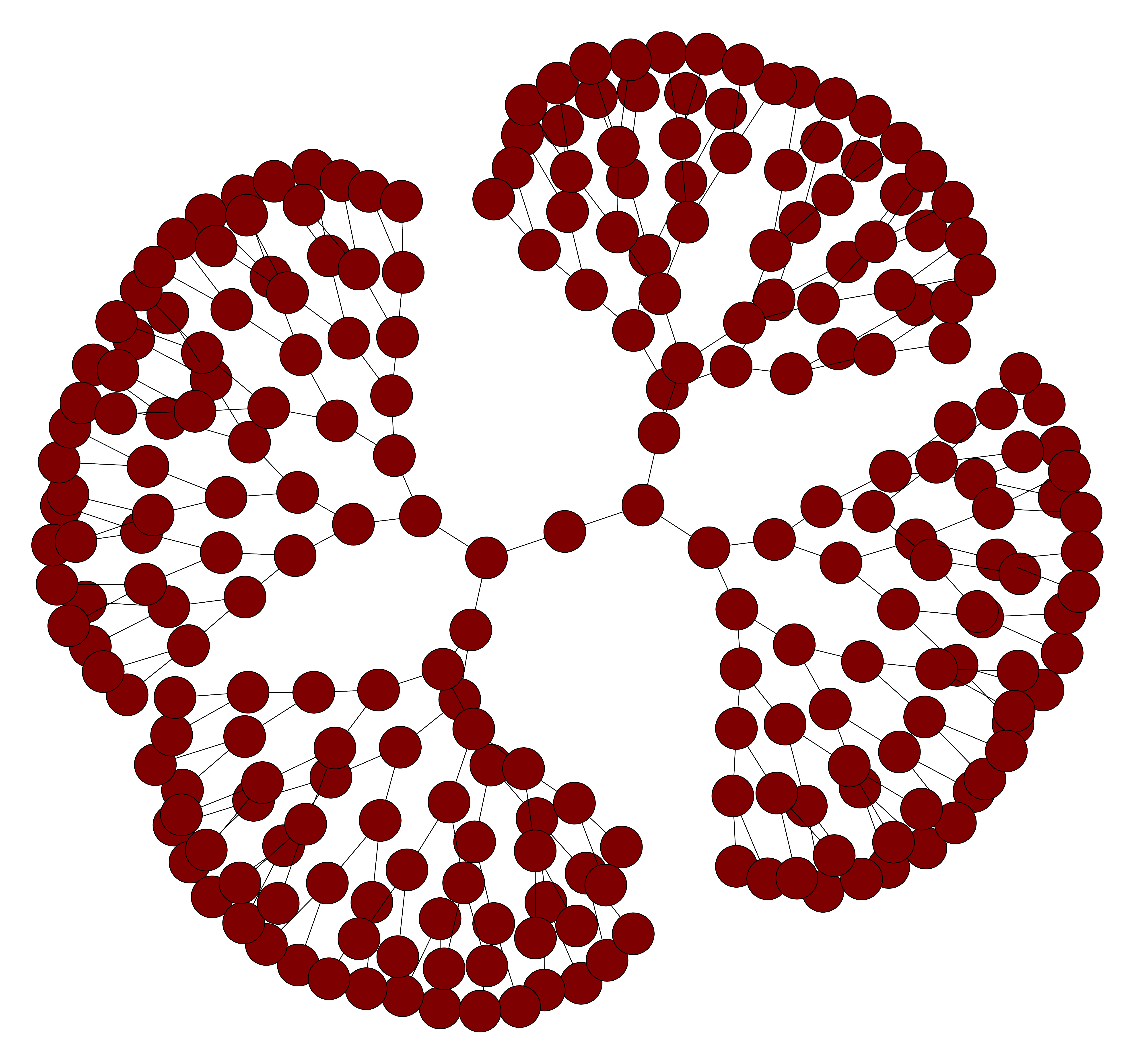}
\caption{\centering \textsc{SmallBinary}}
\end{subfigure}
\begin{subfigure}[h]{0.19\textwidth}
\includegraphics[width=\textwidth]{./figures/planar_10_10_grid_lexm_density.pdf}
\caption{\centering \textsc{SmallPlanar}}
\label{fig:lexm_10_10_plane}
\end{subfigure}
\begin{subfigure}[h]{0.12\textwidth}
\includegraphics[width=\textwidth]{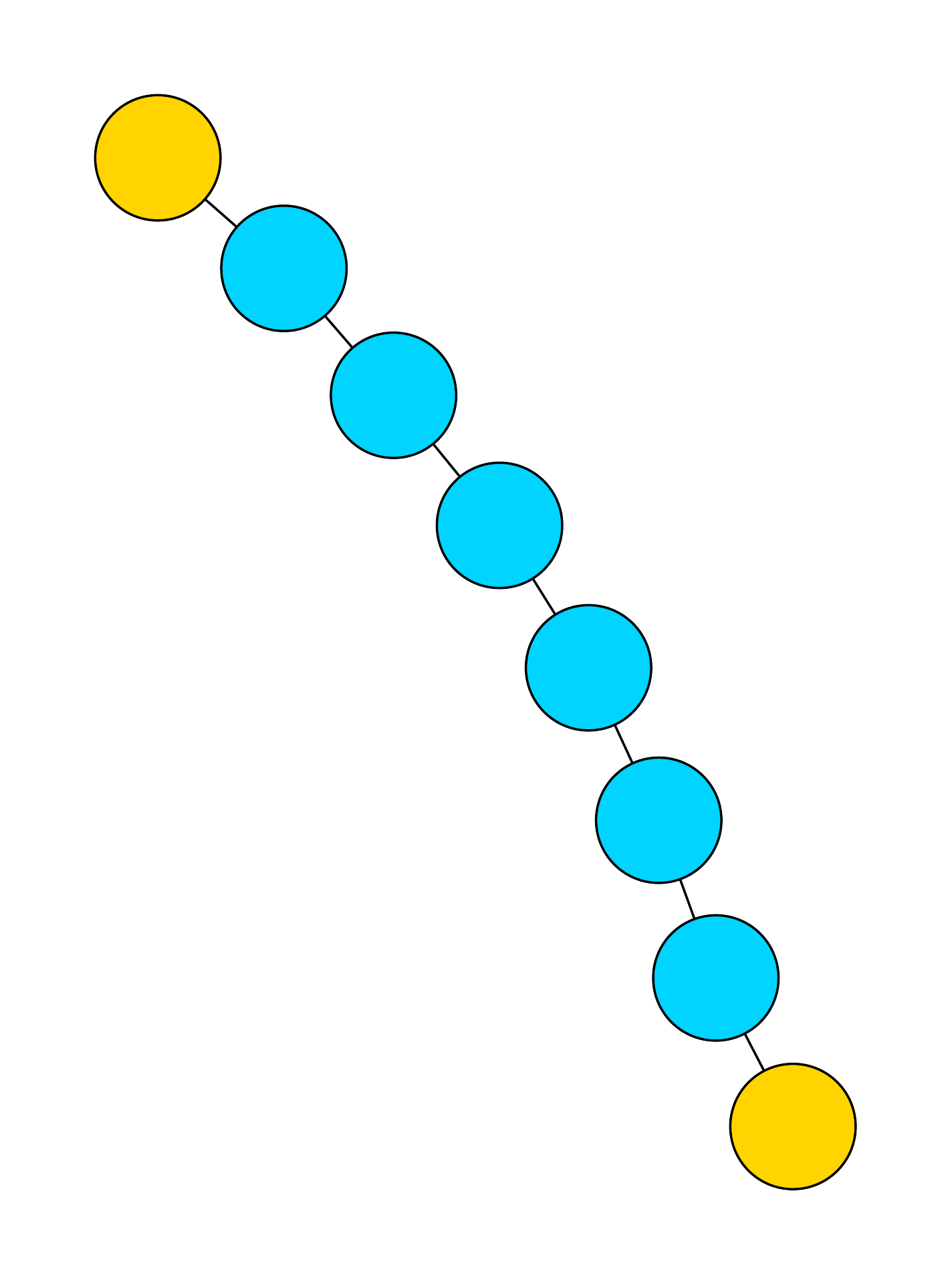}
\caption{\centering \textsc{SmallCycle}}
\end{subfigure}
\begin{subfigure}[h]{0.12\textwidth}
  \includegraphics[width=\textwidth]{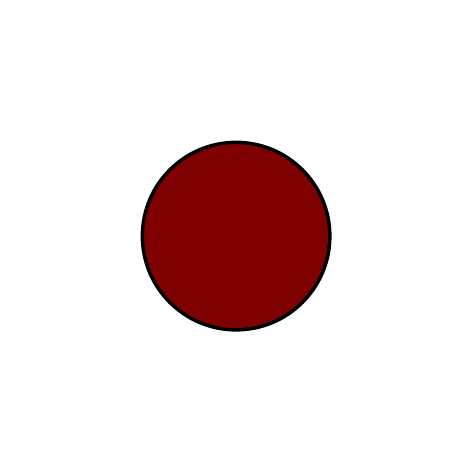}
\caption{\centering \textsc{SmallClique}} 
\end{subfigure}
\begin{subfigure}[h]{0.12\textwidth}
  \includegraphics[width=0.35\textwidth, angle=45]{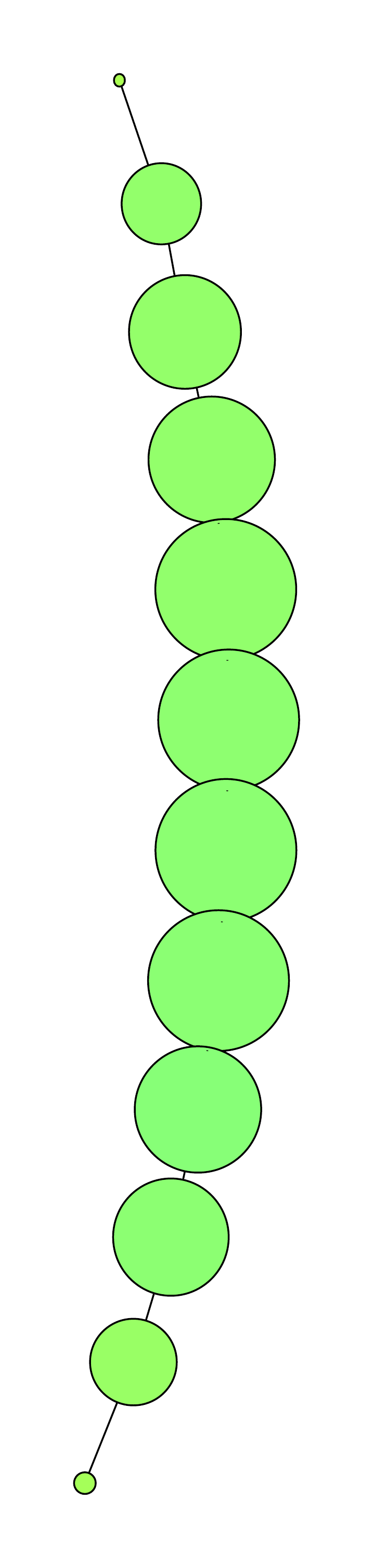}
\caption{\centering \textsc{SmallER}}
\end{subfigure}
\caption{\textsc{lexm} TDs of toy networks.  Bags are colored by the fraction of possible edges present in the bag, with red being denser and blue being less dense.}
\label{fig:toy_td_lexm}
\end{center}
\end{figure*}

More quantitatively, in Tables~\ref{tbl:toy_td_widths} and~\ref{tbl:toy_td_medians}, we
provide basic statistics for TD heuristics applied to each of these
networks.  Table \ref{tbl:toy_td_widths} shows a
summary of our results for the the (maximum) width of TDs produced by
various heuristics (for \textsc{SmallER}, the width given is averaged
over five different instantiations of the network).  We
ignore the issue of tie-breaker choices (e.g., in \textsc{mindeg},
choosing among non-unique minimum degree nodes).  On the whole, the
heuristics do a good job of finding optimal width TDs on
\textsc{SmallBinary}, \textsc{SmallClique}, and \textsc{SmallCycle}
(with the exception of \textsc{metnnd}).  The greedy heuristics have
trouble finding the optimal width TD on \textsc{SmallPlanar}, while
\textsc{lexm} and \textsc{mcs} both find an optimal decomposition on
the grid.  On \textsc{SmallER}, we observe that the greedy heuristics
and \textsc{metnnd} outperform \textsc{lexm} and \textsc{mcs}; this is
in agreement with previously-reported results~\cite{Bodlaender10}.

Table \ref{tbl:toy_td_medians} shows a summary of our results for the
\emph{median} width of TDs produced by various heuristics 
(as defined in Section~\ref{sec:tdprelim}).  The median
width is potentially more useful for revealing structure in realistic
network data since, e.g., it can be used to see whether a TD is
dominated by larger bags or by smaller bags.%
\footnote{Using medians rather than eccentricity can result in different central bags.  However in most of the networks that we studied, the results were very similar.  In particular, the biggest changes occurred in the FB networks where the median shifted towards the heavier end of the path-like TD.  However, these bags were still a part of the thick trunk of the network and thus the results were very similar.  In other networks, the median bag was very close the central eccentric bag, and the main difference is that the median bag tended to have more whisker branches (a branch consisting of one or two bags of small width).  This does not substantially change any of our analysis.}  
If a network is
dominated by bags of small size (such as \textsc{SmallBinary,
  SmallCycle, SmallPlanar}), depending on the internal structure of the
bag, this can indicate several things.  For example, the small bags
could consist of tight clusters or cliques, indicating that the
network has many tightly connected but \emph{small} groups of nodes.
Alternatively, if a small bag's structure is mostly disconnected, this
may indicate the bag is related to small cycles (an example is given~below).  

For \textsc{SmallCycle} and \textsc{SmallPlanar}, the small
bags are cyclical, while for \textsc{SmallBinary} the
small bags all consist of 2-cliques.  \textsc{SmallClique} and the
\textsc{SmallER} have large median widths (though this is trivial in
the case of the clique).  The 100-clique is both trivial and too large
of a clique to be realistic, but \textsc{SmallER} has interesting
bags.  
The results 
in Table~\ref{tbl:toy_td_widths} show the
small median widths of \textsc{SmallPlanar, SmallCycle} and
\textsc{SmallBinary} and the large median widths of \textsc{SmallER}
and \textsc{SmallClique}.  Table~\ref{tbl:toy_td_widths} also
demonstrates that, while there are differences in the widths of the TDs produced by the
heuristics, these differences are reasonably small.%
\footnote{We will see below that most real networks have 
small median width, with smallest bags dominated by cliques, 
intermediate bags dominated by cycles, and with large, connected, 
central bags which resemble bags of \textsc{SmallER}.}  

\begin{table}[h]
\begin{center}
{\footnotesize
\begin{tabular}{l|r|r|r|r|r|r|}
Network & $n$ & $W_{\textsc{mindeg}}$ & $W_{\textsc{minfill}}$ &
$W_{\textsc{nnd}}$ & $W_{\textsc{mcs}}$ & $W_{\textsc{lexm}}$
\\ \hline \hline 
\textsc{SmallBinary}  & 128 &   1 &   1 &   3 &   1 &   1 \\ 
\textsc{SmallPlanar}  & 100 &  13 &  13 &  14 &  10 &  10 \\ 
\textsc{SmallCycle}   &  10 &   3 &   3 &   3 &   3 &   3 \\ 
\textsc{SmallClique} & 100 & 100 & 100 & 100 & 100 & 100 \\ 
\textsc{SmallER}     & 100 &  86 &  85 &  86 &  91 &  89 \\
\end{tabular}
}
\end{center}
\caption{TD heuristic \emph{maximum} widths. The widths of the
  \textsc{SmallClique} and \textsc{SmallER} are relatively large
  (they grow linearly with the network size), the width of 
  \textsc{SmallPlanar} network is of an intermediate size (they 
  grow with the square root of the network size), and the widths of 
  \textsc{SmallBinary} and \textsc{SmallCycle} are small (they stay 
  constant with the size of the networks). The greedy heuristics 
  find smaller width decompositions on \textsc{SmallER}, while 
  \textsc{lexm} and \textsc{mcs} perform better on \textsc{SmallPlanar}. }
\label{tbl:toy_td_widths}
\end{table}

\begin{table}[!htb]
\begin{center}
{\footnotesize
\begin{tabular}{l|r|r|r|r|r|r|}
Network & $n$ & $\tilde{W}_{\textsc{mindeg}}$ & $\tilde{W}_{\textsc{minfill}}$ & $\tilde{W}_{\textsc{nnd}}$ &
$\tilde{W}_{\textsc{mcs}}$ & $\tilde{W}_{\textsc{lexm}}$ \\
\hline \hline
\textsc{SmallBinary}  & 128 &   1 &   1 &   1 &   1 &   1 \\
\textsc{SmallPlanar}  & 100 &   5 &   5 &   5 &  10 &   8 \\ 
\textsc{SmallCycle}   &  10 &   3 &   3 &   3 &   3 &   3 \\
\textsc{SmallClique} & 100 & 100 & 100 & 100 & 100 & 100\\
\textsc{SmallER}     & 100 &  52 &  51 &  49 &  85 &  80 \\
\end{tabular}}
\end{center}
\caption{TD heuristic \emph{median} widths.  This 
  quantity is much smaller than the corresponding widths in several 
  of the networks (although it remains large with \textsc{SmallER}),  
  indicating that these networks are dominated by bags 
  which are much smaller than the largest bag in the network. }
\label{tbl:toy_td_medians}
\end{table}

\begin{figure}[!htb]
\begin{center}
\begin{subfigure}[h]{0.25\textwidth}
\includegraphics[width=\textwidth]{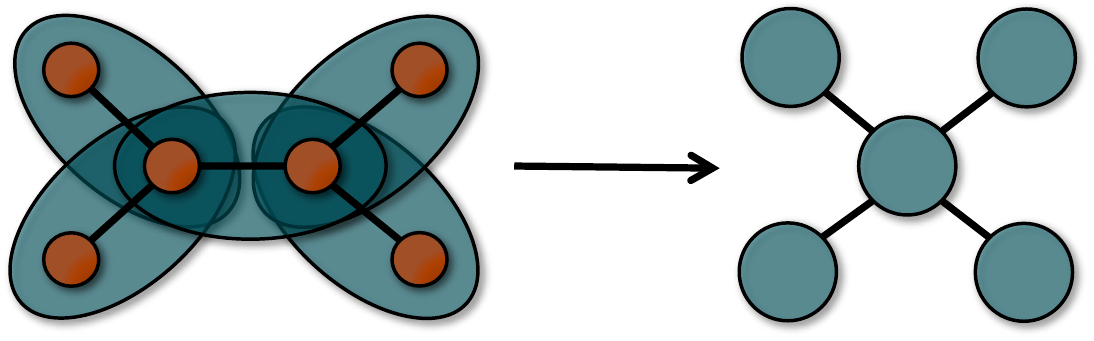}
\caption{ A tree (left) and a optimal TD (right).}
\label{fig:td_example-tree}
\end{subfigure}
\begin{subfigure}[h]{0.2\textwidth}
\includegraphics[width=\textwidth]{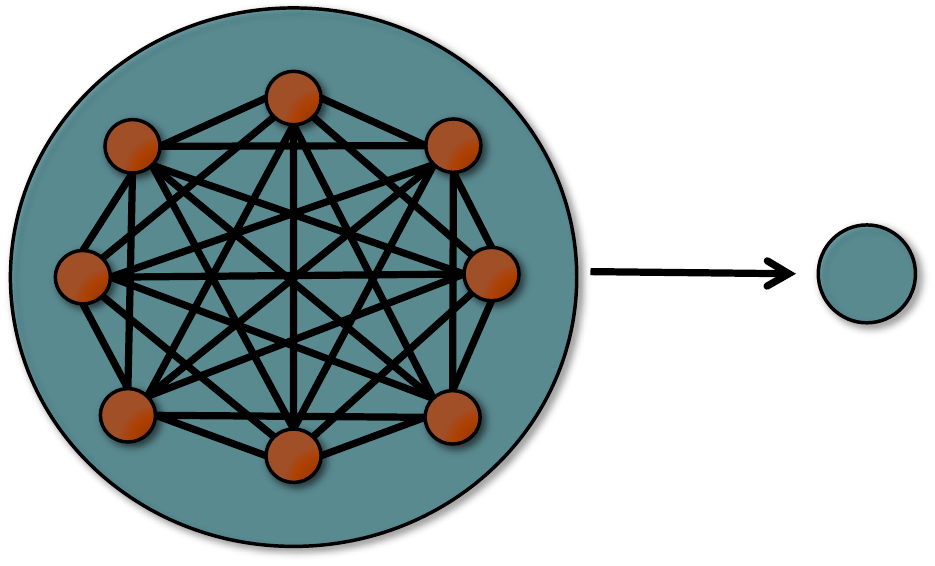}
\caption{ A clique (left) and a optimal TD (right).}
\label{fig:td_example-clique}
\end{subfigure}
\begin{subfigure}[h]{0.45\textwidth}
\includegraphics[width=\textwidth]{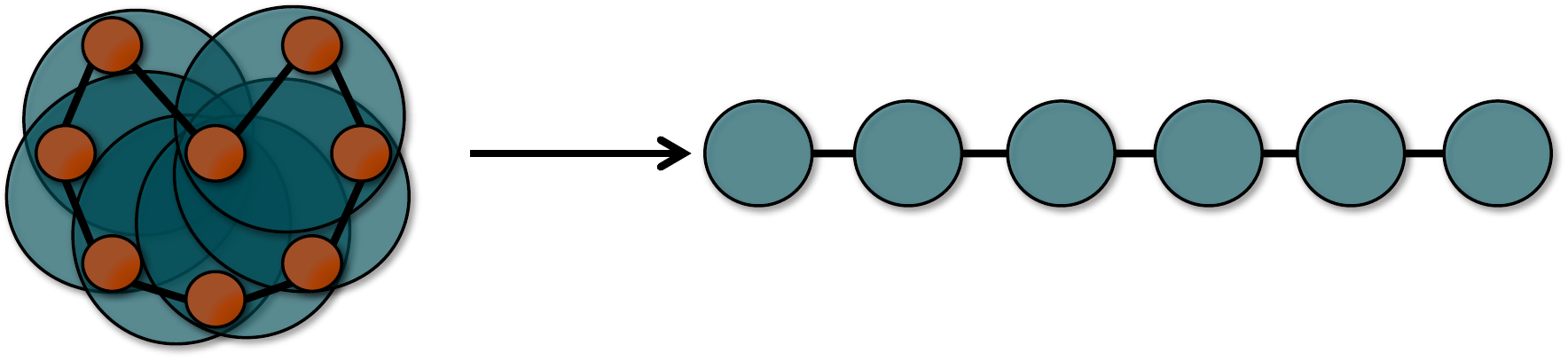}
\caption{An optimal TD of a cycle which  is similar
  to the decomposition found by \textsc{mindeg}.  The center node is
  placed in every bag of the decomposition.}
\label{fig:td_example-ring_mind}
\end{subfigure}
\begin{subfigure}[h]{0.45\textwidth}
\includegraphics[width=\textwidth]{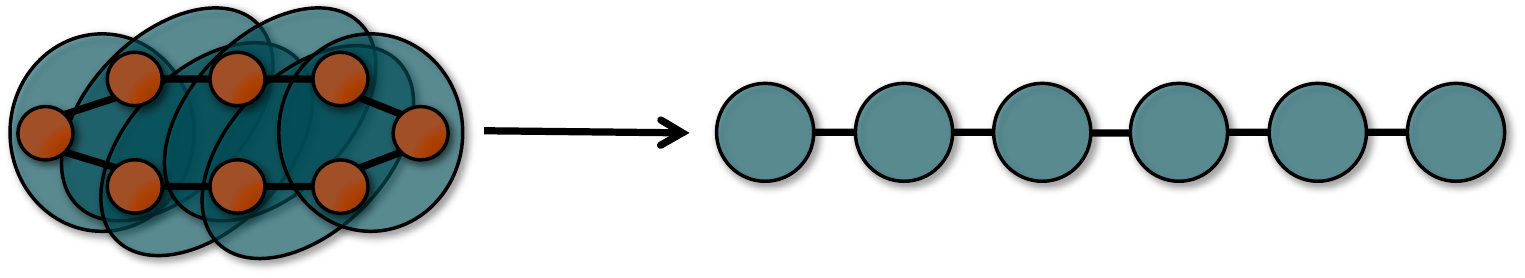}
\caption{ An optimal TD of the cycle which is similar
  to the decomposition found by \textsc{lexm}. Note
  the cycle is flattened and the bags are formed across the
  decomposition.}
\label{fig:td_example-ring_lexm}
\end{subfigure}
\begin{subfigure}[h]{0.45\textwidth}
\includegraphics[width=\textwidth]{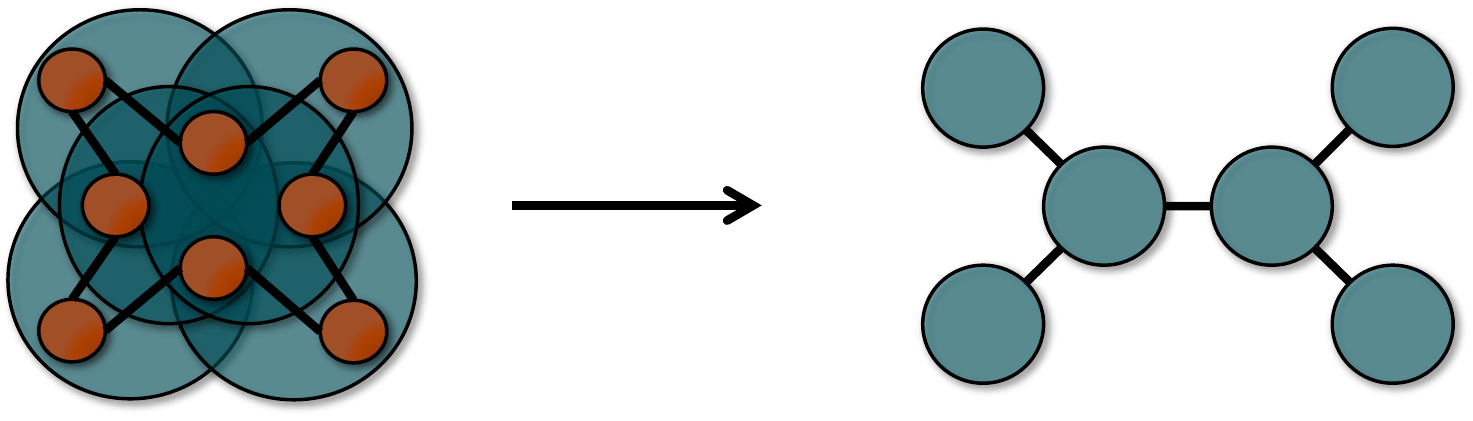}
\caption{ An optimal TD of the cycle which is similar to the
  decomposition found by \textsc{metnnd}.  The cycle is ``pinched'' in
  several places, forming central bags; the remaining pieces of
  the cycle can then be decomposed recursively (pinched in again) or
  using the methods in (b) and (c).}
\label{fig:td_example-ring_metnnd}
\end{subfigure}
\caption{Example TDs.  The tree and the
  clique have a standard optimal TD.  The cycle has many possible minimum width TDs, though all
  place disconnected nodes in the bags.}
\label{fig:td_example}
\end{center}
\end{figure}

\subsection{TDs on clique-like and cycle-like toy networks}

An important aspect of TD heuristics is the difference
between their behavior on (denser) clique-like graphs and (sparser)
cycle-like graphs.  In Figure \ref{fig:td_example}, we illustrate 
this.  First, for reference, in
Figures~\ref{fig:td_example-tree} and~\ref{fig:td_example-clique} we give 
canonical minimum width TDs for a tree and a clique.  To understand the difference between 
cycles and cliques, recall that there are many ways of producing
a TD on a cycle (three of these are illustrated in
Figures~\ref{fig:td_example-ring_mind},
\ref{fig:td_example-ring_lexm}, and \ref{fig:td_example-ring_metnnd}).
One simple way is to produce a tree which is a path.  This can be done
by taking a node $v$ and placing it and its two neighbors in a bag at
one end of the path.  Then, keeping $v$ in every bag, progress around
the cycle sequentially forming the next bag by including $v$ and the
two nodes of the next edge (see
Figure~\ref{fig:td_example-ring_mind}).  Another method produces a
path by ``flattening'' the cycle, and places each edge in a bag with
the node from the other side (see
Figure~\ref{fig:td_example-ring_lexm}).  The \textsc{metnnd} heuristic
``pinches'' the cycle at a few points, and the produces branches from
each of those points (see Figure~\ref{fig:td_example-ring_metnnd}).

There are many differences between these TD heuristics, but an important point is 
that the nodes in the cycle must be placed in bags with nodes they are \emph{not} neighbors of in the original graph.
Different TD heuristics are very different in terms of how 
they make this decision, and its effect can be seen in the TDs constructed 
by these heuristics.

Another important consideration is the interior structure of 
each bag that is produced by a TD heuristic.  
Recall that in \textsc{SmallClique}, the only valid TD (which does not contain unnecessary bags) 
is a single bag containing the entire network.  
Relatedly, if the network is a \emph{k-tree}, formed by overlapping cliques 
(rather than overlapping edges, as in a normal or 2-tree), then the TD will 
have bags which consist of the individual cliques.  
Thus, with cliques, it is the \emph{local structure} (local in the original graph, 
in the sense that it is driven by neighbors of a given node in the original 
graph, in contrast with what is going on in, e.g., \textsc{SmallCycle}) that 
drives the bag formation.
With cycles, on the other hand, this local structure is partially ``lost'' in 
the bags of the TD.  
This is of interest since, as already mentioned, the interior structure of 
bags of different widths is important for understanding what is creating the 
properties of the TD.

As an example, we observe that, for all of the heuristics, the larger bags 
on \textsc{SmallPlanar} have many disconnected nodes and only a few edges.  
This is a signature of ``cyclical'' behavior; and, 
indeed, from the TD perspective, the grid ``looks like'' a set of small, 
regular, overlapping cycles.  
The structure of the TD is formed by the heuristic's method of moving across 
the grid and closing cycles.  
This suggests that a simple metric to measure whether the interior of a bag 
is driven by cycles or is driven by small, tightly connected clusters: 
measure the fraction of edges present in the bag, i.e., the \emph{edge 
density} of the bag.
(We will do this below, and this is why we color many of the visualizations by the density of the bag.)

\subsection{TDs on well-partitionable and poorly-partitionable toy networks}

\textsc{SmallPlanar} (for which there exist good well-balanced partitions) 
and \textsc{SmallER} (for which there do not exist good well-balanced 
partitions) also illustrate differences between the TD heuristics.
For both graphs, the greedy heuristics and \textsc{metnnd} have significantly
smaller median widths than maximum widths.  
This is indicative of heterogeneity in the network: there 
are nodes which are so entangled with other nodes that they must appear 
together in a large bag,  
but there are also nodes which are connected to only a small number of other 
nodes and only need to appear in a few very small bags.  
This can partially be explained by the tendency of the greedy heuristics to 
work from the ``boundary'' of the graph (e.g., 
boundary nodes have smaller degrees) and to pick points to ``eat into'' the 
graph.

This is illustrated in Figure \ref{fig:bag_in_plane} for
\textsc{SmallPlanar}.  Using \textsc{mindeg} as an example, recall
that heuristic works by successively picking a minimum degree node in
the network; thus, when applied to \textsc{SmallPlanar}, it will pick
each of the corner vertices of the grid first.  This then forms small
bags at each corner and, depending on whether it is picking
non-unique nodes at random or in an ordered fashion, it will then
proceed to work in from the periphery of the network.  Indeed, in
Figure \ref{fig:toy_td_mind} and \ref{fig:toy_td_metnnd}, we see that
the TDs for these heuristics have four major arms with small leaves
containing nodes from the border of the grid.  Figure
\ref{fig:bag_in_plane} provides a visualization of where the nodes
from these bags (one of the peripheral bags and one of the central
bags in the TD) are in \textsc{SmallPlanar}.  (See, in particular
Figures~\ref{fig:mind_bag_upper} and~\ref{fig:mind_bag_central} for
\textsc{mindeg}.)

\begin{figure}[!htb]
\begin{center}
\begin{subfigure}[h]{0.20\textwidth}
  \includegraphics[width=\textwidth]{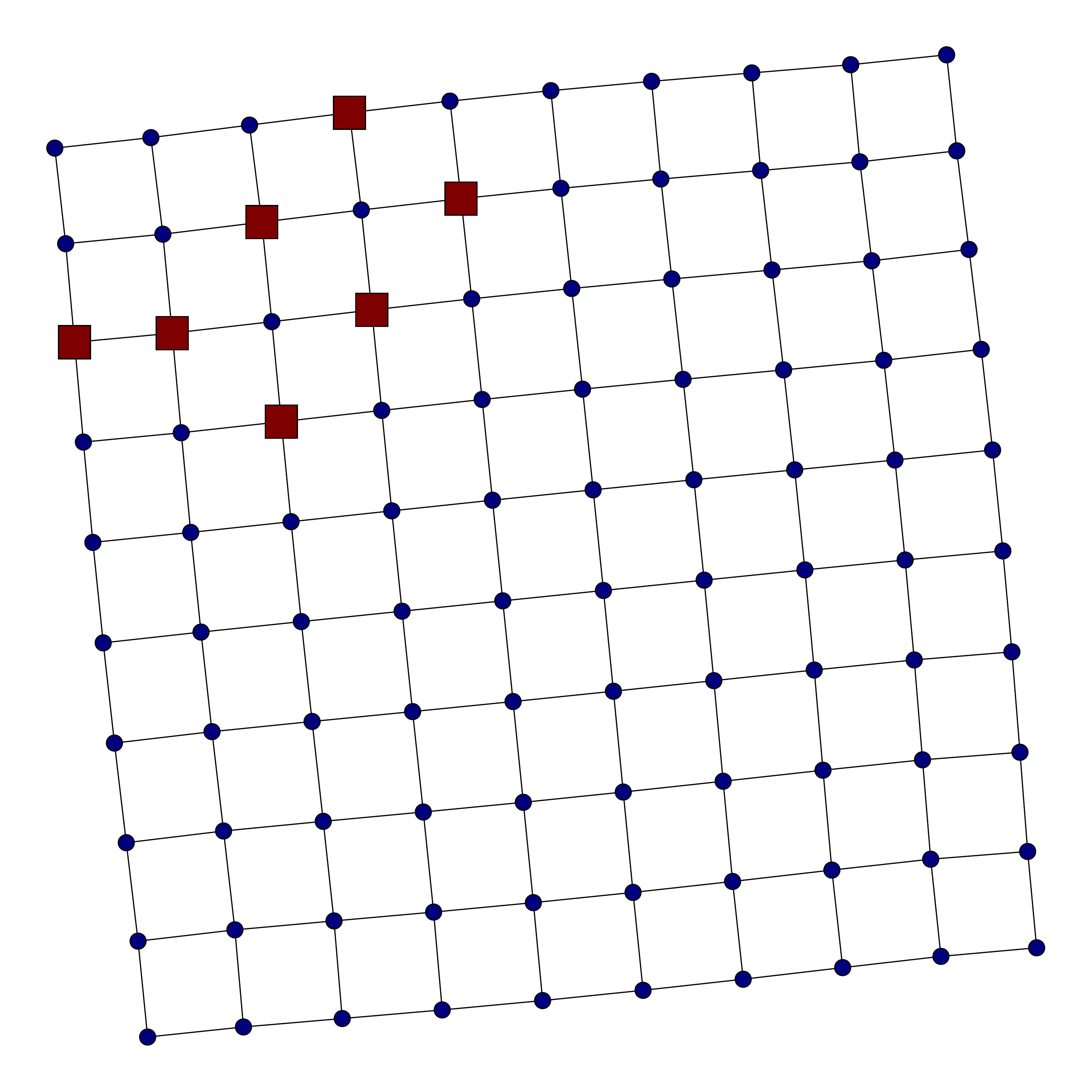}
  \caption{\centering \textsc{mindeg} bag from upper left arm of Figure
    \ref{fig:mind_10_10_plane}. }
  \label{fig:mind_bag_upper}
\end{subfigure}
\begin{subfigure}[h]{0.20\textwidth}
  \includegraphics[width=\textwidth]{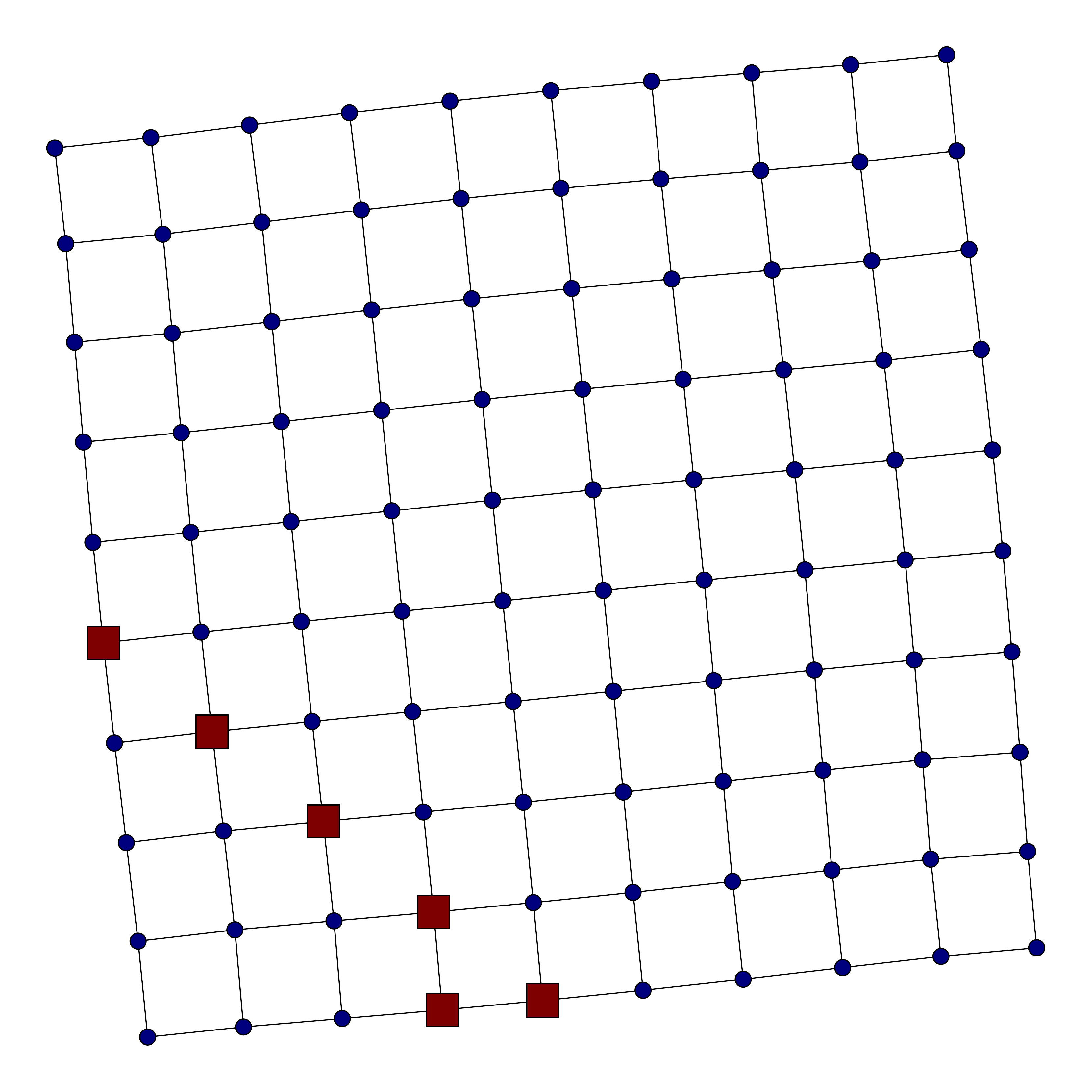}
  \caption{\centering \textsc{lexm} bag from lower right of Figure
    \ref{fig:lexm_10_10_plane}. }
  \label{fig:lexm_bag_lower}
\end{subfigure}
\begin{subfigure}[h]{0.20\textwidth}
  \includegraphics[width=\textwidth]{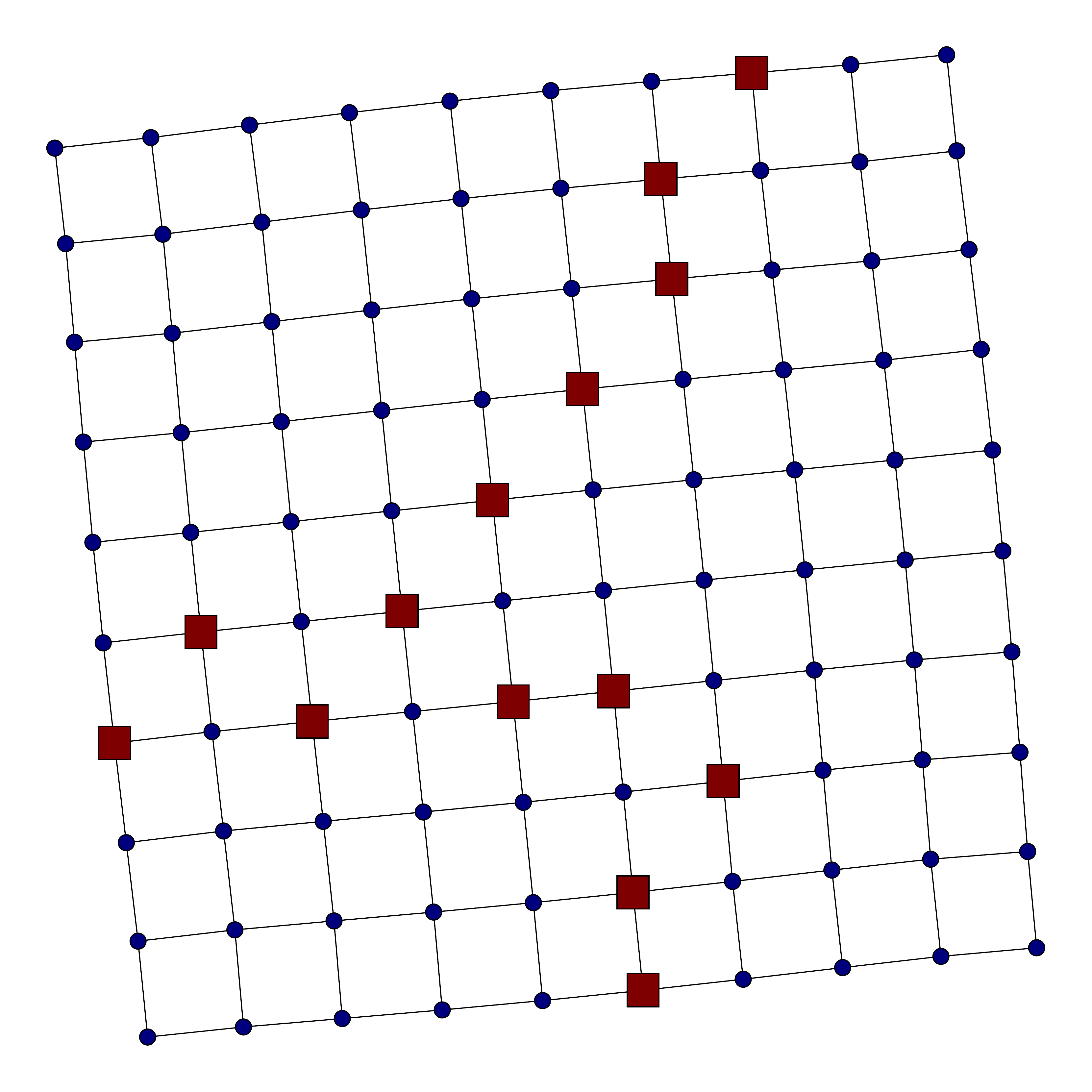}
  \caption{\centering \textsc{mindeg} central bag in Figure
    \ref{fig:mind_10_10_plane}. }
  \label{fig:mind_bag_central}
\end{subfigure}
\begin{subfigure}[h]{0.20\textwidth}
  \includegraphics[width=\textwidth]{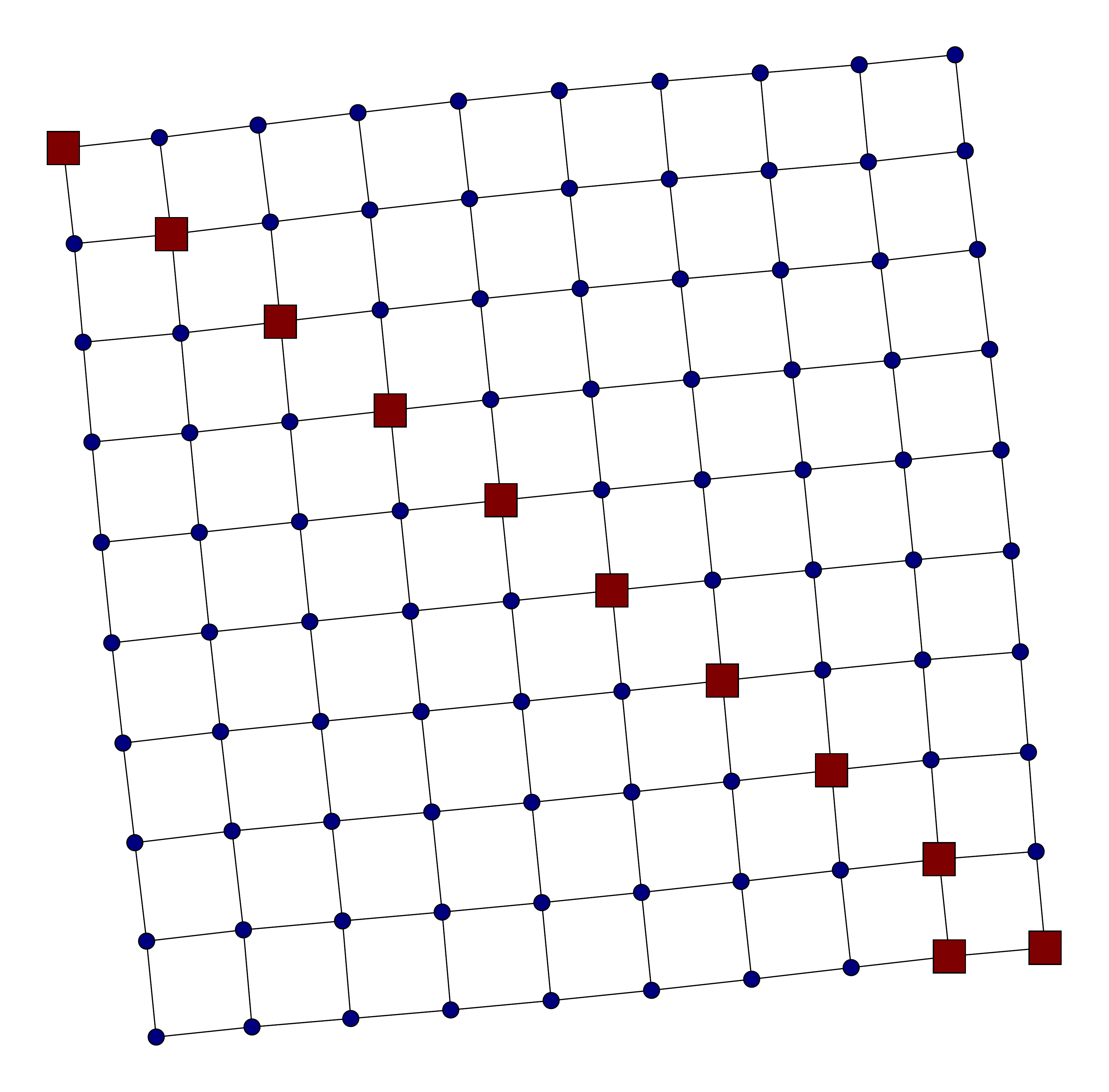}
  \caption{\centering \textsc{lexm} central bag in Figure
    \ref{fig:lexm_10_10_plane}. }
  \label{fig:lexm_bag_central}
\end{subfigure}
\caption{Representative bags from an arm in \textsc{mindeg} and an arm
  in \textsc{lexm}, as well as the associated central bags.  In
  \textsc{mindeg}, each arm progresses from a different corner of
  \textsc{SmallPlanar}. However, when these bag lines converge, the
  central bags end-up containing pieces of each line, as in Figure
  \ref{fig:mind_bag_central}.  In \textsc{lexm}, the line proceeds
  diagonally across the grid from the lower left corner to the upper
  right in a regular manner, as in Figure \ref{fig:lexm_bag_central}.  
  This results in smaller central bags and produces a path 
  decomposition. See the main text for more details. }
\label{fig:bag_in_plane}
\end{center}
\end{figure}

The \textsc{lexm} and \textsc{mcs} heuristics, in contrast, find the
minimum possible width for the grid, but the TDs---as illustrated in
Figure~\ref{fig:toy_td_lexm} for \textsc{lexm}---that are produced are
long, path-like trunks.  This is due to they way that \textsc{lexm}
picks a starting node and then works across the graph in the style of
a breadth-first search.  With \textsc{SmallPlanar}, it starts at one
corner and then moves across the network to form a minimal
triangulation.  Although the (maximum) width is minimal, the median
widths of these networks are relatively large, as most of the bags are
roughly the same size (see Figures~\ref{fig:lexm_bag_lower}
and~\ref{fig:lexm_bag_central} for the results of \textsc{lexm}).  

With \textsc{SmallER} (which is harder to visualize since it doesn't embed
well in two dimensions), the \textsc{mindeg} and \textsc{metnnd} algorithms 
also eat in from the ``boundary'' of the network, where here ``boundary'' 
means nodes with \emph{slightly} smaller degrees or \emph{slightly} better 
cuts (slightly smaller due to random fluctuations).  
As with the very different \textsc{SmallPlanar}, this produces several arms 
and then a few central bags.  
In \textsc{SmallER}, the greedy heuristics produce a better TD in terms of 
width than the \textsc{lexm} and \textsc{mcs} heuristics, both search based 
heuristics.  
These similarities and (substantial) 
differences between TD heuristics in \textsc{SmallER} (compared with 
\textsc{SmallPlanar}) are apparent in Figures~\ref{fig:toy_td_mind}, 
\ref{fig:toy_td_metnnd}, and~\ref{fig:toy_td_lexm}.

\subsection{Summary of TD results on toy networks}

Overall, the greedy heuristics, e.g., \textsc{mindeg} or \textsc{metnnd}, 
seem to produce a better representation of the large-scale structure of 
\textsc{SmallPlanar} and \textsc{SmallER} than the \textsc{lexm} and 
\textsc{mcs} heuristics in two ways.  In
\textsc{SmallER}, the greedy heuristics find decompositions with both
smaller maximum as well as smaller median widths.  (Since most real networks have a
randomized aspect to their generation, this indicates that greedy
heuristics may be more useful.)  On \textsc{SmallPlanar}, the median
width is smaller and the greedy heuristics do a better job of
``capturing'' all four corners of the grid.  In other words, the
resulting tree decomposition has four branches, each of which is
tied to a specific corner of the network, while \textsc{lexm} and
\textsc{mcs} TDs capture two of the corner structures.  (Although the
maximum width is smaller with \textsc{lexm} and \textsc{mcs}, the
ability to capture what is an obvious visual feature of a simple
network is of potential interest.) 
In the rest of the paper, we will be considering significantly larger and 
more complicated networks than these toy examples.  
With these larger networks, the \textsc{metnnd} and \textsc{amd} heuristics, 
as implemented using INDDGO~\cite{Groer12}, are the most scalable, compared 
with the basic greedy algorithms (\textsc{mindeg} or \textsc{minfill}).  
The \textsc{amd} heuristic is very related to the \textsc{mindeg} heuristic 
(recall that \textsc{amd} picks minimum nodes based on an easy-to-compute 
\emph{approximation} of node degree), and it gives similar results to 
\textsc{mindeg}.  
The the most consistent difference between the two heuristics 
seems to be the number of central/overlapping bags produced.
Thus, we will often show results only for the \textsc{amd} heuristic as a 
matter of visual convenience.

\section{Tree decompositions of synthetic networks}
\label{sec:synth_results}

In this section, we will describe the results of using a variety of TD 
heuristics on a set of synthetic networks.  
We focus our attention on two simple classes of random graphs: the popular 
Erd\H{o}s-R\'{e}nyi (ER) random graphs (in 
Section~\ref{sec:synth_results-er}); and a power law (PL) extension of the 
basic ER model (in Section~\ref{sec:synth_results-pl}).  
(We emphasize, though, that similar qualitative results also hold for many
other random graph models---in their extremely sparse regimes.)
This will allow us to begin to understand how TDs behave in random graph models 
with a very simple random structure.  
Importantly, we will focus on \emph{extremely} sparse graphs.  
For the ER model, this means values of the connection probability $p$ that 
lead to the graph not even being fully-connected (in which case we will 
consider the giant component), while for the PL model this means values of 
the degree heterogeneity parameter $\gamma$ that are typically used to
describe many realistic networks and that lead to analogously sparse graphs.

ER graphs are often presented as ``strawmen,'' since they obviously do not 
provide a realistic model for many aspects of real-world networks (e.g., the 
heavy-tailed degree distributions and the non-zero clustering coefficient 
present in many real networks). 
Indeed, ``vanilla ER'' graphs that are often considered (e.g., ER graphs 
with densities that are sufficiently large that the graph is 
fully-connected) are \emph{not} tree-like---either by the metric notion of 
$\delta$-hyperbolicity or by the cut-based notion of TDs.
Recent work has shown, however, that with respect to their large-scale 
structure, \emph{extremely} sparse ER networks do capture several subtle but 
ubiquitous properties of interest in realistic networks: 
first, the small-scale versus large-scale isoperimetric structure of the 
NCP~\cite{LLDM09_communities_IM,Jeub15}; second, a size-resolved version of 
$\delta$-hyperbolicity that is consistent with large-scale metric 
tree-likeness~\cite{Adcock13_icdm}; and third, a non-trivial core-periphery 
structure with respect to $k$-core decompositions~\cite{Adcock13_icdm}.  
(In particular, in the sparsest regime of the ER networks that we consider, 
\textsc{ER(1.6)}, a very shallow core-periphery structure appears---whereas 
none exists at the higher densities.)  
Importantly, for all three of these properties, similar results were seen 
with other random graph models, such as PL random graphs in the regime of the degree heterogeneity parameter 
that is commonly-used.  
Prior work has also provided evidence that these extremely sparse random 
graphs have non-trivial tree-like structure (at least relative to much denser 
ER graphs) when viewed with respect to the cut-like notion of
tree-likeness~\cite{Adcock13_icdm}.

Here, we provide a much more detailed analysis of this phenomenon for TD 
heuristics applied to ER and PL graphs.  
We will be particularly interested in similarities between extremely sparse 
ER graphs and PL graphs with respect to the core-periphery structure (e.g., from $k$-core and related decompositions) of a network.  
Among other things, we show that this core-periphery structure is captured with the \textsc{amd} TD.  
Of particular interest is the how the core-periphery structure relates to 
central (low eccentricity) or perimeter (high eccentricity) bags in the TD.

\subsection{TDs of ER Networks}
\label{sec:synth_results-er}

Here, we give a summary of the empirical results of an analysis 
of TDs on ER random graphs, with an emphasis on the behavior as the 
connectivity parameter $p$ is varied.  
In the very sparse to extremely sparse regime, \textsc{ER} 
networks have non-trivial global structural changes as $p$ is
varied~\cite{erdos60random,bollobas85_rg}.  
In particular, for our subsequent results, there are three regimes of $p$ 
that are of interest: if $p < \frac{1}{n}$, then the largest connected 
component is $O(\log{n})$ in size, and the small components are likely 
trees; if $\frac{1}{n} < p < \frac{\log{n}}{n}$, then the graph has a giant 
component (i.e., a constant fraction of the size of the network is 
connected), and the remaining small components of size $O(\log{n})$ are 
likely trees; and if $p > \frac{\log{n}}{n}$, then almost surely the 
network is fully-connected, the degrees are very near their expected value,
and there are no good-conductance clusters (of any size).  
We are interested in these last two regimes, and we consider synthetic 
graphs (\textsc{ER(1.6)} through \textsc{ER(32)}---values of $p$ between
$1.6/5000$ and $32/5000$, for graphs with $n=5000$ nodes) that go from
extremely sparse to somewhat denser.  
Table~\ref{tbl:networks-basic-stats} provides basic statistics for
these graphs.

\subsubsection{Visualization and basic statistics}

We start with Figure~\ref{fig:er_td_viz} and 
Table~\ref{tbl:er_density_results}, which show the basic features of
the TDs of \textsc{ER} networks.  
Figure~\ref{fig:er_td_viz} presents a visualization of part of the output of 
a TD with the \textsc{amd} heuristic, colored by density of bag subgraph, 
for the sparsest (\textsc{ER(1.6)}) and densest (\textsc{ER(32)}) networks 
in our ER suite.  
Results are similar to 
those of \textsc{metnnd}.  
Observe that there is a much greater heterogeneity in the density of bags 
for \textsc{ER(1.6)} than for \textsc{ER(32)}.  
For the former, there are many small bags which are cliques; while for the 
latter, there are fewer small bags, and the bags are much sparser in 
general.  
This suggests (and we have verified by inspection) that the sparser \textsc{ER(1.6)} has greater structural 
heterogeneity than the denser \textsc{ER(32)}.

A more detailed understanding of this can be obtained from the summary
statistics in Table~\ref{tbl:er_density_results}.  Several
observations are worth making.  First, the number of bags in the TD
tends to decrease as the density $p$ increases (with the exception of
the sparsest regime, where the giant component is smaller).  This
is because the network is mostly placed into one bag, 
and only a few bags are needed to take care of the
remaining nodes.  Second, the TD itself, viewed as a graph, has
smaller diameter as the density $p$ increases.  Third, the maximum and
median width increases with the average degree of the network.
Indeed, the width increases quickly with the average degree, with the
largest bag (at the ``center'' of the TD) containing $77\%$ of the
nodes in the network in \textsc{ER(32)}. Finally, the median density
of the bags decreases dramatically as the density of the original
graph increases.  This initially-counterintuitive phenomenon is
easily-explained: for extremely sparse ER, the
TDs have many small bags, which only need small numbers of edges to have
a reasonably high edge density.  With the dense graphs, many nodes
have to be placed in each bag, and this requires quadratically more edges per bag to achieve a similar density.  

\begin{figure*}[!htb]
\begin{center}
\begin{subfigure}[h]{0.45\textwidth}
\includegraphics[width=\textwidth]{./figures/5000_0_00016_amd_density_box.pdf}
\caption{\centering \textsc{ER(1.6)}, the largest bag in this figure
  contains 80 nodes.}
\end{subfigure}
\begin{subfigure}[h]{0.45\textwidth}
\includegraphics[width=\textwidth]{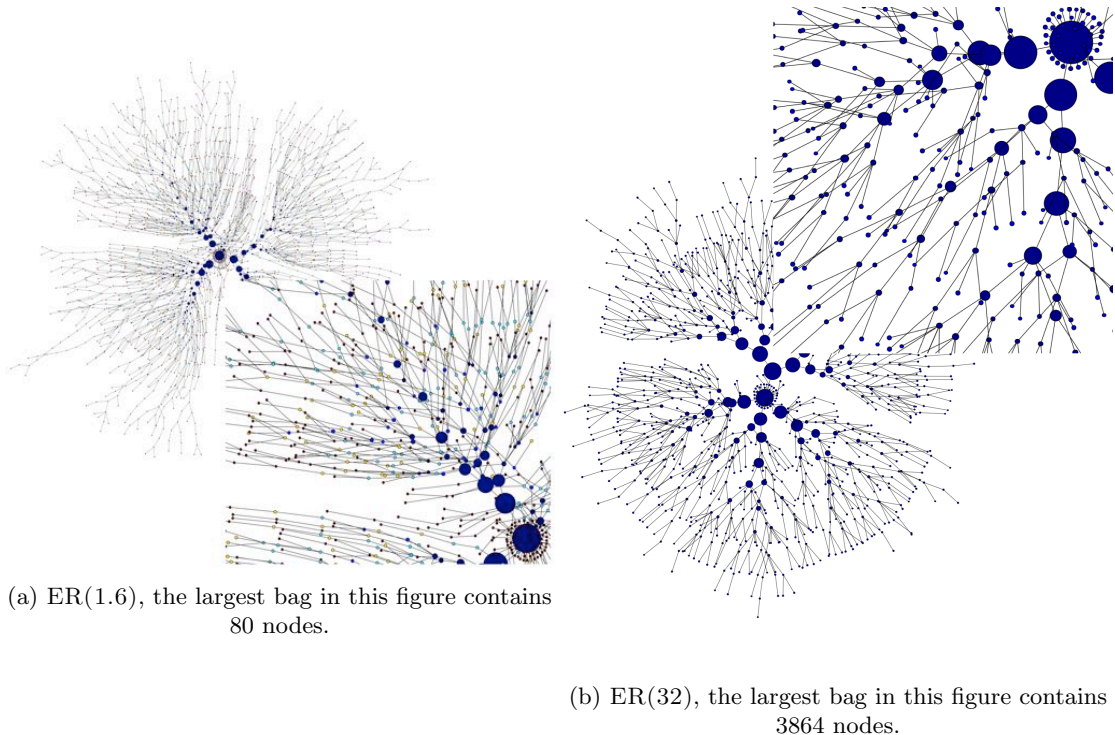}
\caption{\centering \textsc{ER(32)}, the largest bag in this figure
  contains 3864 nodes.}
\end{subfigure}
\caption{Visualization of \textsc{ER(1.6)} and \textsc{ER(32)}
  \textsc{amd} tree decompositions, colored by the density of the bag
  subgraph.  For visualization purposes, the two networks are not
  drawn to the same scale.  The bags in \textsc{ER(32)} have widths
  that are approximately 50 times larger that that of
  \textsc{ER(1.6)}.  The blowups show the upper-left corner of the
  visualization in greater details.  The blowups show the color of
  some of the smaller bags that are in the peripheral part of each
  TD. In \textsc{ER(1.6)}, many of the very the small bags are red
  (meaning they contain a clique, the vast majority of which are
  simply a single edge).  Slightly larger bags are light blue or
  yellow (indicating an edge density of ca. $0.25$ (light blue) to
  ca. $0.75$ (yellow)).  In \textsc{ER(32)}, all of the bags of the TD
  are dark blue, indicating that these bags are all very sparse,
  regardless of whether they are peripheral or central to the TD.
  The statistics in Table \ref{tbl:er_density_results} confirm this.}
\label{fig:er_td_viz}
\end{center}
\end{figure*}


\begin{table}
\centering
\begin{subtable}{.5\textwidth}
\centering

{\footnotesize
\begin{tabular}{l|r|r|r|r|r|}
Network & $N_{\textsc{amd}}$ & $E_{\textsc{amd}}$ & $W$ &
$\tilde{W}$ & $\tilde{D}$ \\
\hline \hline
\textsc{ER(1.6)} & 3127 & 44 &   79 &  1 & 1.0  \\
\textsc{ER(1.8)} & 3457 & 38 &  157 &  1 & 1.0  \\
\textsc{ER(2)}   & 3760 & 38 &  235 &  2 & 0.67 \\
\textsc{ER(4)}   & 3777 & 35 & 1093 &  3 & 0.40  \\
\textsc{ER(8)}   & 2787 & 29 & 2208 &  8 & 0.20  \\
\textsc{ER(16)}  & 1856 & 28 & 3142 & 17 & 0.12  \\
\textsc{ER(32)}  & 1136 & 22 & 3863 & 33 & 0.06 \\
\end{tabular}}

\caption{\textsc{ER} networks}
\label{tbl:er_density_results}
\end{subtable}%
\begin{subtable}{.5\textwidth}
\centering

{\footnotesize
\begin{tabular}{l|r|r|r|r|r|}
Network & $N_{\textsc{amd}}$ & $E_{\textsc{amd}}$ & $W$ &
$\tilde{W}$ & $\tilde{D}$ \\
\hline \hline
\textsc{PL(2.5)}  & 4672 & 32 & 219 & 1 & 1.0 \\
\textsc{PL(2.75)} & 4500 & 39 & 148 & 1 & 1.0 \\
\textsc{PL(3.0)}  & 3974 & 36 & 96  & 1 & 1.0 \\
\end{tabular}}

\caption{\textsc{PL} networks}
\label{tbl:pl_density_results}
\end{subtable}

\caption{Basic \textsc{amd} TD statistics for \textsc{ER} and \textsc{PL}
  networks. $N_{\textsc{amd}}$ gives the number of bags in the TD,
  $E_{\textsc{amd}}$ gives the maximum eccentricity
  (diameter) of the TD, $W$ and $\tilde{W}$ are the
  maximum and median width of the TD, and $\tilde{D}$ is the
  median bag density.}
\label{tbl:density_results}
\end{table}

We would next like to look in more detail at the
structure of the TDs generated on these different ER networks (e.g., what 
changes as we move from the central, large bags of the TD to the smaller, 
peripheral bags of the TD) as well as the internal structure of each bag.  
Recall, first, that, in a very sparse \textsc{ER} graph with expected degree greater 
than $2 \log 2$, but still sufficiently sparse, there are three different 
parts of the random network (two parts which may be viewed as core-like, one 
part which may be viewed as periphery-like)~\cite{Percus08}.  
The core-like part of these graphs is bi-connected, and it has an 
expander-like inner core (i.e., a set of nodes of ``higher'' degree), 
surrounded by an outer core which has long chains of nodes (forming sparse 
cycles).
The third, peripheral, part of the network consists of tree ``whiskers'' 
that hang off the biconnected core.  
A similar structure has been observed empirically when looking at 
low-conductance clusters/communities in a wide range of large social and 
information networks~\cite{LLDM09_communities_IM} and also when looking at 
the Gromov hyperbolicity and $k$-core properties of these real-world
networks~\cite{Adcock13_icdm}.  
This contrasts sharply with the denser \textsc{ER} graphs, which are much 
more regular in terms of their degree variability, core structure, etc.  
Our results (here on extremely sparse ER graphs and below on PL graphs and 
many real-world graphs) demonstrate that TD heuristics can reflect this 
core-periphery structure.

\subsubsection{Internal bag structure}

Next, Figure~\ref{fig:er_dense_bag}
presents visualizations of three typical \textsc{amd} bags for
\textsc{ER(1.6)} and \textsc{ER(32)}, respectively.  In each case, the
three bags are the most central (lowest eccentricity) bag in the TD
(which we call the central bag), a typical bag that is a leaf in the
TD (a periphery bag), and a typical bag that is in
between these two in the TD (an intermediate bag).  
The color-coding is by $k$-core number, with high
core nodes being red and low core nodes being blue.  
Note that the central bag
for \textsc{ER(1.6)} is disconnected and consists of almost all
singletons, while the central bag for \textsc{ER(32)} is
well-connected; and that the intermediate and peripheral bags for
\textsc{ER(1.6)} are small, the latter consisting of only a single
edge, while for \textsc{ER(32)} both the intermediate and the
peripheral bag have non-trivial internal structure.

\begin{figure}[!ht]
\begin{center}
\begin{subfigure}[h]{0.15\textwidth}
\includegraphics[width=\textwidth]{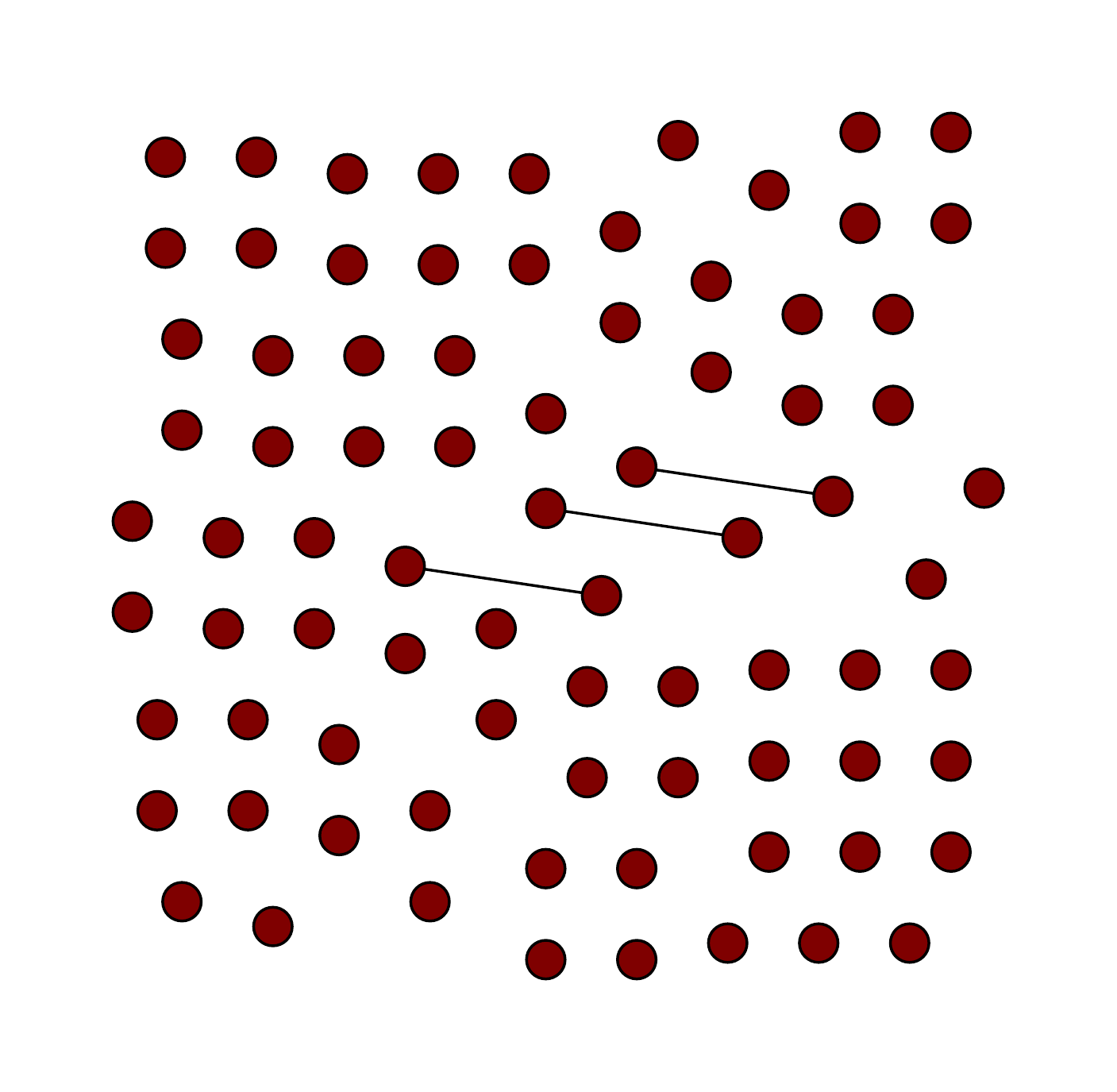}
\caption{\centering \textsc{ER(1.6)} central bag subgraph.}
\label{fig:er_sparse_bag-central}
\end{subfigure}
\begin{subfigure}[h]{0.15\textwidth}
\includegraphics[width=\textwidth]{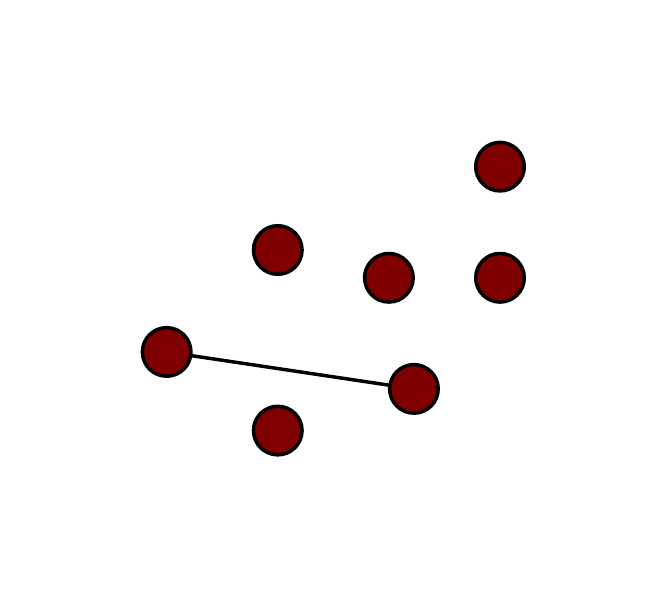}
\caption{\centering \textsc{ER(1.6)} intermediate bag subgraph.}
\end{subfigure}
\begin{subfigure}[h]{0.15\textwidth}
\includegraphics[width=\textwidth]{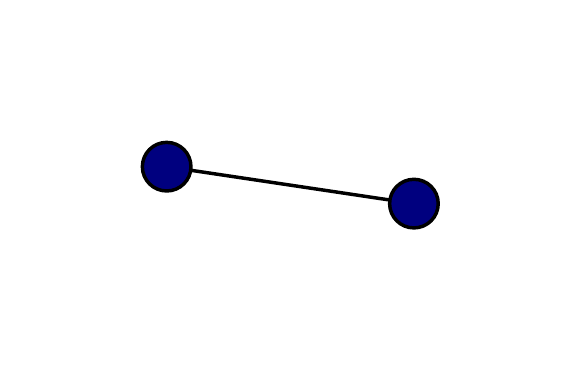}
\caption{\centering \textsc{ER(1.6)} peripheral bag subgraph.}
\end{subfigure}
\begin{subfigure}[h]{0.15\textwidth}
\includegraphics[width=\textwidth]{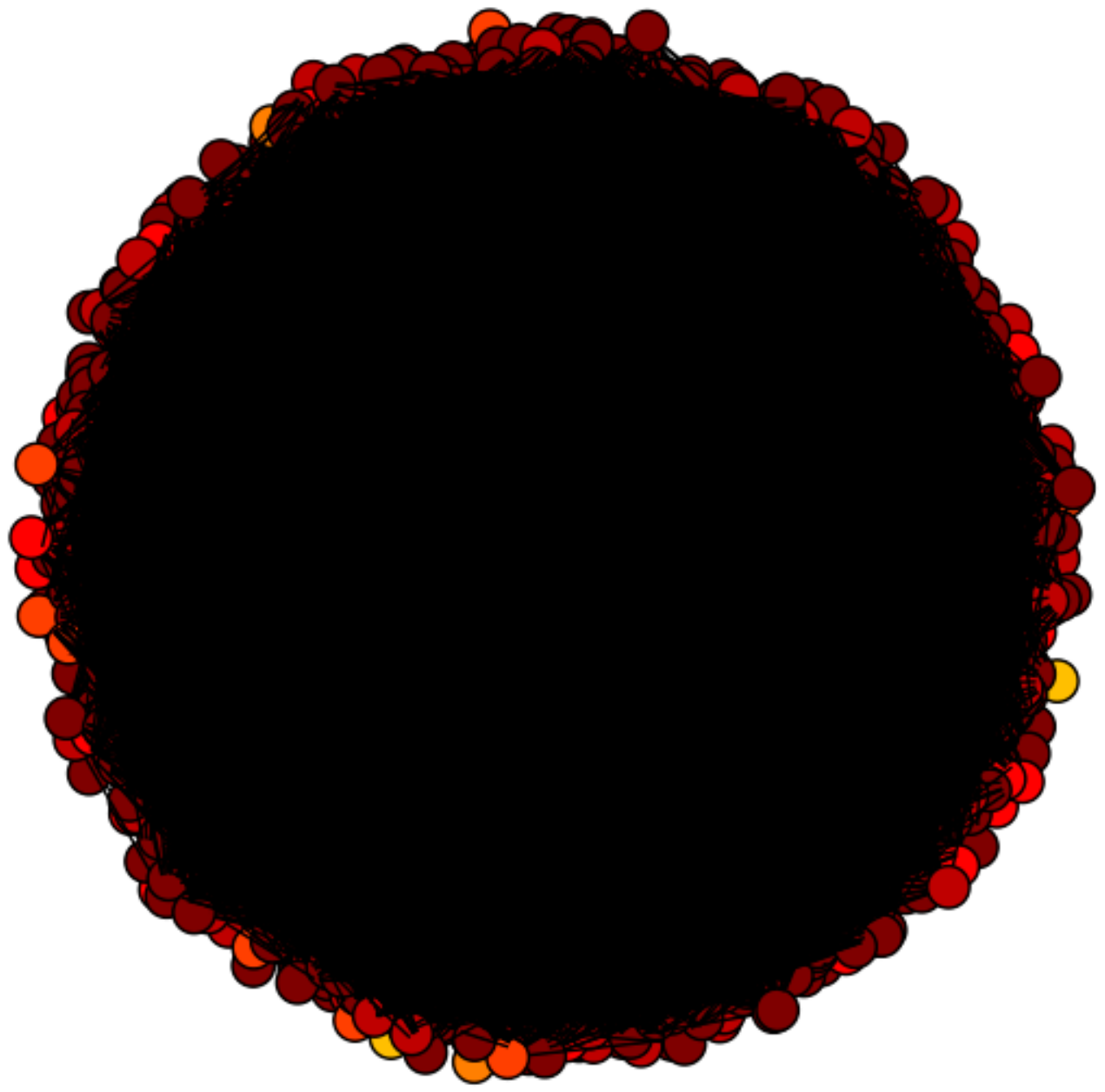}
\caption{\centering \textsc{ER(32)} central bag subgraph.}
\end{subfigure}
\begin{subfigure}[h]{0.15\textwidth}
\includegraphics[width=\textwidth]{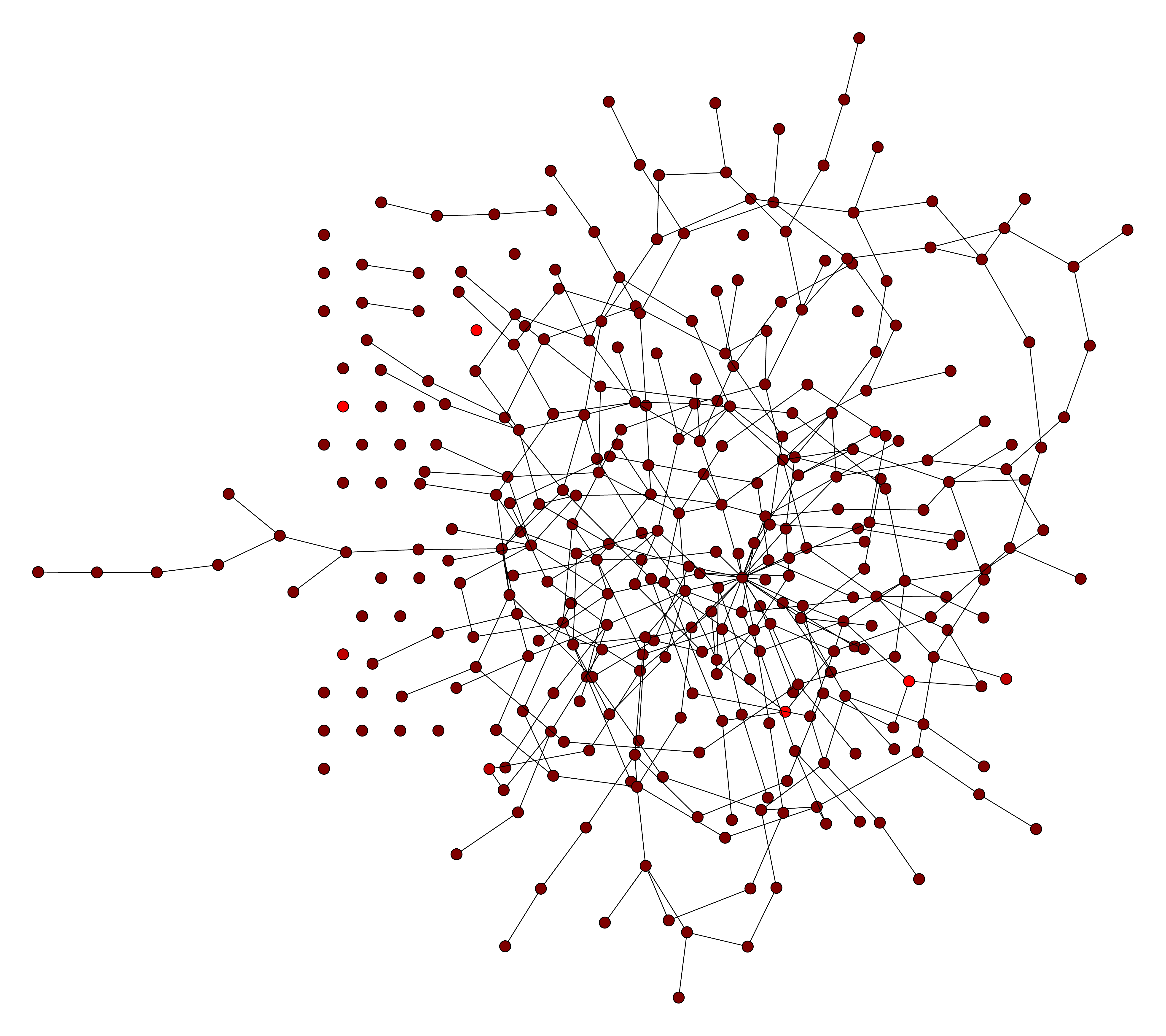}
\caption{\centering \textsc{ER(32)} intermediate bag subgraph.}
\end{subfigure}
\begin{subfigure}[h]{0.15\textwidth}
\includegraphics[width=\textwidth]{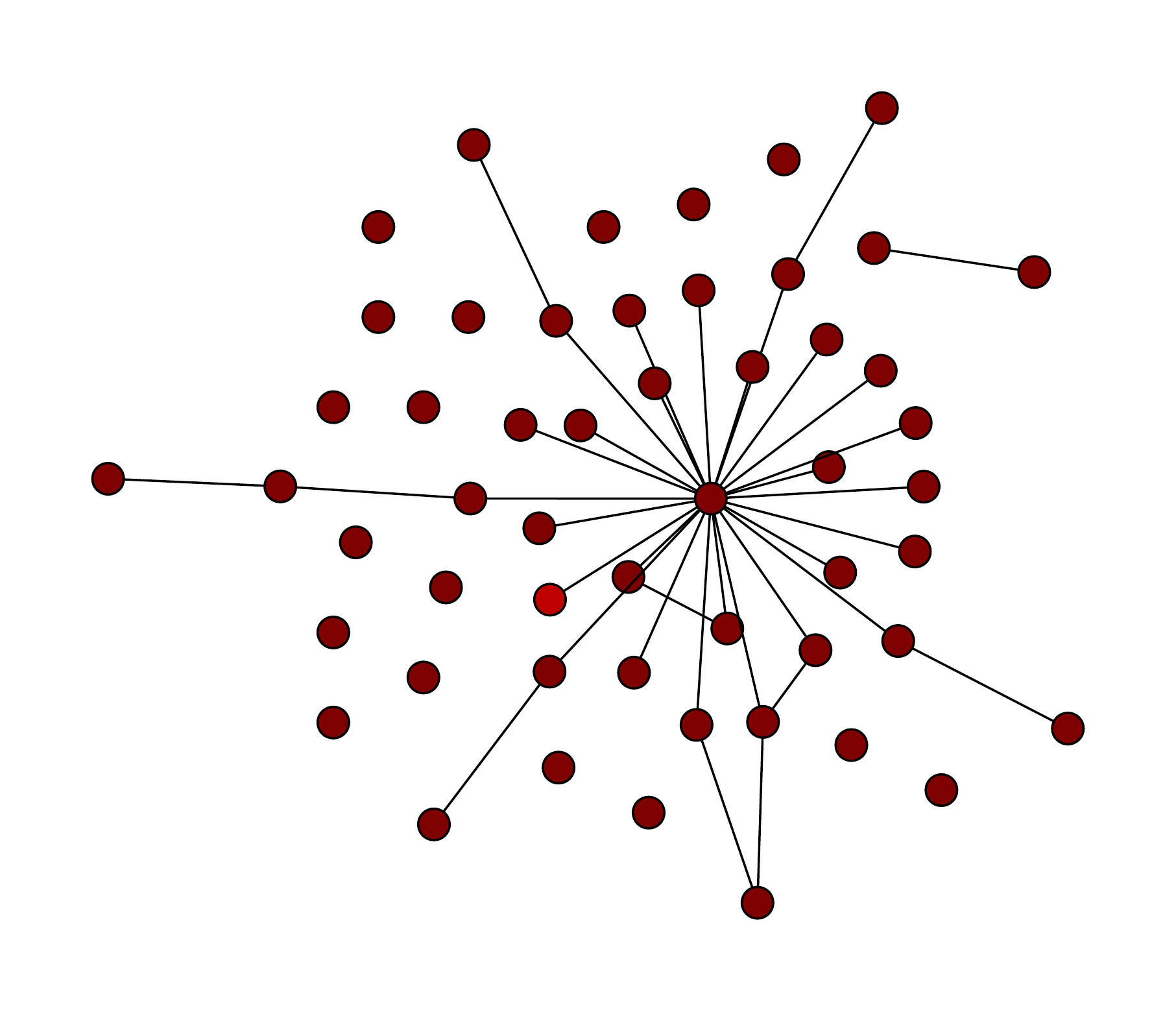}
\caption{\centering \textsc{ER(32)} peripheral  bag subgraph.}
\end{subfigure}
\caption{Bag subgraphs of a \textsc{amd} TD of \textsc{ER(1.6)} and \textsc{ER(32)}
  graphs, colored by the $k$-core number of the node (red is high $k$,
  blue is low $k$).  The central bag is the largest bag in the TD and
  one of the bags of minimum eccentricity; the peripheral bag is a
  leaf in the TD graph, and it achieves the minimum width in the TD;
  the intermediate bag is in between these two extremes.}
\label{fig:er_dense_bag}
\end{center}
\end{figure}

The increased density of \textsc{ER(32)} over \textsc{ER(1.6)} is the
obvious cause of these differences, but it is worth considering what
structures, produced by the increased density, affect the formation of
the \textsc{amd} TD.  Recall from the toy graphs that
heuristic TDs of cycles produced bags which had disconnected
nodes. There were several different ways of producing the
decomposition, but any TD of a small width on a cycle includes
disconnected nodes in most bags.  The more complex
\textsc{SmallPlanar} has many small overlapping cycles.  In that case,
the heuristics have to put many nonadjacent nodes into a bag.  
Essentially, cycles force distant nodes into the same
bag, and many overlapping cycles will force many distant nodes
into the same bag.

This intuition suggests (and we have confirmed by inspection) that a bag 
with many disconnected nodes, as in the central bag of \textsc{ER(1.6)} 
shown in Figure~\ref{fig:er_sparse_bag-central}, is due to a large number of
overlapping cyclical structures.
The intermediate bags of \textsc{ER(1.6)} contain nodes from the long, 
overlapping cycles of the outer core (and as these cycles do not overlap as 
much in periphery, these bags have fewer nodes), while the peripheral bags 
each contain a single edge, capturing the small trees on the periphery of
the network (see also Figure \ref{fig:er_td_viz}).  
The coloring of the nodes indicates the core-periphery structure of the 
subgraph induced by the bags.  
In \textsc{ER(1.6)} there is only a 1-core (blue) and a 2-core (red), thus 
the red nodes in the central bags are all in the 2-core, while 
the peripheral trees are in the 1-core, which agrees with \cite{Percus08}.

On the other hand, in \textsc{ER(32)}, whose core-periphery structure spans 
from a 7-core (blue) to a 23-core (red) although almost all of the nodes 
(94\%) are in the 23-core, the central bag contains a relatively 
tightly-connected mass of 77\% of the nodes in the network.  
This begins to look more like \textsc{SmallER}, which is a very dense ER 
network.  
The intermediate bags contain sparser structures (with some of the 
disconnected nodes and edges that are indicative of cyclical structures); 
and, although the peripheral bags still contain the smallest structures, in 
\textsc{ER(32)} they no longer contain only a single edge.
This indicates that even the sparsest regions contain cycles and other 
complicated structures (but very few triangles, which agrees with the small 
clustering coefficient of these networks).

\subsubsection{Large-scale organization}

To provide a more quantitative evaluation of these ideas and to 
characterize better the large-scale organization of these synthetic 
networks, consider Figures~\ref{fig:bag_hist}, \ref{fig:bag_density},
and~\ref{fig:bag_k_core}.  
These figures plot bag cardinality histograms, average bag density 
versus bag cardinality (this is width + 1), and average 
$k$-core versus bag eccentricity for two ER networks (as well as a suite of 
PL and real-world networks).  We will refer to other subfigures below, but 
for now consider only Figures~\ref{fig:bag_hist-er_sparse} 
and~\ref{fig:bag_hist-er_dense},
Figures~\ref{fig:bag_density-er_sparse}
and~\ref{fig:bag_density-er_dense}, and
Figures~\ref{fig:bag_k_core-er_sparse}
and~\ref{fig:bag_k_core-er_dense} for results on \textsc{ER(1.6)} and
\textsc{ER(32)}, respectively.

\begin{figure*}[!htb]
\begin{center}
\begin{subfigure}[h]{0.24\textwidth}
  \includegraphics[width=\textwidth]{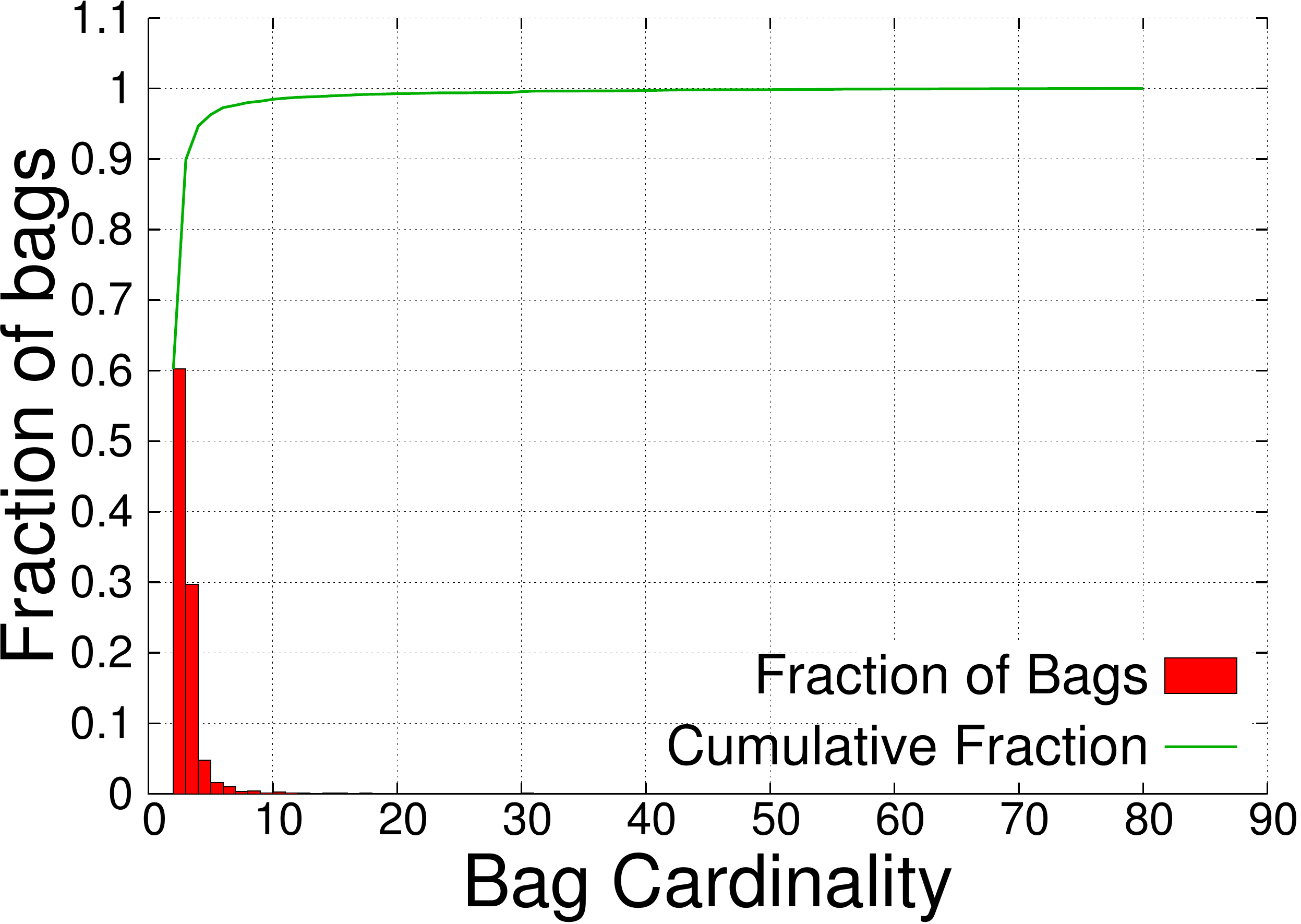}
\caption{\centering \textsc{ER(1.6)}.}
\label{fig:bag_hist-er_sparse}
\end{subfigure}
\begin{subfigure}[h]{0.24\textwidth}
\includegraphics[width=\textwidth]{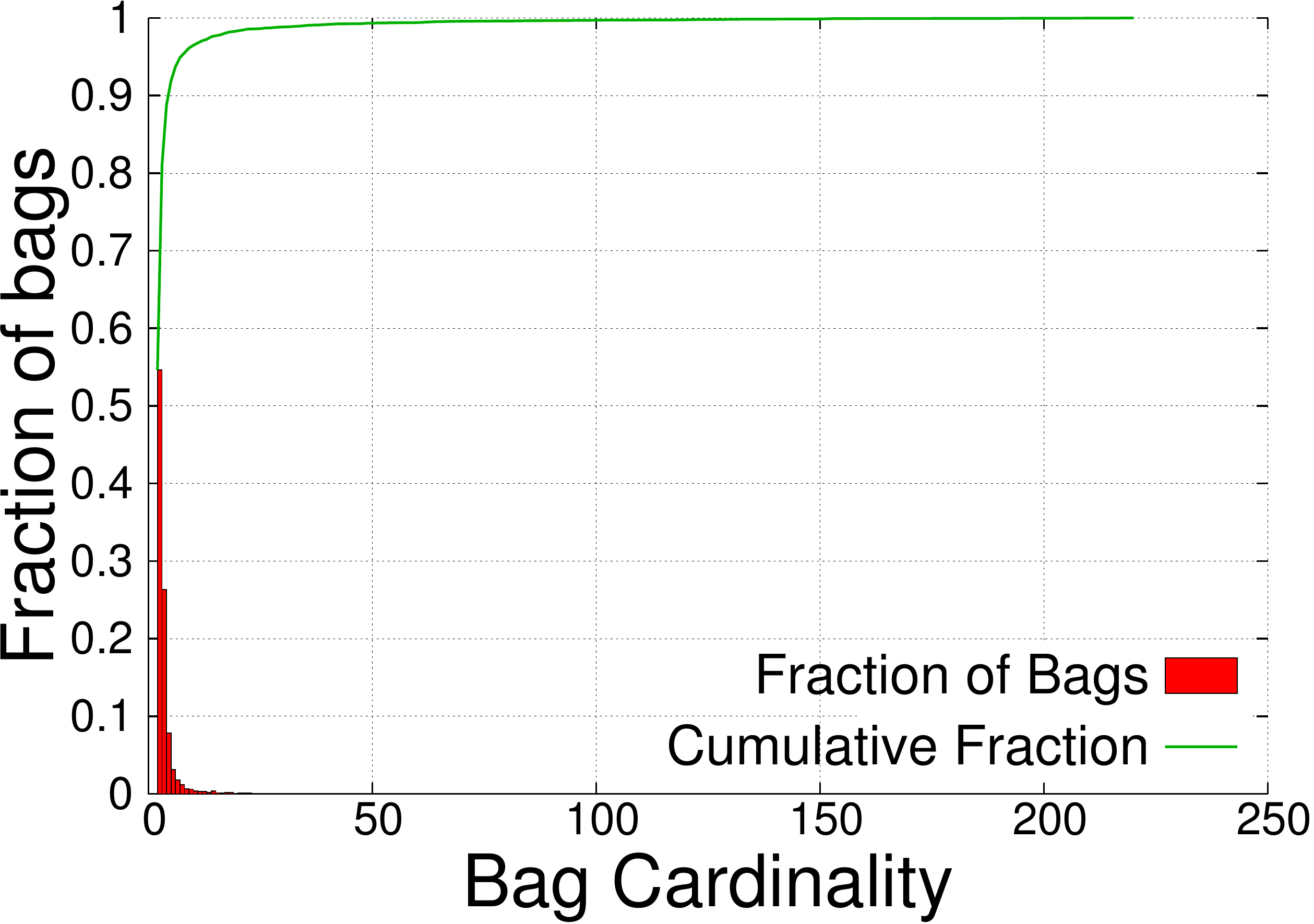}
\caption{\centering \textsc{PL(2.5)}.}
\end{subfigure}
\begin{subfigure}[h]{0.24\textwidth}
\includegraphics[width=\textwidth]{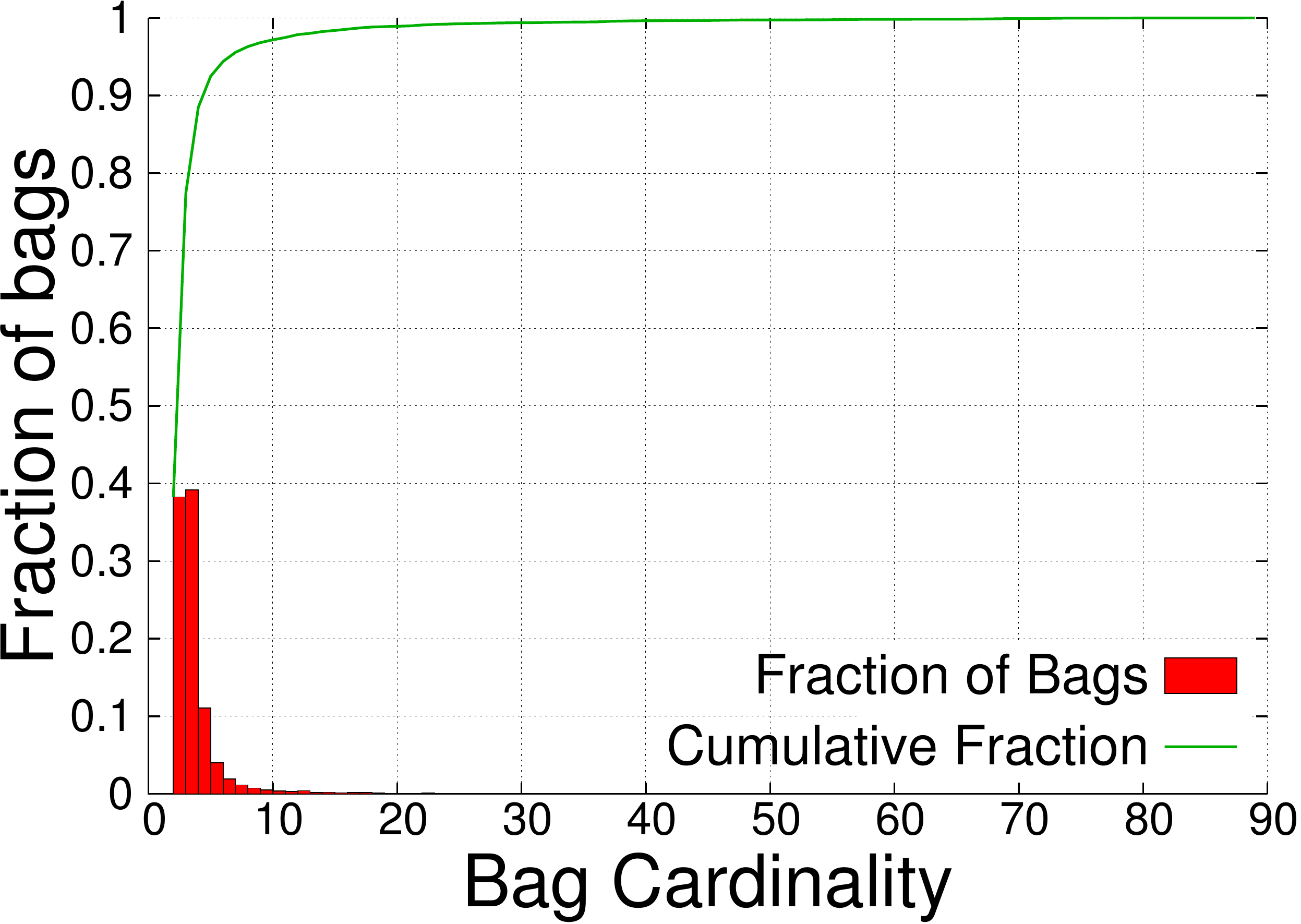}
\caption{\centering \textsc{as20000102}.}
\end{subfigure}
\begin{subfigure}[h]{0.24\textwidth}
\includegraphics[width=\textwidth]{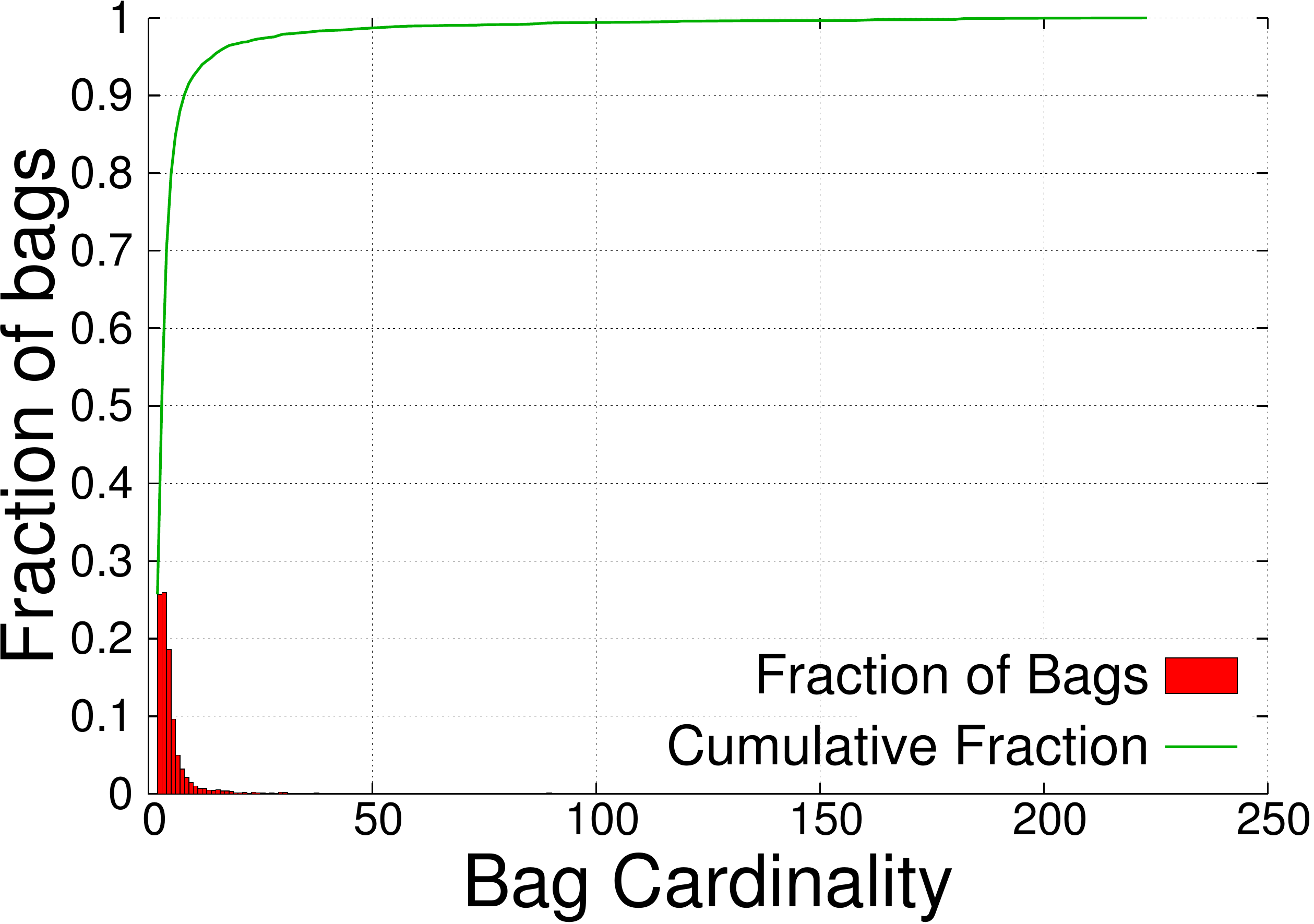}
\caption{\centering \textsc{CA-GrQc}.}
\end{subfigure}
\begin{subfigure}[h]{0.24\textwidth}
  \includegraphics[width=\textwidth]{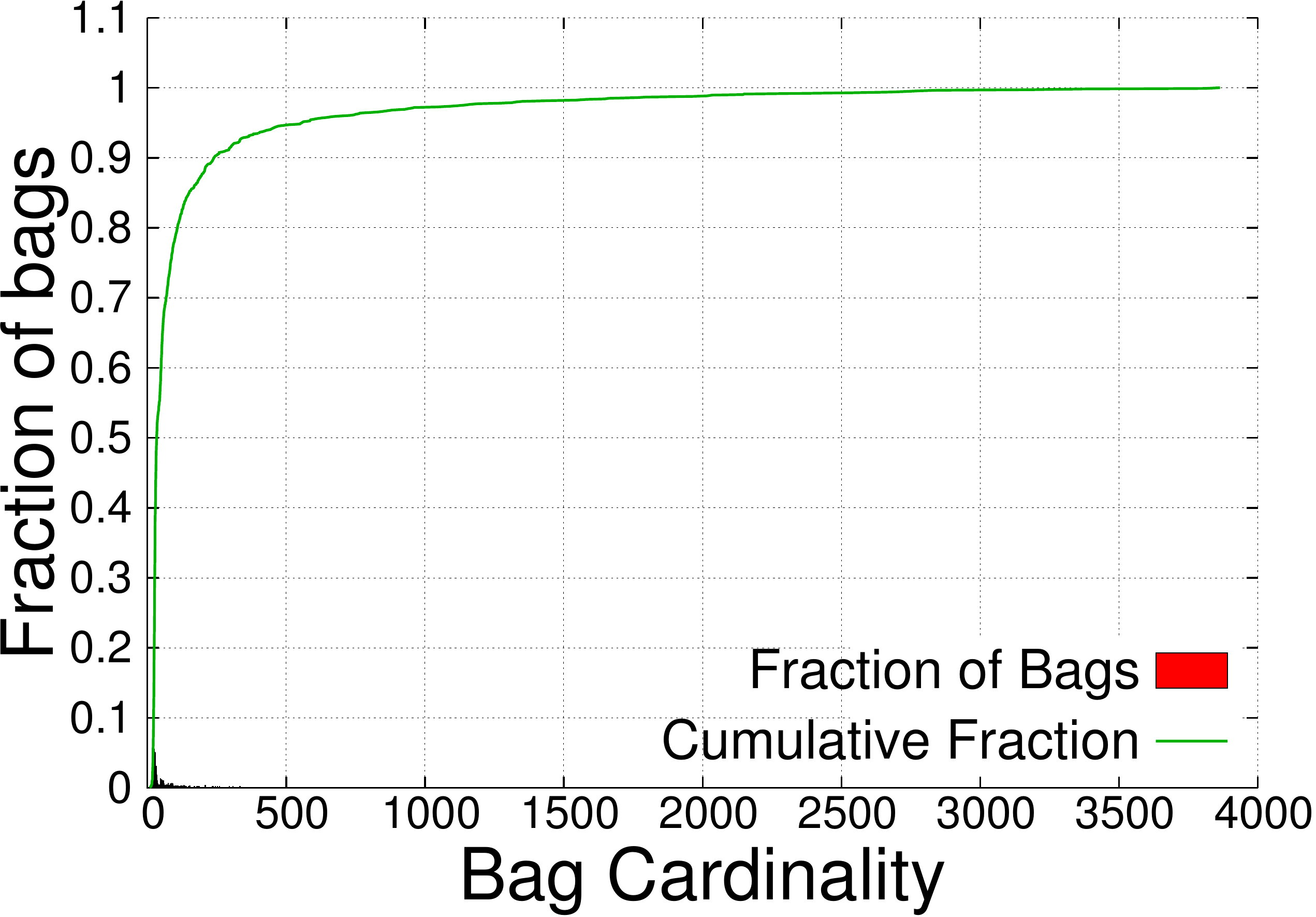}
\caption{\centering \textsc{ER(32)}.}
\label{fig:bag_hist-er_dense}
\end{subfigure}
\begin{subfigure}[h]{0.24\textwidth}
  \includegraphics[width=\textwidth]{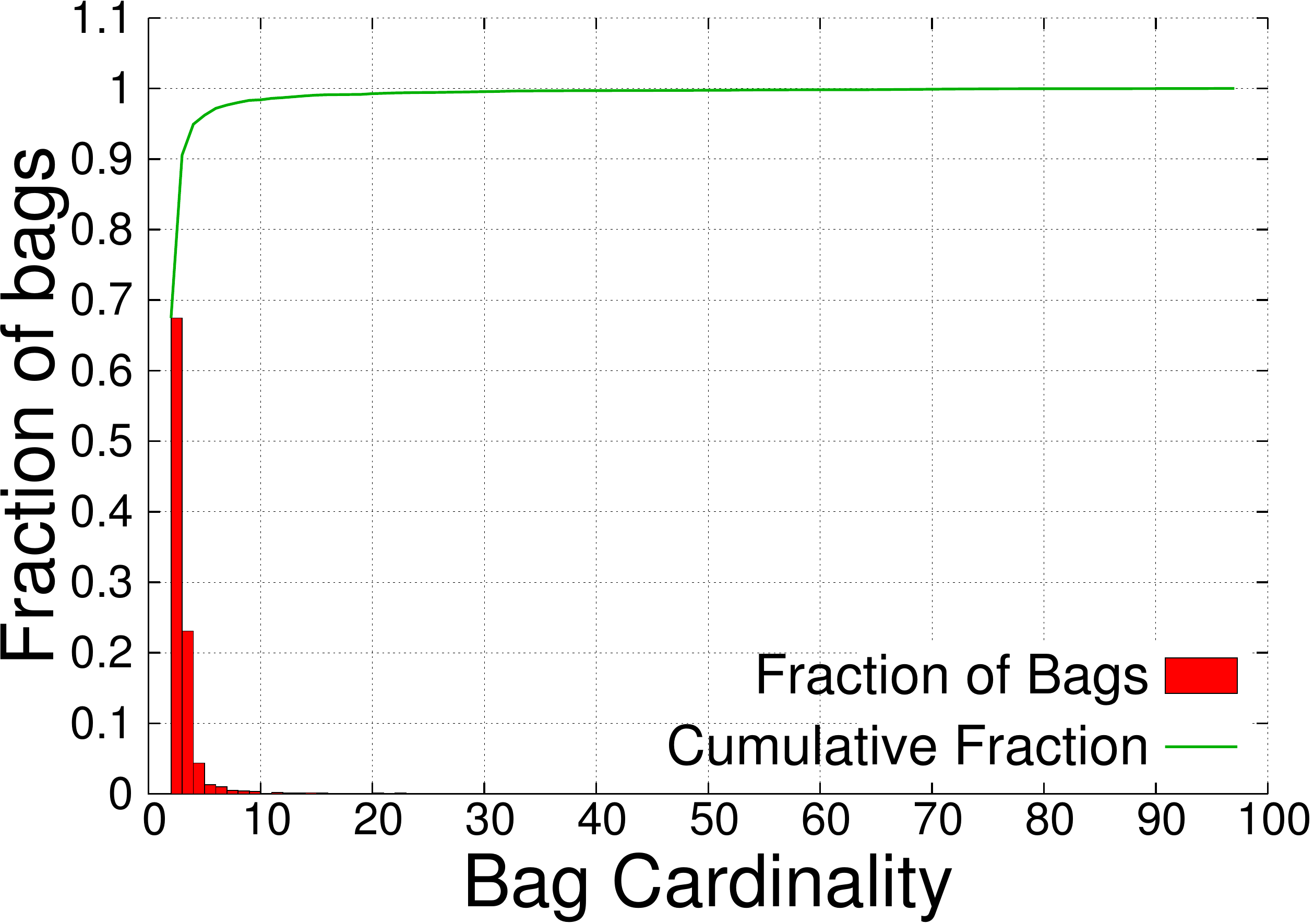}
\caption{\centering \textsc{PL(3.0)}.}
\end{subfigure}
\begin{subfigure}[h]{0.24\textwidth}
  \includegraphics[width=\textwidth]{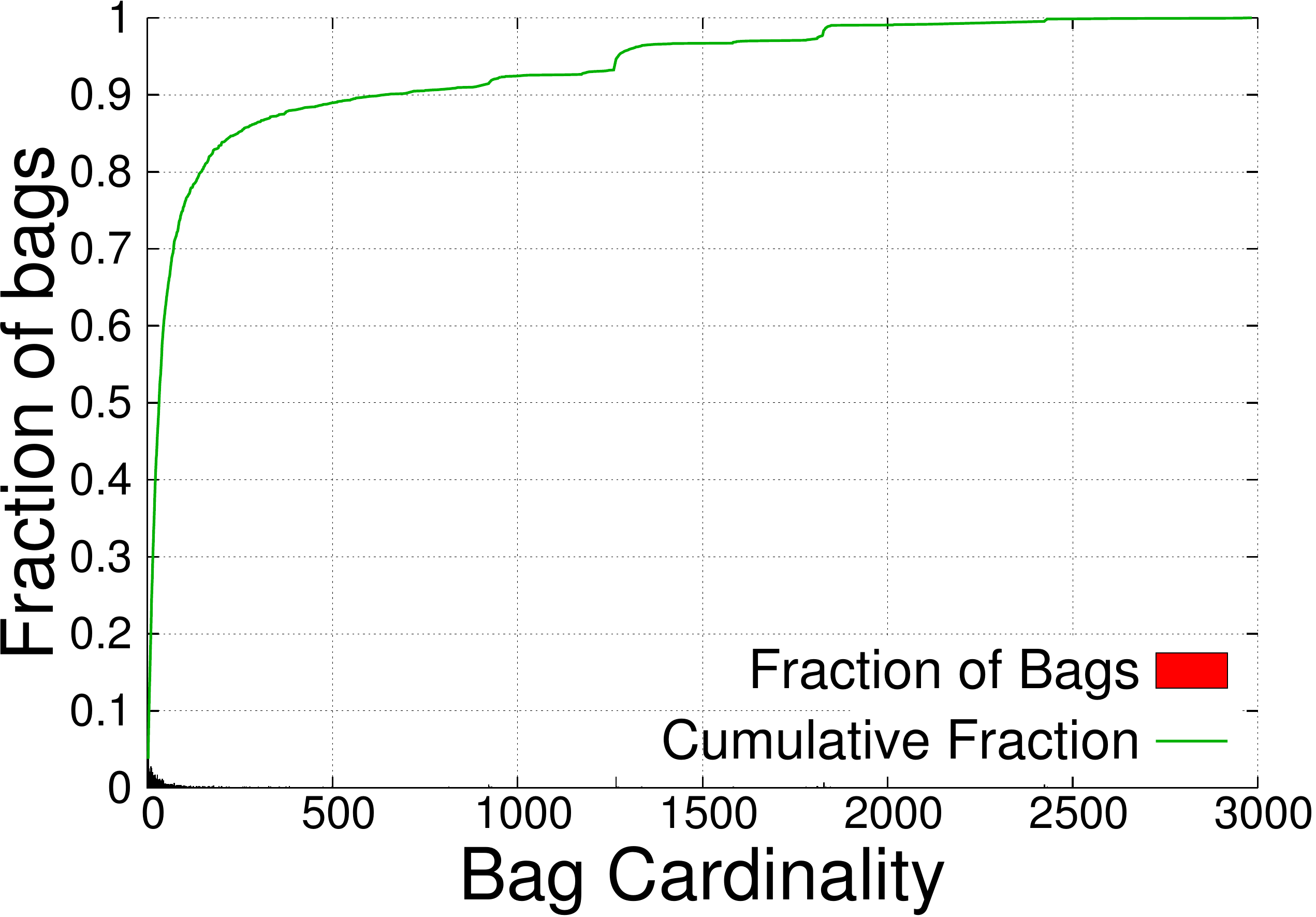}
\caption{\centering \textsc{FB-Lehigh}.}
\end{subfigure}
\begin{subfigure}[h]{0.24\textwidth}
  \includegraphics[width=\textwidth]{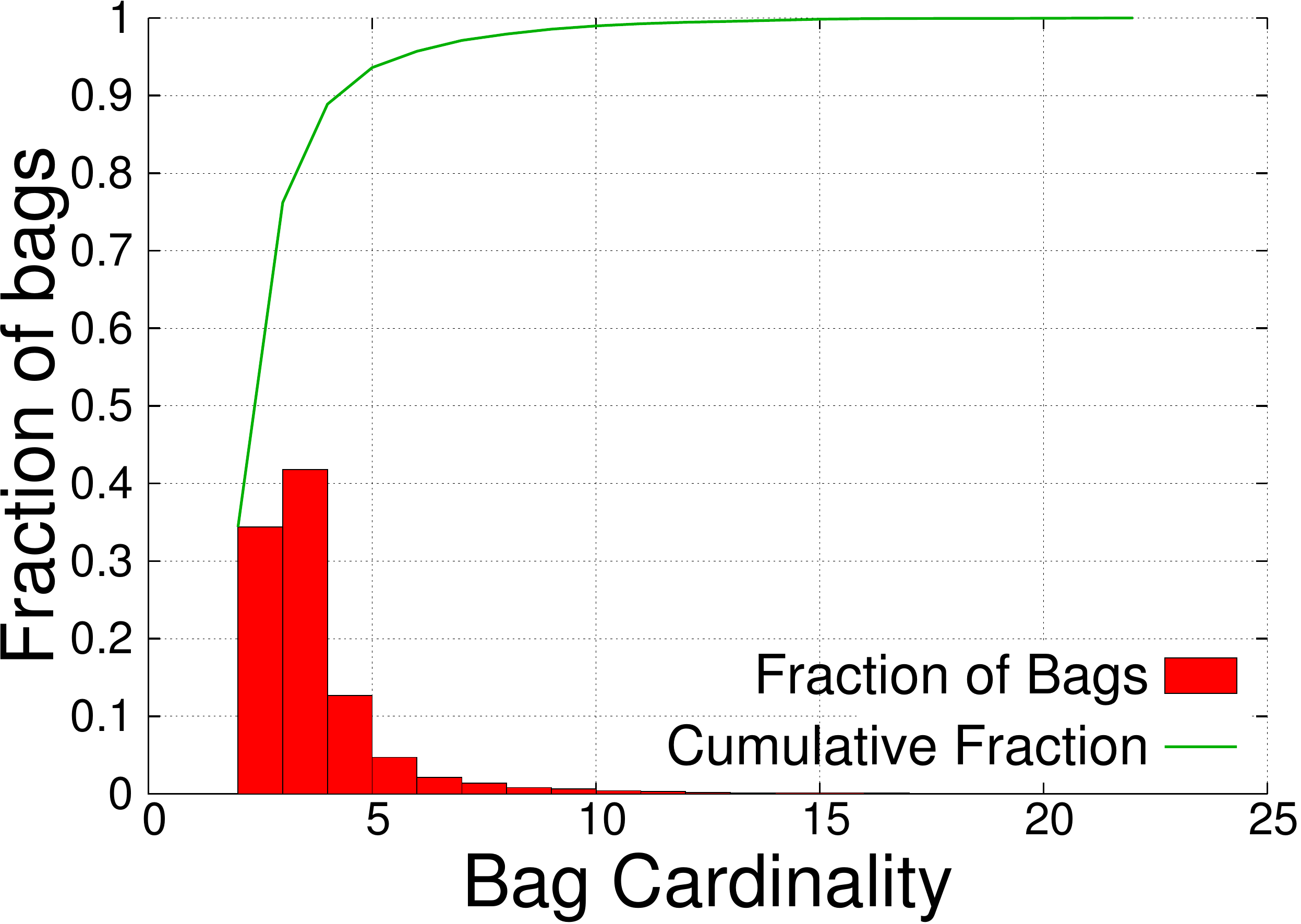}
\caption{\centering \textsc{PowerGrid}.}
\end{subfigure}
\caption{Bag cardinality histograms with cumulative fraction of bags
  for a representative set of networks.  For all of the networks,
  there are many more small cardinality bags than large cardinality
  bags.  This is consistent with a TD structure which has a few
  central bags which quickly taper and branch off into many small
  peripheral bags.  As we will see in Figures \ref{fig:bag_density}
  and \ref{fig:bag_k_core}, in networks with a strong core-periphery
  structure (e.g., not \textsc{PowerGrid}), these peripheral bags tend
  to have a low average $k$-core and high relative density.
  \textsc{FB-Lehigh} has the most large bags, due to the tendency of
  the \textsc{FB} networks to form long, path-like trunks in its TDs.
  \textsc{PowerGrid} has the smallest tapering effect; although the
  largest bags are still at the center of the decomposition, there is
  only a small change in size from the largest bags to the smallest,
  presumably since this network has the weakest core-periphery
  structure.}
\label{fig:bag_hist}
\end{center}
\end{figure*}

\begin{figure*}[!htb]
\begin{center}
\begin{subfigure}[h]{0.24\textwidth}
  \includegraphics[width=\textwidth]{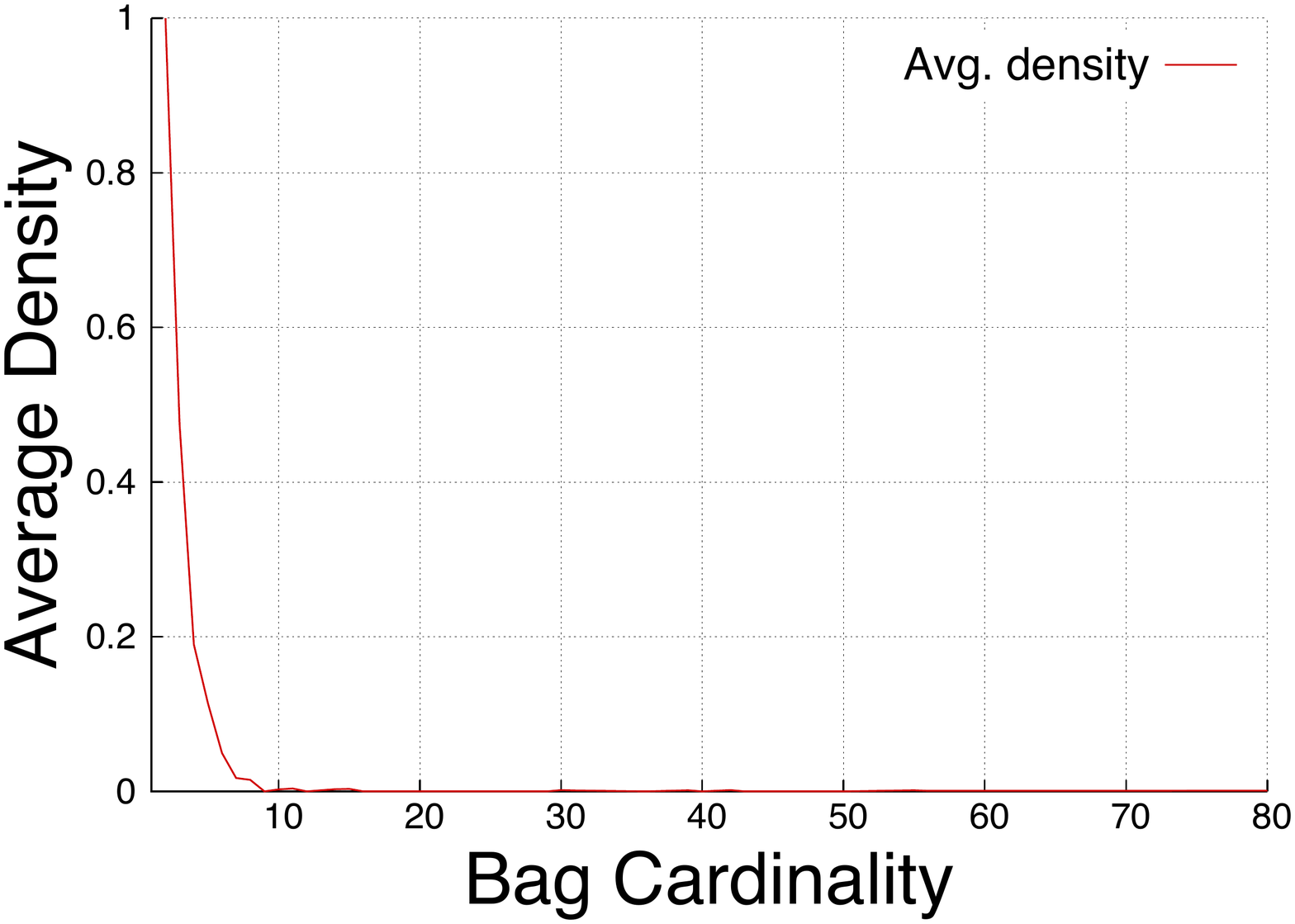}
\caption{\centering \textsc{ER(1.6)}}
\label{fig:bag_density-er_sparse}
\end{subfigure}
\begin{subfigure}[h]{0.24\textwidth}
\includegraphics[width=\textwidth]{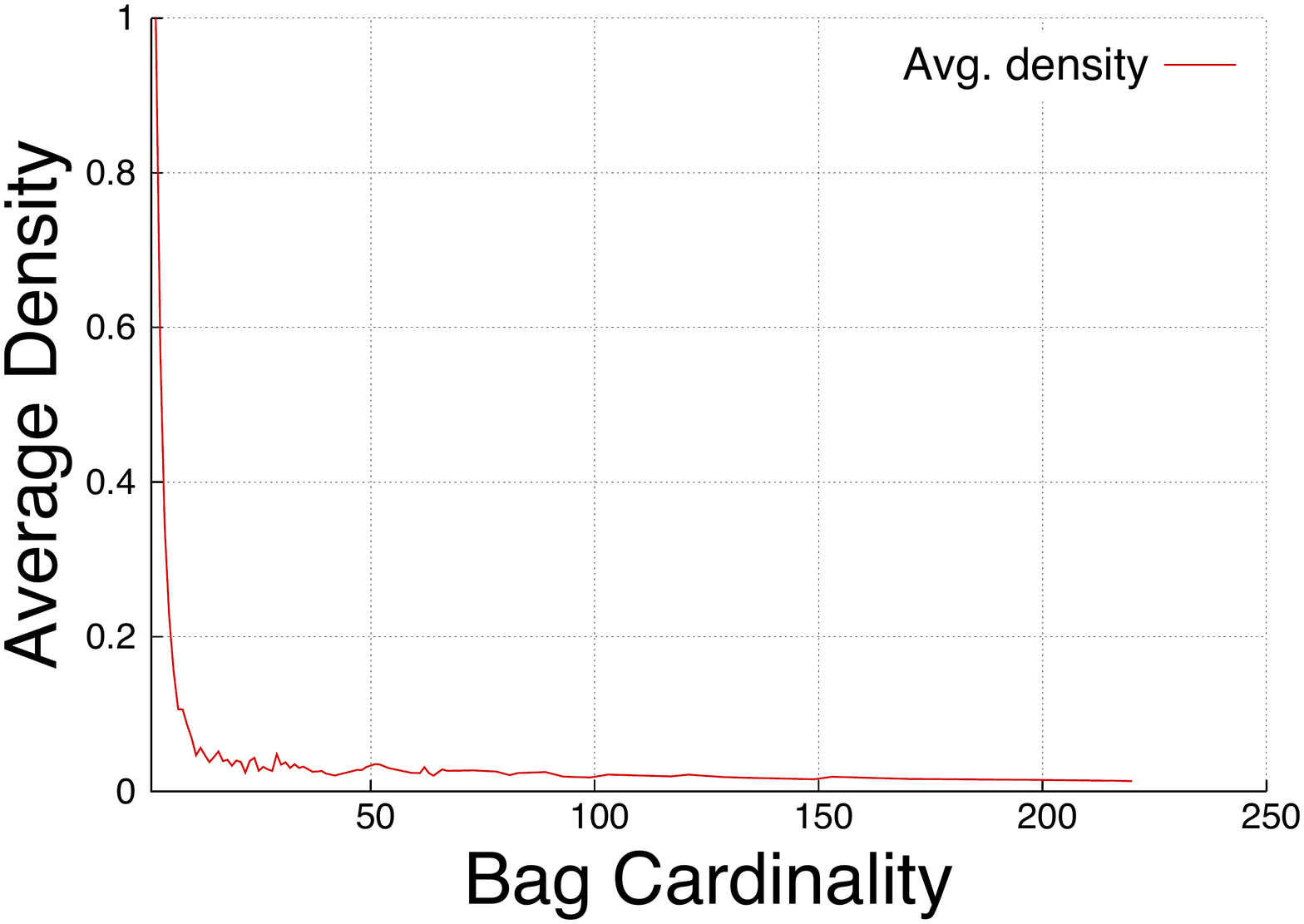}
\caption{\centering \textsc{PL(2.5)}}
\end{subfigure}
\begin{subfigure}[h]{0.24\textwidth}
\includegraphics[width=\textwidth]{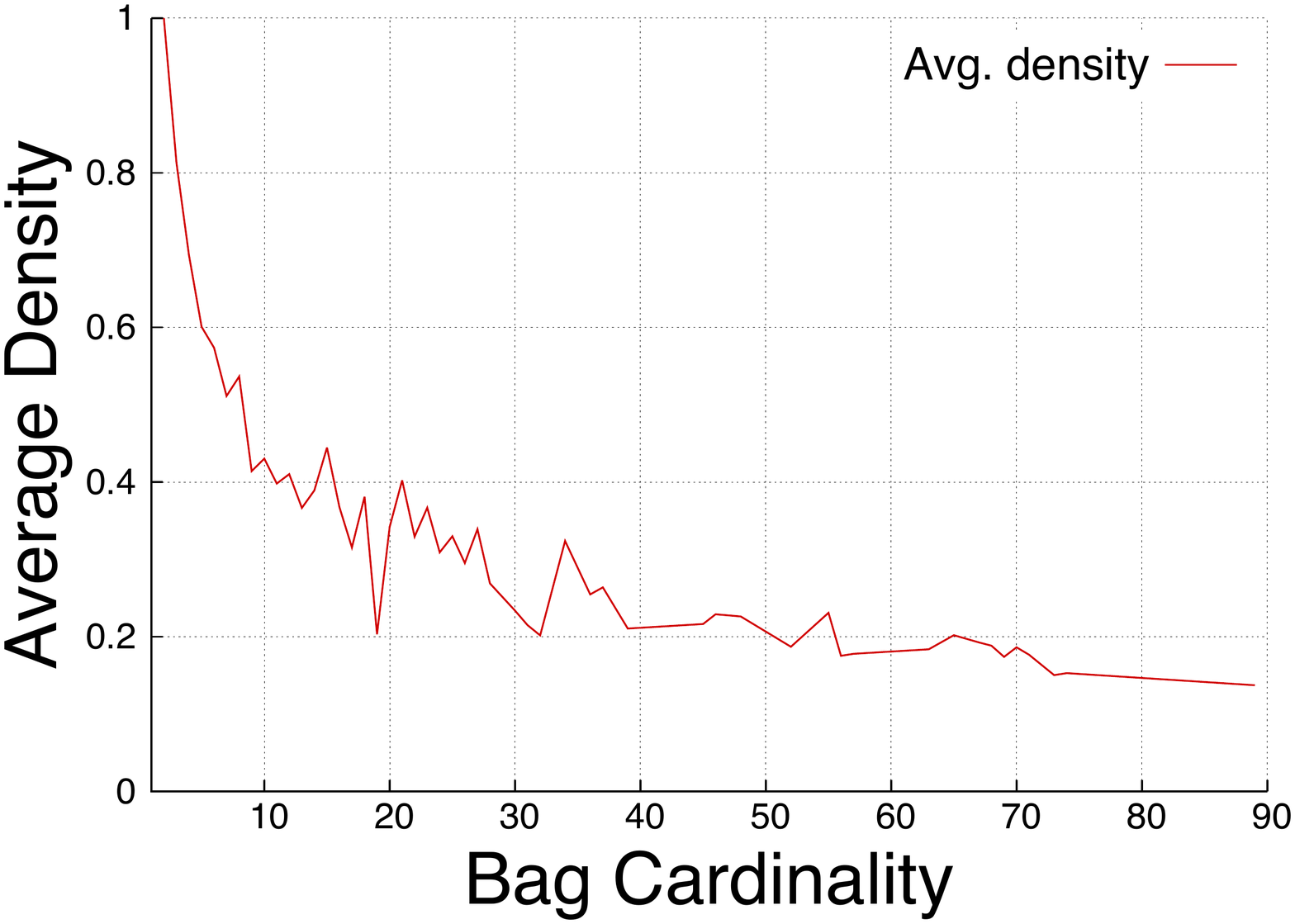}
\caption{\centering \textsc{as20000102}}
\end{subfigure}
\begin{subfigure}[h]{0.24\textwidth}
\includegraphics[width=\textwidth]{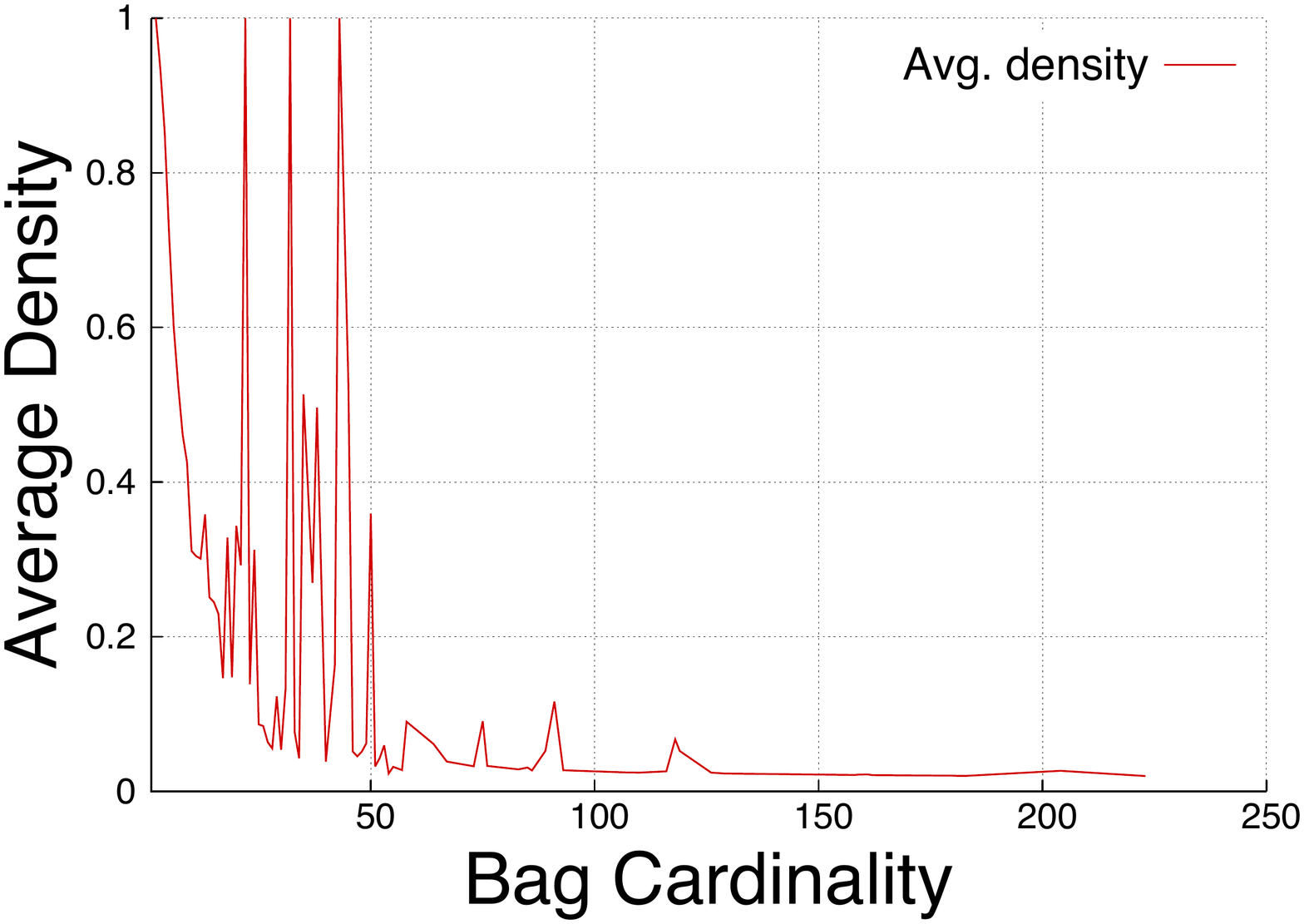}
\caption{\centering \textsc{CA-GrQc}}
\end{subfigure}
\begin{subfigure}[h]{0.24\textwidth}
  \includegraphics[width=\textwidth]{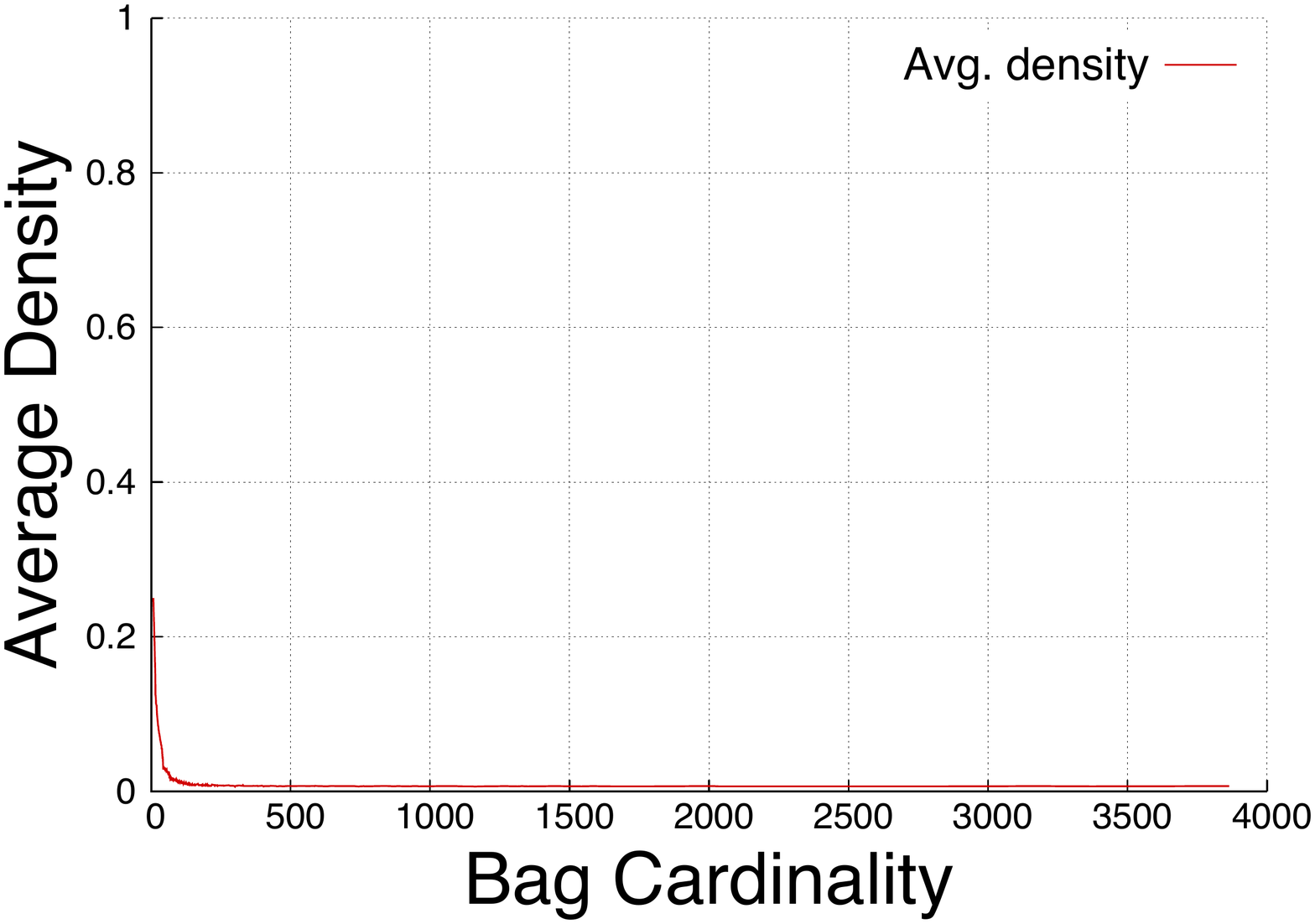}
\caption{\centering \textsc{ER(32)}}
\label{fig:bag_density-er_dense}
\end{subfigure}
\begin{subfigure}[h]{0.24\textwidth}
  \includegraphics[width=\textwidth]{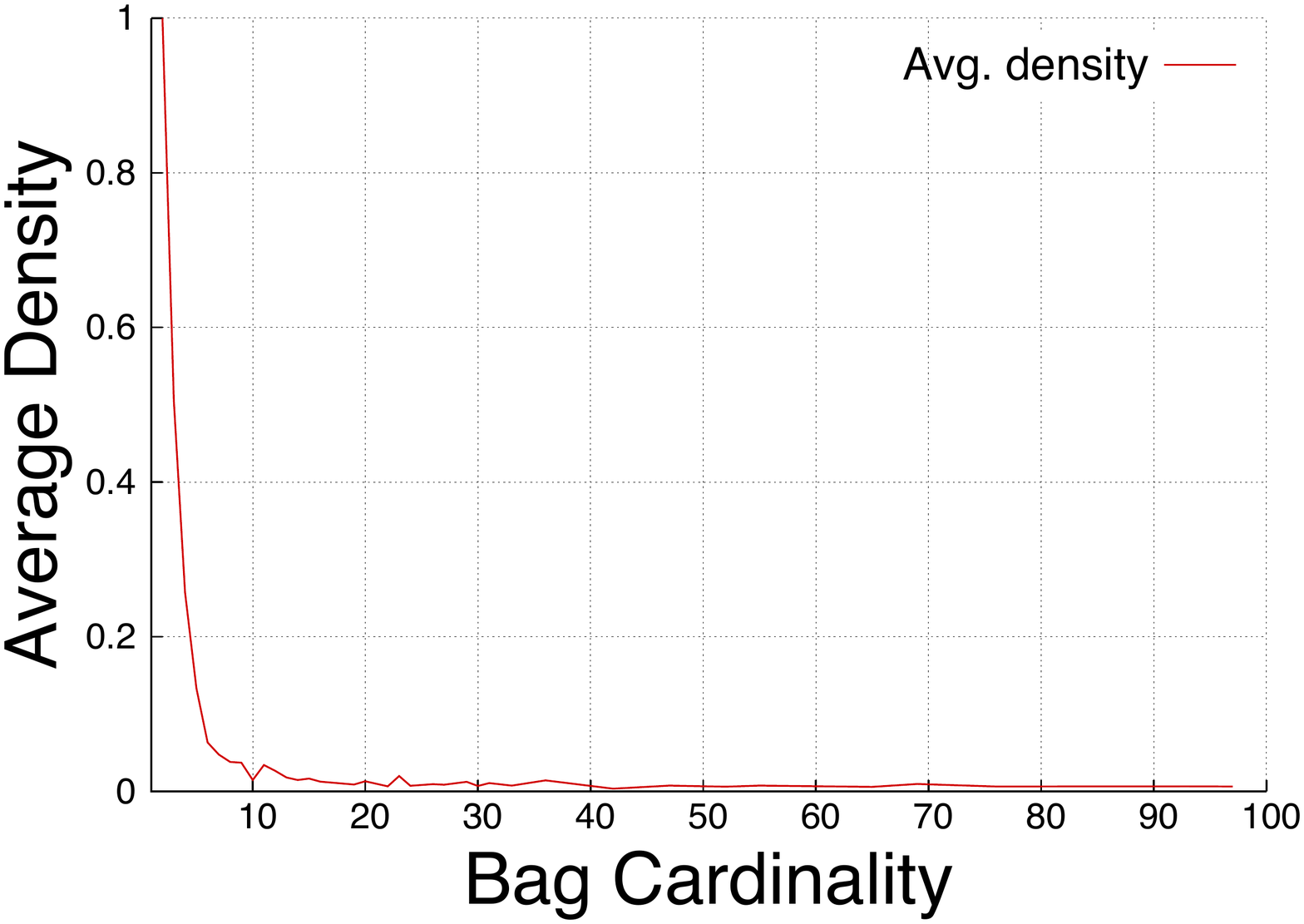}
\caption{\centering \textsc{PL(3.0)}}
\end{subfigure}
\begin{subfigure}[h]{0.24\textwidth}
  \includegraphics[width=\textwidth]{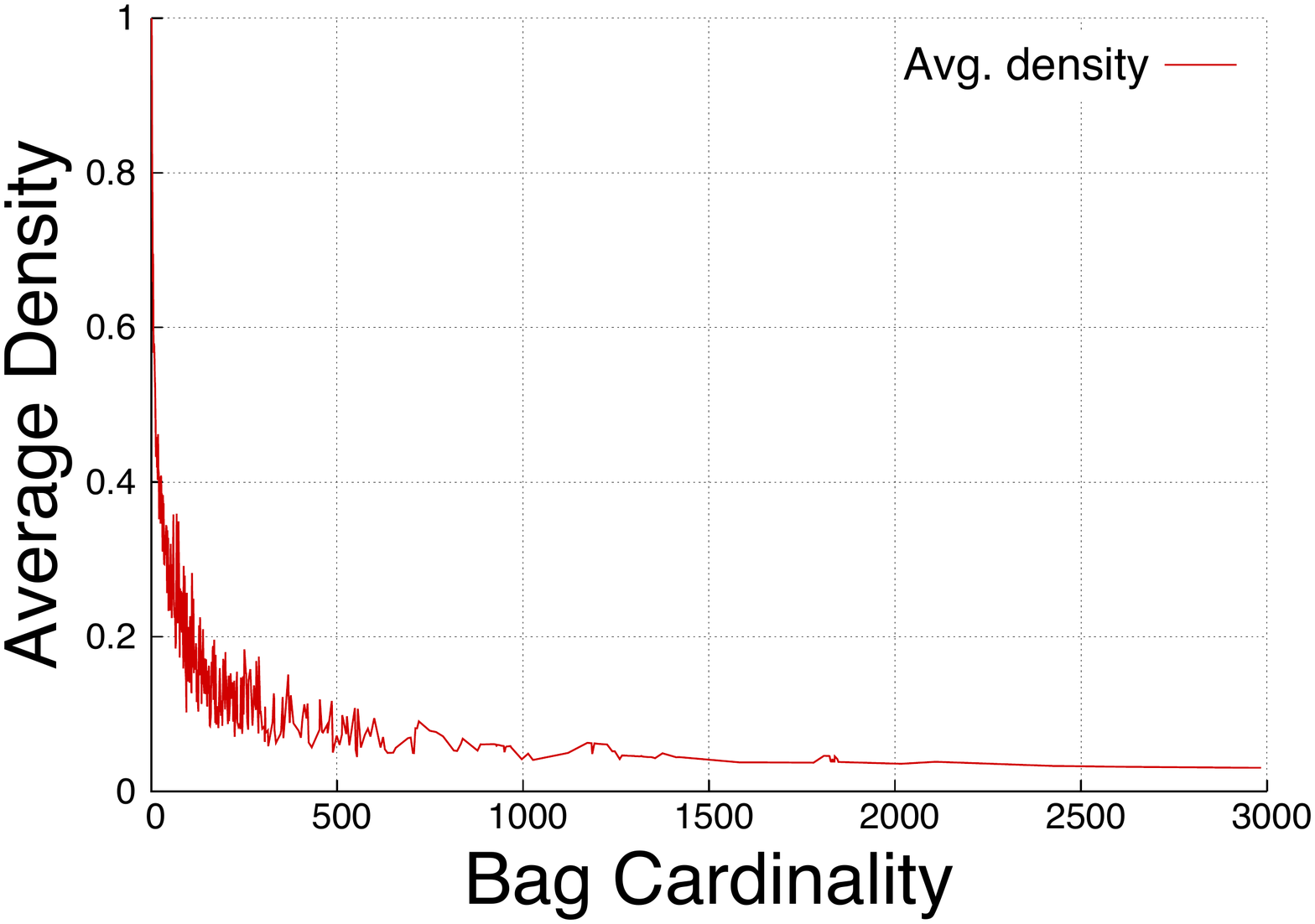}
\caption{\centering \textsc{FB-Lehigh}}
\end{subfigure}
\begin{subfigure}[h]{0.24\textwidth}
  \includegraphics[width=\textwidth]{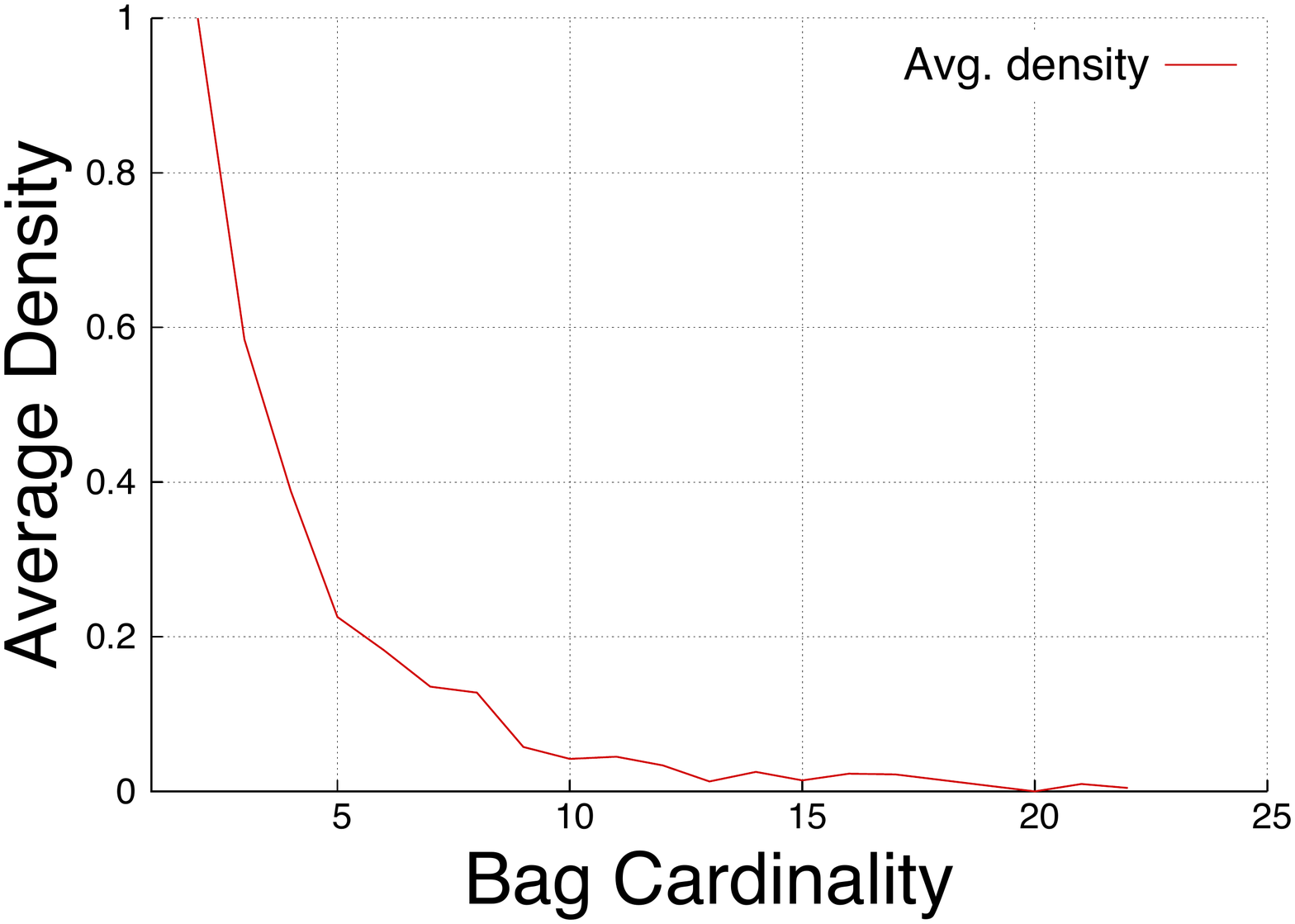}
\caption{\centering \textsc{PowerGrid}}
\end{subfigure}
\caption{Average bag density versus bag cardinality plots for a
  representative set of networks.  In the \textsc{PL} networks the
  small bags are dense---in the extreme case consisting of a single
  edge---like the \textsc{ER(1.6)} network; but the largest bags are
  larger and mostly connected, similar to the intermediate and central
  bags of \textsc{ER(32)}.  The real-world networks all show denser
  bags even at large size scales as compared to the synthetic
  networks, and this is due to the increased clustering present in
  these networks.}
\label{fig:bag_density}
\end{center}
\end{figure*}

\begin{figure*}[!ht]
\begin{center}
\begin{subfigure}[h]{0.24\textwidth}
  \includegraphics[width=\textwidth]{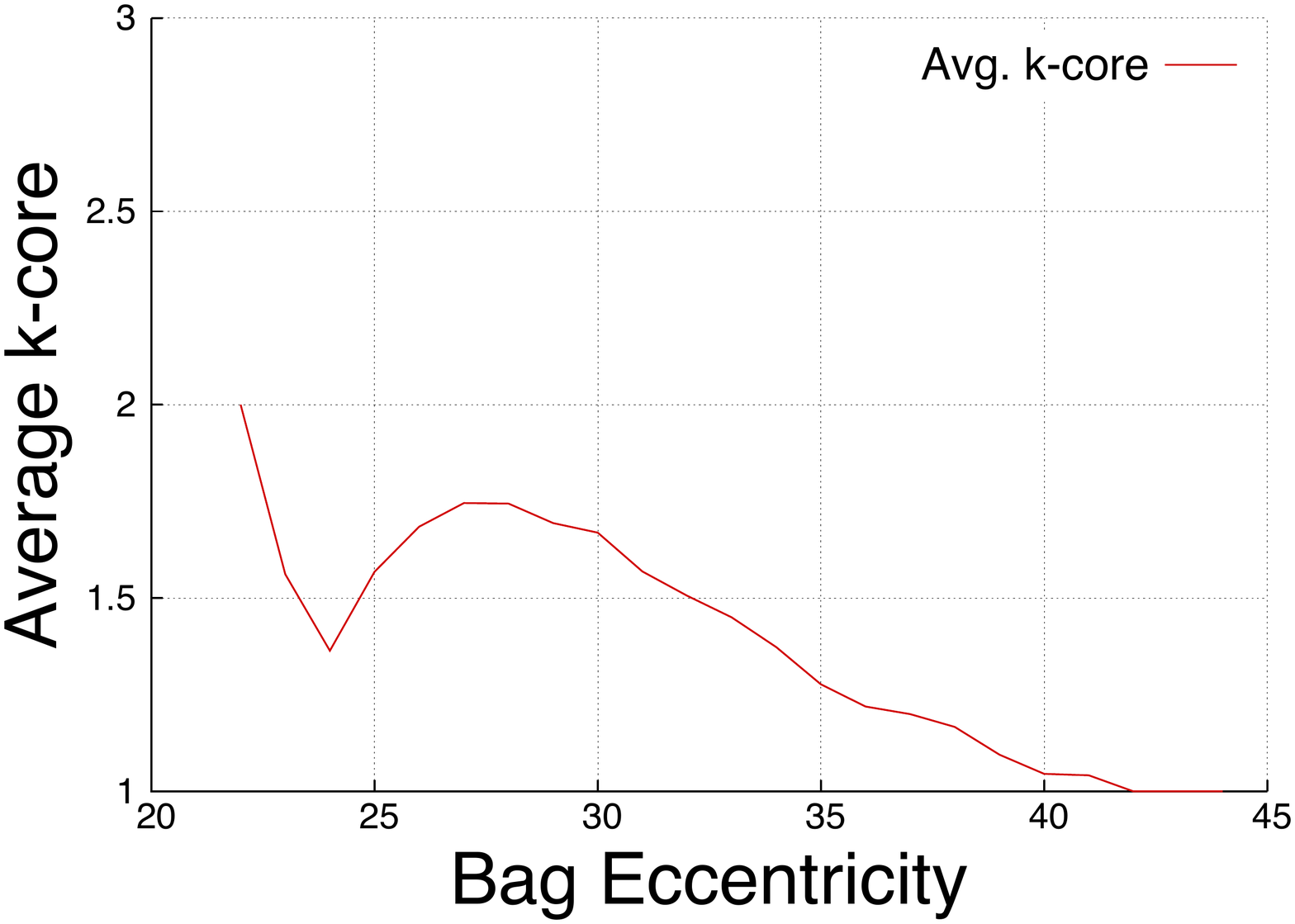}
\caption{\centering \textsc{ER(1.6)}}
\label{fig:bag_k_core-er_sparse}
\end{subfigure}
\begin{subfigure}[h]{0.24\textwidth}
\includegraphics[width=\textwidth]{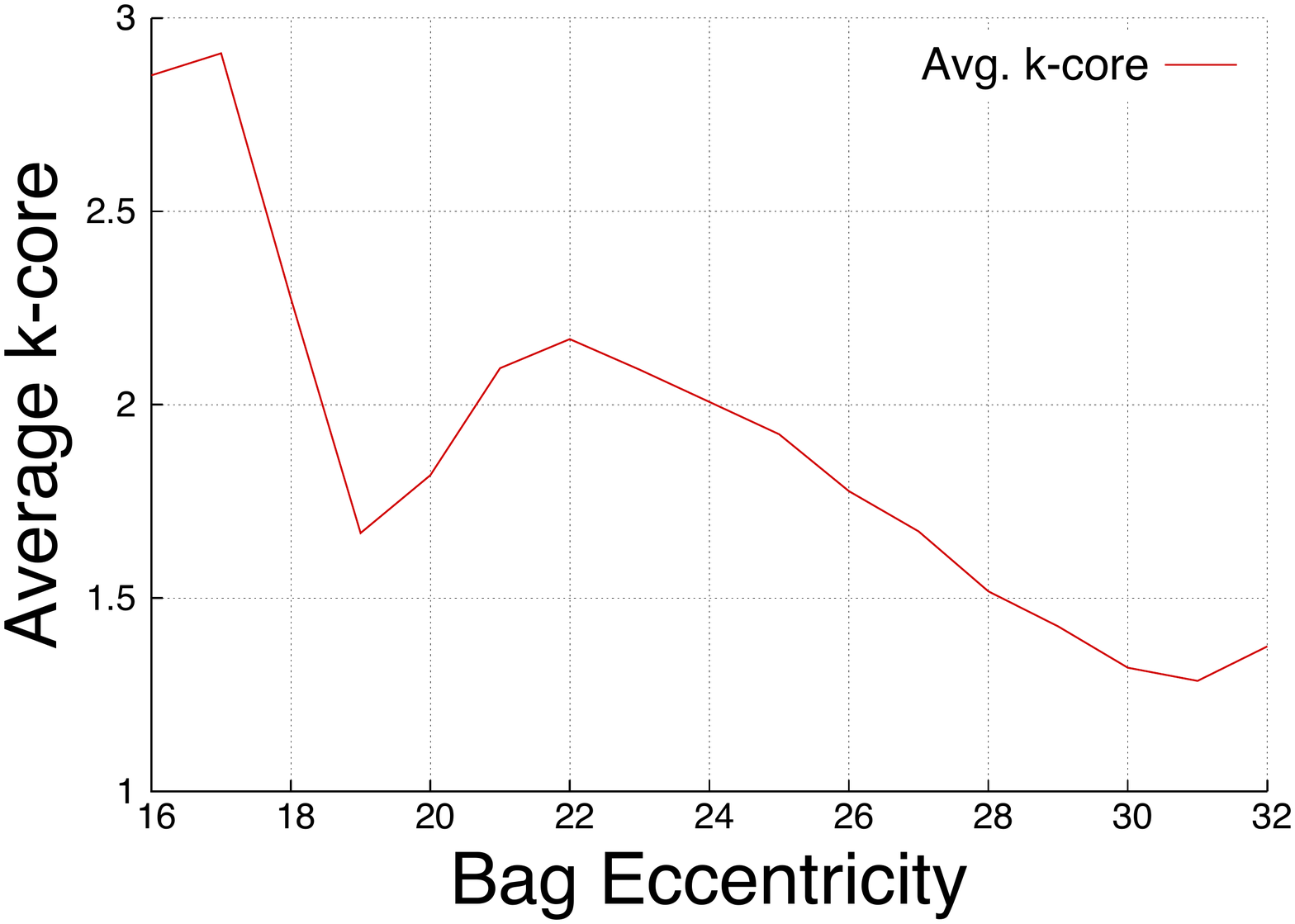}
\caption{\centering \textsc{PL(2.5)}}
\label{fig:pl_25_k_core}
\end{subfigure}
\begin{subfigure}[h]{0.24\textwidth}
\includegraphics[width=\textwidth]{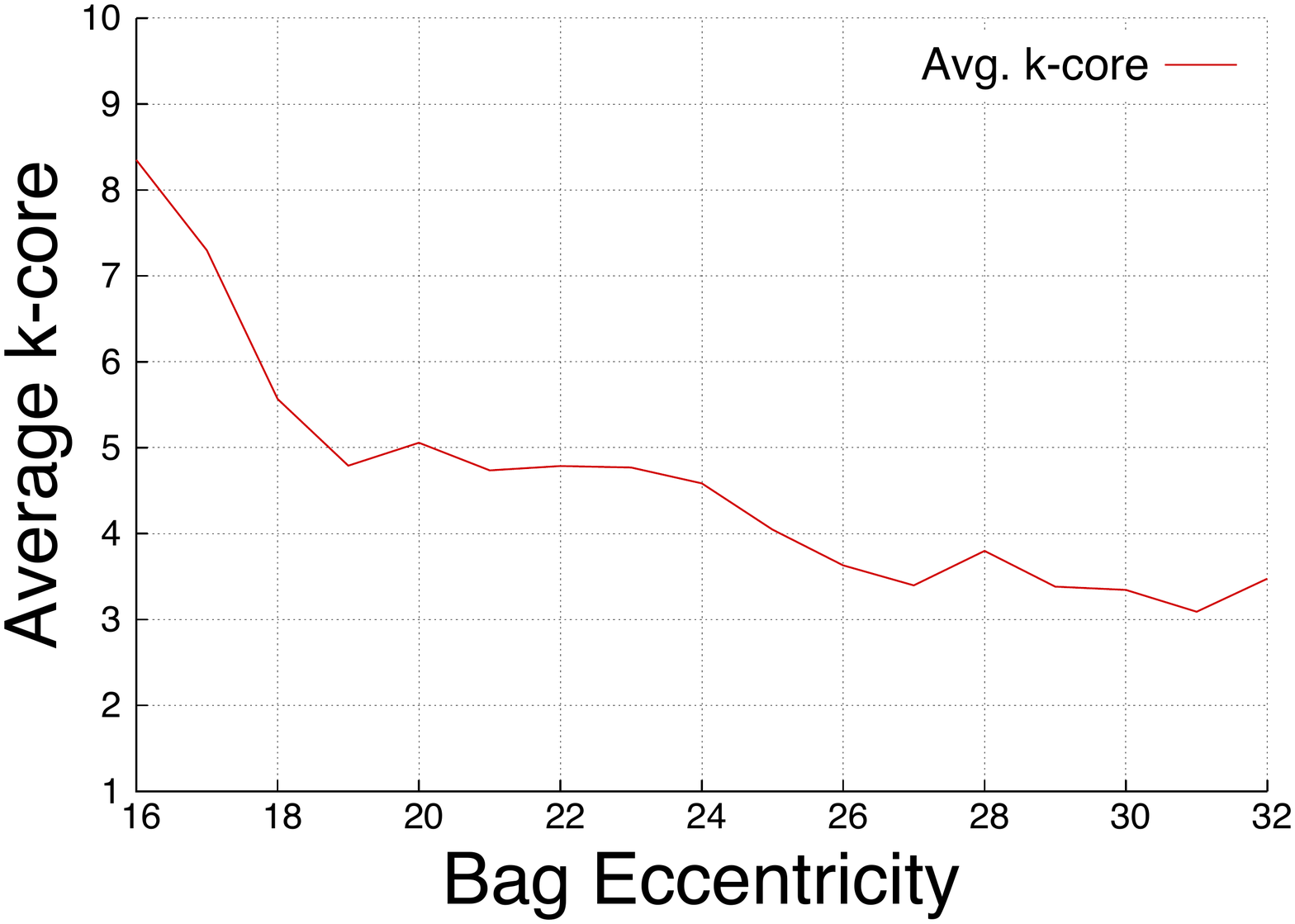}
\caption{\centering \textsc{as20000102}}
\label{fig:as_k_core}
\end{subfigure}
\begin{subfigure}[h]{0.24\textwidth}
\includegraphics[width=\textwidth]{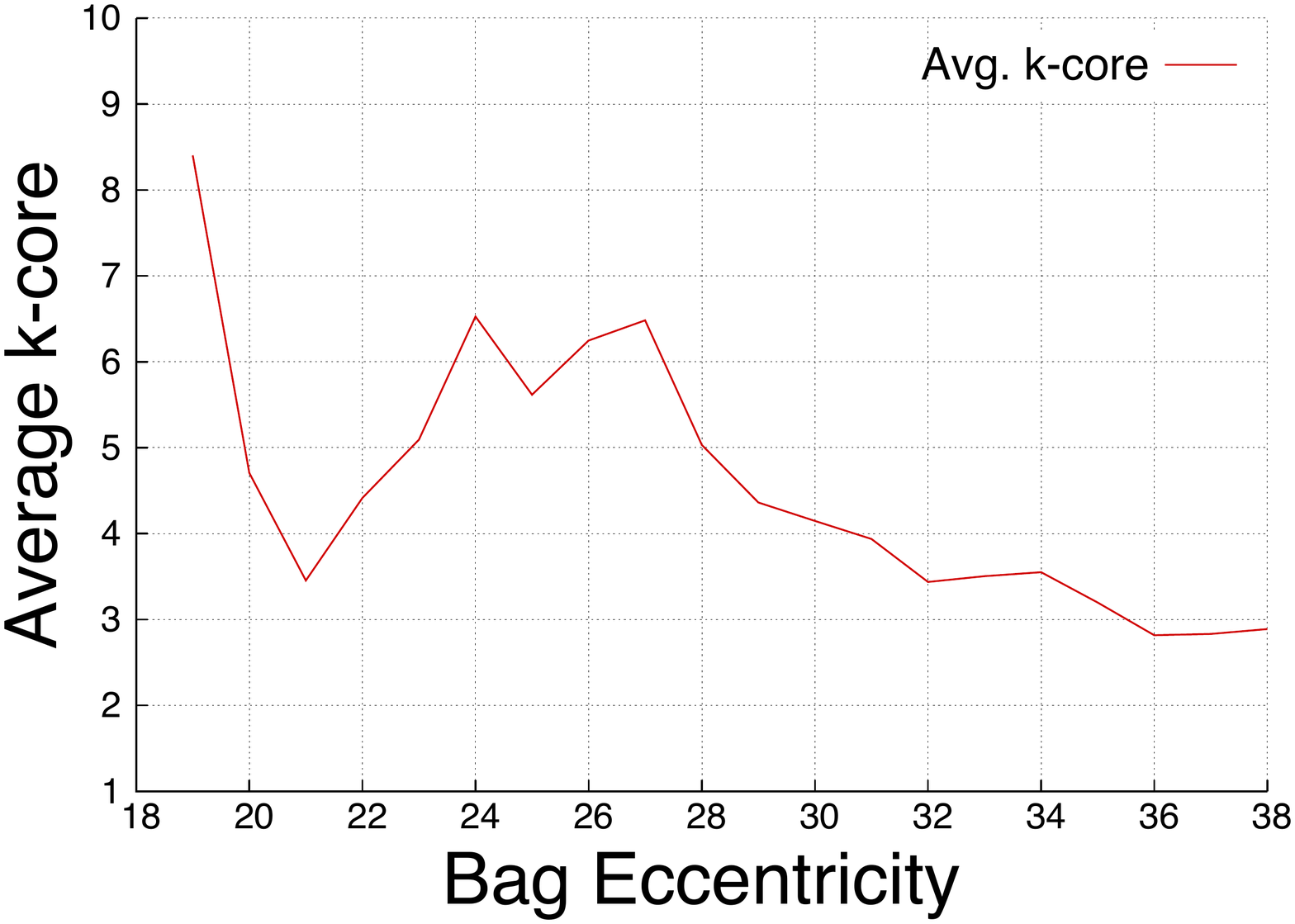}
\caption{\centering \textsc{CA-GrQc}}
\label{fig:cagrqc_k_core}
\end{subfigure}
\begin{subfigure}[h]{0.24\textwidth}
  \includegraphics[width=\textwidth]{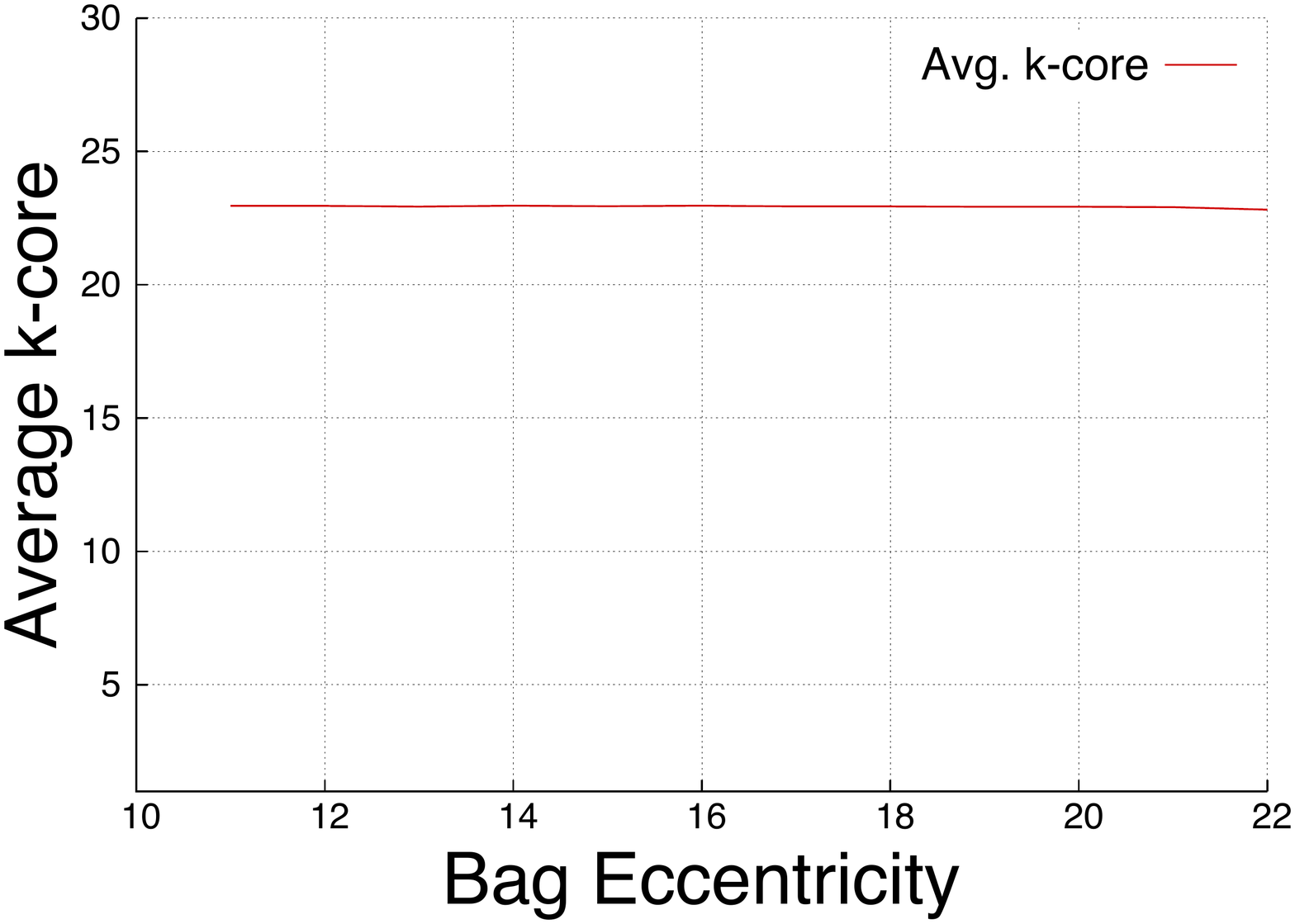}
\caption{\centering \textsc{ER(32)}}
\label{fig:bag_k_core-er_dense}
\end{subfigure}
\begin{subfigure}[h]{0.24\textwidth}
  \includegraphics[width=\textwidth]{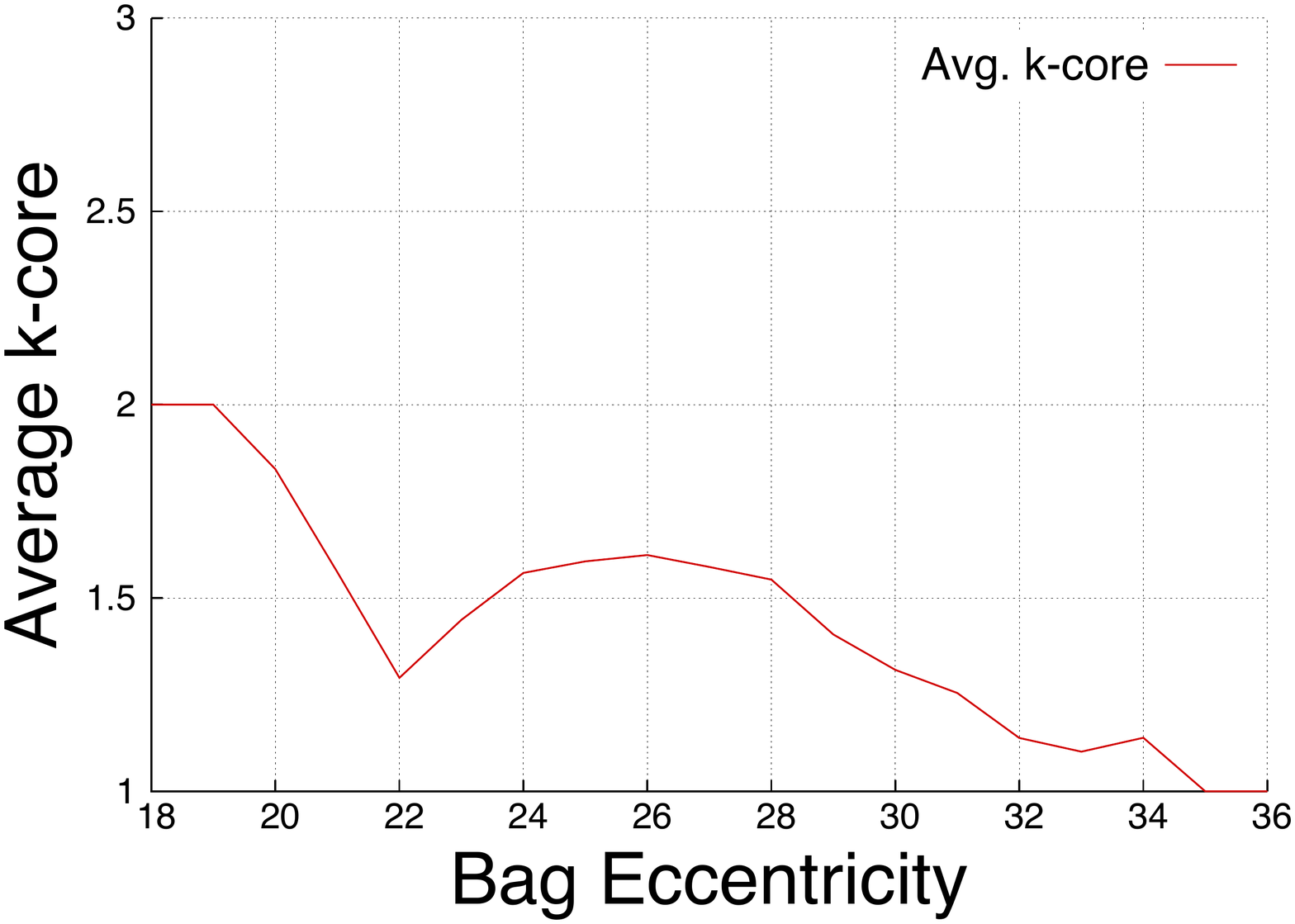}
\caption{\centering \textsc{PL(3.0)}}
\label{fig:pl_30_k_core}
\end{subfigure}
\begin{subfigure}[h]{0.24\textwidth}
  \includegraphics[width=\textwidth]{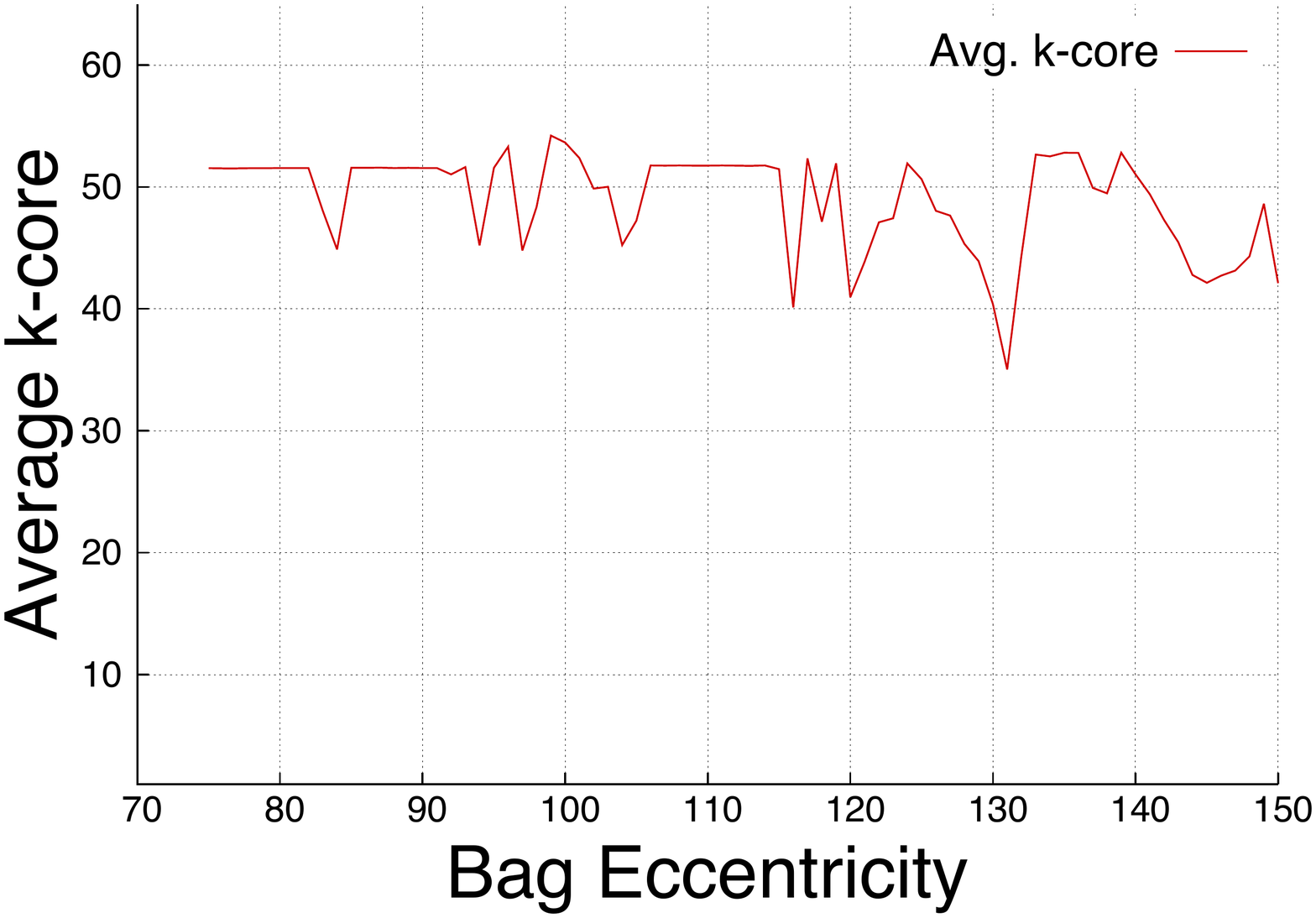}
\caption{\centering \textsc{FB-Lehigh}}
\label{fig:lehigh_k_core}
\end{subfigure}
\begin{subfigure}[h]{0.24\textwidth}
  \includegraphics[width=\textwidth]{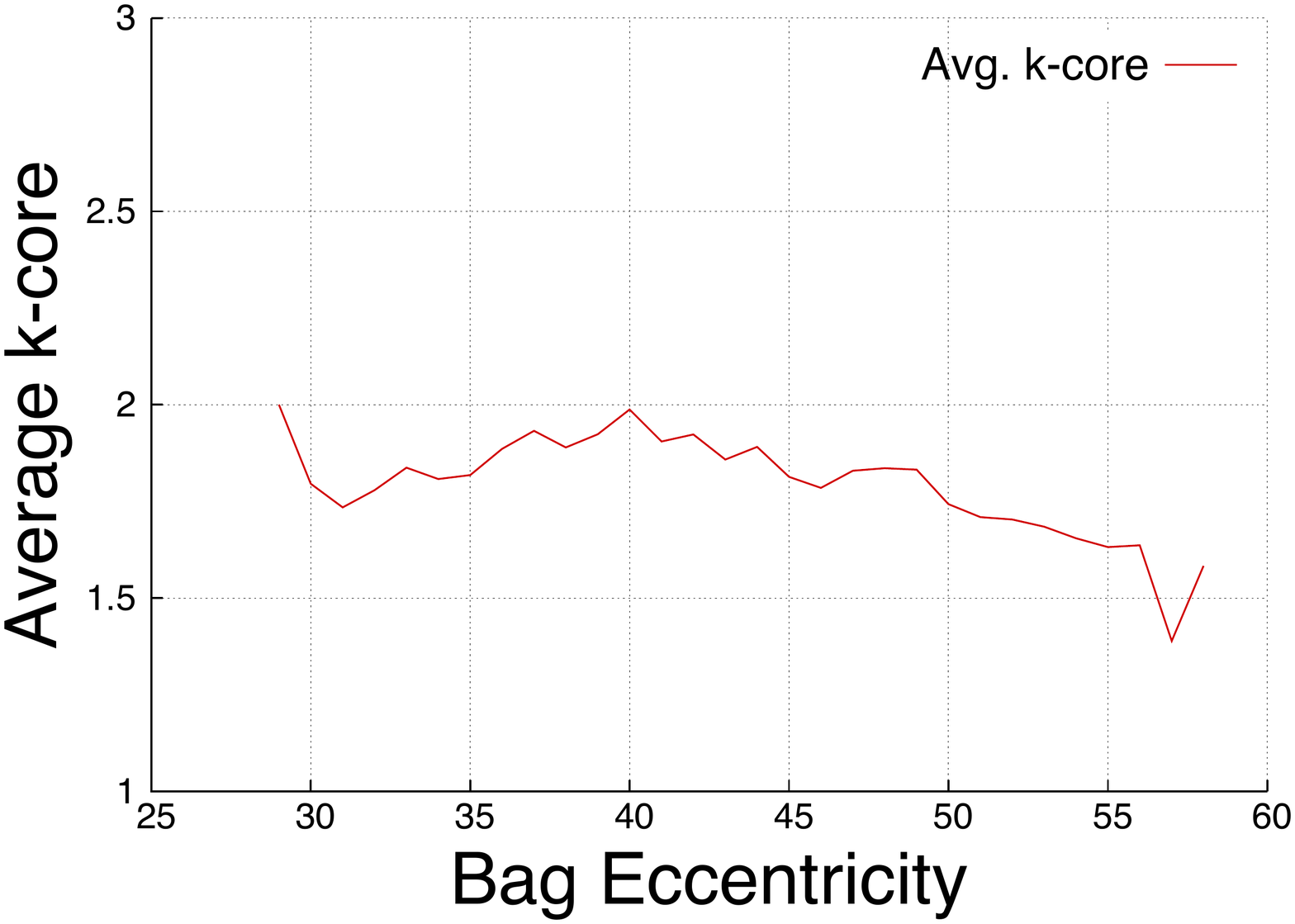}
\caption{\centering \textsc{PowerGrid}}
\label{fig:power_k_core}
\end{subfigure}
\caption{Average $k$-core versus bag eccentricity for a representative
  set of networks.  The correlation between the core-periphery
  structure and the central-perimeter bags can be seen in a downward
  slope in these plots.  Networks with no prominent core-periphery
  structure (\textsc{ER(32)} and \textsc{PowerGrid}, for two very
  different reasons) have a flat plot here; while networks with
  moderate core-periphery structure (the \textsc{PL} graphs and
  \textsc{ER(1.6)}) have a downward sloping line, but relatively
  shallow (i.e., not deep) cores.  \textsc{as20000102} and
  \textsc{CA-GrQc} both have prominent, deep core-periphery structures
  that reveal themselves in this plot.  The dips that show up at small
  eccentricities in several of the synthetic networks and
  \textsc{CA-GrQc} are due to the many small ``whiskers'' (in the
  sense of~\cite{LLDM09_communities_IM}) that hang off of the core
  bag.  \textsc{FB-Lehigh} also has a deep core-periphery structure
  (in the sense of $k$-core decompositions); but because of the long
  path-like nature of the TD and since most of the nodes are in the
  deepest cores, the plot is flat with larger downward dips as the bag
  eccentricity increases.}
\label{fig:bag_k_core}
\end{center}
\end{figure*}

We saw in Figure~\ref{fig:er_sparse_bag-central} the central bag for
\textsc{ER(1.6)}, and we interpreted it in terms of the output of
\textsc{amd} TD as due to overlapping cycles; the histograms in
Figure~\ref{fig:bag_hist-er_sparse} show that for \textsc{ER(1.6)}
(and the networks in the rest of Figures~\ref{fig:bag_hist}) 
there are only a very few such central bags.  On the other hand,
Figure \ref{fig:bag_hist} also shows that there are many very small
bags in \textsc{ER(1.6)}.  Two features we noticed about these two
types of bags are that there is a change in edge density between the
large and small bags and that there is a change in the $k$-cores
represented between the large and small bags.  
Figure~\ref{fig:bag_density} and
Figure~\ref{fig:bag_k_core}, where the average edge density of a bag is
plotted against the bag cardinality, show two ways of measuring this.  These figures show that small
peripheral bags are dense (relative to their small size---in the
extreme case, this could be a single edge), and they contain low
$k$-core nodes, as indicated, e.g., by the downward slope of the plots
in Figure~\ref{fig:bag_k_core-er_sparse}.

Many of the results for \textsc{ER(32)} are very different than for
\textsc{ER(1.6)}.  The histograms in Figure
\ref{fig:bag_hist-er_sparse} show that this network has a much larger
proportion of high-width bags than \textsc{ER(1.6)}.  The largely
homogeneous core-periphery structure of dense
\textsc{ER} networks should also be clear since the nodes, regardless of 
bag size, are mostly in the deepest core ($k = 23$).
These trends can be seen by comparing the density of the smallest bags
in \textsc{ER(1.6)} and \textsc{ER(32)} in
Figures~\ref{fig:bag_density-er_sparse}
and~\ref{fig:bag_density-er_dense}.  The flat plot of the average
$k$-core in Figure~\ref{fig:bag_k_core-er_dense}, which holds steady
close to the value of the maximum $k$-core, indicates the lack of a
core-periphery structure in the network.

Putting all of these results together, we can conclude that when it exists 
(e.g., in extremely sparse ER graphs), the core-periphery structure of 
\textsc{ER} networks is captured by the \textsc{amd} TD; and when the 
core-periphery structure does not exist (e.g., for ER graphs for other even 
moderately sparse values of $p$), the large width of the TD indicates that 
the most of the network is in the largest bag, which is analogous to 
most of the nodes being in the core of the network.

\subsection{TDs of PL Networks}
\label{sec:synth_results-pl}

Here, we give a summary of results of an analysis of TDs on PL random
graphs, with an emphasis on the behavior as the degree heterogeneity
parameter $\gamma$ is varied.
Recall that Table~\ref{tbl:networks-basic-stats} provides basic statistics 
for the PL graphs.
PL graphs are a class of ER-like random graphs, except that degree 
heterogeneity is exogenously-specified.  Previous work has shown that PL 
graphs have important similarities with extremely sparse ER graphs, when one
is interested in small-scale versus large-scale tree-like
structure~\cite{LLDM09_communities_IM,Adcock13_icdm,Jeub15}.  
In particular, the increased degree heterogeneity produces a large-scale
core-periphery structure in the PL networks, similar to the extremely sparse 
ER networks, but these PL networks also have some of the characteristics of
denser ER networks (e.g., the core is more strongly connected and the 
diameter of the network is smaller).  

We start with Table~\ref{tbl:pl_density_results}, which show the basic
features of the TDs of \textsc{PL} networks.  \textsc{PL(3.0)} has the
least amount of degree heterogeneity and has similar characteristics
to \textsc{ER(1.6)}, while the lower degree exponents
(\textsc{PL(2.75), PL(2.5)}) have characteristics similar to both the dense
and sparse \textsc{ER} networks.  Most notably, the maximum width
increases (as it would if the density increased), while the median
width and median bag density stay the same (low and high,
respectively), as in the \textsc{ER(1.6)}.  In the previous section, we
saw that the low median width and high density was related to the
presence of a core-periphery structure in the network.  As we will
see, this is also true of the PL networks, and the \textsc{amd} TDs are
again able to capture this structure.


Among other things, we find that for the \textsc{PL} networks,
for a given average degree, the presence of very high degree nodes
that tend to link to each other means that the density of the high-width 
bags, i.e., the core of the network, is greater, making
it more like the cores of the denser \textsc{ER} networks.  On the
other hand, the peripheries of these \textsc{PL} networks are still
very sparse, and TD bags including them look more like the peripheries
of the sparse \textsc{ER} networks.  The periphery results are
reflected in the results presented in
Table~\ref{tbl:pl_density_results}, where we see that the median width
is low and the median density is high (many bags with only a single
edge within).  The visualizations of the central, intermediate, and
peripheral bags from TDs of \textsc{PL} networks in
Figure~\ref{fig:pl_25_bags} reflect
this.  In particular, the central bag for \textsc{PL(2.5)} looks
somewhat like the intermediate \textsc{ER(32)} bags, while the central
bag for \textsc{PL(3.0)} is much less well-connected; and the
peripheral bags for both \textsc{PL} graphs look like the
\textsc{ER(1.6)} bags.  The TD reflects the
core-periphery structure via the central-peripheral bags as reflected
in the downward slope of Figure \ref{fig:pl_25_k_core}.

\begin{figure}[!ht]
\begin{center}
\begin{subfigure}[h]{0.16\textwidth}
\includegraphics[width=\textwidth]{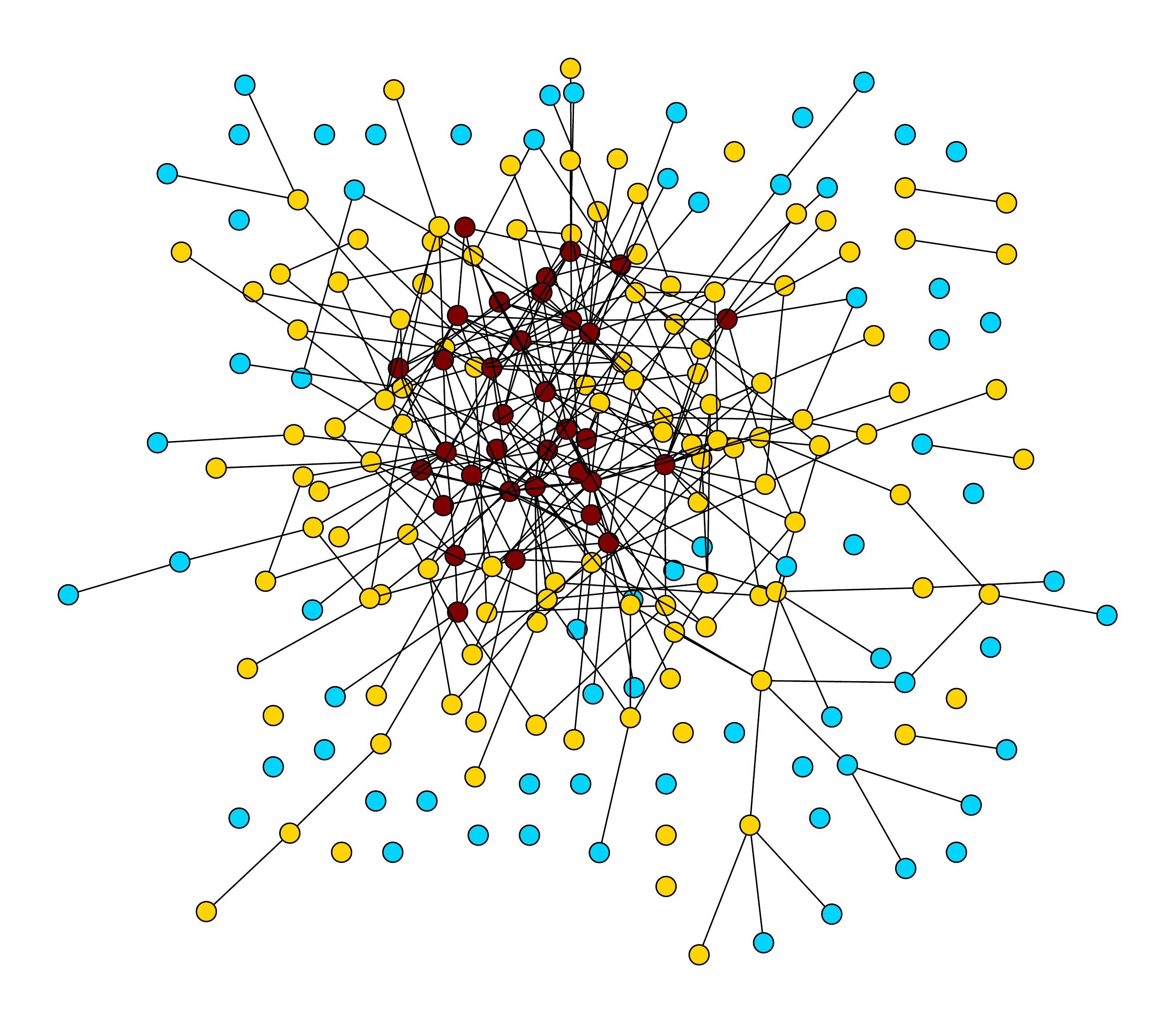}
\caption{\centering \textsc{PL(2.5)} central bag subgraph.}
\end{subfigure}
\begin{subfigure}[h]{0.16\textwidth}
\includegraphics[width=\textwidth]{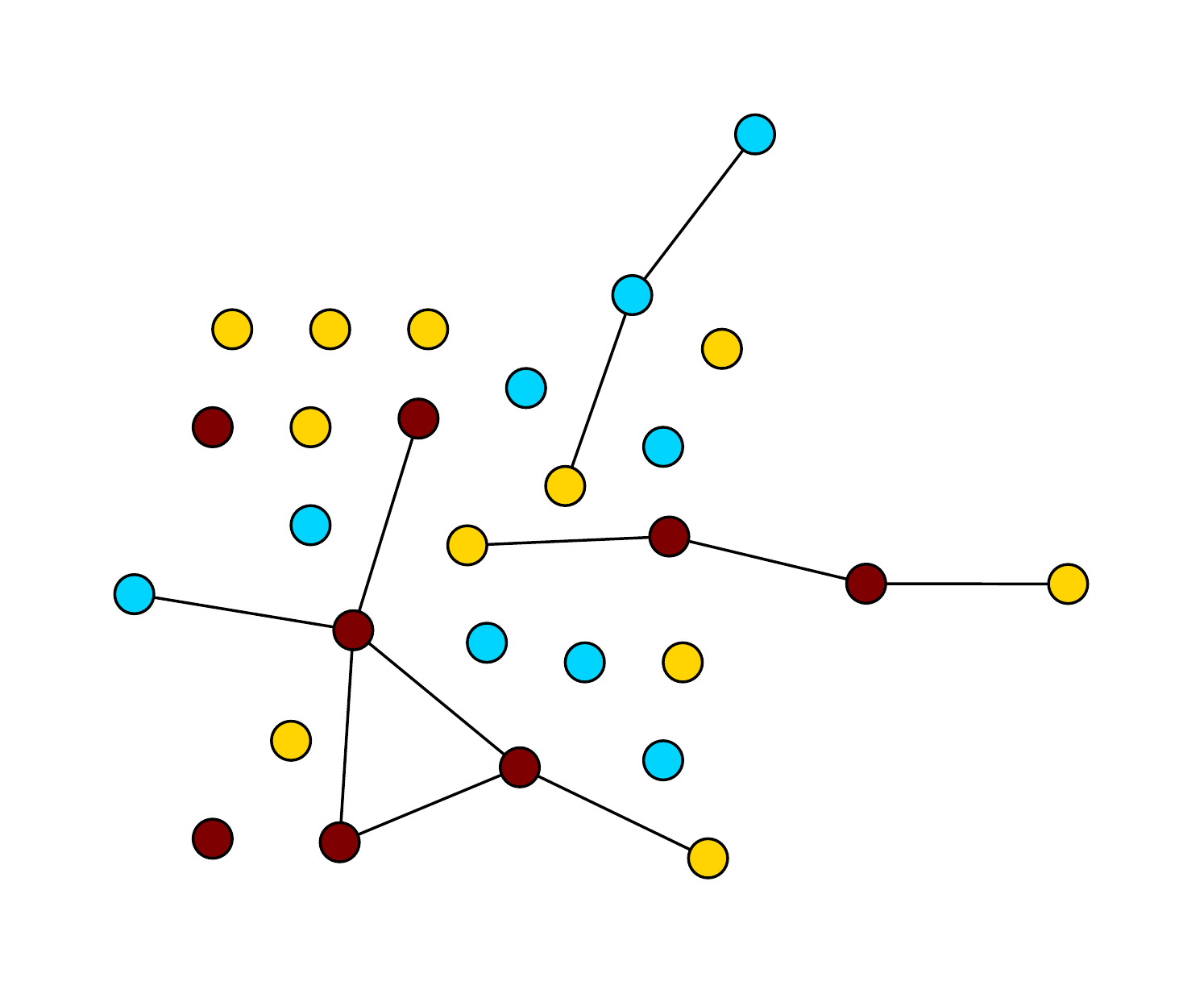}
\caption{\centering \textsc{PL(2.5)} intermediate bag subgraph.}
\end{subfigure}
\begin{subfigure}[h]{0.14\textwidth}
\includegraphics[width=\textwidth]{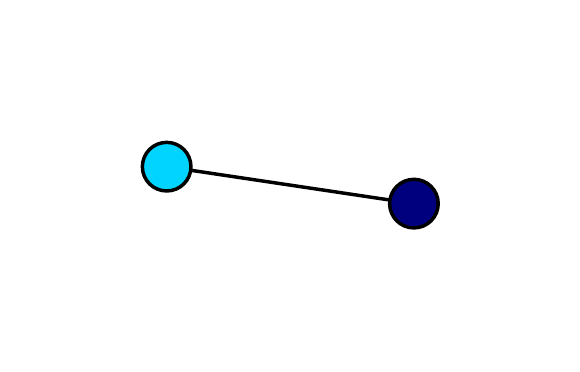}
\caption{\centering \textsc{PL(2.5)} peripheral  bag subgraph.}
\end{subfigure}
\begin{subfigure}[h]{0.16\textwidth}
\includegraphics[width=\textwidth]{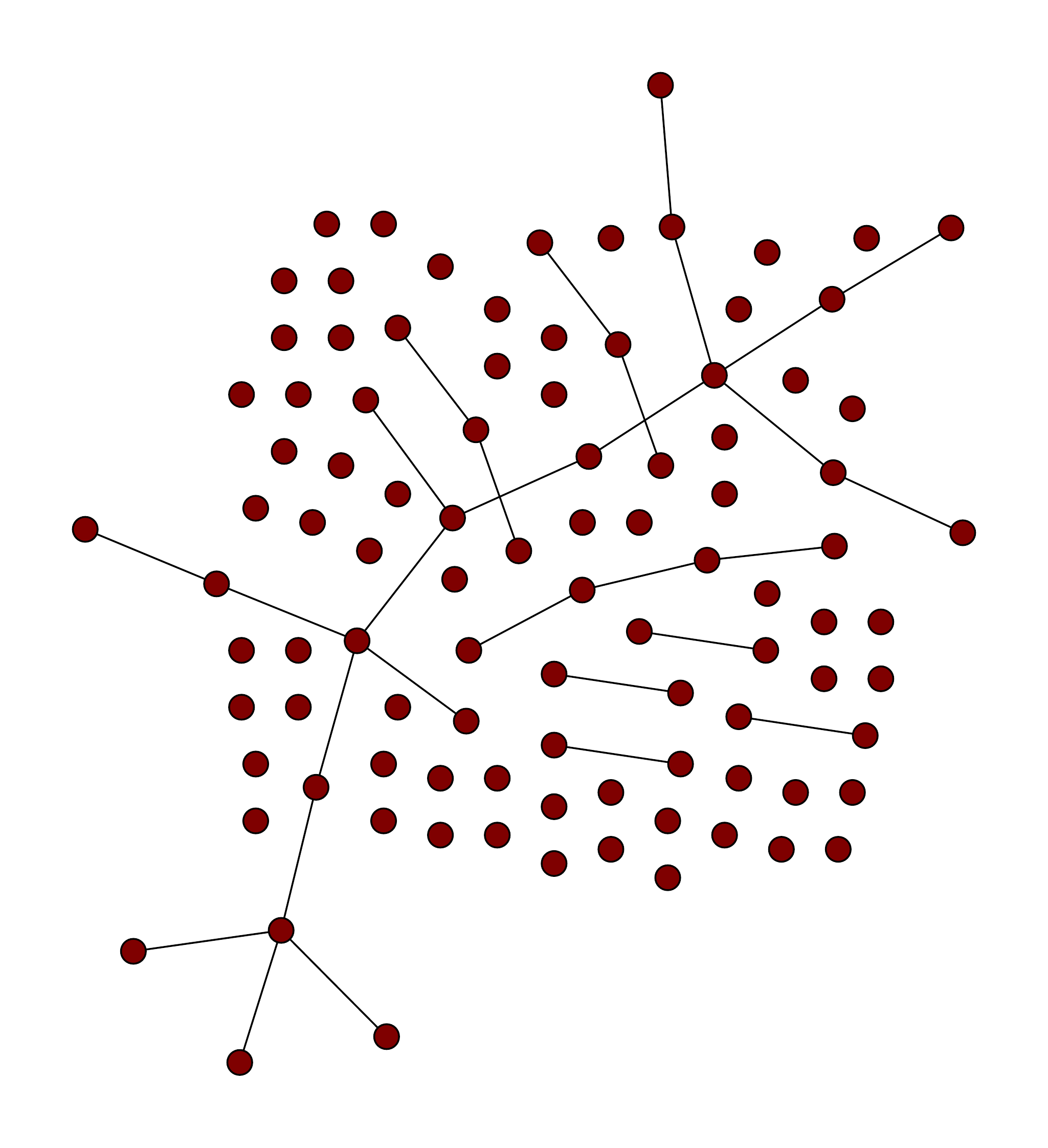}
\caption{\centering \textsc{PL(3.0)} central bag subgraph.}
\end{subfigure}
\begin{subfigure}[h]{0.16\textwidth}
\includegraphics[width=\textwidth]{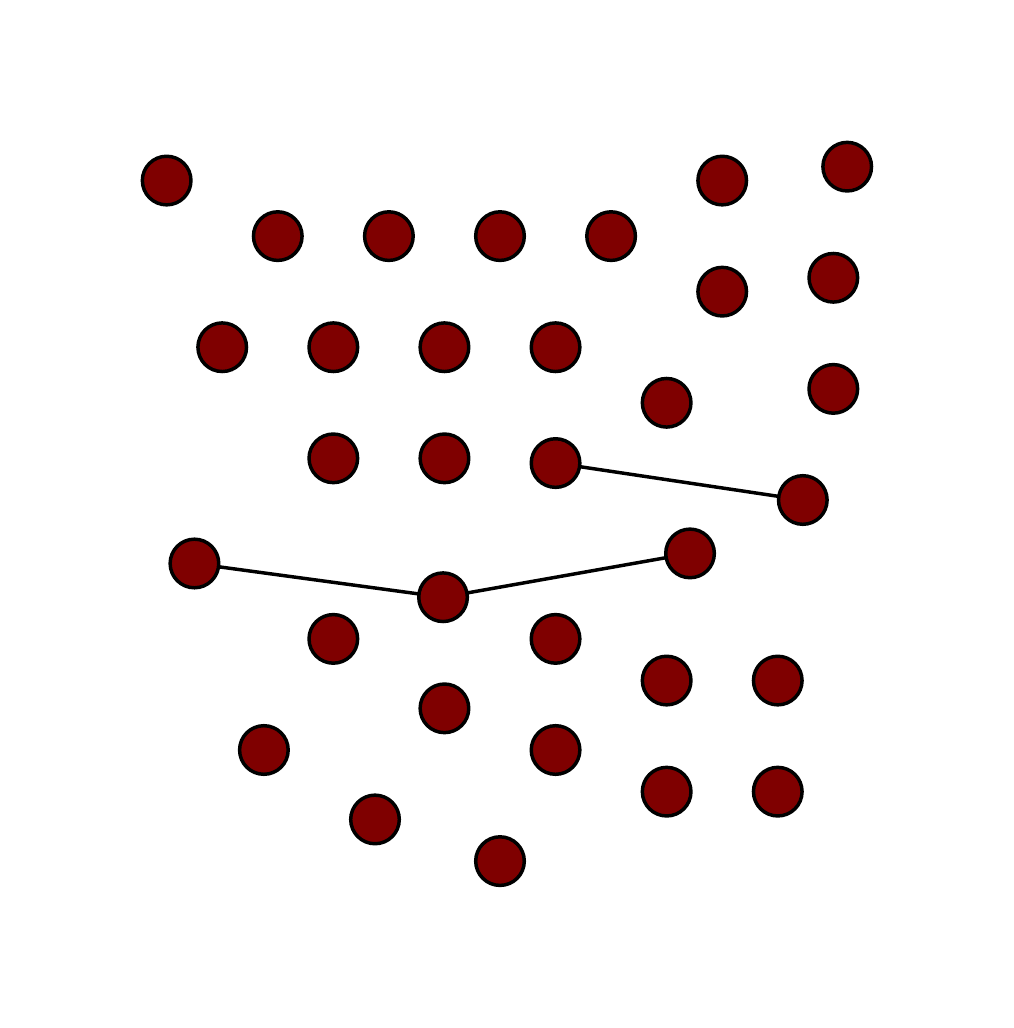}
\caption{\centering \textsc{PL(3.0)} intermediate bag subgraph.}
\end{subfigure}
\begin{subfigure}[h]{0.14\textwidth}
\includegraphics[width=\textwidth]{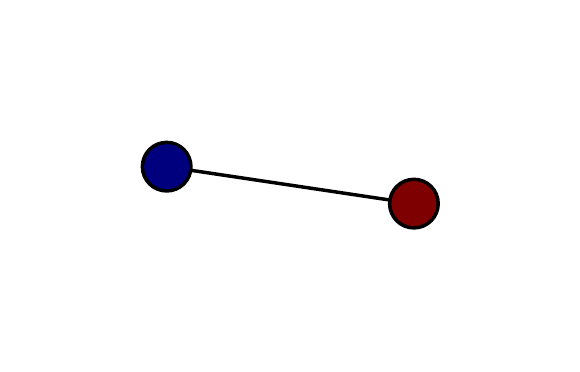}
\caption{\centering \textsc{PL(3.0)} peripheral  bag subgraph.}
\end{subfigure}
\caption{\textsc{PL(2.5) and PL(3.0)} bag subgraphs, colored by $k$-core number of
the node.}
\label{fig:pl_30_bags}\label{fig:pl_25_bags}
\end{center}
\end{figure}

As the power law exponent $\gamma$ is increased, recall that the
amount of degree heterogeneity in the resulting network is reduced, i.e.,
the number of high degree nodes specified by the power
law degree distribution is decreased.
As a necessary consequence of maintaining this distribution as the
nodes are connected, the high degree nodes are likely to be connected
to other high degree nodes.  This causes a core-periphery structure to
emerge (see \cite{Adcock13_icdm} for empirical measurements between the
relationship between $\gamma$ and the $k$-core structure).

For example, the core-periphery structure of \textsc{PL(3.0)} is
shallower than that of \textsc{PL(2.5)}, as seen in Figure
\ref{fig:bag_k_core}.  Similarly, the width of the \textsc{PL(2.5)}
\textsc{amd} TD is larger than that of the \textsc{PL(3.0)}
\textsc{amd} TD.  In all cases these widths are less
than the corresponding \textsc{ER(2)} network, whereas one might
expect these networks to have larger widths because of the increased
core-periphery structure.  This occurs because there are several
factors to consider as $\gamma$ is decreased.  The core \emph{does}
become denser as more edges are added to the core, causing these nodes
to become more difficult to separate; but most of those extra edges
come from the outer regions of the expander-like core, thus shrinking
the size of the core and increasing the size of the periphery.  In
other words, when only a few medium degree nodes are added, then there
are still cyclical structures in the core, as is observed in
\textsc{ER(1.6)}, except smaller; but as higher degree nodes are
added, the core becomes denser and begins to become larger, as this
forces larger and larger pieces of the core to be placed in the same
bag, as is observed in \textsc{ER(32)}.

The TD results on \textsc{ER} and \textsc{PL} graphs demonstrate that
TDs (in particular, with the \textsc{amd} heuristic) can capture the
core-periphery structure of two common random network models.  In both
cases, to the extent that there was a core-periphery structure (which
itself depended on the sparsity parameter $p$ or the degree
heterogeneity parameter $\gamma$), the central and peripheral bags in
the TD from \textsc{amd} were correlated with this structure.  The
peripheral bags were smaller and much sparser than the central bags
and contained nodes from the shallow (low) $k$-cores of the network.
With the exception of the tree-like periphery of the extremely sparse
\textsc{ER} and the \textsc{PL} networks, the structures observed in
TDs of the random network models were largely driven by loosely
connected core structures (e.g., overlapping loops in sparser regions and
expander-like cores in the denser regions).  This is consistent both
with the results on the toy networks and previous work involving the
structure of ER networks \cite{erdos60random, Percus08}.  (The one exception
to this is the denser \textsc{ER(32)}, where the central bag contained 
77\% of the network; in this case, the network does not exhibit the
core-periphery structure of the other networks looked at in this
section.)  We will see how we obtain similar results when applying the
\textsc{amd} heuristic to real-world networks.

\section{Tree decompositions of real-world networks }
\label{sec:real_results}

In this section, we will describe the results of using a variety of TD 
heuristics on a set of real-world networks.  
Our goal is to use the insights from the previous sections to evaluate the 
performance of existing TD heuristics on real social and information 
networks and to understand how those TDs can be used to obtain an improved 
understanding of the properties of these realistic networks.

Our main results in this section are three-fold.
First, in Section~\ref{sec:real_results-core_periphery},
we summarize results of a detailed empirical evaluation of the 
\textsc{amd} TD heuristic applied to our suite of realistic networks.%
\footnote{We should note that we ran these computations with many different TD heuristics.  In most of this section, however, we only show results from (the most scalable) \textsc{amd} heuristic.  This is simply for brevity.  There were some differences from heuristic to heuristic, but we feel this one is representative of the type of behavior found.}
The main focus is to illustrate how these TDs capture 
previously-identified core-periphery structure, and also to
illustrate how the internal structure of TD bags can be understood in 
terms of large-scale cycles and small-scale clustering in the
original graph.
Second, in Section~\ref{sec:real_results-ncp},
we evaluate the ability of \textsc{amd} to identify 
small-scale good-conductance communities such as those 
previously-identified by the NCP with local spectral 
methods~\cite{LLDM09_communities_IM,Jeub15}.
We show connections between bags that are more 
peripheral in the TD and small good-conductance communities responsible for 
``dips'' in the NCP.
Third, in Section~\ref{sec:real_results-gt},
we illustrate that TD heuristics can be used to identify 
certain other types of large-scale non-conductance-based 
``ground truth'' communities.
In particular, we will show connections between bags that are more
central in the TD and large-scale community-like (by a ``ground truth'' 
metric but not by conductance quality) clusters.

\subsection{Results on identifying core-periphery structure}
\label{sec:real_results-core_periphery}

Here, we will describe the results of an empirical evaluation of the
\textsc{amd} TD heuristic applied to our suite of real-world networks (those
in Table \ref{tbl:networks-basic-stats}). 
We will begin in Table~\ref{tbl:td_real_results} with a brief survey of all 
of our real networks, and we will then focus on four representative 
networks: \textsc{as20000102}, \textsc{CA-GrQc}, \textsc{FB-Lehigh}, and 
\textsc{PowerGrid}.
The first three all exhibit some form of 
previously-recognized core-periphery structure~\cite{Adcock13_icdm}, while 
\textsc{PowerGrid} is known to lack a strong core-periphery structure 
(basically since it is heavily tied to the underlying 
locally-Euclidean geometry of the Earth~\cite{Adcock13_icdm}).  
Note, though, that the Facebook networks are very core-heavy, in the sense 
that they have many nodes in deep cores, essentially because of their 
significantly higher average degree (see, e.g., \cite{Adcock13_icdm} and 
Table \ref{tbl:networks-basic-stats}). 
(Thus, informed by previous results on $k$-core 
decompositions and related tree-like 
techniques~\cite{Adcock13_icdm, ST08, LLDM09_communities_IM, MSL11},  
we expect to see evidence of the core-periphery structure in the TDs associated 
with \textsc{as20000102}, \textsc{CA-GrQc}, and, to a lesser extent, 
\textsc{FB-Lehigh}, but a lack of substantial core-periphery structure in the 
TD of \textsc{PowerGrid}.)

\subsubsection{Overview of core-periphery results for all networks}

In Table \ref{tbl:td_real_results}, we present 
the number of bags in the \textsc{amd} TD ($N_{\textsc{amd}}$),
the maximum eccentricity (diameter) of the TD ($E_{\textsc{amd}}$), 
the maximum and median width of the TD ($W$ and $\tilde{W}$, respectively),
and median bag density ($\tilde{D}$). 
These measurements provide us with an idea of how large is the most connected 
part of the network (maximum width); how numerous are small bags (median 
width), which is indicative of areas of the network that have small 
separators, in our case small peripheral regions of the network; and whether 
the small separators are more clique-like or consist of mostly disjoint 
nodes (median density), with disjoint nodes being indicative of cycles and 
clique-like structures being indicative of more meaningful communities.  
A large maximum width combined with a low median width is 
evidence for a deep core and a shallow periphery, and high median density 
is evidence for a periphery based on more community-like separators, rather 
than more disparate separators.  
These observations assume that high density bags are mostly 
small-width bags.
This assumption is plausible, given that as the width $w$ of a bag 
increases, the number of edges required to maintain a constant density 
increases like $w^2$; and in many cases we have confirmed this assumption 
indirectly or by direct observation.
For example, see our discussion of bag density and bag width below, as 
well as Figure \ref{fig:bag_density} below for empirical evidence that high 
density bags are generally the smallest width bags.  

\begin{table}[!h]
\hskip1em
\begin{center}
{\footnotesize
\begin{tabular}{l|r|r|r|r|r|}
Network & $N_{\textsc{amd}}$ & $E_{\textsc{amd}}$ & $W$ &
$\tilde{W}$ & $\tilde{D}$ \\
\hline \hline
\textsc{CA-GrQc} & 3014 & 39 & 222 & 2 & 1.0 \\  
\textsc{CA-AstroPh} & 10708 & 78 & 3616 & 5 & 1.0 \\
\textsc{as20000102} & 6364 & 33 & 88 & 2 & 1.0 \\
\textsc{Gnutella} & 6475 & 33 & 1629 & 2 & 0.67 \\
\textsc{Email-Enron} & 26781 & 78 & 2237 & 3 & 1.0 \\
\hline \hline
\textsc{FB-Caltech} & 395 & 30 & 357 & 18 & 0.53 \\
\textsc{FB-Haverford} & 516 & 56 & 891 & 37.5 & 0.38 \\
\textsc{FB-Lehigh}  & 1919 & 151 & 2983 & 31 & 0.32 \\
\textsc{FB-Rice} & 1481 & 76 & 2553 & 31 & 0.37 \\
\textsc{FB-Stanford} & 4809 & 100 & 6674 & 16 & 0.38 \\
\hline \hline
\textsc{PowerGrid} & 4666 & 59 & 21 & 2 & 0.67 \\
\textsc{Polblogs} & 899 & 49 & 294 & 6 & 0.57 \\
\textsc{road-TX} & $1.25 \times 10^6$ & 170 & 197 & 3 & 0.5 \\
\textsc{web-Stanford} & $2.10 \times 10^5$ & 500 & 1419 & 5 & 0.83 \\
\end{tabular}}
\end{center}
\caption{Statistics for TDs of real networks. Notation is the same as in Table~\ref{tbl:density_results}.}
\label{tbl:td_real_results}
\end{table}

Networks based on an underlying Euclidean geometry (e.g., \textsc{road-TX},
\textsc{PowerGrid}) have low maximum widths and low median widths, which 
indicates that they do \emph{not} have a strong core-periphery structure.  
While these networks have many small width bags, which is indicative of 
tree-like sections of the network, the internal subgraphs have a low 
median density (e.g., as compared to certain ER networks).
More social networks, such as \textsc{Polblogs} and the Facebook networks,
all have higher average degrees and, consequently, higher widths, with lower 
median widths.  This is one indicator of a core-periphery structure.  
As the median widths are higher in these networks, as compared to the other 
real networks (although lower than the ER 
networks), the peripheral structure in these social graphs tends to be 
denser than in the other networks.
Also, the median density, while low compared to the other real networks, is 
very high compared to the median density of the densest \textsc{ER} networks 
in Table \ref{tbl:er_density_results}.  
Thus, although the periphery is more difficult to separate into small
good-conductance community-like clusters than some of the other real 
networks, e.g., \textsc{CA-GrQc} or \textsc{CA-AstroPh}, it is still formed 
from more community-like pieces than similar ER (or PL) networks.  

Observe also that the two web networks, \textsc{Gnutella} and 
\textsc{web-Stanford}, have high widths, like the Facebook networks, but 
lower median widths; and that they have higher median densities (especially 
when compared to \textsc{ER(4), ER(8)}, and \textsc{ER(16)}).
This indicates that these networks have a sparser, more tree-like periphery 
than other social networks, which is also consistent with previous
results~\cite{LLDM09_communities_IM}.  
Importantly, and also consistent with previous
results~\cite{LLDM09_communities_IM}, is that the sparsest, most tree-like 
peripheries belong to the collaboration, email, and autonomous systems 
networks (\textsc{CA-GrQc, CA-AstroPh, Email-Enron, as20000102}).
These networks all have low median widths, high median densities, and high 
maximum widths, indicating that they exhibit the cleanest core-periphery 
structure, also consistent with the 
upward-sloping NCPs~\cite{LLDM09_communities_IM,Jeub15}.

\subsubsection{More details on core-periphery structure of four representative networks}

We will now look 
at several representative networks in greater detail.  
Let us start by discussing Figures~\ref{fig:bag_hist}, \ref{fig:bag_density}, 
and~\ref{fig:bag_k_core} from Section \ref{sec:synth_results}.  
Figure \ref{fig:bag_hist} clearly shows that most of the bags in the TDs 
are small-width bags.
In fact, proportionally, \textsc{FB-Lehigh} has the largest fraction of 
large bags, and yet 80\% of the bags are below width 200 (in a graph with 
5073 nodes, where the TD has a maximum width bag of 2983).  
Since the bags and edges of TDs form separators in the network, this indicates 
that there are many relatively small separators.
These are the largest in Facebook networks, where the separators tend to have 
around 100 nodes, while in most other networks many of the separators have 
around 10 nodes.
This is consistent with typical views of core-periphery structure, with a few 
more highly-connected nodes in the core and many less well-connected nodes in 
the periphery.  
That is, in order to separate off most pieces of the periphery (where 
``piece'' is defined by the end of branches in the TD), only 10 or fewer 
nodes are needed for most of the social/information networks, while ca. 
100 nodes are needed for the Facebook networks.

In Figure \ref{fig:bag_density}, we see the average edge density for bags of 
a given cardinality plotted against the bag cardinality, showing
that small-width bags have 
high densities.  
An important distinguishing feature of the three representative real networks 
(that are not tied to an underlying Euclidean geometry) is that the curve has a
heavier tail than in the synthetic networks.  
This indicates that separators, up to much larger size scales, are less 
disparate (e.g., are denser or clumpier) than in the synthetic networks.  
In \textsc{PowerGrid}, on the other hand, the underlying Euclidean geometry 
leads the density to falls off more quickly.
It falls off similarly to the sparse ER network, except that
the tail of the curve is shorter.  
In this case, only the smallest bags have tight separators.  

In Figure \ref{fig:bag_k_core}, we consider the relationship between the 
core structure and low eccentricity (central) bags, and we compare that 
with the relationship between the periphery structure and the high 
eccentricity (perimeter) bags in the TD.
Figures~\ref{fig:as_k_core} and~\ref{fig:cagrqc_k_core}
show that for \textsc{as20000102} and \textsc{CA-GrQc} there is a clear 
downward trend as the bag eccentricity is increased.  
This indicates that low eccentricity bags contain more high $k$-core nodes 
on average and that the high eccentricity bags contain more low $k$-core 
nodes on average.  
Figure \ref{fig:lehigh_k_core} shows that \textsc{FB-Lehigh}, due to 
its greater density, has a mostly flat profile until the most extreme 
reaches of eccentricity are met, at which point some of the bags begin to 
contain nodes of a lower $k$-core.  
Thus, the core-periphery structure is present in \textsc{FB-Lehigh}, but 
the core-periphery structure is moderated by a very large core which 
produces long path-like sets of nodes that in turn lead to large core 
bags and hence a much larger eccentricity.
(This is typical of the results for most of the Facebook 
networks, which is consistent with their flat NCP~\cite{Jeub15}.)  
Finally, \textsc{PowerGrid}, which is not expected to have exhibit a 
correlation between $k$-core structure and bag eccentricity, has a flat 
profile.  

To illustrate these findings, we present 
visualizations in Figures~\ref{fig:cagrqc_bags} and \ref{fig:power_bags}.
Shown are a central or very deep core bag, a perimeter or 
very peripheral bag, and an intermediate bag, for each of our four 
networks.
These figures show the community-like nature of typical bags for the 
three information networks, i.e., \textsc{as20000102}, 
\textsc{CA-GrQc}, and \textsc{FB-Lehigh}, as well as the more disparate 
separators of the \textsc{PowerGrid}.
The coloring of the visualizations in these figures is by $k$-core:
the red nodes are in deep (high) $k$-cores while the blue nodes are in 
shallow (low) $k$-cores.

\begin{figure}[!htb]
\begin{center}
\begin{subfigure}[h!]{0.15\textwidth}
\includegraphics[width=\textwidth]{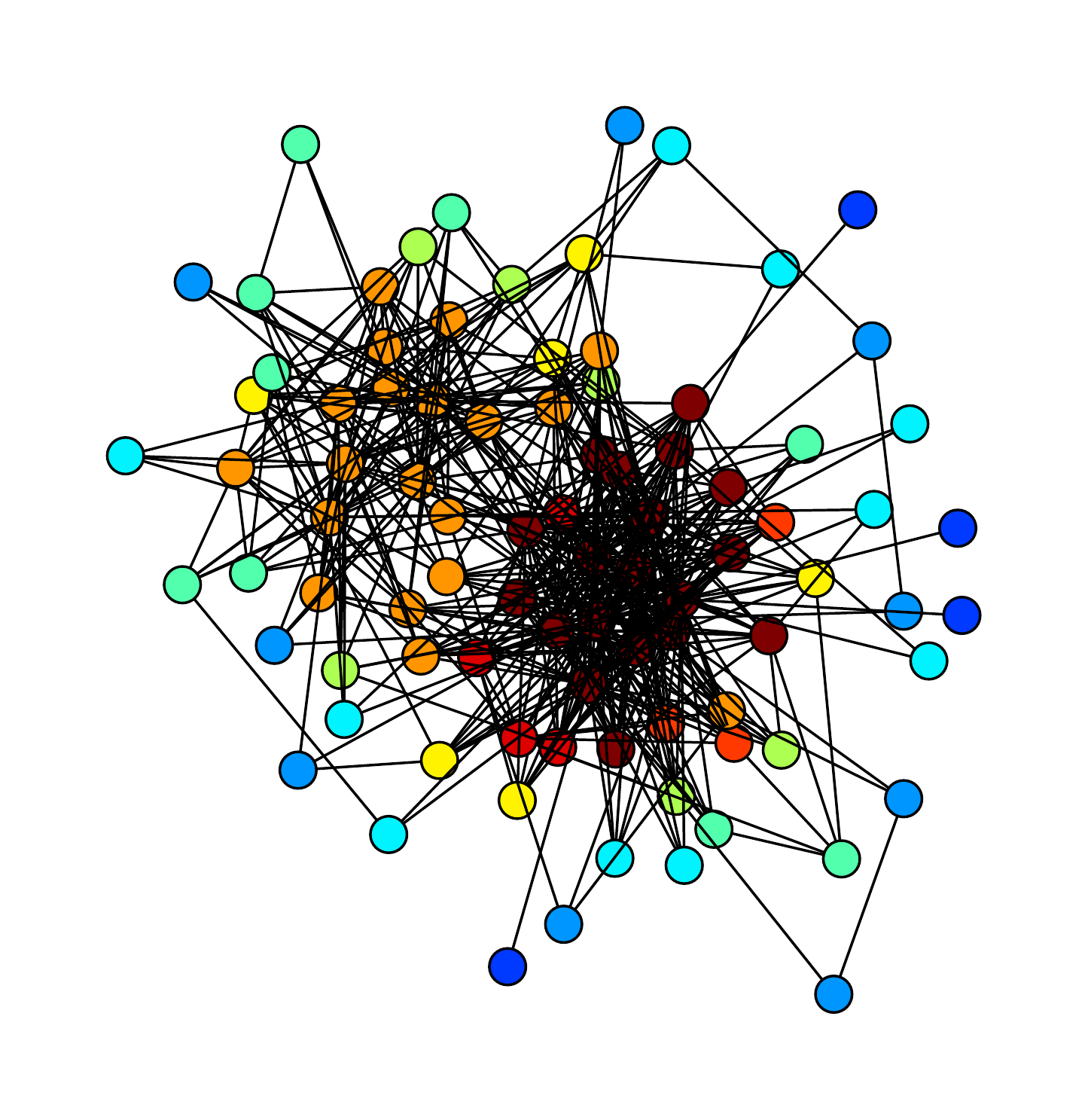}
\caption{\centering \textsc{as20000102} central bag}
\end{subfigure}
\begin{subfigure}[h!]{0.15\textwidth}
\includegraphics[width=\textwidth]{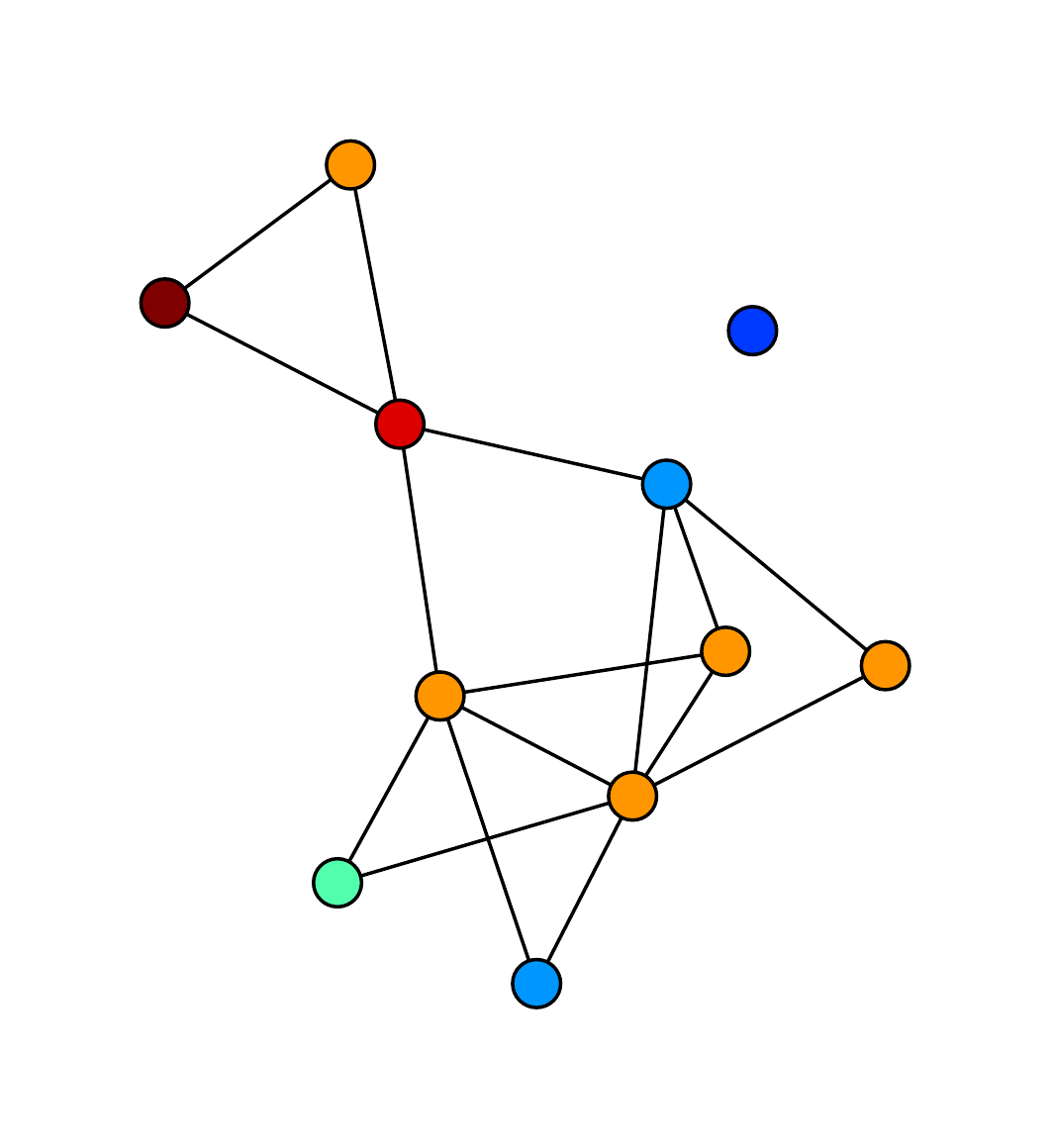}
\caption{\centering \textsc{as20000102} intermediate bag}
\end{subfigure}
\begin{subfigure}[h!]{0.15\textwidth}
\includegraphics[width=\textwidth]{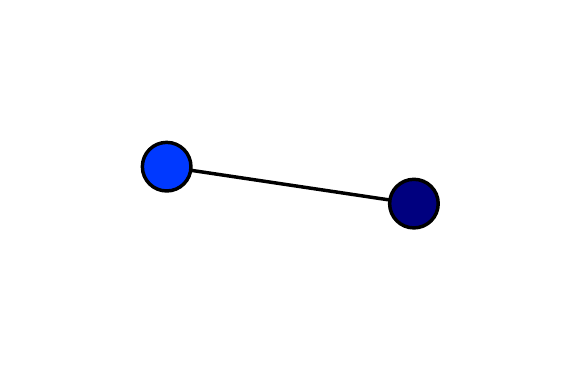}
\caption{\centering \textsc{as20000102} perimeter bag}
\end{subfigure}
\begin{subfigure}[h!]{0.15\textwidth}
  \includegraphics[width=\textwidth]{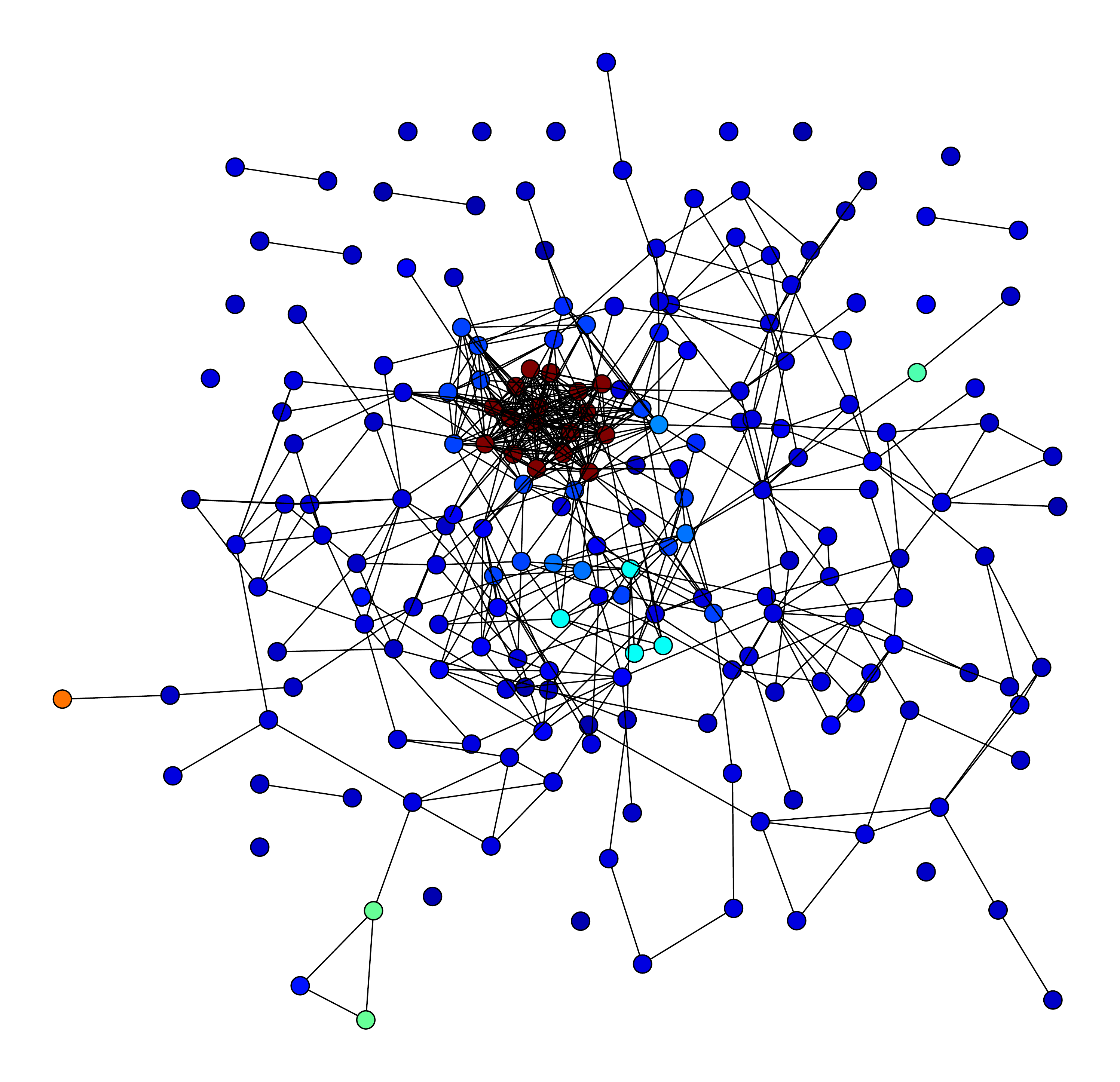}
\caption{\centering \textsc{CA-GrQc} central bag}
\end{subfigure}
\begin{subfigure}[h!]{0.15\textwidth}
\includegraphics[width=\textwidth]{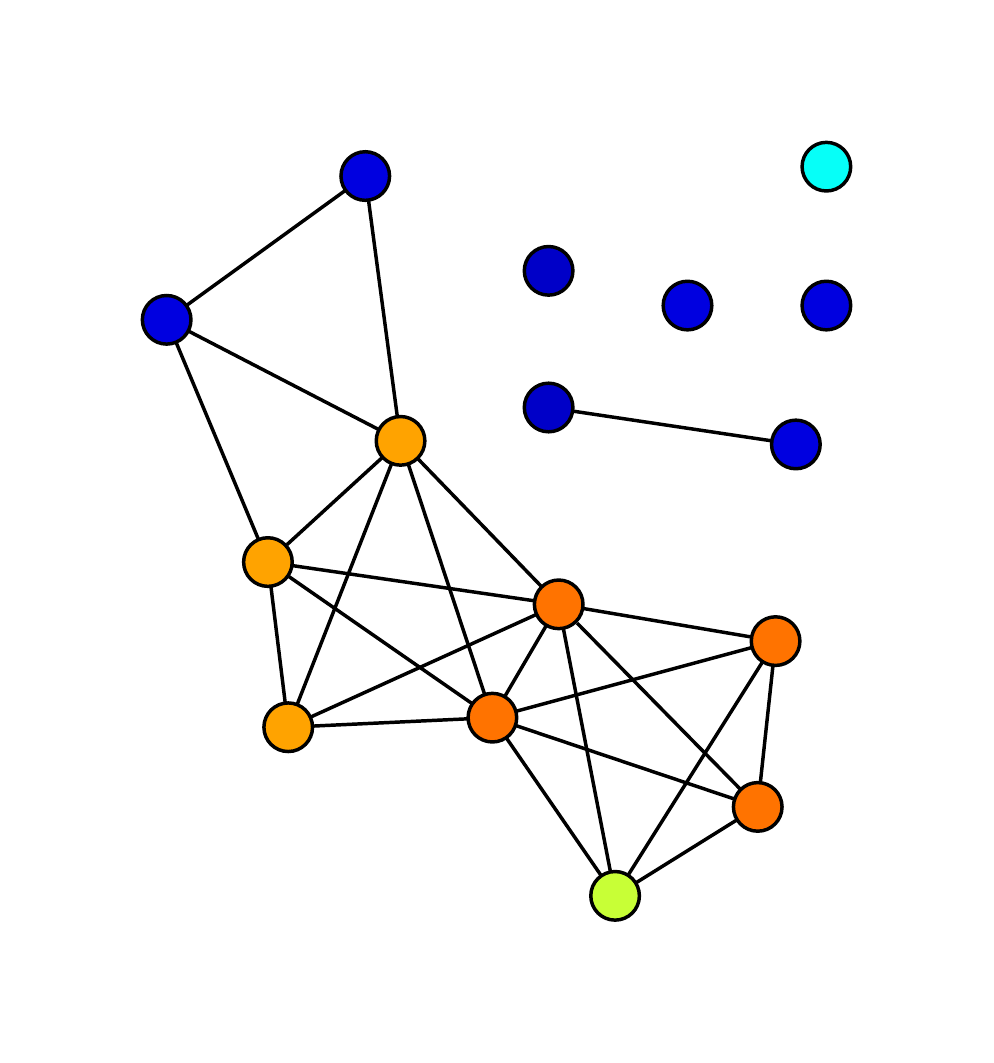}
\caption{\centering \textsc{CA-GrQc} intermediate bag}
\end{subfigure}
\begin{subfigure}[h!]{0.15\textwidth}
\includegraphics[width=\textwidth]{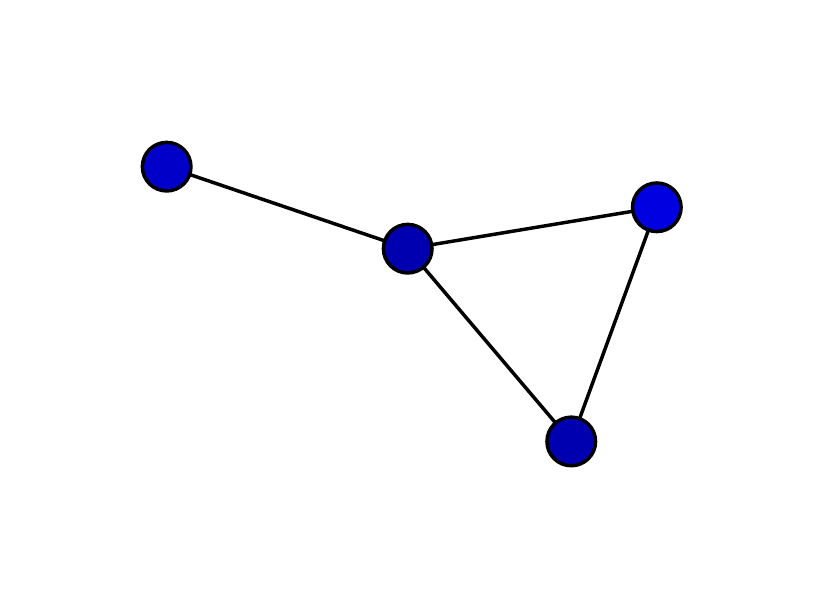}
\caption{\centering \textsc{CA-GrQc} perimeter bag}
\end{subfigure}
\caption{\textsc{as20000102}  and \textsc{CA-GrQc} \textsc{amd} bag subgraphs, colored by $k$-core number, with red indicating deep/high $k$-cores and blue indicating shallow/low $k$-cores.}
\label{fig:cagrqc_bags}\label{fig:as_bags}
\end{center}
\end{figure}

\begin{figure}[!htb]
\begin{center}
\begin{subfigure}[h]{0.15\textwidth}
\includegraphics[width=\textwidth]{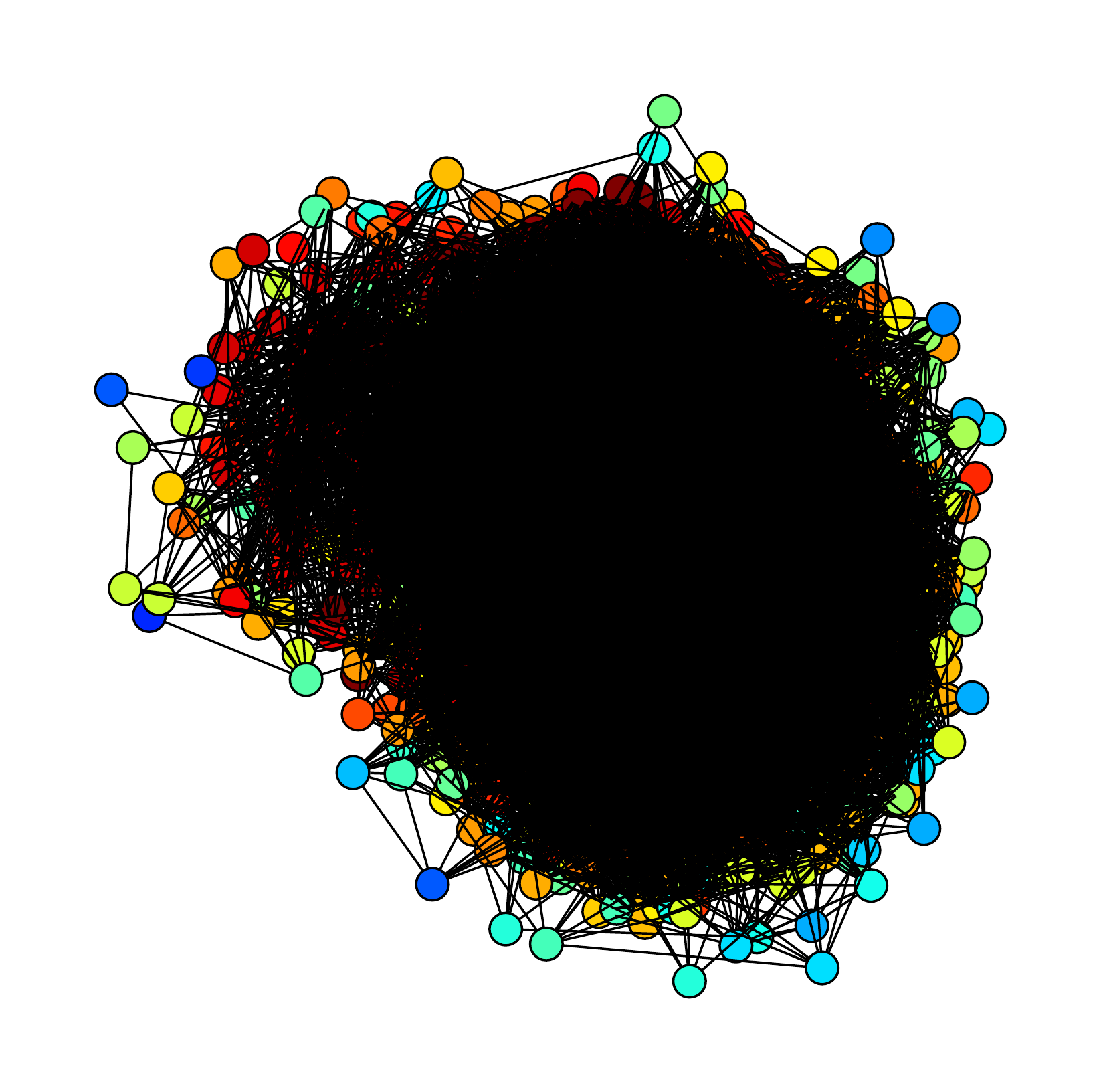}
\caption{\centering \textsc{FB-Lehigh} central bag}
\end{subfigure}
\begin{subfigure}[h!]{0.15\textwidth}
\includegraphics[width=\textwidth]{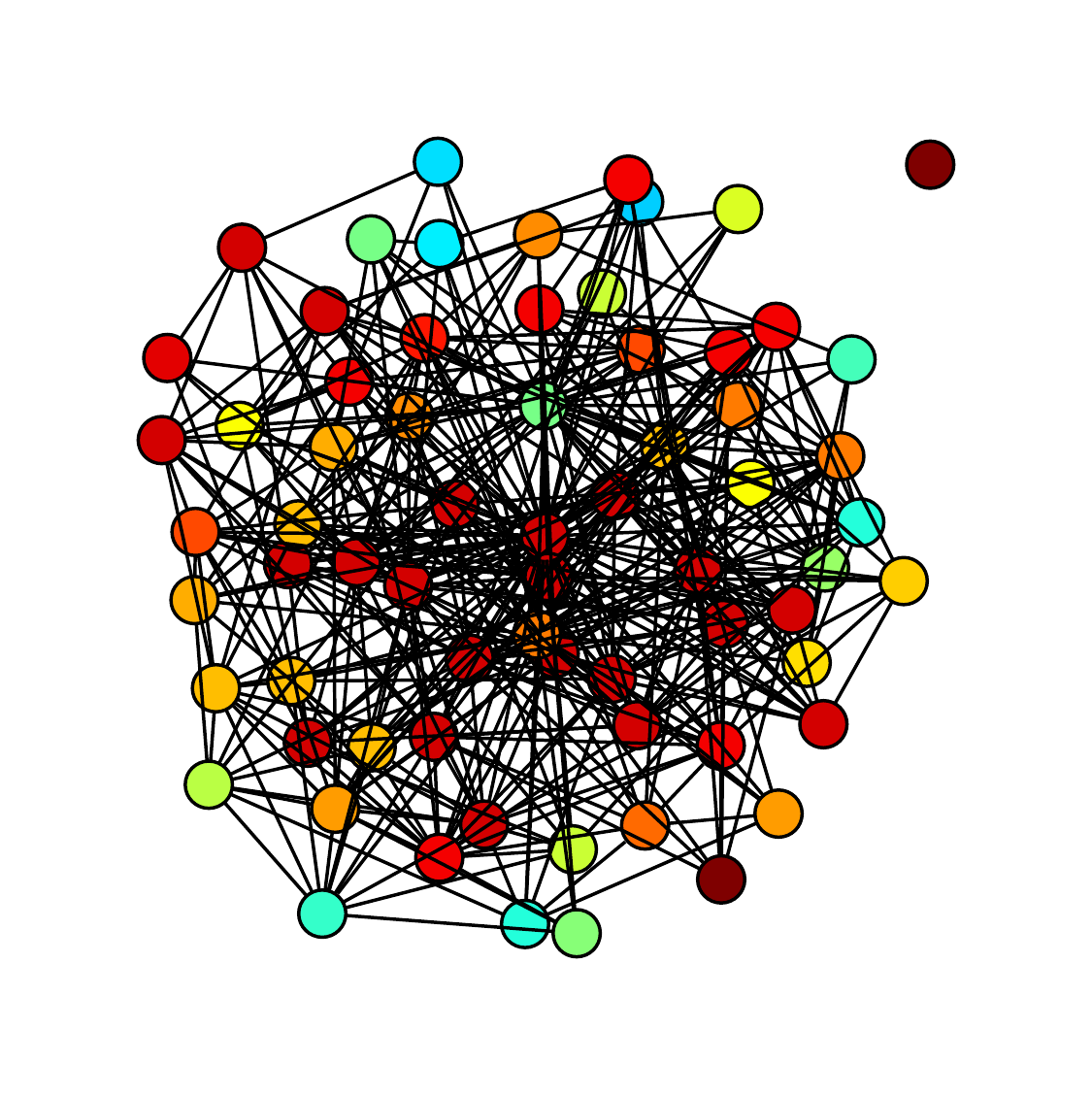}
\caption{\centering \textsc{FB-Lehigh} intermediate bag}
\end{subfigure}
\begin{subfigure}[h!]{0.15\textwidth}
\includegraphics[width=\textwidth]{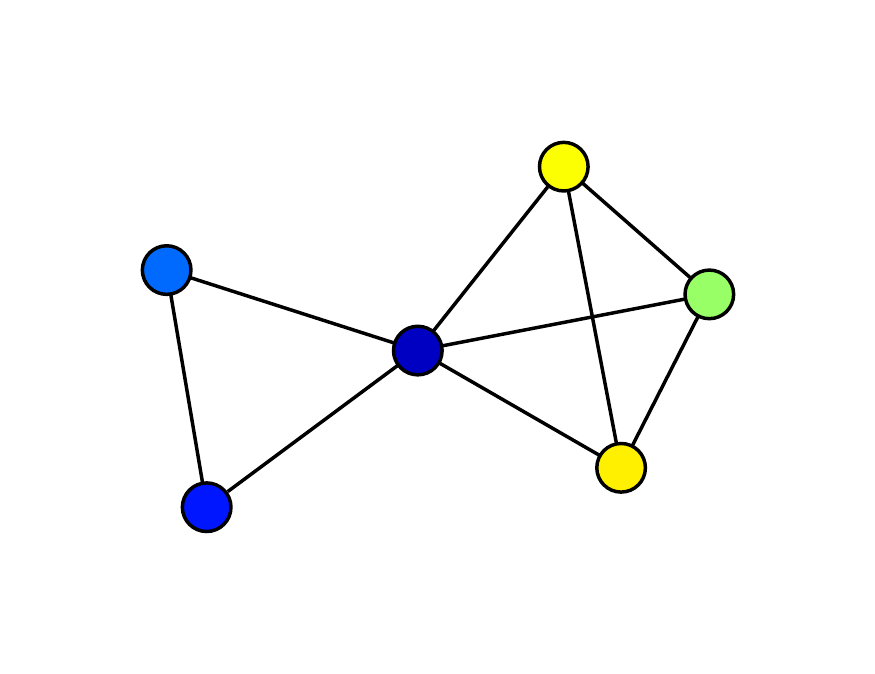}
\caption{\centering \textsc{FB-Lehigh} perimeter bag}
\end{subfigure}
\begin{subfigure}[h!]{0.15\textwidth}
\includegraphics[width=\textwidth]{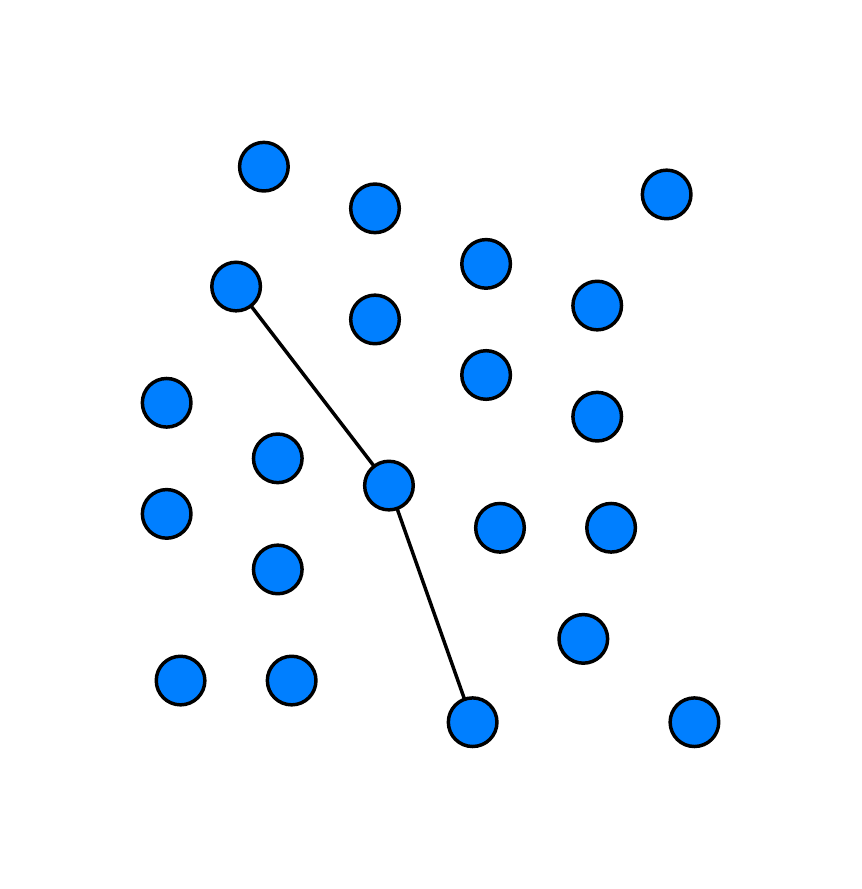}
\caption{\centering \textsc{PowerGrid} central bag}
\end{subfigure}
\begin{subfigure}[h!]{0.15\textwidth}
\includegraphics[width=\textwidth]{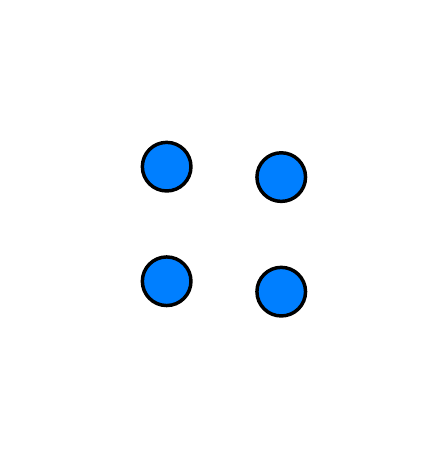}
\caption{\centering \textsc{PowerGrid} intermediate bag}
\end{subfigure}
\begin{subfigure}[h!]{0.15\textwidth}
\includegraphics[width=\textwidth]{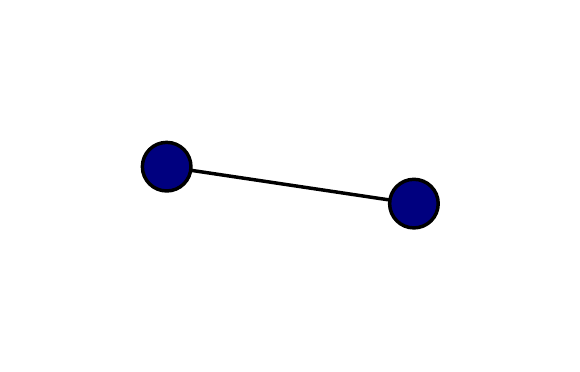}
\caption{\centering \textsc{PowerGrid} perimeter bag}
\end{subfigure}
\caption{\textsc{FB-Lehigh} and \textsc{PowerGrid} \textsc{amd} bag subgraphs, colored by $k$-core number, with red indicating deep/high $k$-cores and blue indicating shallow/low $k$-cores.}
\label{fig:power_bags}\label{fig:lehigh_bags}
\end{center}
\end{figure}

One final observation we would like to make is to address the ``dips'' in 
the average $k$-core curves shown in Figure \ref{fig:bag_k_core} (e.g.,
the dip in Figure \ref{fig:cagrqc_k_core} at a bag eccentricity of 21 or in 
Figure \ref{fig:lehigh_k_core} throughout).
These dips are due to what we will call ``twigs,'' where a twig is a small 
(low width and short) branch off of a much larger (high width and long) 
trunk-like structure of the TD.  
For example, in \textsc{FB-Lehigh} and in the other Facebook networks, the 
high average degree results not only in high widths, but in larger 
collections of bags of high width.  
These are arranged in a long path (a ``trunk'') with many branches at either 
end.  
Along this main trunk, there are occasional twigs which contain peripheral 
nodes.  
Since the trunk is long, the average $k$-core at the point where the twig is
attached is slightly lower, resulting in the dip in the curve.
In Figure \ref{fig:twig}, we provide a visualization of the twigs 
responsible for three of these dips.

\subsubsection{Summary of large-scale core-periphery and tree-like structure in real-world networks}

These empirical observations suggest that many realistic 
social/information networks have a non-trivial core-periphery structure; and 
that in many cases this is caused by many small overlapping cluster-like or 
moderately clique-like structures.
That is, there is local non-tree-like (combinatorial) structure that ``fits together'' into 
a global core-periphery structure that is tree-like (in a metric and/or cut
sense) when viewed from large size scales.  
This is in sharp contrast with many models and intuitions.
Most obviously, this is in contrast with the random networks (in particular, 
the not extremely sparse ER networks and to a lesser extent the PL networks, 
but many other more popular random generative models) which have a locally 
tree-like, but globally loopy structure.
Less obviously, this is also in sharp contrast with networks such as 
\textsc{PowerGrid}, \textsc{PlanarGrid}, and \textsc{road-TX} that are strongly 
tied to an underlying Euclidean geometry.
Said another way, many realistic social/information networks have a more 
tightly-connected core-like structure than is present in typical random 
networks, and they have peripheral and intermediate regions that are 
``clumpier'' than these random networks.  
While these claims are perhaps intuitive, our empirical observations 
demonstrate that they can be meaningfully identified with TDs and interpreted as leading to large-scale 
cut-based tree-like structure.

Interestingly, aside from the local clumpiness, the real-world 
social/information networks do have a core-periphery structure 
that is reminiscent of that which is also seen in extremely sparse ER 
graphs and PL graphs with greater degree heterogeneity.
(This too is consistent with prior results suggesting that extreme sparsity 
coupled with randomness/noise is responsible for the dips in the 
NCP~\cite{LLDM09_communities_IM,Jeub15}.)
It is also worth emphasizing that in most of the intermediate bags of the 
real social/information networks, there are still a small number of 
disconnected nodes.
This indicates that there are still a small number of alternate paths, which 
are disparate from the clusters, to the rest of the nodes in the network.  
The most prominent exceptions to these general observations are networks that 
either do not have a strong core-periphery structure, e.g., \textsc{PowerGrid} 
that is tied to a two-dimensional underlying Euclidean geometry, or networks 
that have a relatively low clustering coefficient, e.g., \textsc{Gnutella09}.  
In both of these cases (but for different reasons), the internal subgraphs of 
the intermediate and peripheral bags have a larger number of disconnected 
nodes than the other realistic networks.

\begin{figure}[!htb]
\begin{center}
\begin{subfigure}[h!]{0.3\textwidth}
\includegraphics[width=\textwidth]{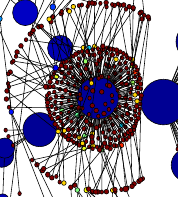}
\caption{\centering \textsc{CA-GrQc} twigs on central bag.}
\end{subfigure}
\begin{subfigure}[h!]{0.3\textwidth}
\includegraphics[width=\textwidth]{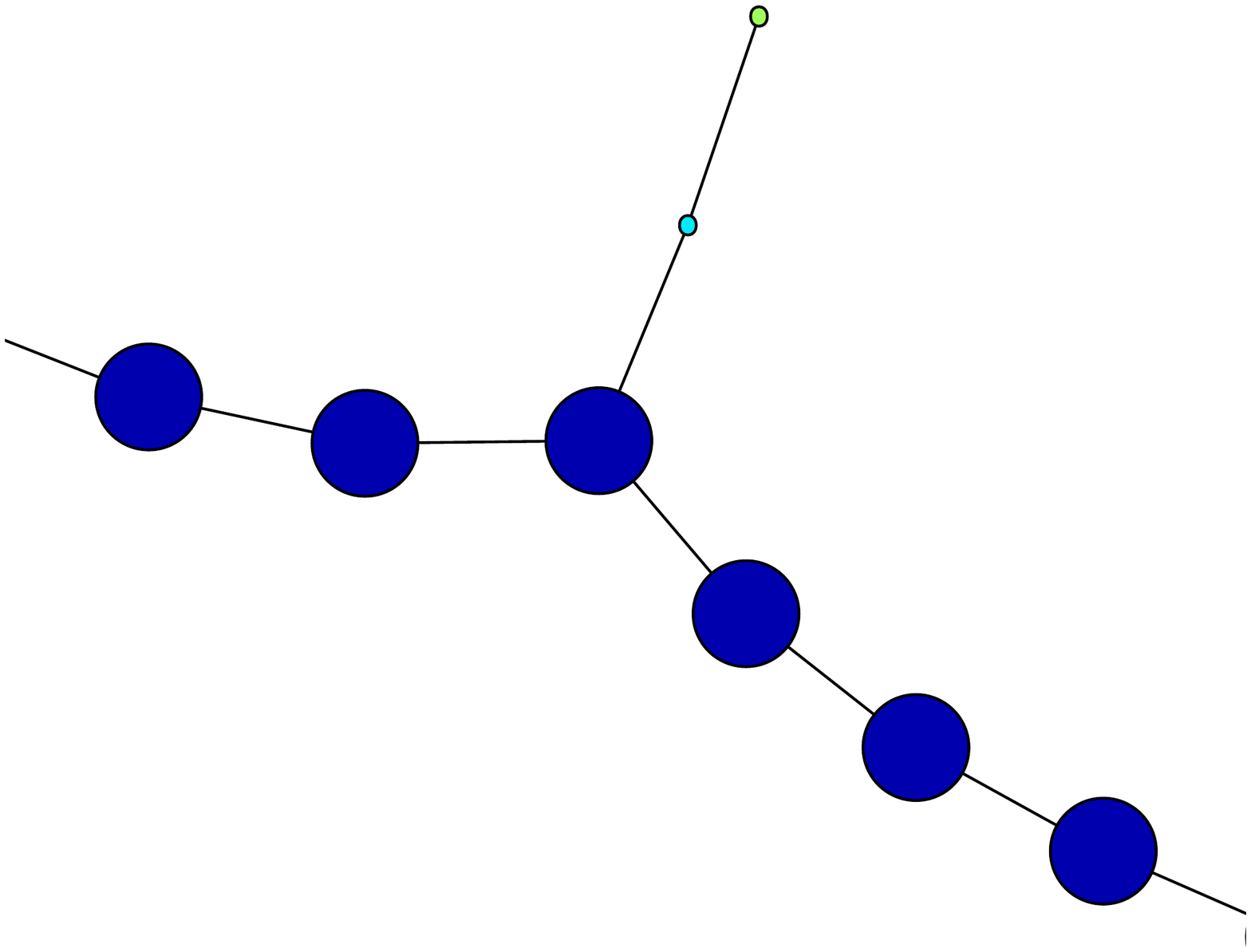}
\caption{\centering \textsc{FB-Lehigh} small twig on central trunk.}
\end{subfigure}
\begin{subfigure}[h!]{0.3\textwidth}
\includegraphics[width=\textwidth]{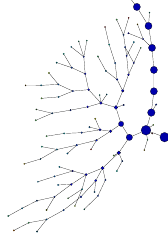}
\caption{\centering \textsc{FB-Lehigh} large branching twig on central
  trunk.}
\end{subfigure}
\caption{Twigs on \textsc{FB-Lehigh} and \textsc{CA-GrQc}.
  Bags are colored by the density, blue indicating
  low density and red indicating high density.  A small twig and a
  larger, branching twig on the \textsc{FB-Lehigh} trunk are shown.
  These twigs, combined with the long, path-like trunk,
  cause dips in the $k$-core eccentricity plot.  In
  \textsc{CA-GrQc}, the concentration of the twigs on one or two
  central bags causes only a single, large dip as compared to the
  multiple dips in \textsc{FB-Lehigh}.  Synthetic networks in
  Figure \ref{fig:bag_k_core}, (\textsc{PL(2.5), PL(3.0),} and
  \textsc{ER(1.6)}) also have twigs similar to \textsc{CA-GrQc}.}
\label{fig:twig}
\end{center}
\end{figure}

\subsection{Connections with good-conductance communities results}
\label{sec:real_results-ncp}

Here, we will consider how the peripheral part of the tree-like 
core-periphery structure identified by TDs relates to low-conductance 
clusters/communities that were previously-identified by the NCP 
method~\cite{LLDM09_communities_IM,Jeub15}.
To do so, observe that one way to determine whether a TD ``captures'' 
clustering/community structure is to see if those clusters/communities are 
well-localized in the TD.  By ``well-localized,'' we mean here that the
cluster/community is contained in a relatively small number of (contiguous) 
bags.  
We followed previous personalized page rank (PPR) local spectral procedures~\cite{andersen06local}
to generate a set of candidate
clusters~\cite{LLDM09_communities_IM,Jeub15}.
Then, given a set of candidate clusters, we looked at how many bags in the 
TD contain at least one node from this cluster, i.e., we measured how 
well-localized the community is in the TD.  
As a crude threshold of whether a cluster/community is localized, we consider 
it to be localized if it is contained in fewer bags than there are 
nodes in the community.%
\footnote{To understand this threshold, consider the
  following example: if a community of size $n$ is a tree, e.g.,
  whiskers in \textsc{ER(1.6)}, then it will be contained in $n$ bags
  in the (ideal) TD; if the community is a ``clique
  whisker,'' i.e., a clique connected to the rest of the network by
  only one edge, it will be contained in just one or two bags; and if
  the community contains deep core nodes which are connected to many
  nodes outside of the community, the community will be spread across
  many bags in the network.
  Other measures of TD locality showed similar
  results.}  
We apply this method using the \textsc{amd} heuristic.

Our results for several real-world and synthetic networks are presented in
Figures~\ref{fig:synth_td_local}--\ref{fig:planar_td_local}.  
For each figure/subfigure, the horizontal axis represents community 
size in number of nodes (on log scale), and the vertical axis is either
the conductance of the best community found using the PPR method (recall
that a low conductance 
represents a better community) or the number of number of bags that contain 
members of the community (again, on log scale).  
In the bag plots, the red line represents the number of bags 
which contain a node \emph{from the community in the corresponding NCP plot}, 
and the green dashed line represents the locality threshold.
When the number of bags for a given community is localized by our
definition, the red plot will be below the green threshold.

\begin{figure}[!htb]
\begin{center}
\begin{subfigure}[h]{0.23\textwidth}
\includegraphics[width=\textwidth]{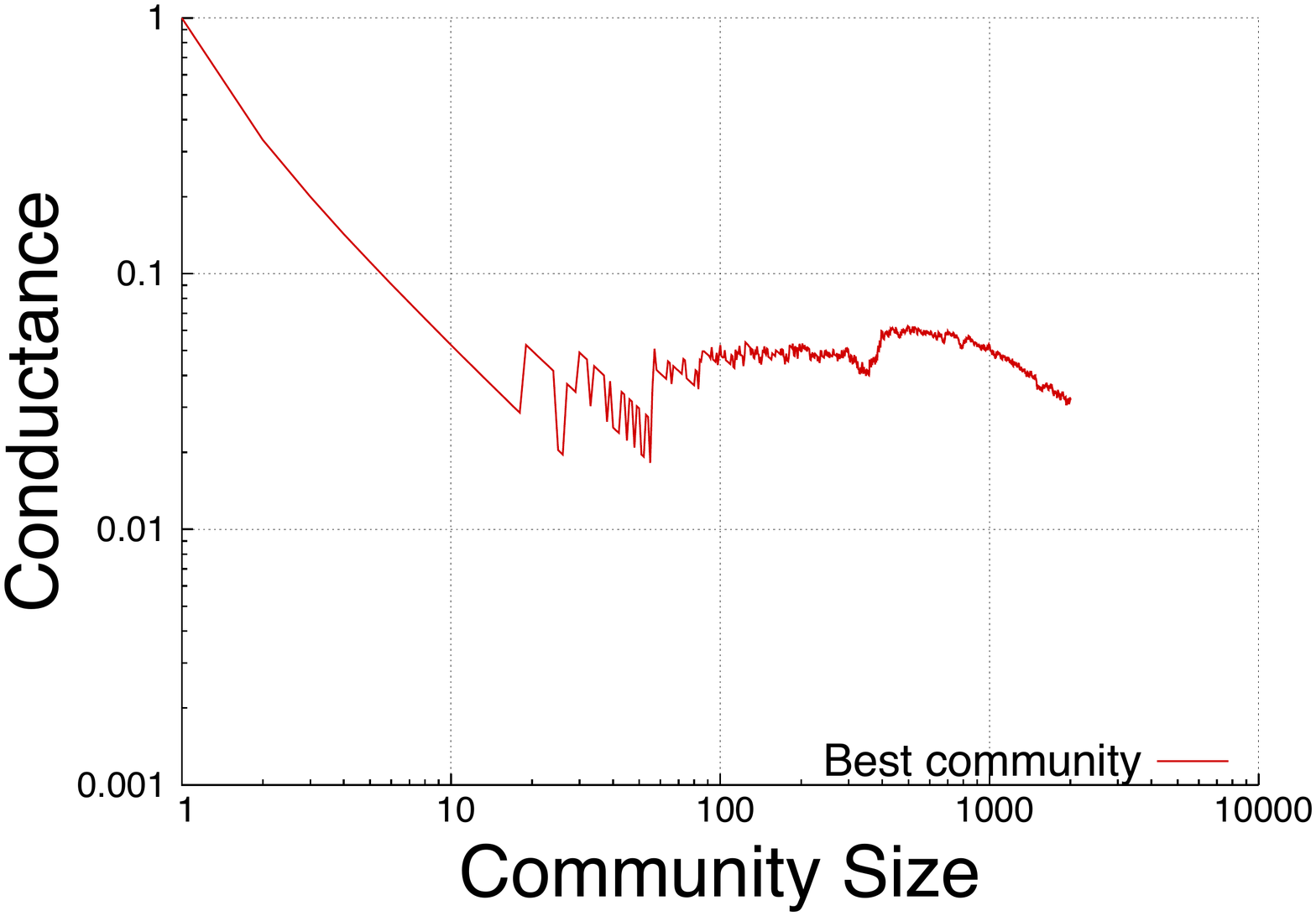}
\caption{\centering \textsc{ER(1.6)} NCP plot}
\end{subfigure}
\begin{subfigure}[h]{0.23\textwidth}
\includegraphics[width=\textwidth]{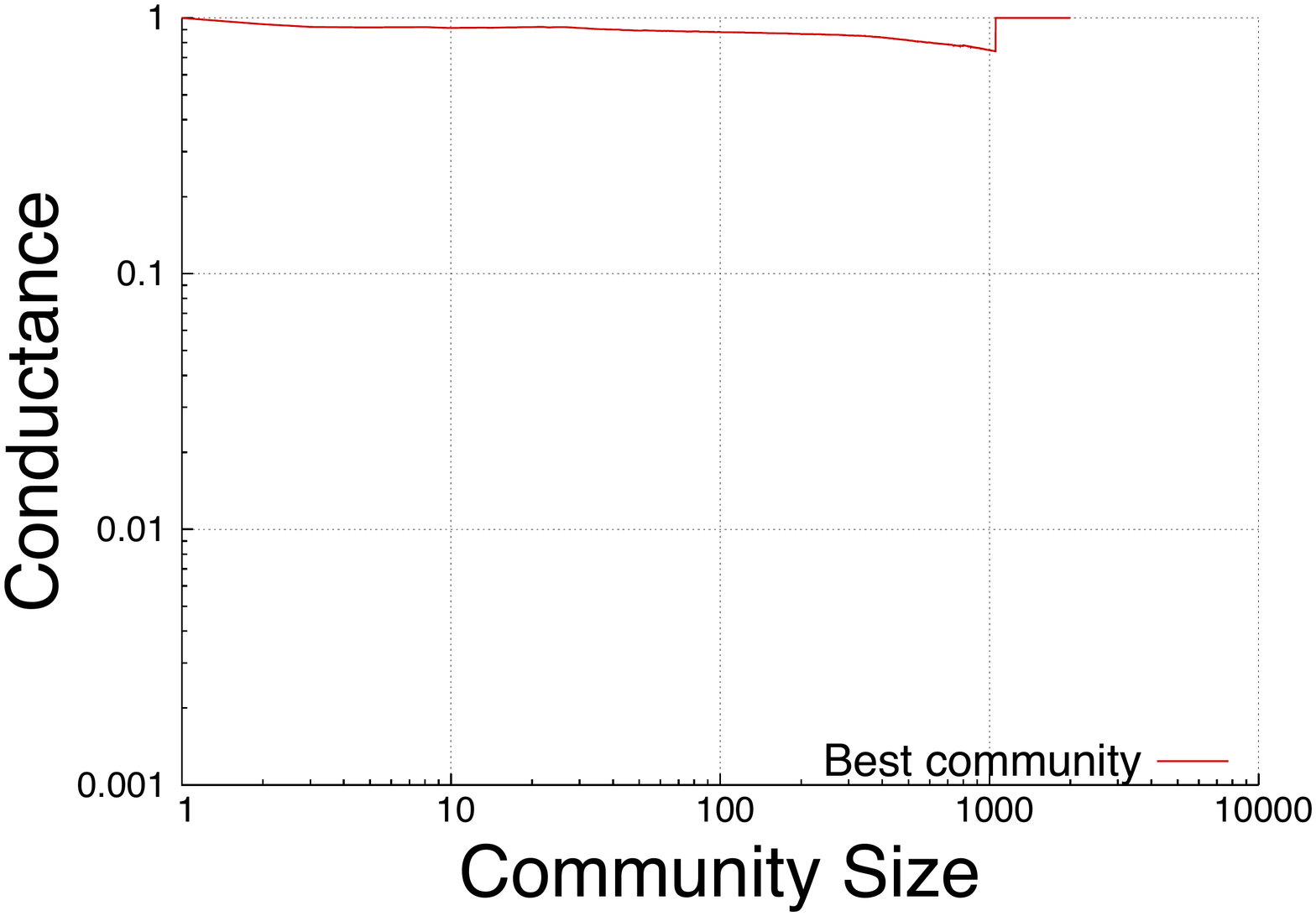}
\caption{\centering \textsc{ER(32)} NCP plot}
\end{subfigure}
\begin{subfigure}[h]{0.23\textwidth}
\includegraphics[width=\textwidth]{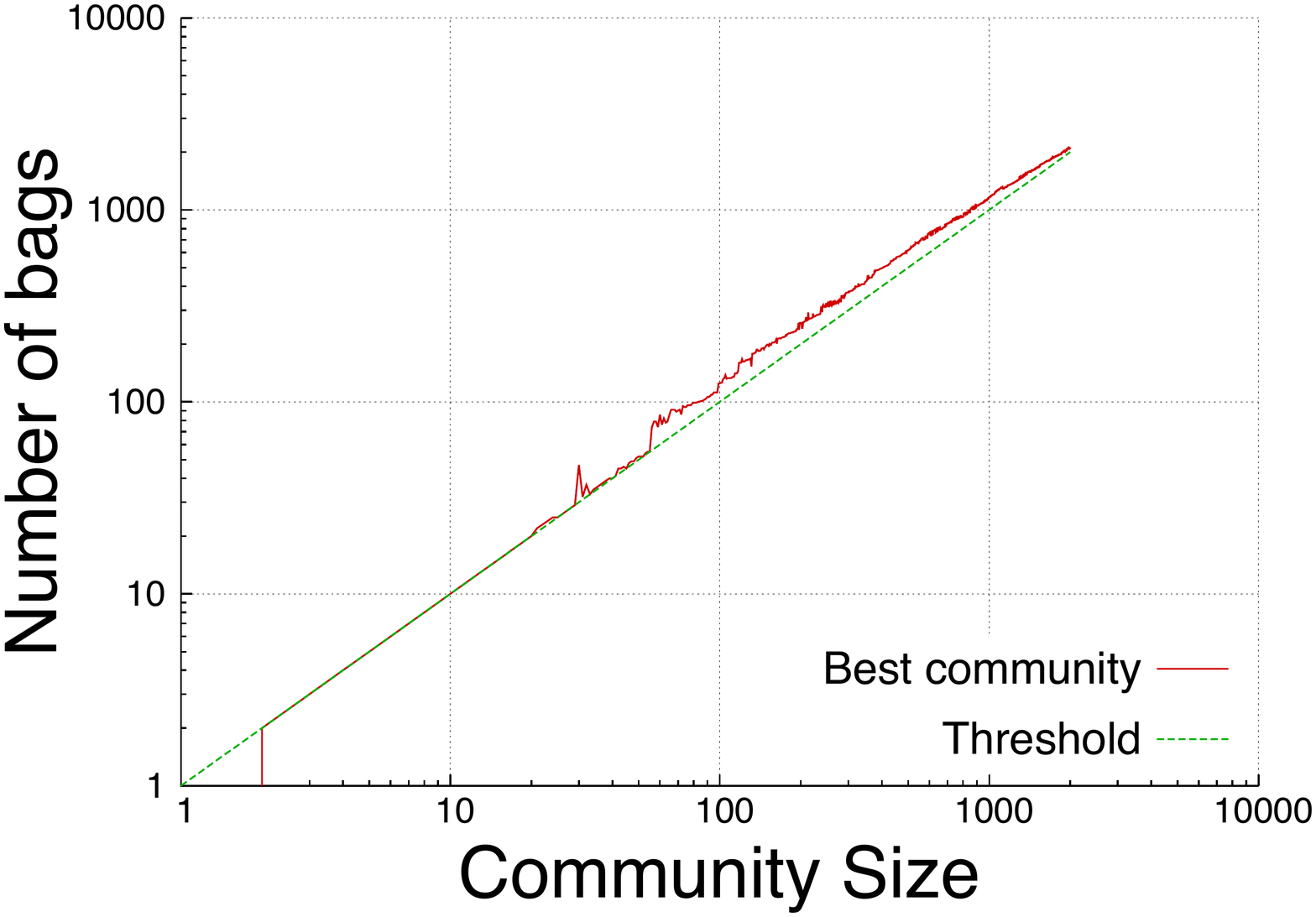}
\caption{\centering \textsc{ER(1.6)} bag localization}
\end{subfigure}
\begin{subfigure}[h]{0.23\textwidth}
\includegraphics[width=\textwidth]{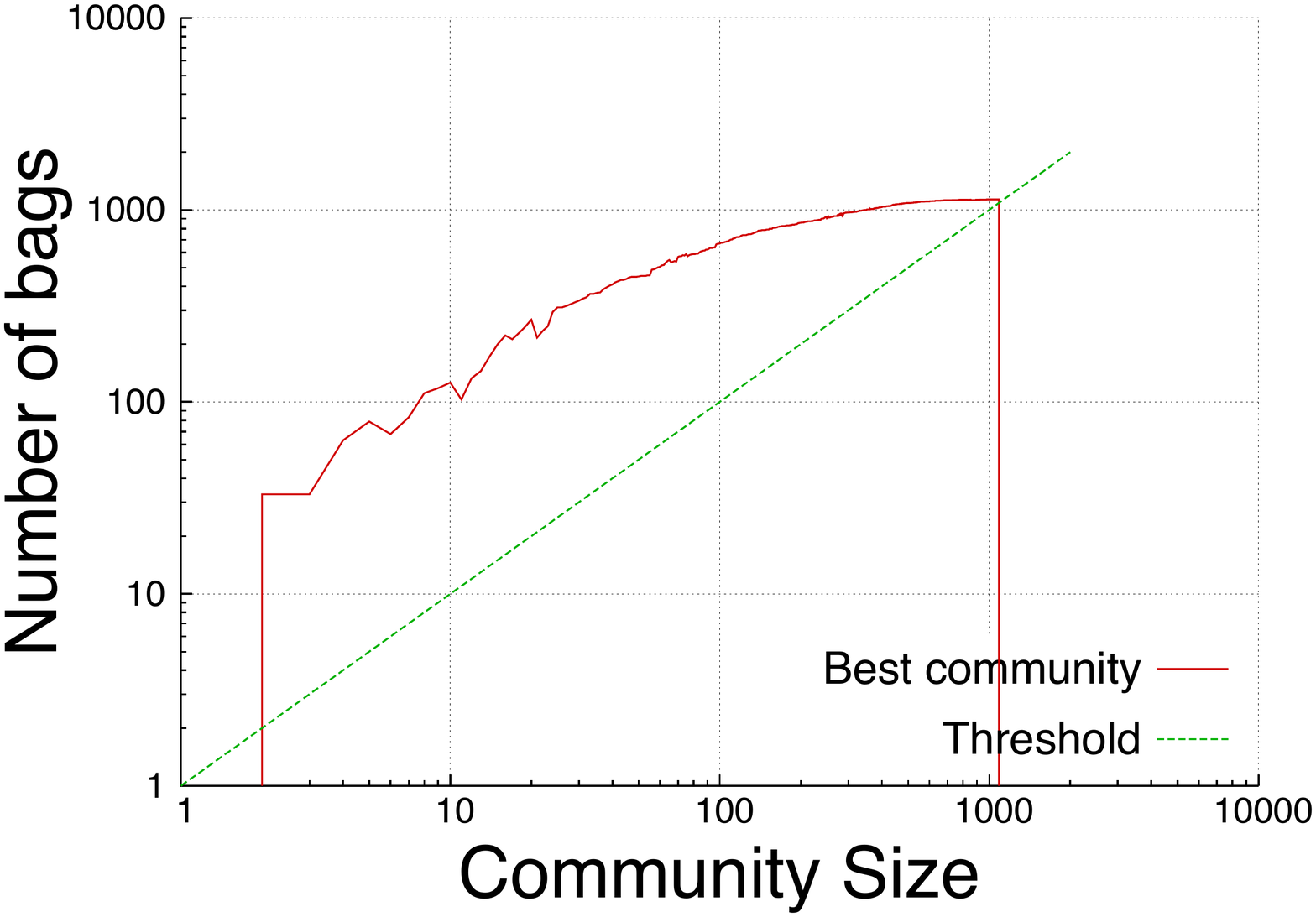}
\caption{\centering \textsc{ER(32)} bag localization}
\end{subfigure}
\caption{\textsc{ER(1.6} and \textsc{ER(32)} NCP plots and tree 
localization plots.  The localization threshold is plotted in green.}
\label{fig:synth_td_local}
\end{center}
\end{figure}

\begin{figure}[!htb]
\begin{center}
\begin{subfigure}[h]{0.23\textwidth}
\includegraphics[width=\textwidth]{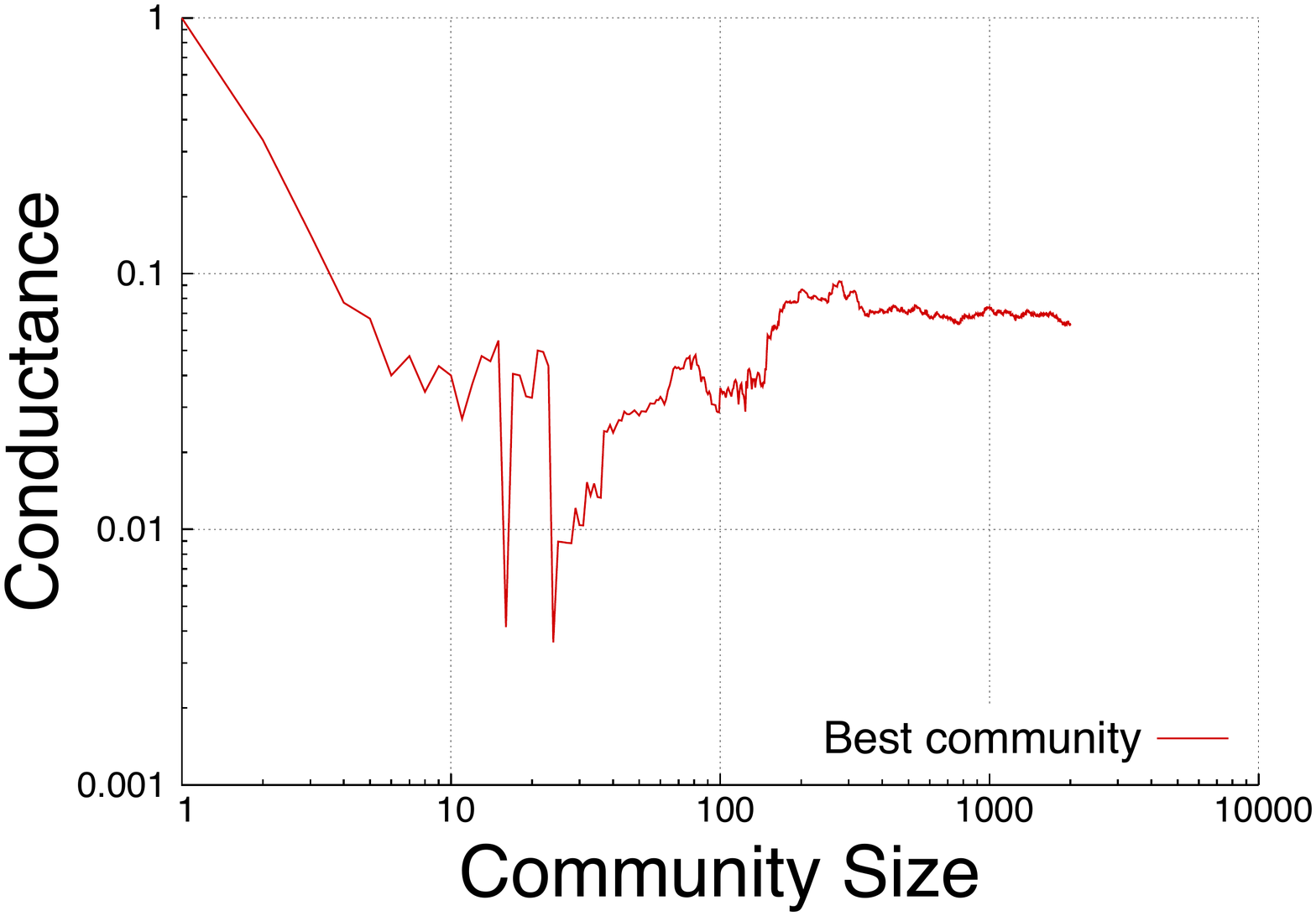}
\caption{\centering \textsc{CA-GrQc} NCP plot}
\end{subfigure}
\begin{subfigure}[h]{0.23\textwidth}
\includegraphics[width=\textwidth]{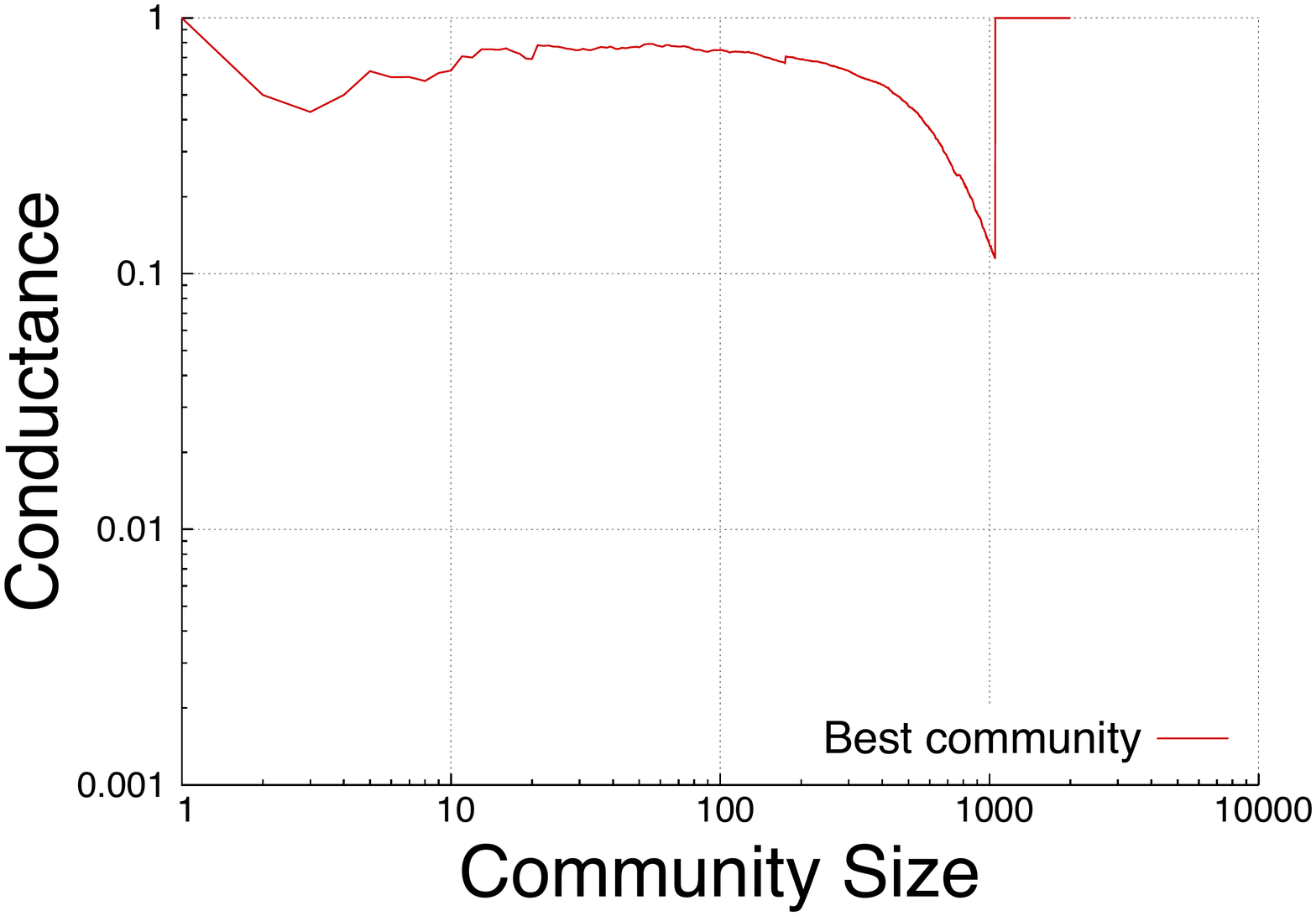}
\caption{\centering \textsc{FB-Lehigh} NCP plot}
\label{fig:lehigh_ncp}
\end{subfigure}
\begin{subfigure}[h]{0.23\textwidth}
\includegraphics[width=\textwidth]{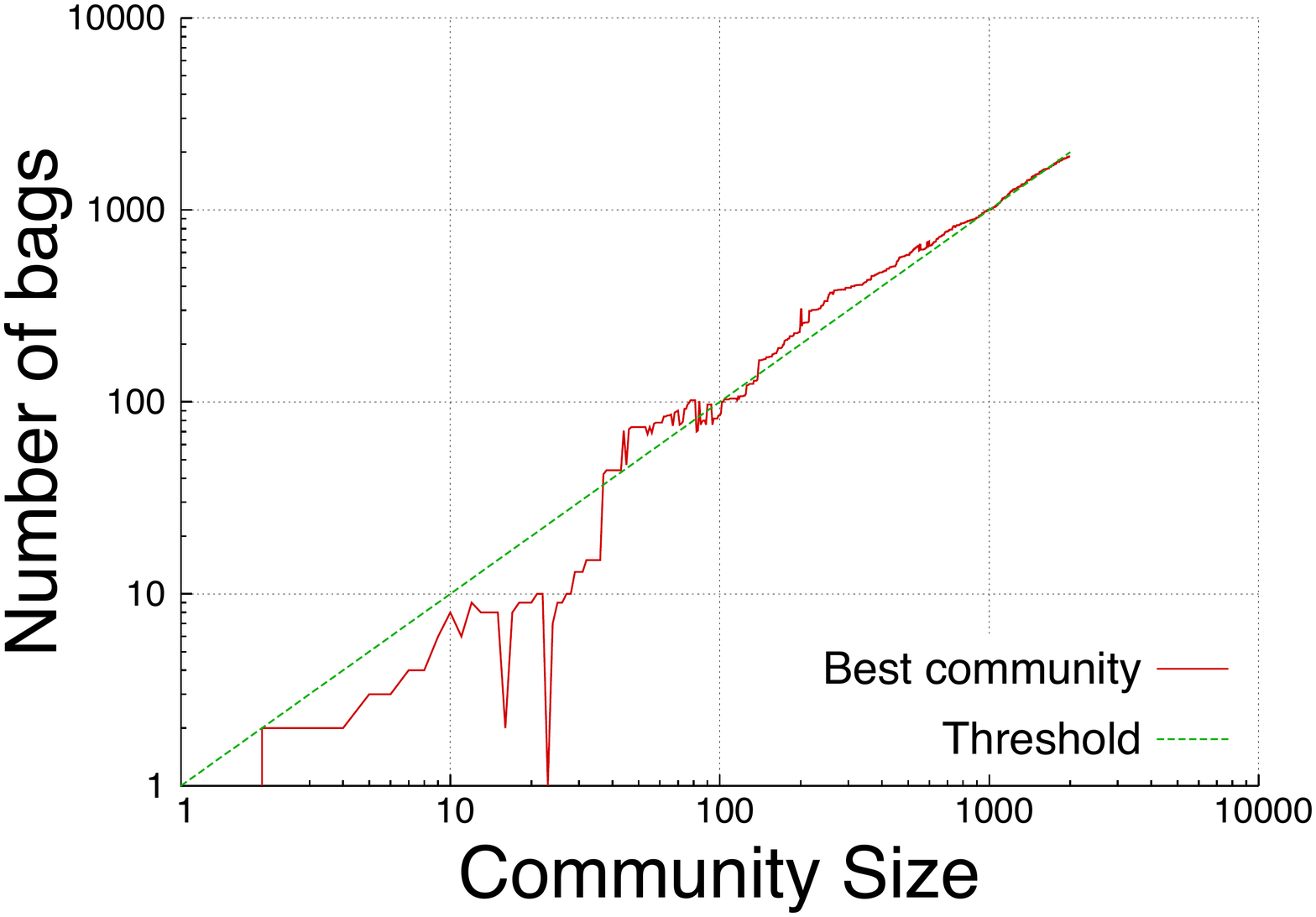}
\caption{\centering \textsc{CA-GrQc} bag localization}
\end{subfigure}
\begin{subfigure}[h]{0.23\textwidth}
\includegraphics[width=\textwidth]{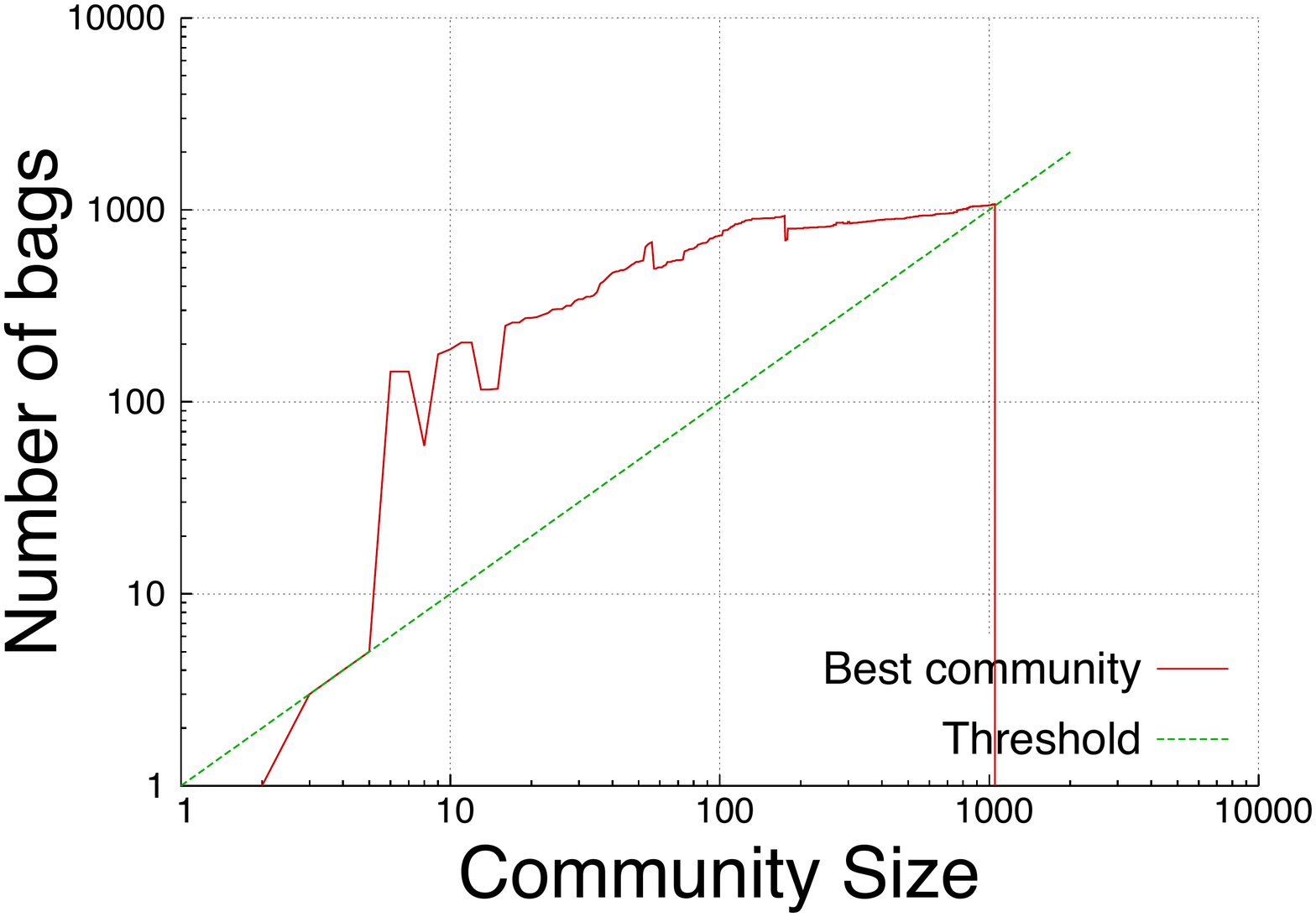}
\caption{\centering \textsc{FB-Lehigh} bag localization}
\end{subfigure}
\caption{\textsc{CA-GrQc} and \textsc{FB-Lehigh} NCP plots and tree
localization plots.  The localization threshold is plotted in green.}
\label{fig:collaboration_td_local}
\end{center}
\end{figure}

\begin{figure}[!htb]
\begin{center}
\begin{subfigure}[h]{0.23\textwidth}
\includegraphics[width=\textwidth]{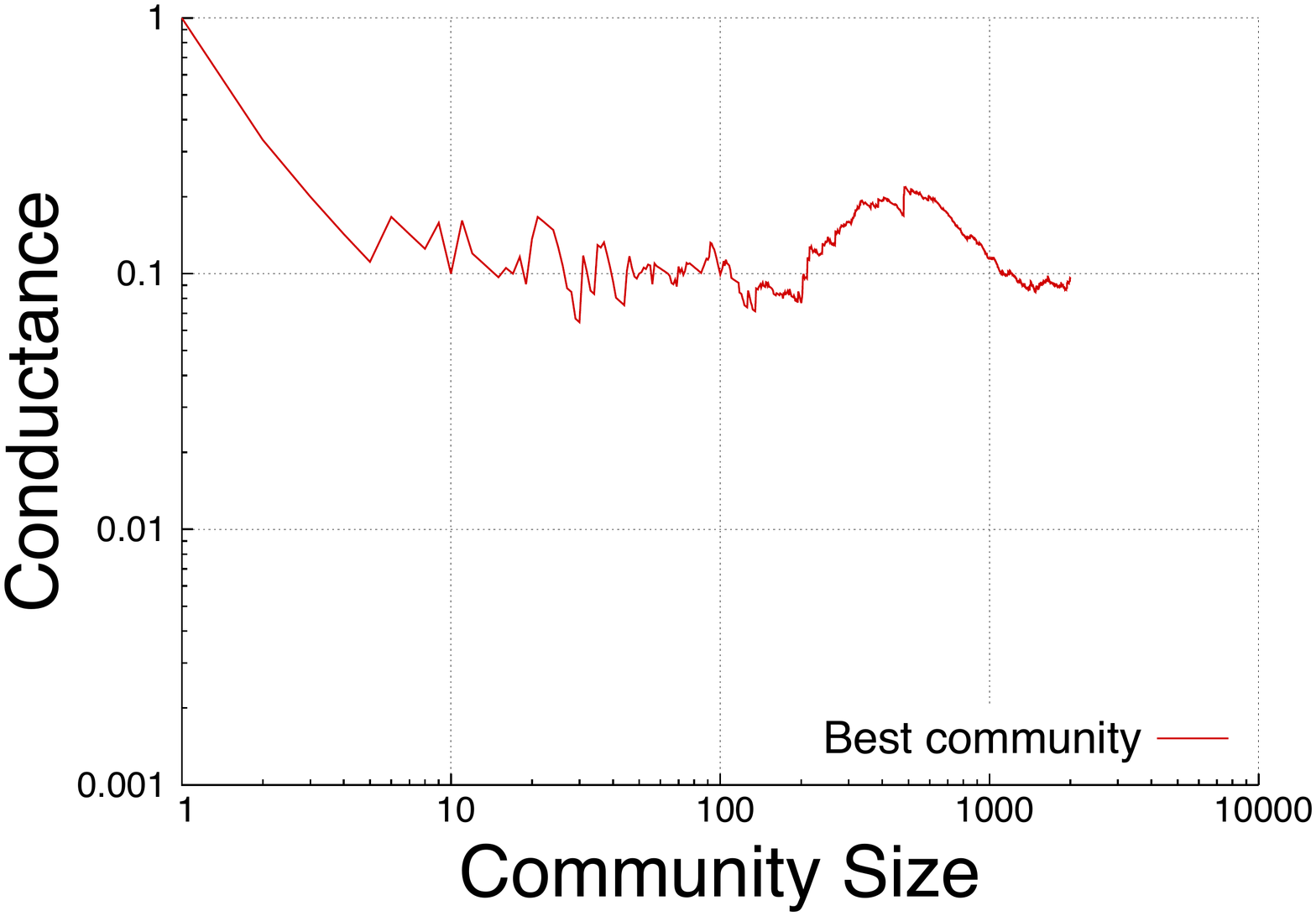}
\caption{\centering \textsc{as20000102} NCP plot}
\end{subfigure}
\begin{subfigure}[h]{0.23\textwidth}
\includegraphics[width=\textwidth]{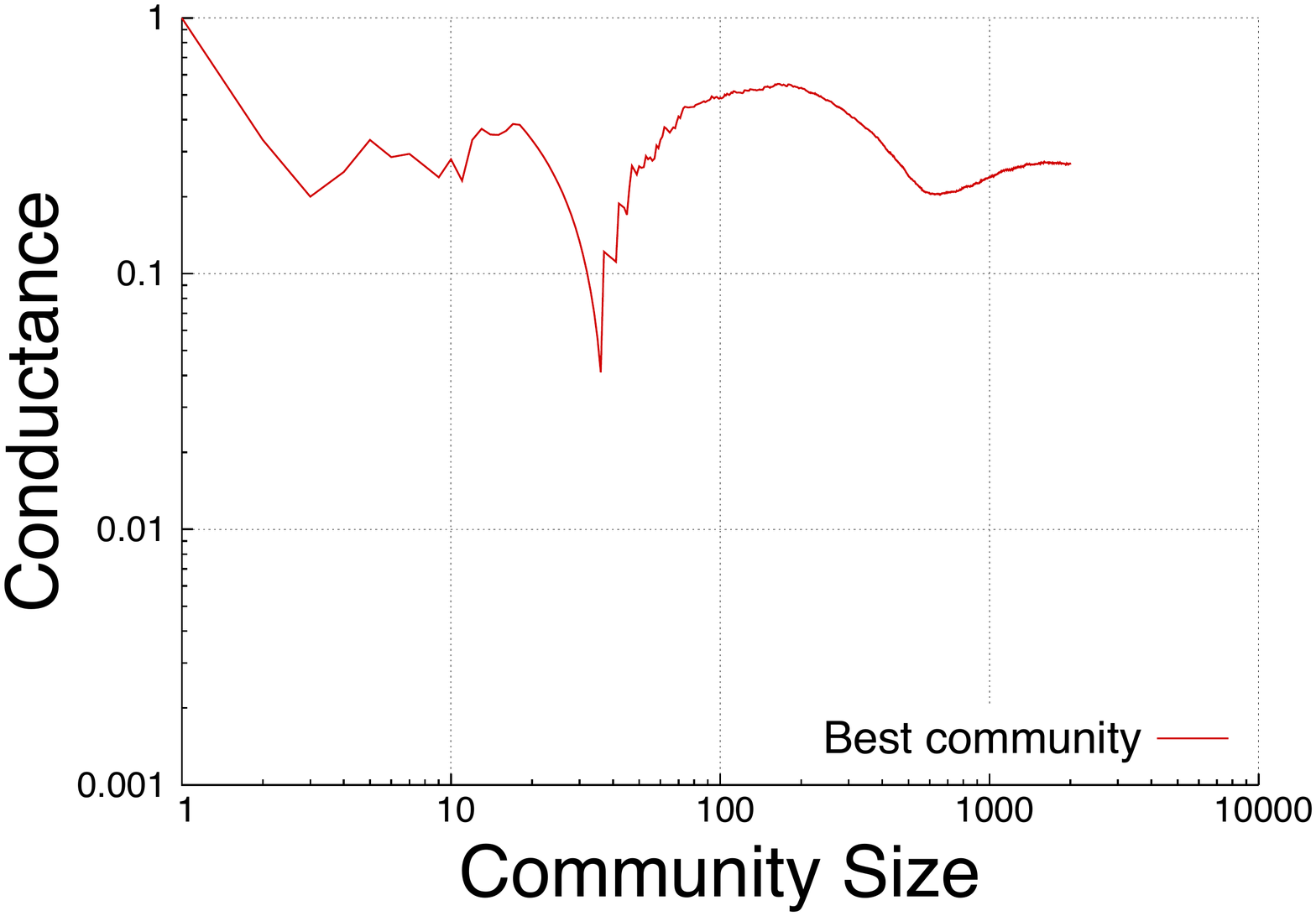}
\caption{\centering \textsc{Gnutella09} NCP plot}
\end{subfigure}
\begin{subfigure}[h]{0.23\textwidth}
\includegraphics[width=\textwidth]{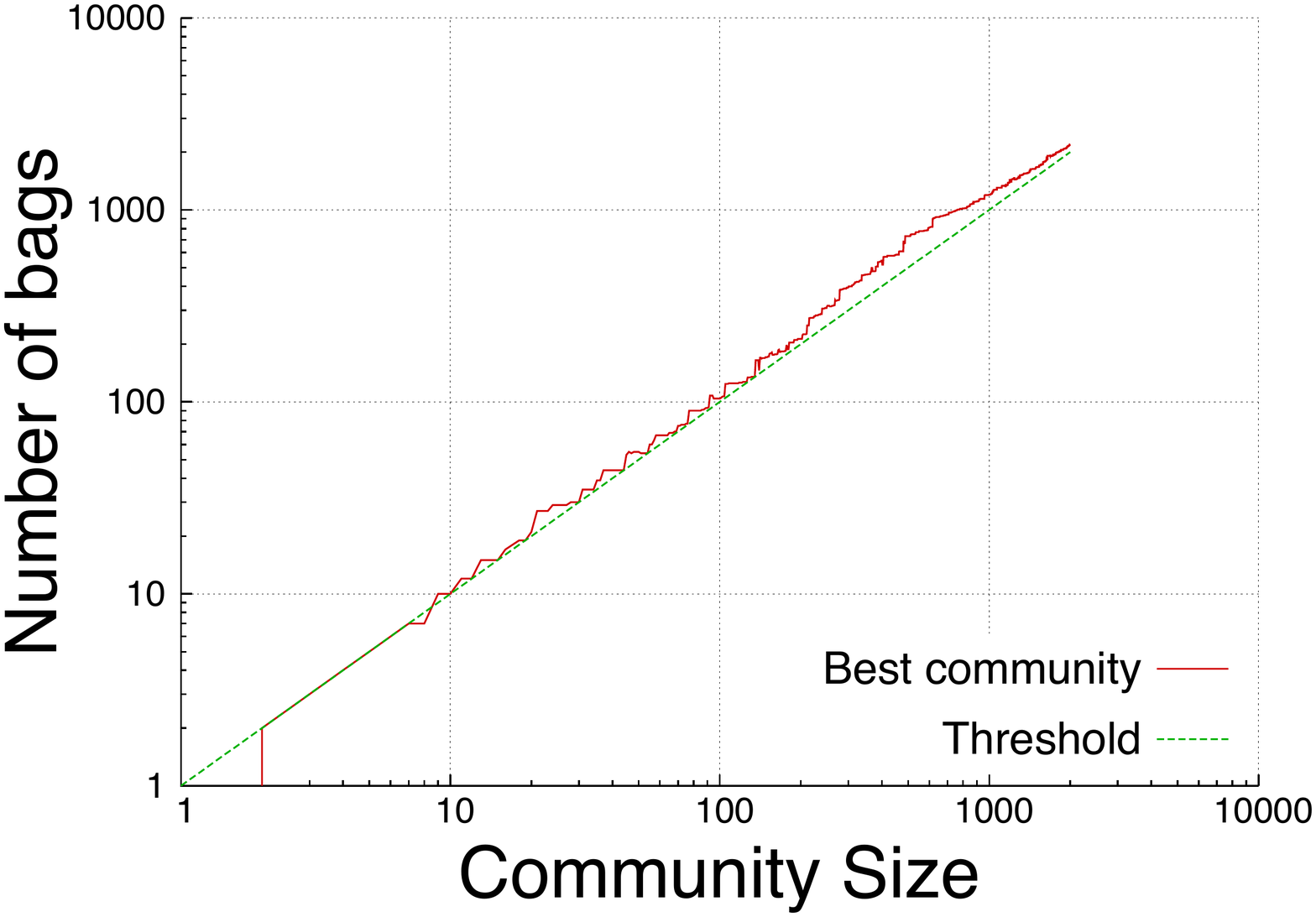}
\caption{\centering \textsc{as20000102} bag localization}
\end{subfigure}
\begin{subfigure}[h]{0.23\textwidth}
\includegraphics[width=\textwidth]{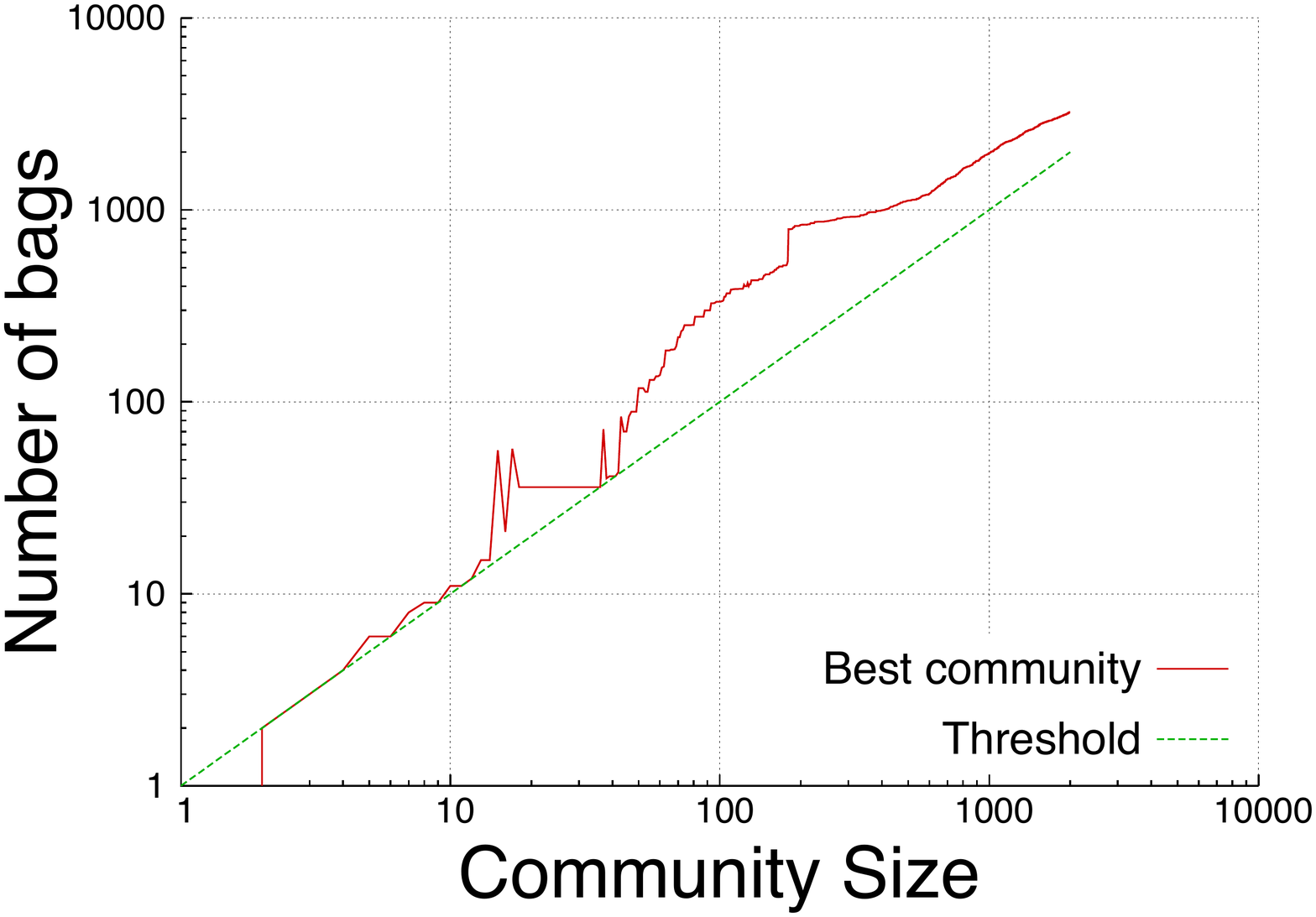}
\caption{\centering \textsc{Gnutella09} bag localization}
\end{subfigure}
\caption{\textsc{as20000102} and \textsc{Gnutella09} NCP plots and tree 
localization plots.  The localization threshold is plotted in green.}
\label{fig:internet_td_local}
\end{center}
\end{figure}

\begin{figure}[!htb]
\begin{center}
\begin{subfigure}[h]{0.23\textwidth}
\includegraphics[width=\textwidth]{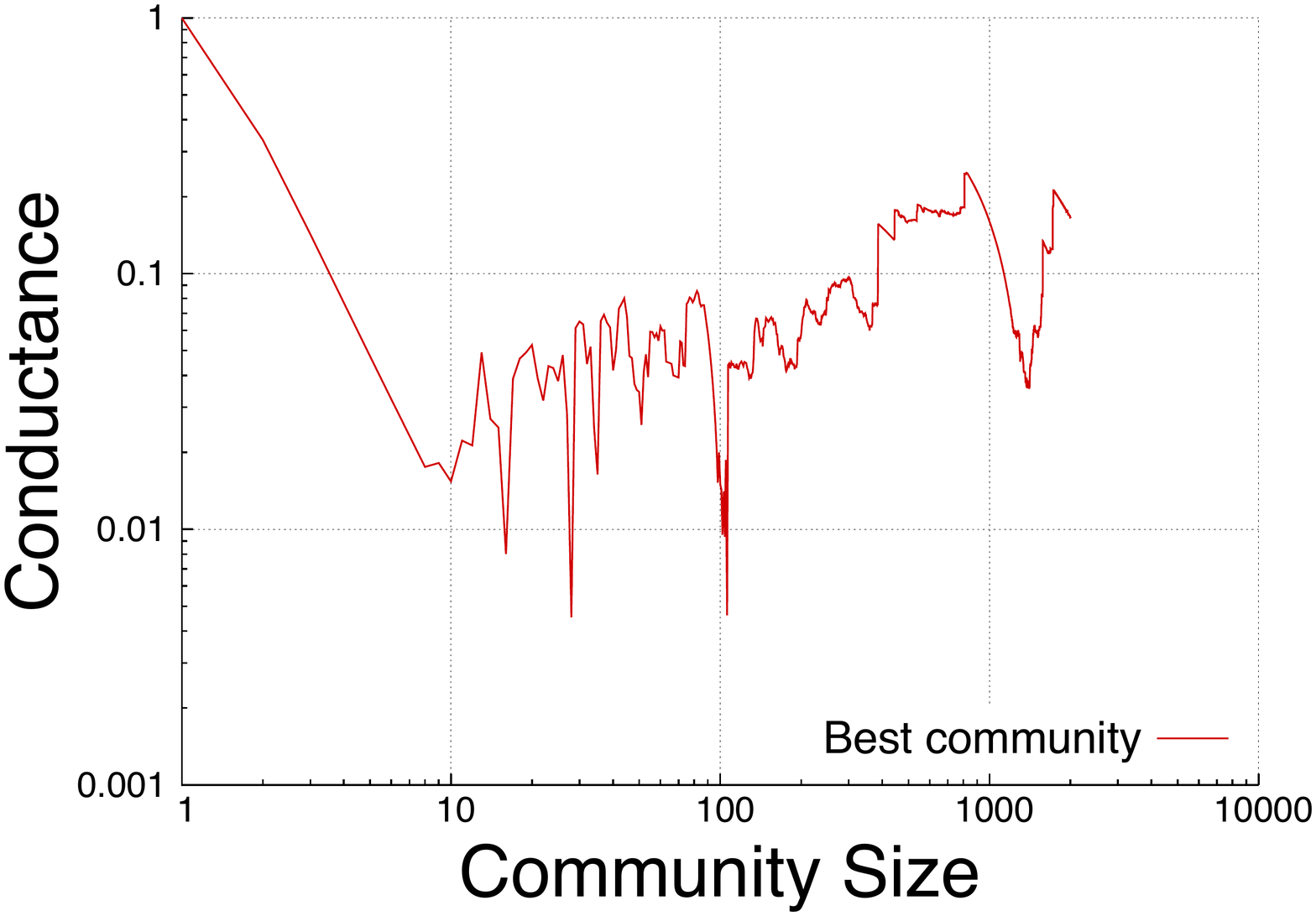}
\caption{\centering \textsc{Email-Enron} NCP plot}
\end{subfigure}
\begin{subfigure}[h]{0.23\textwidth}
\includegraphics[width=\textwidth]{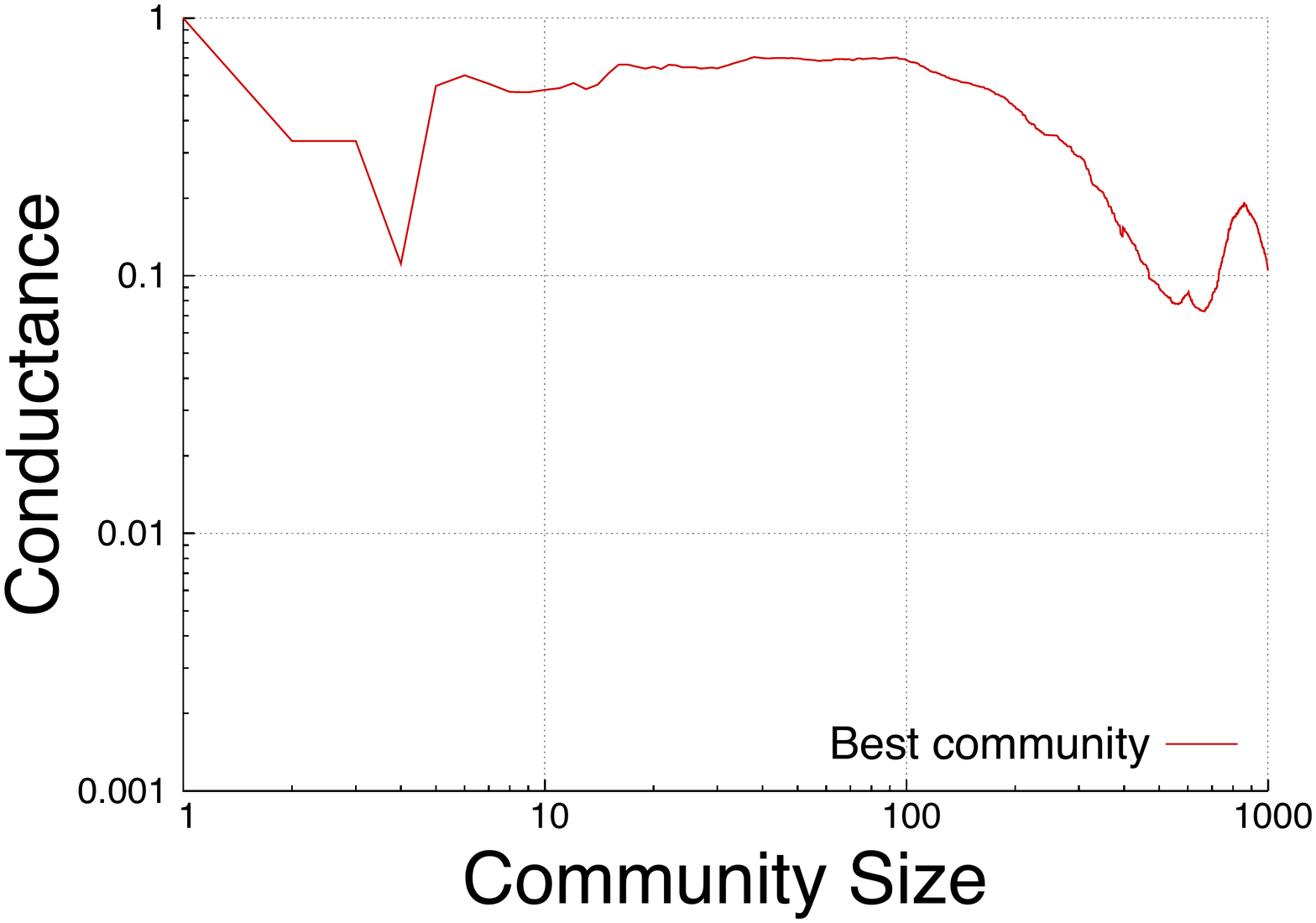}
\caption{\centering \textsc{Polblogs} NCP plot}
\end{subfigure}
\begin{subfigure}[h]{0.23\textwidth}
\includegraphics[width=\textwidth]{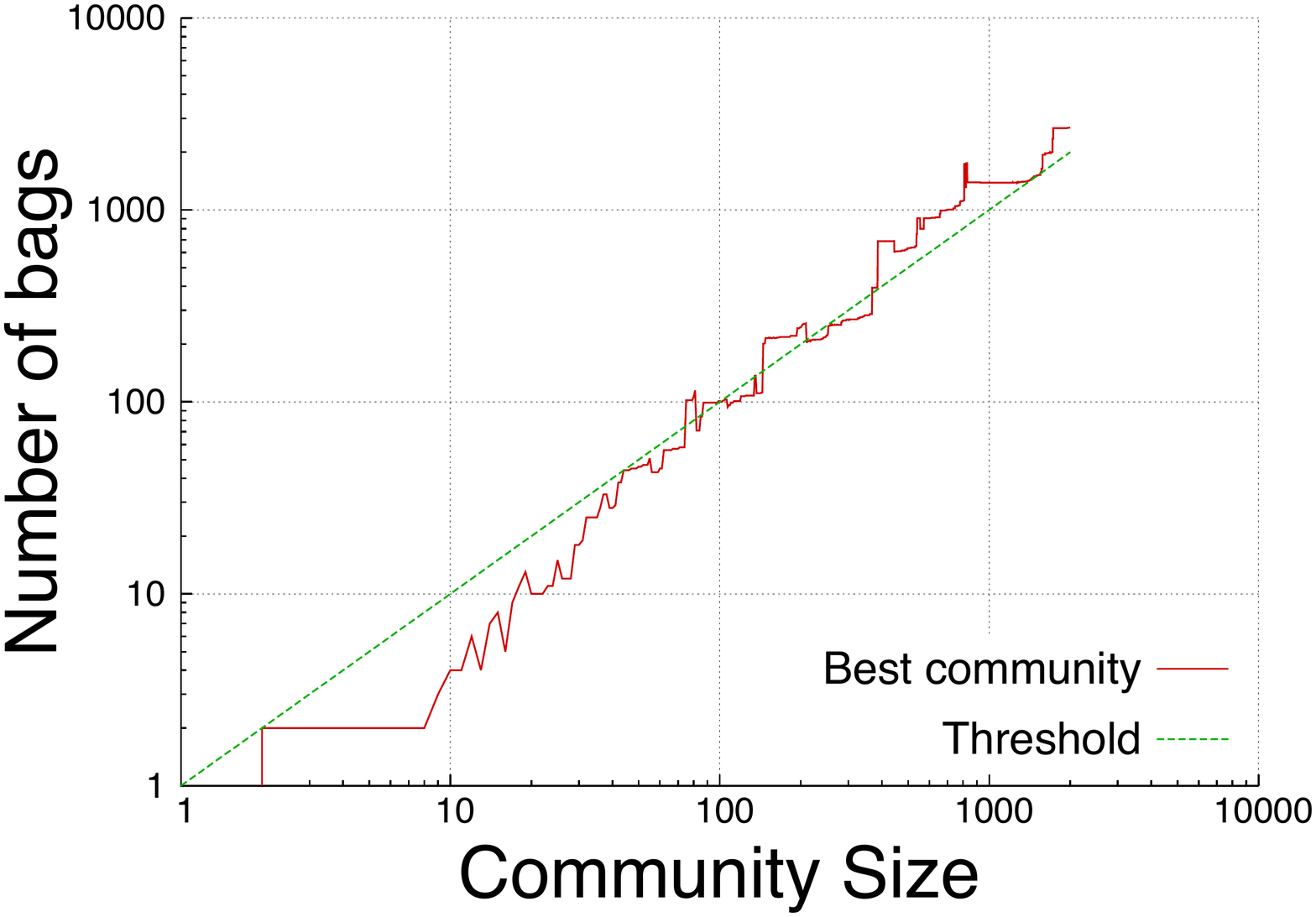}
\caption{\centering \textsc{Email-Enron} bag localization}
\end{subfigure}
\begin{subfigure}[h]{0.23\textwidth}
\includegraphics[width=\textwidth]{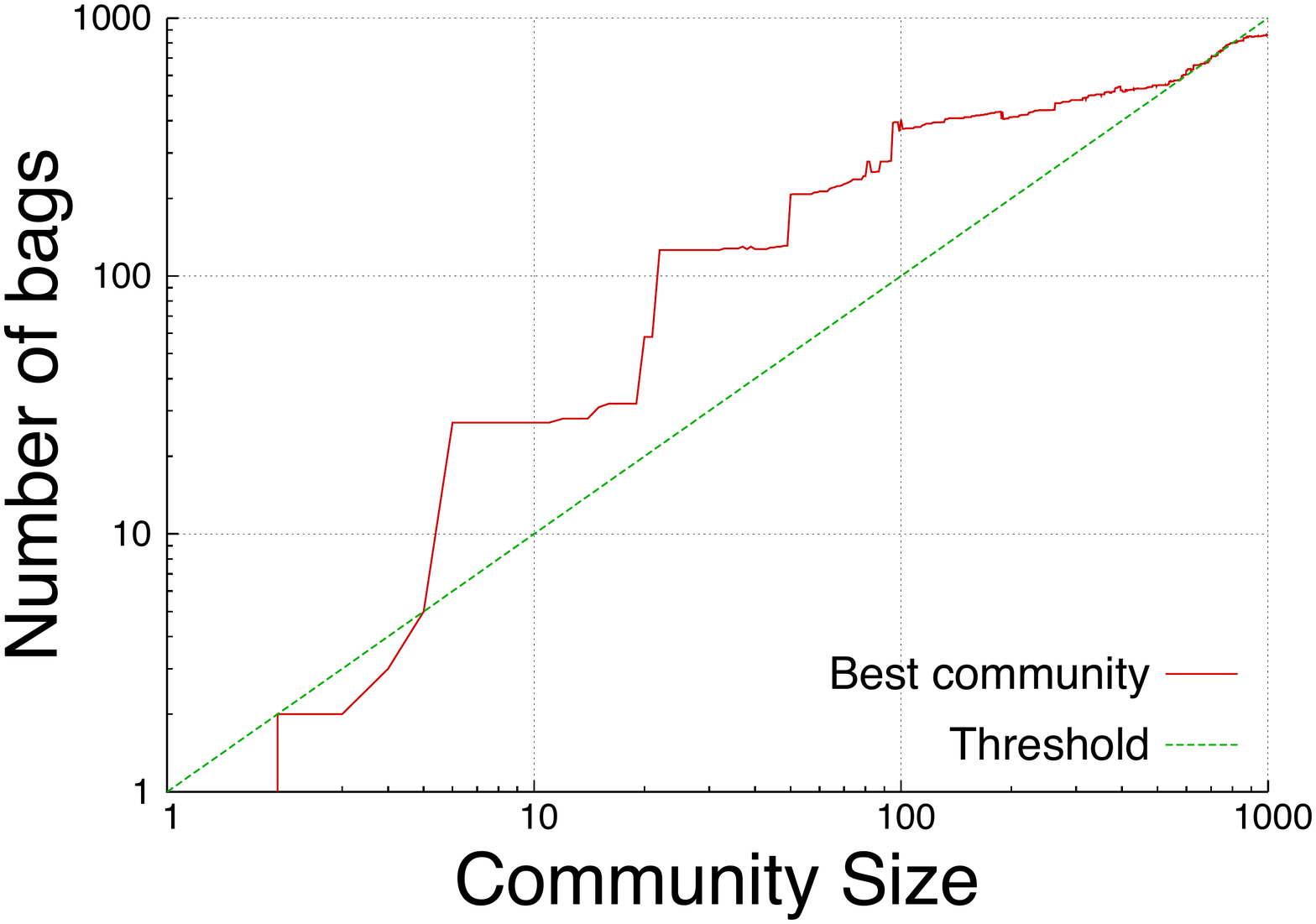}
\caption{\centering \textsc{Polblogs} bag localization}
\end{subfigure}
\caption{\textsc{Email-Enron} and \textsc{Polblogs} NCP plots and tree 
localization plots.  The localization threshold is plotted in green.}
\label{fig:misc_td_local}
\end{center}
\end{figure}

\begin{figure}[!htb]
\begin{center}
\begin{subfigure}[h]{0.23\textwidth}
\includegraphics[width=\textwidth]{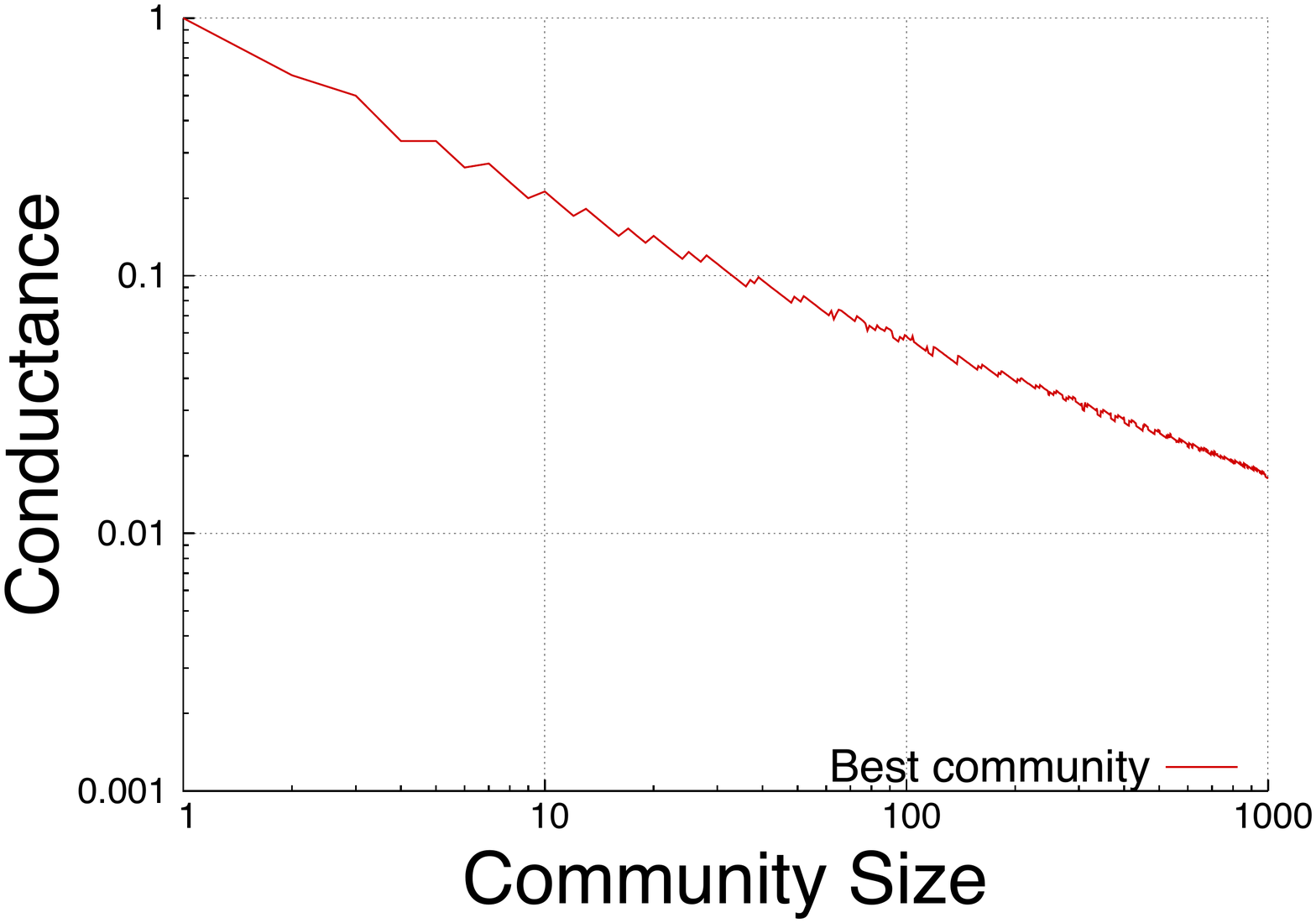}
\caption{\centering \textsc{Planar} NCP plot}
\end{subfigure}
\begin{subfigure}[h]{0.23\textwidth}
\includegraphics[width=\textwidth]{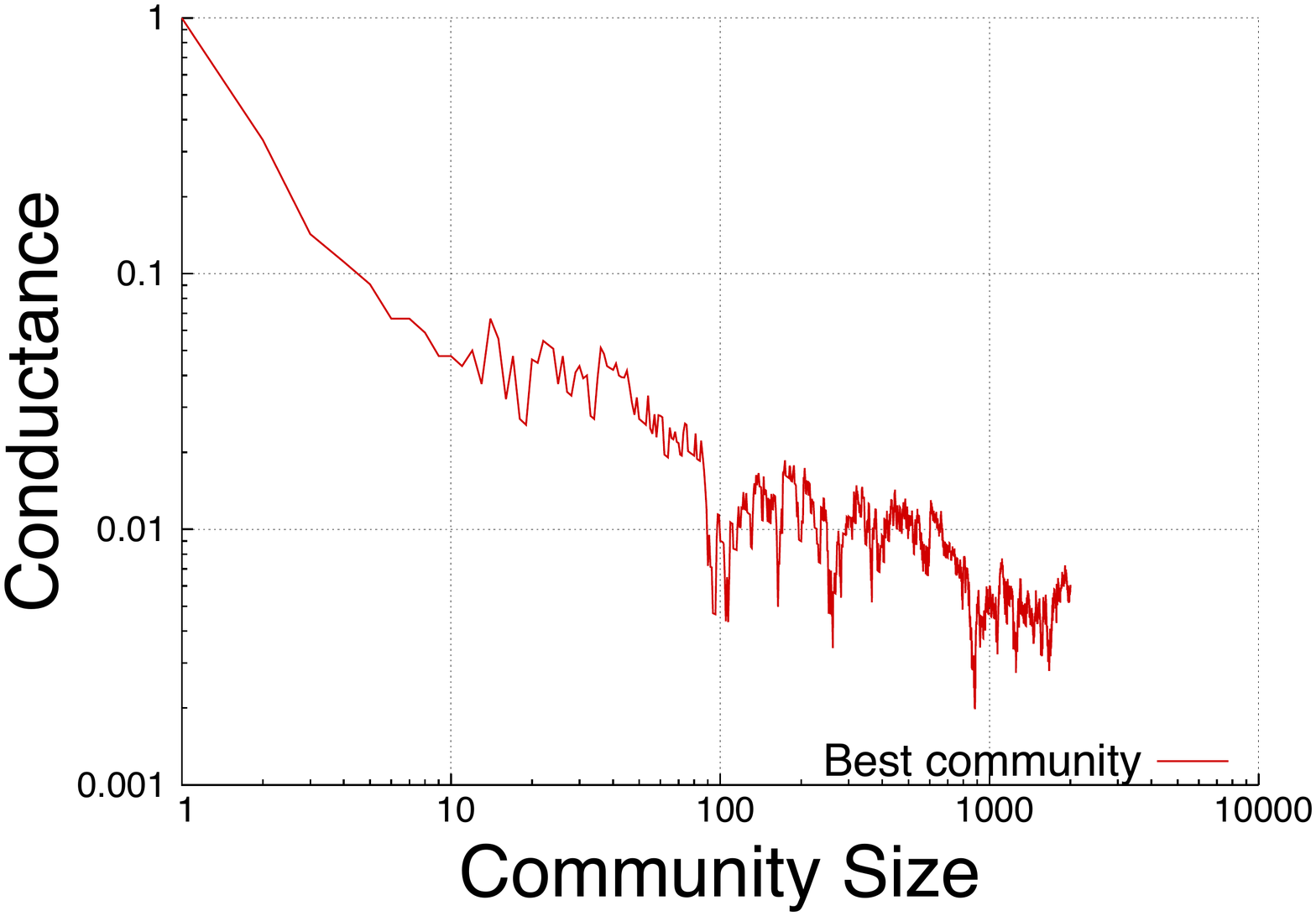}
\caption{\centering \textsc{PowerGrid} NCP plot}
\end{subfigure}
\begin{subfigure}[h]{0.23\textwidth}
\includegraphics[width=\textwidth]{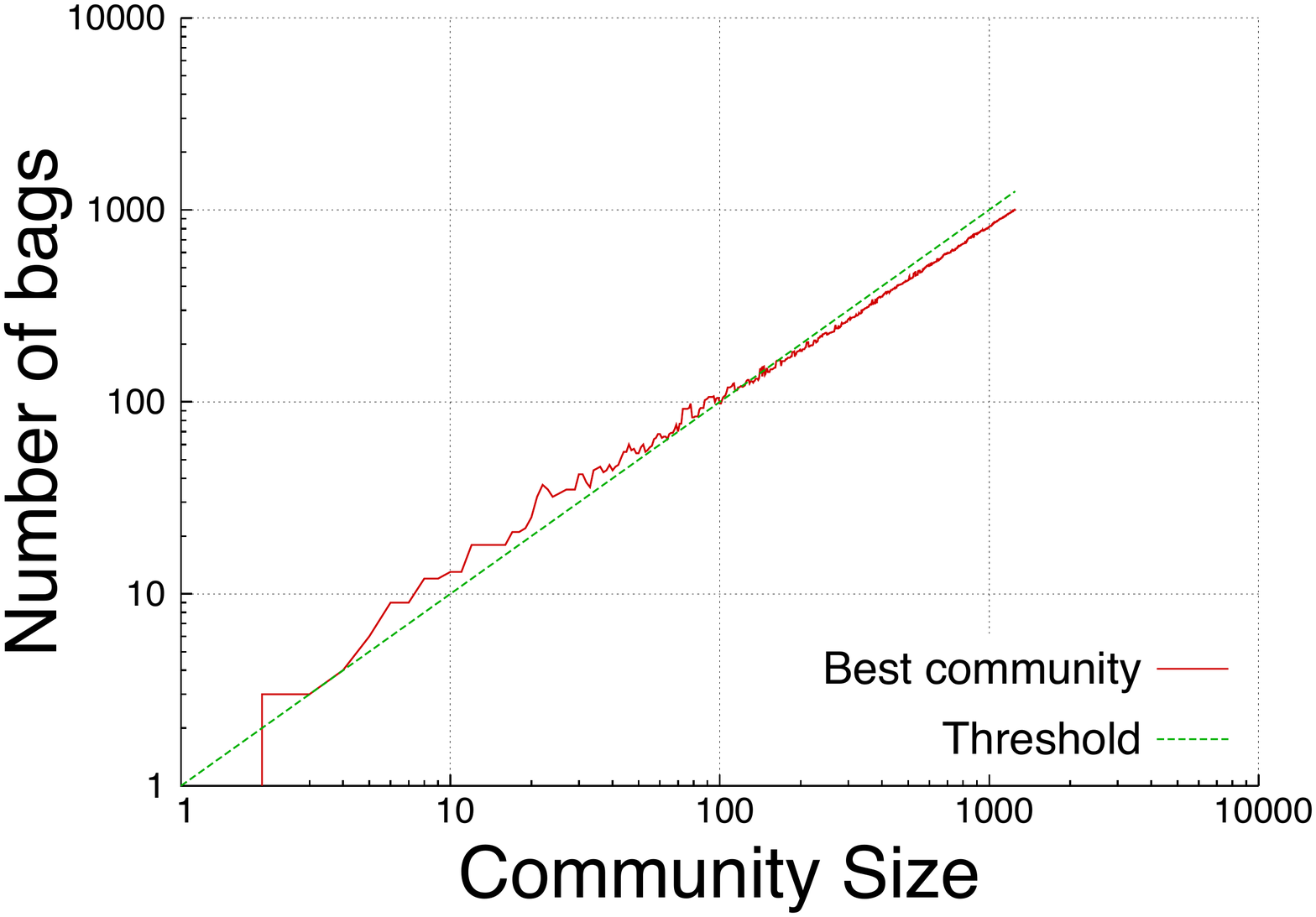}
\caption{\centering \textsc{Planar} bag localization}
\end{subfigure}
\begin{subfigure}[h]{0.23\textwidth}
\includegraphics[width=\textwidth]{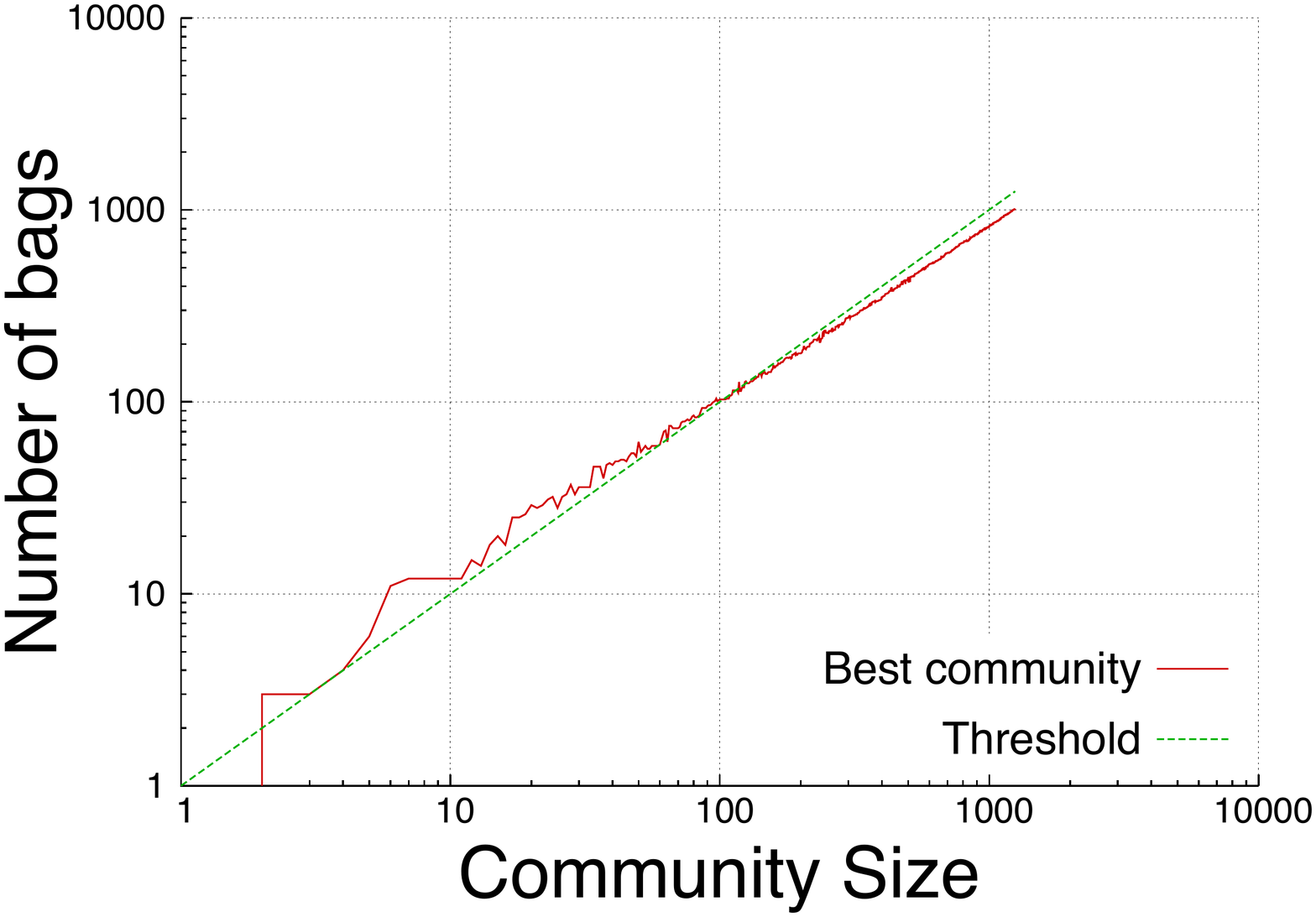}
\caption{\centering \textsc{PowerGrid} bag localization}
\end{subfigure}
\caption{\textsc{Planar} and \textsc{PowerGrid} NCP plots and tree 
  localization plots.  The localization threshold is plotted in green.}
\label{fig:planar_td_local}
\end{center}
\end{figure}

As a reference, consider the extremely sparse and somewhat denser 
\textsc{ER} networks, which are shown in Figure~\ref{fig:synth_td_local}.  
Since it is so sparse, \textsc{ER(1.6)} does have some very small 
good-conductance clusters.
As shown in the figure, however, the small ``communities'' are contained in 
roughly the same number of bags as there are in the community.  
This is expected, as these communities are largely peripheral tree-like
whiskers in the network (Section \ref{sec:synth_results}).  
For larger communities, which include core nodes, the localization is
slightly above the line defining our threshold.  
On the other hand, for the denser \textsc{ER(32)}, there are no 
good-conductance clusters at any size, and bag localization is 
above the line defining the localization threshold, indicating that 
the localization is poor at all size scales.

For the small and intermediate-sized clusters in many of the real networks (including many of those from~\cite{LLDM09_communities_IM}), the smaller good-conductance clusters found 
using the PPR method are reasonably well-localized within the TD, while the 
larger poorer-conductance clusters are not.
Consider, e.g., \textsc{CA-GrQc} in Figure~\ref{fig:collaboration_td_local} 
as an example.  
On the other hand, both large and small clusters found with the PPR method
applied to the denser graphs from the \textsc{Facebook100} set (i.e., those 
that do not have even small-cardinality good-conductance 
clusters~\cite{Jeub15}) are not well-localized in the TD.  
Consider, e.g., \textsc{FB-Lehigh} in Figure~\ref{fig:collaboration_td_local}
as an example.  
Figure~\ref{fig:internet_td_local} shows \textsc{as20000102} and 
\textsc{Gnutella09}, which also shows NCP plots that do not yield small good 
conductance clusters, and which shows that the outputs of the PPR method are 
not particularly well-localized in the TD.  
Figure~\ref{fig:misc_td_local} shows that \textsc{Email-Enron} does have 
some of its small good-conductance clusters well-localized, and it
also shows that the output of the PPR algorithm applied to \textsc{Polblogs} 
leads to medium-to-large clusters with poor conductance values that are 
poorly-localized in the~TD.

Finally, although networks with an underlying Euclidean geometry are of less 
interest for social/information network applications, for completeness it is 
worth considering how these TD methods apply to them.
Figure~\ref{fig:planar_td_local} presents results for \textsc{Planar} and 
\textsc{PowerGrid}.  
Both of these networks have downward-sloping NCP plots which are different
from the other social and information networks, 
reflecting the Euclidean geometry underlying 
these networks.  
In both cases, fairly uninteresting results are obtained, suggesting that 
the localization metric we propose is more interesting for realistic social 
graphs with non-trivial tree-like core-periphery structure.

Although our results demonstrate that good-conductance clusters/communities 
in several realistic social graphs are well-localized in TDs found with 
existing heuristics, it is not obvious how to address the reverse question 
of finding good-conductance communities from a TD.  
One could attempt to look at all or some large number of combinations of 
bags in the TD. 
Since one is usually interested in well-connected communities/clusters, the 
running intersection property of TDs could be used to restrict attention to 
connected subsets of a TD.
There are, however, two obvious issues.
First, there does not exist an obvious analogue of the ``sweep cut'' used in 
the spectral partitioning method for finding the best community from a TD.  
Second, as a related practical matter, the presence of high degree (or deep 
core) nodes in the intermediate and central bags of a TD cause bags to be 
poor conductance communities.  
These nodes have many connections and increase the ``surface area'' of most 
cuts, even if there is only a small number of them in a cluster.  
We observed that, in the clusters we found using the PPR method, each 
cluster is typically well-represented by a set of small bags plus a couple of
nodes in the larger bags.  
If we then attempt to form clusters by combining bags, we get \emph{all} of 
the nodes in the larger bags, including deep core nodes.  
Additional methods of filtering nodes for the larger bags, such as ordering 
by node degree or $k$-core combined with a sweep cut, may improve these 
results.

\subsection{Results on identifying ground-truth communities}
\label{sec:real_results-gt}

Here, we will consider other ways in which the output of TDs can be useful 
in identifying clusters/communities of interest to the domain analyst.
In particular,
we describe 
two examples from the demographic data associated with the Facebook100 
dataset~\cite{Traud12}.

Consider, first, Figure \ref{fig:Haverford_age}, where we show the
\textsc{amd} TD of the \textsc{FB-Haverford} network, and where each
bag is colored-coded by the average graduation year of the constituent
nodes.  There is a large
linear or trunk-like structure that dominates the large-scale structure
of the TD.  We observe that there is a strong
overlap between the nodes that comprise successive bags in that trunk, and
we note that this trunk-like structure is typical of most of the
\textsc{Facebook100} networks (but is not seen in most other social
graphs we have considered).  Also, each
end of the long trunk correlates strongly with graduation year, and
there is a gradual change in the average graduation year of each
bag as we move across the trunk.  Thus, to the extent that one accepts
graduation year as some sort of easily-quantifiable ``ground truth''
community, the large bags in the TD of this network seem to be
capturing a legitimate ground-truth structure in the network. This
fits well with prior results that report that in most of the Facebook
networks graduation year is best predictor of the existence of edges
between two nodes~\cite{Traud12}.

Consider, next, Figure~\ref{fig:Caltech_res}.  
It is known that for a small number of the \textsc{Facebook100} networks 
(e.g., \textsc{FB-Caltech}, \textsc{FB-Rice}, and \textsc{FB-UCSC}), 
residence hall rather than class year is the best edge 
predictor~\cite{Traud12}.  
Thus, we considered the \textsc{amd} TD of (the students-only subset of) \textsc{FB-Caltech}.  
In this case, a single simple trunk-like structure is not dominant, but 
there are several relatively large peripheral branches, and many of the 
peripheral branches are dominated by a particular residence hall.  
In Figure~\ref{fig:Caltech_res}, the bags are colored by the fraction 
of students in residence hall 170 (chosen arbitrarily).
These examples are of particular interest since good-conductance clusters 
do not exist in \textsc{Facebook100} networks~\cite{Jeub15}.

\begin{figure}
\centering
\begin{subfigure}{.5\textwidth}
\centering

\includegraphics[width=0.90\textwidth]{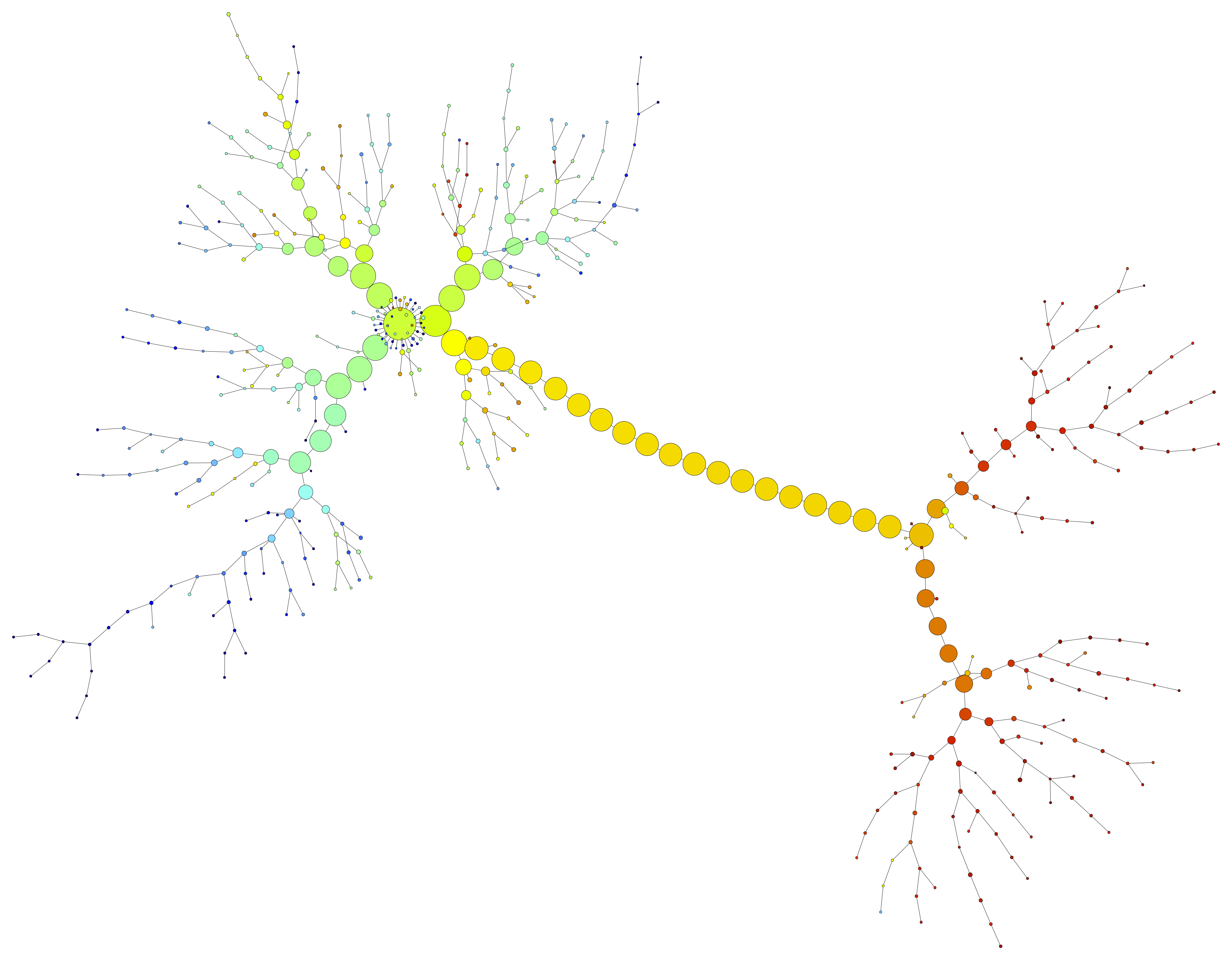}

\caption{\textsc{amd} TD of \textsc{FB-Haverford}, colored by
  graduation year (red = freshman, blue = alumni).  The long,
  path-like trunk of this (and most other) Facebook networks is driven by
  the propensity of students to be friends with students of a similar
  graduation year.  
}
\label{fig:Haverford_age}
\end{subfigure}%
\begin{subfigure}{.5\textwidth}
\centering

\includegraphics[width=0.90\textwidth]{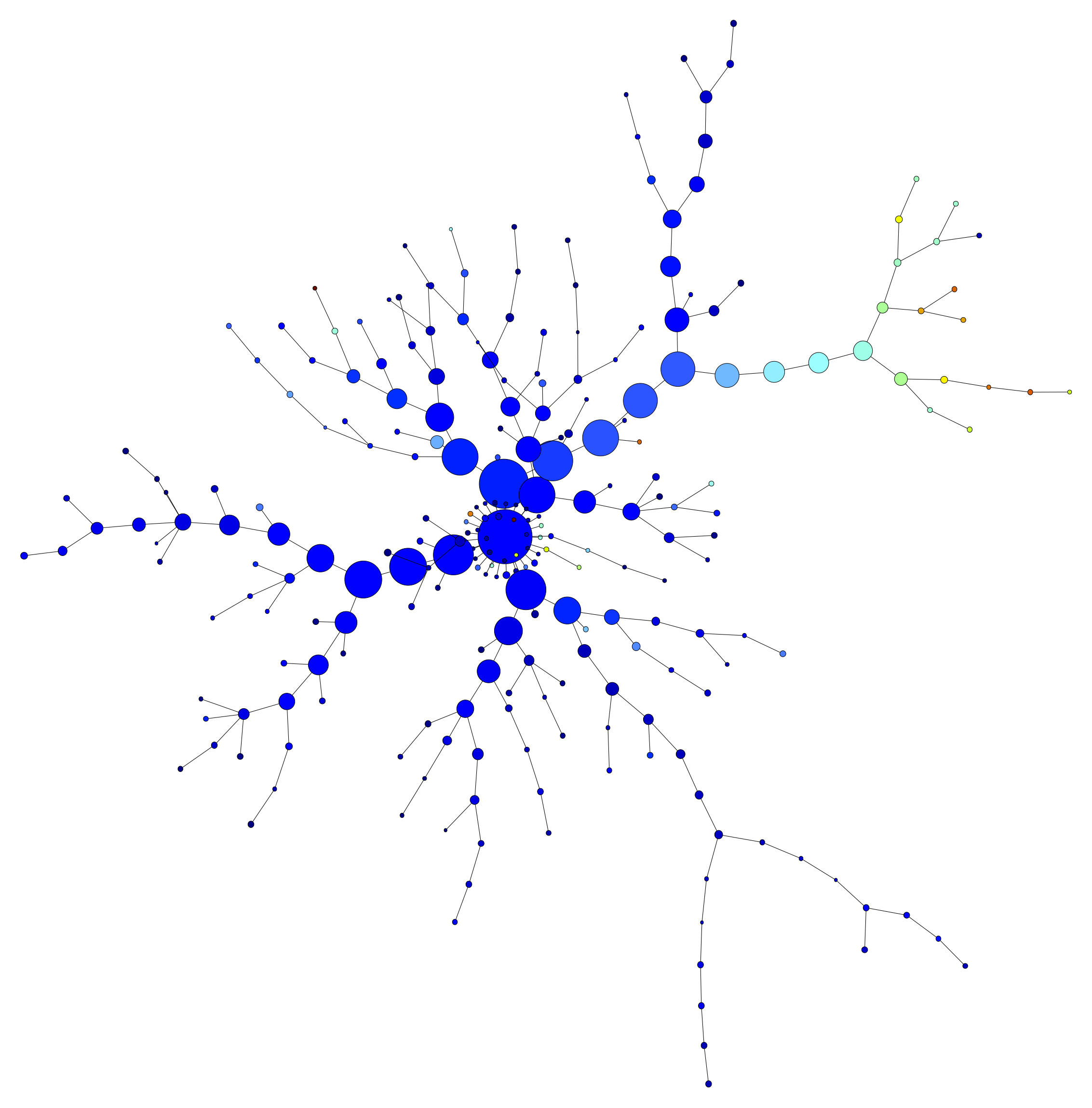}

\caption{\textsc{amd} TD of (the students-only subset of)
  \textsc{FB-Caltech}, colored by the fraction of students in
  residence hall 170 (blue = no nodes belong to residence, ..., red =
  all nodes belong to residence).  
}
\label{fig:Caltech_res}
\end{subfigure}

\caption{\textsc{amd} TD of \textsc{FB-Haverford} and \textsc{FB-Caltech}.
\textsc{FB-Haverford} is presented, rather than \textsc{FB-Lehigh} (which has similar large-scale TD structure), because its smaller eccentricity (56 rather than 150) makes it easier to visualize.
For \textsc{FB-Caltech}, this is a graphical representation of data presented in Table~\ref{tbl:community_precision} for the \textsc{amd} TD and residence hall 170.
}
\label{fig:Haverford_Caltech_residences}
\end{figure}

By looking at bags where the concentration of a particular community
node is higher than the incidence of that community throughout the TD,
we can form a very simple classification rule.
In particular, given residence hall
$X$, we collected all bags whose fraction of nodes which were
listed as belonging to residence $X$ was higher than the fraction of
nodes belonging to that hall in the network (the incidence in the
network is given in column $F$ in Tables \ref{tbl:community_recall}
and \ref{tbl:community_precision}).  We then kept the largest
contiguous set of bags and used membership in this set as the
classifier.  
Although this is an overly simple classifier, the goal in this section 
is simply to provide a baseline about how the residence communities are 
located in the TD.  

We performed this procedure on the students-only restriction of
\textsc{FB-Caltech}.  Tables \ref{tbl:community_recall} and
\ref{tbl:community_precision} provide a summary of the classification
results using this method on \textsc{FB-Caltech} network, with the
(anonymized) listed residence hall for that student as the community.
Table~\ref{tbl:community_recall} shows the fraction of
the ``ground truth'' community captured by the largest contiguous set
of bags described above.  This is analogous to the recall of
classifying the community using this branch in the TD.  Table
\ref{tbl:community_precision} shows the fraction of the nodes in the
union of all bags in this largest contiguous set which belong to the
community.  This is analogous to the precision of classifying the
community using this branch.  Since \textsc{FB-Caltech} is very
small, we can use a much larger variety of TD as classifiers
than is possible for larger networks, and we present results for all
of these TD classifiers. 

\begin{table}[!htb]
\begin{center}
{\footnotesize
\begin{tabular}{l|r|r|r|r|r|r|r|}
Hall & $F$ & \textsc{mindeg} & \textsc{minfill} & \textsc{lexm} & \textsc{mcs} & \textsc{amd} & \textsc{metnnd} \\
\hline \hline
None & .134 & .270 & .257 & .270 & .324 & .284 & .297 \\ 
165  & .066 & .472 & .528 & .556 & .528 & .472 & \textbf{.861} \\ 
166  & .090 & .736 & .736 & \textbf{.925} & .811 & .642 & .792 \\ 
167  & .134 & \textbf{.642} & .566 & .453 & .585 & .491 & .566 \\ 
168  & .116 & .746 & .762 & \textbf{.952} & .889 & .825 & .143 \\ 
169  & .136 & .726 & .712 & \textbf{.904} & .877 & .658 & .740 \\ 
170  & .090 & .725 & .725 & .739 & .783 & \textbf{.855} & .362 \\ 
171  & .136 & .714 & .673 & .776 & \textbf{.857} & .592 & .429 \\ 
172  & .098 & .630 & .630 & .534 & .699 & .548 & \textbf{.863} \\ 
\end{tabular}}
\caption{Fraction of each \textsc{FB-Caltech} residence hall captured
  in the largest contiguous set of ``frequent'' bags.  A frequent bag
  is a bag where the fraction of students who belong to the given
  residence hall is greater than the fraction of students who belong
  to that residence in the entire network. Column $F$ gives the
  fraction of students who identified as being in the associated hall
  (i.e., the threshold for being a frequent bag
  for that residence hall).  This procedure was also performed for the
  nodes which did not have a residence hall listed for comparison.}
\label{tbl:community_recall}
\end{center}
\end{table}

\begin{table}[!htb]
\begin{center}
{\footnotesize
\begin{tabular}{l|r|r|r|r|r|r|r|}
Hall & $F$ & \textsc{mindeg} & \textsc{minfill} & \textsc{lexm} & \textsc{mcs} & \textsc{amd} & \textsc{metnnd} \\
\hline \hline
None & .134 & .065 & .064 & .068 & .071 & .071 & .100 \\ 
165  & .066 & .055 & .057 & .052 & .052 & .059 & \textbf{.086} \\ 
166  & .090 & .104 & .110 & .111 & .111 & .092 & \textbf{.129} \\ 
167  & .134 & \textbf{.102} & .097 & .092 & .087 & .088 & .092 \\ 
168  & .116 & .137 & .140 & .150 & .147 & \textbf{.157} & .043 \\ 
169  & .136 & .150 & .151 & .147 & .163 & .138 & \textbf{.187} \\ 
170  & .090 & .134 & .139 & .145 & .140 & \textbf{.171} & .103 \\ 
171  & .136 & \textbf{.105} & .099 & .100 & .104 & .084 & .074 \\ 
172  & .098 & .131 & .129 & .100 & .137 & .115 & \textbf{.199} \\ 
\end{tabular}}
\caption{Fraction of the nodes contained in the largest contiguous set
  of frequent bags for a given residence hall which \emph{actually
    belong} to the given residence hall.  A frequent bag is a bag
  where the fraction of students who belong to the given residence
  hall is greater than the fraction of students who belong to that
  residence in the entire network. Column $F$ gives the fraction of
  students who identified as being in the associated hall 
  (i.e., the threshold for being a frequent bag for that
  residence hall).  This procedure was also performed for the nodes
  which did not have a residence hall listed for comparison.}
\label{tbl:community_precision}
\end{center}
\end{table}

Although the communities do seem to be well-captured by the TDs, 
there are also many other nodes in the same bags as
these communities (see Table \ref{tbl:community_precision}).
Although the only bags selected were bags where the
residence hall in question was over-represented, combining these bags
actually resulted in a lower concentration of residents than were
present in the network for some residence halls (see Table
\ref{tbl:community_precision} where the values are lower the $F$ for a
given residence hall).  This occurs since the non-resident nodes in
each of these bags are different, while the resident nodes are largely
the same for each bag in the branch.

In terms of heuristic performance, the \textsc{mindeg, minfill,} and
\textsc{amd} seem to have similar performance given a residence,
although there seems to be a larger gap between \textsc{amd} than the
other heuristics.  This is not surprising as these are all greedy
heuristics which work by reducing fill (or minimum degree, which is a
proxy for fill) in each step.  \textsc{lexm} and \textsc{mcs} also
seem to behave similarly, and they have the best performance
in terms of recall (Table \ref{tbl:community_recall}).
\textsc{metnnd} has a different profile from the other networks and
seems to do the best in terms of precision (Table
\ref{tbl:community_precision}).  
These results are comparable to what can be obtained with other simple 
classification rules, and they suggest that TDs could be useful in these 
types of machine learning applications.

Overall, these results demonstrate that for these 
realistic social/information networks, several types of plausible ``ground 
truth'' communities are well-correlated with the large-scale structure
identified by existing TD heuristics.  
This striking since these heuristics make local greedy decisions about how 
to form the TDs, and it suggests that improved results could be obtained in 
this application by considering TD heuristics designed for graphs with this
type of structure.

\section{More details on tree decomposition methods}
\label{sec:theory_new}

In this section, we consider the question of whether TDs and their 
treewidths can be related to other parameters for tree-like structure, 
specifically the Gromov $\delta$ hyperbolicity.%
\footnote{As we mentioned in Section~\ref{sec:background}, this is not the main focus of our paper, but there has 
been recent theoretical and empirical interest in this and related 
questions; see, 
e.g.,~\cite{BKM01, WZ11, Dourisboure05, lokshtanov, GM09, KLNS12, Dra13, Ata16, abuata_diss}.}
It might appear that there is no relation between TDs and $\delta$ (since, 
e.g., treewidth and $\delta$ take on opposite extremal values on cliques and
cycles), but there are in fact structural characterizations for when they 
align. 
We will present here our new theoretical results on relating TDs and 
$\delta$-hyperbolicity.  
Although this result is a relatively-straightforward extension of previous 
work~\cite{Muller12_TR}, and although most of the rest of the paper can 
be understood without this result, we include it here for completeness: 
first, since motivating prior work in~\cite{Adcock13_icdm} demonstrates 
an empirical connection between the cut-based tree-like notion from TDs 
and the metric-based tree-like notion from $\delta$-hyperbolicity; and
second, since our results in Section~\ref{sec:real_results} demonstrate
the inadequacy of a na\"{\i}ve optimization of treewidth and the importance 
of large cycles for realistic social graphs.

\subsection{Treewidth, Treelength, and Hyperbolicity}
\label{sec:theory_new-mainresults}

We start with the following definition, which provides another quality 
measure of a TD; this was first introduced by Dourisboure and 
Gavoille~\cite{Dourisboure05}. 
See also~\cite{DL07}.

\begin{definition}\label{def:tl}
Let $\mathcal{T} = \left (\{X_i\}, T=(I,F) \right )$  be a tree 
decomposition of a graph $G$. 
The \emph{length} of $\mathcal{T}$ is defined to be 
$\max_{i \in I, x,y \in X_i} d_G(x,y)$, where $d_G(x,y)$ is the shortest 
path distance in $G$. 
Analogously to treewidth, the \emph{treelength} of $G$, denoted $tl(G)$, is 
the minimum length achieved by any tree decomposition of $G$.
\end{definition}

\noindent
It is straight-forward to see that the treelength is at most the diameter of 
$G$. 
Like with treewidth, finding a tree decomposition achieving minimum length 
(and in fact the treelength itself) is NP-hard~\cite{lokshtanov}. 
Given this, one might ask whether treelength and treewidth can be 
simultaneously approximated. 
For general graphs, Dourisboure and Gavoille proved a negative result.

\begin{theorem}\label{thm:dg}~\cite{Dourisboure05} 
Any algorithm computing a tree decomposition approximating the treewidth (or 
the treelength) of an $n$-vertex graph by a factor $\alpha$ or less does not 
give an $\alpha$-approximation of the treelength (resp. the treewidth) unless 
$\alpha = \Omega(n^{1/5})$. 
\end{theorem}

The specific examples used by~\cite{Dourisboure05} to prove their 
negative result are modifications of the $2$-dimensional mesh (i.e., a
lattice), which---due to long induced cycles---is not $\delta$-hyperbolic 
for small values of $\delta$. 
This suggests that the situation might be very different for ``real-world'' 
graphs---which have small diameter and which have non-trivial embedding 
properties into low-dimensional hyperbolic spaces. 
(This is an open area of research more generally.) 
Chepoi et al.  \cite{CDEHV08} showed that if 
$tl(G) \leq \lambda$, then $G$ is $\lambda$-hyperbolic, and that a 
$\delta$-hyperbolic graph $G$ on $n$ vertices satisfies 
$tl(G) \leq 17 + 12\delta + 8\delta \log_2 n$.  
Unfortunately, for many real networks of interest, this is \emph{not} an 
improvement on the trivial bound of diameter as their diameter alone will 
be less than $O(\log_2 n)$. 
We conjecture that under minimal additional conditions, a 
$\delta$-hyperbolic graph with diameter $D$ has treelength at most a 
function of $\log_2 D$, a vast improvement on both known~bounds. 

We turn to the question of using additional 
structural properties to characterize the interplay between $\delta$, 
$tw(G)$, and $tl(G)$. 
The following theorem is our main result; this theorem follows from the 
work of M\"{u}ller on atomic TDs~\cite{Muller12_TR}, and
its proof is in 
Section~\ref{sec:theory_new-mainproof}.

\begin{theorem}~\cite{ReidlSullivan14}
\label{thm:holes}
Say a subgraph $H$ of $G$ is geodesic if $d_H(u,v) = d_G(u,v)$ for all $u,v \in V(H)$. Let
$\nu(G)$ be the length of a longest geodesic cycle in $G$. 
Then 
$$ \delta(G) \leq tl(G) \leq (tw(G)+1)\cdot \nu(G) .$$ 
Further, this result is tight---there is a graph class $\mc G$ of unbounded treewidth and containing arbitrarily long geodesic cycles such that $\delta(G) = \Theta(tw(G)\cdot \nu(G))$ for every graph $G \in \mc G$. 
\end{theorem}

\noindent
In other words, if we can eliminate long distance-preserving cycles and 
obstructions to low treewidth (large grid minors), then $G$ will 
embed well in low-dimensional hyperbolic~space. 

\subsection{Proof of Theorem~\ref{thm:holes} }
\label{sec:theory_new-mainproof}

Before we can give the proof of Theorem \ref{thm:holes}, we need a few additional definitions. First, given a rooted tree~$T$ and a node $s \in T$, define $T_s$ to be the subtree of $T$ with root $s$: 
$$T_s := T[\{ t \in T \mid s~\text{is an ancestor of}~t \} ].$$ 
For a graph $G = (V,E)$ with tree decomposition $(\{X_i\},T)$ where $T$ is rooted arbitrarily, for $s \in T$ define 
$G_s := G[\bigcup_{t \in T_s} X_t]$ to be the graph induced by those bags that
are equal to or below $X_s$ in the decomposition. We will write $N(S)$ for the neighbors of a set $S$ -- more precisely, $N(S) = \{ u \in V \;|\; (u,s) \in E \text{ for some } s \in S\} \setminus S$. Finally, for notational convenience, for $x \in V$ and $e \in E$,  we will write $G - x$ for the graph $(V\setminus\{x\}, E)$ and $G - e$ for the graph $(V, E\setminus\{e\})$.  
We now define a special type of tree decomposition (so-called {\em atomic tree decompositions}), 
and give a crucial property of all vertices that co-occur in one of its bags. 

\begin{definition}\label{atomicTD} [atomic tree decomposition, as in~\cite{Bellenbaum_12}] 
	Let $G$ be a graph
	on $n$ vertices. The \textit{fatness} of a tree decomposition of $G$ is the $n$-tuple
	$(a_0,  \ldots, a_n)$, where $a_h$ denotes the number of bags that have exactly $n-h$
	vertices. A tree decomposition of lexicographically minimal fatness is called	
	an \textit{atomic tree decomposition}.
\end{definition}

\begin{proposition}\label{prop:connected-or-else} {\em [Lemma 3.9 in M\"uller~\cite{Muller12_TR}]}
	Let $(\{X_i\},T)$ be an atomic tree decomposition of a connected graph~$G=(V,E)$.
	Then for any two distinct vertices $x,y$ that occur together in some
	bag $X_t$, either $(x,y) \in E$ or there exists a neighbor $s$ of $t$ in $T$
	such that $\{x,y\} \subseteq V_s \cap V_t$.
\end{proposition}

We also need the following proposition, which follows from Lemmas 3.7 and 3.8 in M\"uller~\cite{Muller12_TR}.  
\begin{proposition}\label{prop:good-component}
	Let $(\{X_i\},T)$ be an atomic tree decomposition of a connected graph~$G$, $e = (s,t) \in E(T)$ be any edge and let $T_t$ be the connected component
	of $T - e$ rooted at~$t$, and set $X = X_s \cap X_t$. Then there exists a connected component $C_t$ in $G_t \setminus X$ such that $N(C_t) = X$ and $X_t \subseteq C_t \cup X$.
\end{proposition}

Finally, we are ready to give a bound on treelength in terms of a graph's treewidth and its longest geodesic cycle. 
Our proof relies heavily on tools from~\cite{Muller12_TR}.

\begin{theorem}\cite{ReidlSullivan14}\label{thm:tl_nu}
	For any graph~$G=(V,E)$ it holds that $\tl(G) \leq \nu(G) \cdot (\tw(G)+1)$
	where~$\nu(G)$ is the length of the longest geodesic cycle in~$G$.
\end{theorem}
\begin{proof}
We will prove a stronger statement, namely that {\em any atomic tree decomposition}
of a two-connected graph has treelength at most~$\nu(G) \cdot(\tw(G)+1)$.
Let us first show how this proves the lemma for graphs that are not two-connected.

Assume~$G$ is not two-connected and $x \in V$ is a cut vertex ($G - x$ has at least two connected components). Let $H_1,\dots,H_\ell$ be
the connected components of $G - x$. If we prove that the graphs $G[H_i \cup \{x\}]$,
$1 \leq i \leq \ell$ have tree decompositions $\mc T_i$ with treelength bounded as in
the statement of the theorem, then we can easily construct a tree decomposition
for $G$ with the same property: we simply introduce a single new bag 
$V_x = \{x\}$ and connect it to an arbitrary bag containing $x$ in each of the individual tree
decompositions $\mc T_i$ (since these graphs all contain
the vertex $x$ such a bag must exist). Note that the treelength of this decomposition
is simply $\max_{1 \leq i \leq \ell} \tl(\mc T_i)$ since the bag $V_x$ we added contains
only the vertex $x$ and thus cannot increase the treelength.
Since we will show the statement for two-connected graphs in the following, we recursively decompose the graph $G$
over cut vertices until the remaining connected components are all two-connected and
then construct a tree decomposition of $G$ as described above.

We may now assume $G$ is two-connected. Given an atomic tree decomposition $(\{X_i\},T)$ of~$G$, we show that for every
two vertices $x,y$ that occur in a common bag $X := X_t$,  $x$ and $y$ are
connected by a path whose length depends only on $|X|$ and $\nu(G)$. To this end, let $\mc C_X$ be the collection of geodesic cycles in~$G$ that
have at least one vertex in $X$. We first show that if $G[\mc C_X]$ is connected
and $X \subseteq V(\mc C_X)$, then every pair of vertices in $X$ is connected 
by a path of length at most $|X|\cdot \nu(G[\mc C_X])$.

Consider $x,y \in X$. Start a breadth-first search (bfs) from~$x$ that stops as soon as it reaches~$y$. 
Let $L_1, L_2, \cdots L_p$ be
the layers of the bfs-tree where $L_1 = \{x\}$ is the starting layer.
We claim that for all $L_i$ with $L_i \cap X \neq \emptyset$, there is a $j$ such that 
$i < j \leq i + \nu(G[\mc C_X])$ and  $L_j \cap X \neq \emptyset$. Consider such an $L_i$, and denote by $X_l \subseteq X$
those vertices of $X$ that are contained in $\bigcup_{k = 1}^{i} L_k$. Denote 
by $X_r = X \setminus X_l$ those vertices of $X$ that have not been visited until
step~$i$. If there exists a geodesic cycle $C$ in $\mc C_X$ with vertices in both
$X_l$ and $X_r$ we are done -- the bfs will have seen all of $C$ in at most $\nu(G[\mc C_X])$ steps  (and thus found $C \cap X_r$). 
Otherwise, since $\mc C_X$ is connected, there
exist two geodesic cycles $C_l, C_r \in \mc C_X$ with $C_l \cap X_l \neq \emptyset$,
$C_r \cap X_r \neq \emptyset$ and $C_r \cap C_l \neq \emptyset$. Since the bfs
will visit all vertices of $C_r \cup C_l$ in at most $(|C_r| + |C_l|)/ 2 \leq \nu(G[\mc C_X])$
steps, the claim follows. Therefore the number of layers $p \leq \nu(G[\mc C_X]) \cdot |X_t|$ and thus the distance
between $x$ and $y$ is bounded by $\nu(G[\mc C_X]) \cdot |X_t|$ as claimed.

Therefore, if we show that for every bag $X$, the set $\mc C_X$ of geodesic cycles
touching $X$ induces a connected graph $G[V(\mc C_X)]$, we are done: then every vertex
pair $x,y \in X$ is indeed connected by a path of length at most 
$\nu(G[V( \mc C_X)]) |X|$, which (by the definition of treewidth and the fact $C_X$ is a family of 
geodesic cycles) is bounded by $\nu(G) (\tw(G)+1)$.

We first prove that for any choice of $X := X_t$ and any pair of vertices $x,y \in X$, $x$ and $y$ lie 
on some cycle of $G$. By Proposition~\ref{prop:connected-or-else},
the vertices $x,y$ are either connected by an edge (in which case we are done: $G$ is
two-connected, so every edge lies on some cycle) or
there exists some node $s \in N_T(t)$ such that $\{x,y\} \subseteq V_s \cap V_t$.
In the latter case, we invoke Proposition~\ref{prop:good-component}: for $i \in \{s,t\}$ we can find 
connected components $H_i$ of $G_i \setminus X$ such that
$N(H_i) = X$ and $V_i \subseteq H_i \cup X$. Therefore, there exist two
$x$-$y$-paths: one inside $H_s$ and another in $H_t$, hence $x$ and $y$ lie on a 
cycle. 

Since the set of geodesic cycles forms a basis for the cycle space of a graph (see Theorem 3.1 of~\cite{GS09}),
it follows that for every $t \in T$, $G[V(\mc C_{X_t})]$ is connected. The
distance between any vertices in $X_t$ is thus bounded by $\nu(\mc C_{X_t}) \cdot |X_t|$, implying that 
$tl(G)$ is at most $\nu(G) \cdot (\tw(G)+1)$, as claimed.
\end{proof}

Finally, we put all the pieces together and show why these bounds are tight. \\

\noindent \textbf{Proof of Theorem~\ref{thm:holes}}
This follows directly from Theorem~\ref{thm:tl_nu},  Chepoi's result that hyperbolicity is at most the treelength~\cite{CDEHV08}, and the observation that for any  non-negative integers $n$ and $k$, the $k$-subdivision of the $n \times n$ planar grid has treelength $n(k+1)$, treewidth $n$, a longest geodesic cycle of length $4(k+1)$, and hyperbolicity $(n-1)(k+1)-1$. \qed

\section{Discussion and Conclusion}
\label{sec:conc}


Clearly, there is a need to develop TD heuristics that are better-suited for 
the properties of realistic informatics graphs.  This might involve making more sophisticated choices 
than greedily minimizing degree or fill, but it might also involve optimizing other
parameters such as treelength (which has connections with $\delta$-hyperbolicity) 
or minimizing the width of bags that are not central (associated with the deep core).
In addition, it would be interesting to use TDs to help to combine small 
local clusters found with other methods, e.g., local spectral methods, 
into larger overlapping clusters, in order to understand better what might 
be termed the ``local to global'' properties of realistic informatics 
graphs. Since these graphs are not well-described by simple low-dimensional 
structures or simple constant-degree expander-like structures, this coupling 
is particularly counterintuitive, but it is very important for applications 
such as the diffusion of information. 
Finally, given the connections between TDs and graphical models, it would be 
interesting to understand better the implications of our results for 
improved graphical modeling and/or for improved inference on realistic 
network data.
We expect that this will be a particularly challenging but promising 
direction for future work on social (as well as non-social)~graphs.

\vspace{5mm}
\textbf{Acknowledgments.}
We would like to thank Felix Reidl for considerable help in simplifying the
proof of Theorem~\ref{thm:holes}.
We would also like to thank Mason Porter for helpful discussions and for 
providing several of the networks that we considered as well as Dima Krioukov and 
his collaborators for providing us access to their code for generating 
networks based on their hyperbolic model.
In addition, we would like to acknowledge financial support from the Air 
Force Office of Scientific Research, the Army Research Office, the Defense 
Advanced Research Projects Agency, the National Consortium for Data Science, 
and the National Science Foundation. 
Any opinions, findings, and conclusions or recommendations expressed in this publication are
those of the author(s) and do not necessarily reflect the views of any of the above funding agencies. 

\vspace{5mm}

\begin{small}

\bibliographystyle{unsrt}


\end{small}

\end{document}